\documentclass[aps,prx,twocolumn,superscriptaddress,floatfix]{revtex4-2}
\usepackage{microtype}

\usepackage{graphicx}
\usepackage{amsmath,amssymb,amsthm,mathrsfs,amsfonts,amscd,keyval}
\usepackage{mathtools}
\usepackage{nicematrix}
\IfFileExists{dsfont.sty}{\usepackage{dsfont}}{}
\IfFileExists{yfonts.sty}{\usepackage{yfonts}}{}
\usepackage{subfigure,epsfig}
\usepackage{braket}
\usepackage{bm}
\usepackage{enumerate}
\usepackage{enumitem}
\usepackage{multirow}
\usepackage{booktabs}
\usepackage{array,tabularx}
\usepackage[dvipsnames]{xcolor}

\usepackage{tikz}
\usetikzlibrary{calc, 3d}%
\usepackage{tikz-cd}
\usepackage{float}
\IfFileExists{quiver.sty}{\usepackage{quiver}}{}
\usepackage[linesnumbered,ruled,vlined]{algorithm2e}
\usepackage{hyperref}

\definecolor{DarkLinkBlue}{RGB}{0,70,140}
\hypersetup{
    colorlinks=true,
    linkcolor=DarkLinkBlue,
    citecolor=DarkLinkBlue,
    urlcolor=DarkLinkBlue
}

\newcommand{\mc}{\mathcal}

\newcommand{\mbb}{\mathbb}
\newcommand{\mr}{\mathrm}

\newcommand{\LP}{\mathrm{LP}}

\DeclareMathOperator{\coker}{coker}
\DeclareMathOperator{\supp}{Supp}

\DeclareMathOperator{\rs}{rs}
\DeclareMathOperator{\IM}{im}

\renewcommand\vec{\mathbf}

\newcommand{\tr}{\operatorname{tr}}
\renewcommand\vec{\mathbf}

\newcommand{\Lket}[1]{\left\lvert #1 \right\rangle}

\newcommand{\Bmat}{\mathbb{B}}

\newtheorem{m-theorem}{Theorem}

\theoremstyle{definition}
\newtheorem{definition}{Definition}[section]
\newtheorem{example}[definition]{Example}

\newtheorem{counterexample}[definition]{Counter-Example}

\newtheorem{proposition}[definition]{Proposition}

\theoremstyle{plain}
\newtheorem{theorem}{Theorem}

\newtheorem{lemma}[definition]{Lemma}
\newtheorem{corollary}{Corollary}

\theoremstyle{remark}
\newtheorem*{remark}{Remark}

\providecommand{\ket}[1]{$\left|#1\right\rangle$} %

\SetKwInput{KwIn}{Input}
\SetKwInput{KwOut}{Output}
\SetKwInput{KwResources}{Resources}
\SetKwInput{KwPre}{Preconditions}
\SetKwInput{KwPost}{Postconditions}
\SetKwInput{KwGoal}{Goal}

\SetAlgoNlRelativeSize{-1}      %
\SetAlFnt{\small}               %
\SetKwComment{tcc}{\% }{}       %
\DontPrintSemicolon              %

\SetKw{Prepare}{prepare}
\SetKw{Measure}{measure}
\SetKw{Apply}{apply}
\SetKw{Record}{record}
\SetKw{Reset}{reset}
\SetKw{Wait}{wait}
\SetKw{Barrier}{barrier}
\SetKw{And}{and}
\SetKw{Or}{or}

\SetKwBlock{Parallel}{\textbf{parallel}}{}     %
\SetKwRepeat{Repeat}{repeat}{until}            %

\begin{document}

\title{Logical computation with canonical lifted product codes}

\author{Han Zheng}
\email{hanzuchicago@gmail.com}
\affiliation{Department of Computer Science, The University of Chicago, Chicago, IL 60637, USA}
\affiliation{Pritzker School of Molecular Engineering, The University of Chicago, Chicago, IL 60637, USA}

\author{Guo Zheng}
\email{guozheng@uchicago.com}
\affiliation{Pritzker School of Molecular Engineering, The University of Chicago, Chicago, IL 60637, USA}

\author{Liang Jiang}
\email{liangjiang@uchicago.edu}
\affiliation{Pritzker School of Molecular Engineering, The University of Chicago, Chicago, IL 60637, USA}

\author{Qian Xu}
\email{qxu@oratomic.com}
\affiliation{Institute for Quantum Information and Matter, Caltech, Pasadena, CA 91125, USA}
\affiliation{Walter Burke Institute for Theoretical Physics, Caltech, Pasadena, CA 91125, USA}
\affiliation{Oratomic, Pasadena, CA 91125, USA}

\date{\today}

\begin{abstract}
High-rate quantum low-density parity-check (qLDPC) codes encode many logical qubits with low physical-qubit overhead, but realizing efficient fault-tolerant computation on such dense encodings remains a major challenge. Generic, code-agnostic techniques such as code surgery and gate teleportation apply broadly, but are difficult to make modular, low-overhead, and fully certifiable on complex high-rate codes whose structure is left unexploited. Here we overcome these obstacles by co-designing the code together with its logical instruction set for a broad family of \emph{canonical} lifted-product (LP) codes with cyclic symmetry.
We show that these codes admit a \emph{canonical logical basis}, in which conjugate logical operators are organized into rows and columns of cyclic orbits inherited directly from the underlying classical codes,  analogous to the structure that makes hypergraph-product codes so tractable. This canonical basis enables universal logical processing using magic-state input and a constant number of reusable seed surgery gadgets—only two (resp. four) for a $[[1122,148,\leq 20]]$ (resp. $[[4350,1224,\leq 20]]$) LP code. 
A fixed extractor smaller than half the data block also enables arbitrarily high-weight logical Pauli-product measurements (PPMs) while preserving the code’s cyclic symmetry.
These codes also enable rigorously fault-tolerant parallel logic, including PPMs across all logical qubits and magic-state injection from many surface codes into a full LP block. 
For example, a $[[468,36,20]]$ (resp. $[[952,112,17]]$) code supports simultaneous measurements of arbitrary logically disjoint two-body Pauli products with an ancilla overhead of roughly three (resp. four) times the data block size.
These results advance the frontier of fault-tolerant quantum computation on ultra-high-rate quantum architectures.
\end{abstract}

\maketitle

\section{Introduction}\label{sec: main}

Quantum error correction (QEC)~\cite{gottesman1997stabilizercodesquantumerror, Shor1995ReducingDecoherence, Kitaev1997QuantumComputations, BravyiKitaev1998LatticeBoundary, KnillLaflamme1997TheoryQECC, Calderbank1998GF4, Steane1996ErrorCorrectingCodes} is the foundation of large-scale fault-tolerant quantum computation. Recently, high-rate encodings such as quantum low-density parity-check (qLDPC) codes~\cite{gottesman2013fault, breuckmann2021quantum, panteleev2022asymptoticallygoodquantumlocally, leverrier2022quantumtannercodes, Panteleev_2021} have emerged as a compelling alternative to local topological codes, offering much higher encoding rates and substantially lower physical-qubit overhead by leveraging non-local connectivity~\cite{BravyiCrossGambettaMaslovRallYoder2024HighThreshold, xu2024constant}. As quantum memories, qLDPC codes are now increasingly well understood and coming within experimental reach~\cite{BravyiCrossGambettaMaslovRallYoder2024HighThreshold, xu2024constant, yoder2025tour, tham2026breakeven, wang2026demonstration}. The central challenge is therefore shifting from storing logical information to computing with it: how can one build{, from high-rate encodings,} an efficient fault-tolerant quantum processor that executes complex logic, rather than merely preserving quantum states?

Code-agnostic methods for fault-tolerant logic on general qLDPC codes have recently been developed, such as code surgery~\cite{CohenKimBartlettBrown2022LongRangeConnectivity, cross2025improvedqldpcsurgerylogical, williamson2024lowoverheadfaulttolerantquantumcomputation, he2025extractorsqldpcarchitecturesefficient, Ide_2025, CowtanHeWilliamsonYoder2025ParallelCodeSurgery}, gate teleportation~\cite{gottesman2013fault, nguyen2024quantum}, and batched logical operations~\cite{Xu2025BatchedHighRate}. 
Their generality is appealing, but when applied to complex high-rate codes without exploiting code-specific structure, these methods can still incur substantial overhead and implementation complexity. In code surgery, for example, one introduces an ancilla system whose local measurements realize a target logical operator with space overhead scaling only near-linearly (up to polylogarithmic factors) in the Hamming weight of the operator~\cite{williamson2024lowoverheadfaulttolerantquantumcomputation, cross2025improvedqldpcsurgerylogical, sayginel2025faulttolerantlogicalcliffordgates}. 
Yet if applied naively to large codes with arbitrary high-weight logical operators, such constructions may not yield low-overhead, modular, or rigorously certifiable gadgets. 
This suggests a more promising route: co-designing the code and the logical toolkit so that the code’s structure and symmetries can be directly exploited, yielding modular gadgets, fault tolerance by construction, and a logical gate set that can be certified efficiently.

Hypergraph-product (HGP) codes~\cite{TillichZemor2009HypergraphProduct, bravyi2013homologicalproductcodes} are a representative example of this co-design philosophy. As tensor products of two classical codes, they admit a \emph{canonical} logical basis inherited from classical codewords, realized as conjugate logical operators supported on rows and columns of a 2D grid and intersecting on a distinguished subgrid that can be viewed, conceptually, as representing the encoded logical qubits~\cite{Quintavalle_2023, xu2025fast}. This explicit structure has enabled a rich logical toolkit, including automorphism gates~\cite{berthusen2025automorphismgadgetshomologicalproduct, malcolm2026computing, xu2025fast}, homomorphic CNOTs and Steane gadgets~\cite{xu2025fast}, fold-transversal Clifford gates~\cite{Quintavalle_2023}, fast code switching~\cite{hong2024single, tan2025singleshotuniversalityquantumldpc},  fast code surgery~\cite{chang2026constant}, and efficient extractors~\cite{blue2026full}. However, HGP codes have limited code parameters, particularly in their distance, restricting the practical payoff of this toolkit. This motivates extending these techniques to more modern lifted- and balanced-product codes~\cite{Breuckmann2021BalancedProductQuantumCodes, Panteleev_2021, Panteleev_2022}, which can achieve much better parameters, including linear distances~\cite{Breuckmann2021BalancedProductQuantumCodes, panteleev2022asymptoticallygoodquantumlocally}. 

A balanced-product code may be viewed as a hypergraph product followed by a quotient by a shared group symmetry $G$, but that quotient obscures the tensor-product structure of logical operators. It is therefore unclear whether these codes still admit a canonical logical basis inherited from the classical codes, and hence the structured gate set that such a basis enables~\cite{eberhardt2024logicaloperatorsfoldtransversalgates}.

In this work, we resolve this question for a broad and practically important class of quantum codes. We identify a family of \emph{canonical} lifted-product (LP) codes over the cyclic group ring $R = \mbb{F}_2[x]/(x^l + 1)$, and show that, provided the underlying protograph $A$ contains more than one element and the lift satisfies mild algebraic conditions, they admit a \emph{canonical logical basis} inherited directly from the classical codewords of the base matrices.
Specifically, each logical qubit is associated with a conjugate pair of logical operators derived directly from codewords of the (lifted) classical base code $A$, with the two operators overlapping on exactly one physical qubit. The logical qubits are further organized into logical fibres -- orbits of $l$ qubits whose logical operators are related by the cyclic group action -- and the resulting basis exhibits a row- and column-parallel structure, so that logical operators have an identical form across the rows and columns of the physical fibre grid (see Fig.~\ref{fig:main-overview}(a)).
In effect, this canonical logical basis recovers the HGP structure on a code family with far better parameters.

The symmetry of the canonical basis unlocks a complete native logical instruction set (Fig.~\ref{fig:main-overview}, Table~\ref{tab: main-instruction-set}). \emph{Automorphism} gates cyclically shift every logical fibre in parallel by permuting the physical qubits, while \emph{fold-transversal} Clifford gates, arising from the $ZX$-duality of symmetric LP codes, implement global Hadamard, \(S\), and \(CZ\) gates, recovering the logical instruction set of symmetric HGP codes~\cite{Quintavalle_2023}. Beyond these global operations, we obtain fully addressable logic through graph-surgery gadgets tailored to the canonical basis. The key co-design feature is that many logical operators are equivalent up to physical permutations (Fig.~\ref{fig: canonical-basis}), allowing any low-logical-weight Pauli-product measurement (PPM) to be assembled by bridging a small set of reusable ``seed'' gadgets (Fig.~\ref{fig: seed-surgery}(a)). The number of seed surgery gadgets is only \(r_A\), independent of the lift size \(l\), so arbitrary low-weight logical measurements reduce to constructing and certifying just a handful of small gadgets. For example, the large \(\LP^{3\times7}_{75}=[[4350,1224,\leq\!20]]\) code~\cite{cain2026shorsalgorithmpossible10000} requires only four seed surgery gadgets (Table~\ref{tab: seed-surgery-gadget-resources}). We further construct a fixed \emph{canonical extractor} for arbitrary high-weight logical Pauli products by exploiting the same row/column structure and $ZX$-duality (Fig.~\ref{fig: seed-surgery}(b)). Remarkably, these extractors are themselves smaller than the data block -- for example, the extractors for {\(\LP^{3\times5}_{28}=[[952,112, 17]]\), \(\LP^{3\times5}_{33}=[[1122,148,\leq\!20]]\), and \(\LP^{3\times7}_{75}=[[4350,1224,\leq\!20]]\)} are below half the size of their codes (Table~\ref{tab: LP-canonical-extractor-resources} and Table~\ref{tab: LP-canonical-extractor-numerics}) -- while retaining the same cyclic symmetry as the data code, making them far more structured than generic extractor gadgets~\cite{he2025extractorsqldpcarchitecturesefficient}. {Retaining the cyclic symmetry also makes the extractors easier to implement in hardware (Section~\ref{sec: LP-hardware}).}

Finally, we construct highly parallel logical primitives for LP codes, spanning Pauli-product measurements (PPMs) and magic-state injection, enabling high-throughput logical computation. We first generalize the homomorphic-measurement technique for HGP codes~\cite{xu2025fast} to inter-column PPMs, which act across columns of the canonical logical basis (Fig.~\ref{fig: parallel-surgery}(b)). The required ancilla size does not scale with the code distance $d$ and is much smaller than the data block.

We then extend parallel measurements to arbitrary logically disjoint operators using \emph{soundness-enhanced hypergraph surgery} (Fig.~\ref{fig:main-overview}(d) and Fig.~\ref{fig: parallel-surgery}(a)). Generic constructions require thickening an initial hypergraph ``skeleton'' to $d$ layers, incurring a large ancilla overhead. We prove that large soundness of the data code can substantially reduce the required thickening and hence the ancilla overhead. Furthermore, mixed $XZ$-type and $Y$-type PPMs are enabled by \emph{generalized bridging/augmentation} techniques (Proposition~\ref{prop: bridge_get_full_Clifford-informal}, Theorem~\ref{thm: main-LP-extractor}, and Theorem~\ref{thm: main-bridged-hypergraph-surgery-informal}).
Building on these results, we codesign an LP subfamily with optimal soundness and ancilla overhead. For example, the codes $\LP^{2\times4}_{16}=[[320,64,10]]$, $\LP^{2\times3}_{36}=[[468,36,20]]$, $\LP^{3\times5}_{20}=[[680,80,13]]$, and $\LP^{3\times5}_{28}=[[952,112,17]]$ (Example~\ref{example: main-soundness-codesign}) support parallel PPMs on \emph{all} logical qubits with ancilla sizes only $1$--$3$ times the data block size for same-type $Z$-PPMs and $3$--$4$ times for \emph{arbitrary} two-body PPMs (Tables~\ref{tab: high-rate-two-body-addressable} and~\ref{tab: high-rate-arbitrary-resources}).

Beyond PPMs, we introduce a protocol for parallel magic-state injection, transferring magic states from small-distance surface codes into an entire LP block simultaneously. For example, we give an explicit protocol that injects $132$ $\ket{T}$ states into the $\LP^{3\times5}_{33}=[[1122,148,\leq\!20]]$ code from distance-$7$ surface codes in two batches. The protocol primarily uses a helper (transistor) code comparable in size to the data block while maintaining provable distance-$\geq7$ fault tolerance throughout. Combined with a distillation factory based solely on transversal CNOTs~\cite{xu2025fast,cain2026shorsalgorithmpossible10000}, this protocol supports parallel magic-state generation with low amortized space--time overhead. Together, parallel PPMs and magic-state injection enable parallel execution of general quantum circuits.
```

The remainder of this paper is organized as follows. Section~\ref{sec:summary} summarizes our main results. Section~\ref{sec: main-LP-basis} constructs the canonical logical basis and discusses its symmetries. Sections~\ref{sec: main-graph-surgery} and~\ref{sec: main-LP-parallel-logic} then develop the tailored logical instruction set for the canonical LP codes. Self-contained background and full technical proofs are deferred to the appendices. {Appendix~\ref{app-sec: lp-basis-characterization} derives the logical bases of LP codes. Appendix~\ref{sec: code-surgery-general} develops a code-agnostic framework for universal, addressable, fault-tolerant surgery. Appendix~\ref{sec: canonical-LP-surgery-construction} applies it to canonical LP codes, supporting Sections~\ref{sec: main-graph-surgery} and~\ref{sec: main-LP-parallel-logic}.}

\section{Summary of main results \label{sec:summary}}
\begin{figure*}[t]
    \centering
    \includegraphics[width=\textwidth]{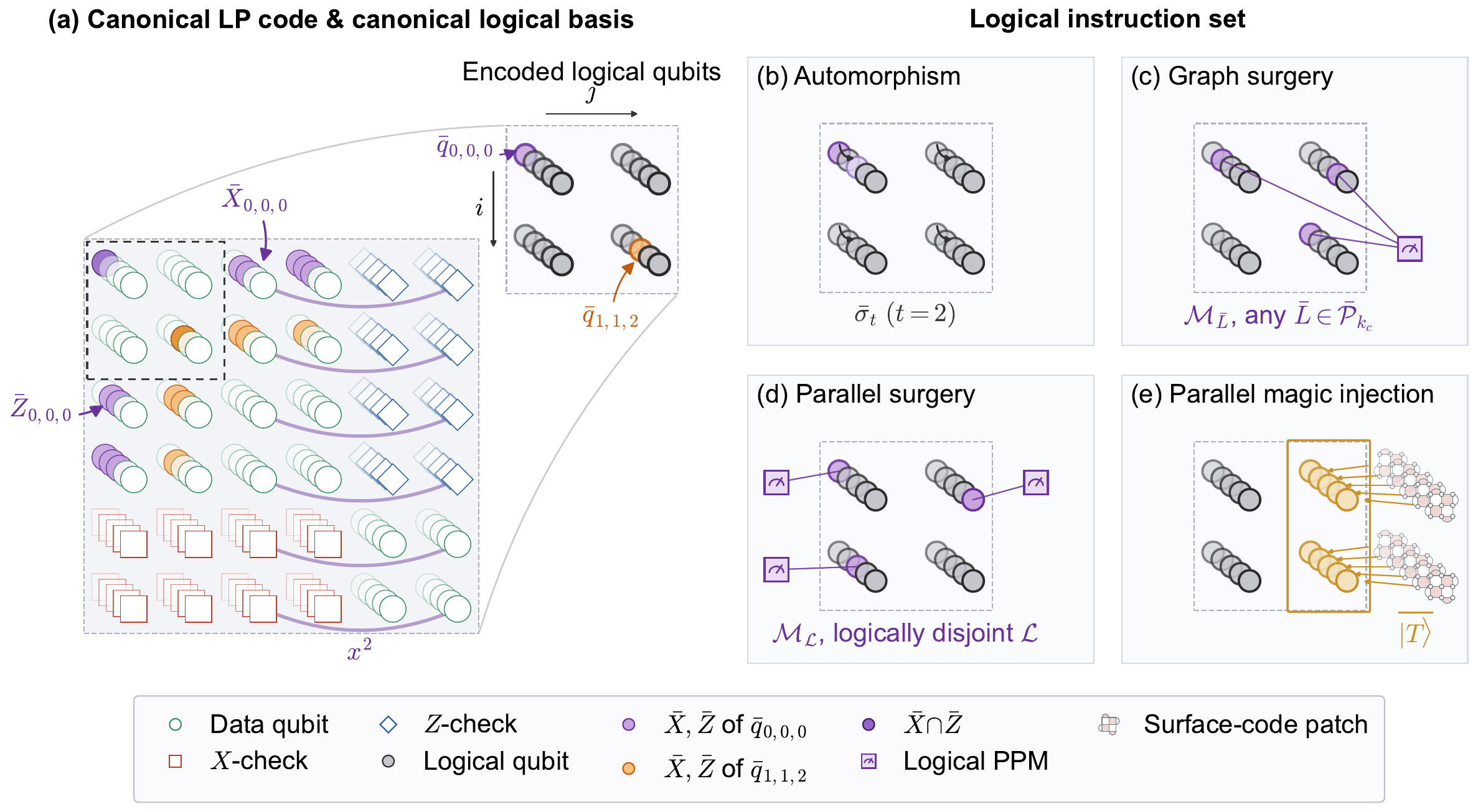}
    \caption{\textbf{Canonical LP codes, their logical basis, and logical instruction set. }
    \textbf{(a)} Illustration of a canonical LP code $\LP_l(A,A^{*})$, drawn for the toy code $\LP^{2\times4}_5=[[100,26,4]]$ with $r_A = 2$: data qubits (green circles) and $X$-/$Z$-checks (red squares/blue diamonds) each form a stack of $l$ copies of the lift, which we refer to as a \emph{physical fibre}. Each purple arc is an $R$-linear edge representing one element of the check matrices; only a subset of the horizontal edges is drawn, all carrying the same element $A_{2,3} = x^{2}$ of Eq.~\eqref{eq:toy_code}, which illustrates the product structure of the code. Gray circles with thick outlines (inset) are the encoded logical qubits, forming an $r_A \times r_A$ grid of \emph{logical fibres} of $l$ logical qubits each: the logical qubit $\bar{q}_{i,j,m}$ is defined by a conjugate pair of canonical logical operators $\bar{X}_{i,j,m}$ and $\bar{Z}_{i,j,m}$, supported on the $i$-th row and the $j$-th column of physical fibres, respectively. Two such pairs are highlighted in one color per logical qubit ($\bar q_{0,0,0}$ purple, $\bar q_{1,1,2}$ orange); each pair overlaps on a single darker qubit inside the dark dashed box that marks the $r_A \times r_A$ information sub-grid of physical fibres on which all conjugate pairs intersect.
    \textbf{(b)--(e) Logical instruction set} (Table~\ref{tab: main-instruction-set}): one primitive per panel acting on the logical qubits of (a), with each purple box denoting a logical Pauli-product measurement (PPM) supported on the logical qubits it is wired to --- \textbf{(b)} the cyclic-shift automorphism $\bar\sigma_t$, drawn for $t=2$; \textbf{(c)} graph surgery of an arbitrary logical Pauli product $\mathcal{M}_{\bar L}$, $\bar{L}\in\bar{\mathcal{P}}_{k_c}$; \textbf{(d)} parallel intra-column surgery, here three logically disjoint $Y$-type PPMs $\{\mathcal{M}_{\bar{L}^{(i)}}\}_{i=1}^{3}$ on the column $j=0$ and the row $i=0$; and \textbf{(e)} parallel magic-state injection into the boxed column $j=1$ (gold) from stacks of surface-code patches.}
    \label{fig:main-overview}
\end{figure*}

\begin{table*}[!t]
\centering
\renewcommand{\arraystretch}{1.25}
\setlength{\tabcolsep}{4pt}
\newcommand{\isI}[1]{\parbox[t]{0.155\textwidth}{\raggedright #1\strut}}
\newcommand{\isA}[1]{\parbox[t]{0.295\textwidth}{\raggedright #1\strut}}
\newcommand{\isS}[1]{\parbox[t]{0.325\textwidth}{\raggedright #1\strut}}
\begin{tabular}{l l l c}
\toprule
\textbf{Instruction} & \textbf{Logical action} & \textbf{Space (ancilla) overhead} & \textbf{Time (QEC cycles)} \\
\midrule
\isI{Cyclic-shift automorphism $\overline{\sigma}_t$} & \isA{cyclic shifts $\bar{q}_{i,j,m}\!\to\!\bar{q}_{i,j,(m+t)\bmod l}$ on all $r_A^2$ logical fibres in parallel} & \isS{none (physical-qubit permutation)} & $1$  \\
\addlinespace[3pt]
\isI{Fold-transversal $\overline{\mathrm{H\text{-}SWAP}}$} & \isA{transversal $H$ $+$ fold permutation: $\bar{H}^{\otimes k_c}$ followed by $\bar{q}_{i, j, m} \leftrightarrow \bar{q}_{j, i, m}$} & \isS{none (depth-$1$ physical circuit)} & $1$ \\
\addlinespace[3pt]
\isI{Fold-transversal $\overline{\mathrm{S\text{-}CZ}}$} & \isA{diagonal $S$ $+$ folded $CZ$: $\bigotimes_{i, m}\bar{S}_{\bar{q}_{i, i, m}} \bigotimes_{i < j, m} \overline{CZ}_{\bar{q}_{i, j, m}, \bar{q}_{j, i, m}}$} & \isS{none (depth-$1$ physical circuit)} & $1$ \\
\addlinespace[3pt]
\isI{Graph surgery by bridging $r_A$ seed surgery gadgets} & \isA{An arbitrary logical PPM: $\mathcal{M}_{\bar{L}}$, where $\bar{L} \in \bar{\mc{P}}_{k_c}$} & \isS{$\mathrm{b}(\bar{L})\cdot\tilde{O}(\omega)$, where $b(\bar{L})$ denotes the logical weight (body) of $\bar{L}$ and $\omega$ is the physical weight of the seed logicals} & $O(d)$ \\
\addlinespace[3pt]
\isI{Graph surgery with canonical extractor} & \isA{An arbitrary logical PPM: $\mathcal{M}_{\bar{L}}$, where $\bar{L} \in \bar{\mc{P}}_{k_c}$} & \isS{$\tilde{O}(n_A l)$} & $O(d)$ \\
\addlinespace[3pt]
\isI{Parallel intra-column {(row)} surgery} & \isA{parallel $\{\mathcal{M}_{\bar{L}^{(t)}}\}_{t}$, where $\{\bar{L}^{(t)}\}_t$ is any collection of logically disjoint  {arbitrary logical operators, including mixed-type and $Y$-type, on a column and a row of logical fibres}} & \isS{ {$O(n_A l m)$, with thickening length $m \leq d$ set by the soundness}} & $O(d)$ \\
\addlinespace[3pt]
\isI{Parallel inter-column surgery} & \isA{parallel $\{\mathcal{M}_{\bar{Z}_{i, j, m} \bar{Z}_{i, j^{\prime}, m}}\}_{i,m}$ (resp. $\{\mathcal{M}_{\bar{X}_{i, j, m} \bar{X}_{i^{\prime}, j, m}}\}_{j,m}$) between two columns (resp. rows) of logical $Z$- (resp. $X$-) operators} & \isS{$O(n_A l)$} & $O(d)$ \\
\addlinespace[3pt]
\isI{Parallel magic injection} & \isA{inject $|\bar{T}\rangle^{\otimes k_c}$ from copies of distance-$d_s$ surface codes in parallel} & \isS{$O(d_s n_A l)$ } &  {$O(r_B^* d_s)$} \\
\bottomrule
\end{tabular}
\caption{\textbf{Logical instruction set for canonical LP codes:} (Section~\ref{sub: main-FT-logic}). 
This set applies generally to any $[[n, k, d]]$ symmetric canonical LP code $\mr{LP}_l(A, A^{*})$ with a base matrix $A \in R^{(n_A - r_A)\times n_A}$, which encodes $k_c = r_A^2 l \leq k$ logical qubits in the canonical logical basis, $\{\bar{q}_{i, j, m}\}_{i \in [r_A], j \in [r_A], m \in [l]}$ (Section~\ref{sub: main-canonical-LP-codes}). 
Note that when $r_A$ is fixed to be a constant, $n = O(l)$ scales linearly with the lift size $l$.
For each instruction, we list its logical action on the canonical logical qubits, the extra space overhead required (asymptotic), and the time cost in units of QEC cycles.
Here, we hide polylog factors using the $\tilde{O}$ notation.
When we describe the logical actions, we use the geometrical picture of the logical qubits in Fig.~\ref{fig:main-overview}(a), where the $k_c$ canonical logical qubits are arranged on {an} $r_A \times r_A$ grid of logical fibres (indexed by $(i, j)$), each consisting of $l$ logical qubits (indexed by $m$). 
}
\label{tab: main-instruction-set}
\end{table*}

In this section, we summarize our main results, which are illustrated in Fig.~\ref{fig:main-overview}. We first introduce a family of \emph{canonical} lifted-product (LP) codes whose logical operators admit an explicit and highly structured \emph{canonical logical basis} (Section~\ref{sub: main-canonical-LP-codes}). We then show that this basis unlocks a complete logical instruction set (Section~\ref{sub: main-FT-logic}), including automorphism and fold-transversal Clifford gates, modular low-overhead graph-based code surgery, and highly parallel logical primitives such as parallel hypergraph-based code surgery and parallel magic-state injection. Together, these primitives form a powerful and universal framework for logical computation. The full instruction set, together with its logical action and space--time cost, is summarized in Table~\ref{tab: main-instruction-set}.

\subsection{Canonical LP codes and logical basis}\label{sub: main-canonical-LP-codes}

Lifted-product (LP) codes are defined over Abelian rings, as a sub-family of general balanced-product (BP) codes~\cite{Breuckmann_2024}. Let $R:= \mathbb{F}_2[x]/(x^l+1) $ be a polynomial ring. Equivalently, $R=\mathbb{F}_2[G]$ for a cyclic group $G$ of order $l$ generated by {the }cyclic shift $x$. Given two classical base matrices \(A \in R^{m_A\times n_A}\) and \(B \in R^{m_B\times n_B}\), an LP code $\LP_l(A, B)$ has parity-check matrices: 
\begin{equation}\label{eq: main-LP-parity-checks}
\begin{aligned}
H^T_Z &= {\mathbb{B}}\left(\begin{array}{c} I_{n_A}\otimes_R B \\ A \otimes_R I_{n_B} \end{array}\right), \\
H_X &= {\mathbb{B}}\left(A\otimes_R I_{m_B},\; I_{m_A}\otimes_R B\right),
\end{aligned}
\end{equation}
{where $\mathbb{B}$ denotes the binarization, which replaces every entry of $R$ by its $l \times l$ circulant matrix.}
As illustrated in Fig.~\ref{fig:main-overview}(a), the physical qubits can be organized into two blocks of tiles: an \(n_A\times m_B\) left block ($L$) and an $m_A\times n_B$ right block ($R$), where each tile contains \(l\) physical qubits. We refer to each such tile of \(l\) qubits as a \emph{physical fibre}. Accordingly, we denote by \(q_{L/R}(i,j,m)\) the \(m\)-th qubit in the \((i,j)\)-th physical fibre of the left/right block.

As will be discussed in Section~\ref{sec: main-LP-basis}, under mild algebraic conditions on \(A\), \(B\), and \(l\), we obtain a family of \emph{canonical} LP codes
that admit the following canonical logical basis (Definition~\ref{def: main-canonical-basis-LP}):
\begin{equation}
    \{\left( \bar{Z}_{i, j, m}, \bar{X}_{i, j, m}\right)\}_{i \in [r_A], j\in [r_B^*], m \in [l]},
\end{equation}
encoding $k_c = r_A\, r_B^*\, l$ logical qubits exclusively on the left block (L) with $r_A:= n_A - m_A$ and $r_B^* := m_B - n_B$, indexed by triplets $(i,j,m)$, which we arrange on $r_A \times r_B^*$ \emph{logical fibres} each of size $l$ (Fig.~\ref{fig:main-overview}(a)). 
Accordingly, we denote by $\bar{q}_{i, j, m}$ the $m$-th logical qubit in the $(i, j)$-th logical fibre.
The basis satisfies the following properties (Proposition~\ref{prop: main-property-canonical-basis}):

\begin{enumerate}[label=(\roman*)]
    \item \textbf{Inheriting classical codewords.} The supports of $\bar{Z}_{i, j, m}$ and $\bar{X}_{i, j, m}$ are given, respectively, by the classical codewords of $A$ and $B^*$, so the physical weights of the logical operators are controlled by the (typically low) weights of the base classical codewords.
    
    \item \textbf{Conjugate pairs.} Each pair of logical operators $\left(\bar{Z}_{i,j,m}, \bar{X}_{i,j,m}\right)$ forms a conjugate pair that intersects at a single physical qubit, $q_L(i,j,m)$. Moreover, $\bar{Z}_{i,j,m}$ is disjoint from every non-conjugate $X$ logical operator $\bar{X}_{i',j',m'}$ with $(i',j',m') \neq (i,j,m)$. Consequently, each conjugate pair defines an independent logical qubit $\bar{q}_{i,j,m}$.
    
    \item \textbf{Cyclic symmetry.} 
    Fix a grid coordinate $(i,j)$. The logical operators $\{\bar{X}_{i,j,m}\}_{m=0}^{l-1}$ (and likewise $\{\bar{Z}_{i,j,m}\}_{m=0}^{l-1}$) form a cyclic orbit under the cyclic group action. 
    Specifically, the physical permutation
    \begin{align}\label{expr: main-cyclic-shift-physical}
        \sigma_t:\ q_{L/R}(i,j,m)\mapsto q_{L/R}(i,j,(m+t)\bmod l),
    \end{align}
    acts on the supports of the canonical basis operators: 
    \begin{align}
        \begin{aligned}
            \sigma_t:\quad
        \supp \bar{Z}_{i,j,m}&\mapsto \supp \bar{Z}_{i,j,(m+t)\bmod l}, \\
        \supp \bar{X}_{i,j,m}&\mapsto \supp \bar{X}_{i,j,(m+t)\bmod l}.
        \end{aligned}
    \end{align}
    
    \item \textbf{Row/column parallel structure.} The $Z$-type (resp. $X$-type) logical operators take the same form across logical qubits on different columns (resp. rows) of the logical fibres. Concretely, let $\sigma_C$ be any column permutation of $[r^*_B]$ and $P_C$ be the corresponding physical permutation such that 
    \begin{align}\label{expr: main-column-permutation-physical}
        P_C: q_L(i, j, m) &\mapsto q_L(i, \sigma_C(j), m).
    \end{align}
    {Similarly, let $\sigma_R$ be any row permutation of $[r_A]$ and $P_R$ be the corresponding physical permutation.} These physical permutations act on the supports of the canonical basis operators: $\supp \bar{Z}_{i,\sigma_C(j),m} = P_C (\supp \bar{Z}_{i,j,m})$ and $\supp \bar{X}_{\sigma_R(i),j,m} = P_R(\supp\bar{X}_{i,j,m})$.

    \item \textbf{$ZX$-duality.} When the code is symmetric, i.e. $B = A^*$, where $A^*$ is the transpose of $A$ with each entry $f(x)$ replaced by $f(x^{-1})$ ($x^{-1} = x^{l-1}$ since $x^l = 1$; see Appendix~\ref{app-sec: premliminaries}), the $Z$- and $X$-type logical operators also take the same form. Concretely, there exists a physical permutation $\tau$, in the form of an involution, 
    \begin{align}\label{expr: main-XZ-duality-involution-physical}
        \tau \quad: q_L(i, j, m) \leftrightarrow q_L(j, i, m); \\
        \quad q_R(i, j, m) \leftrightarrow q_R(j, i, m),
    \end{align}
    such that $\supp \bar{X}_{i, j, m} = \tau (\supp \bar{Z}_{j, i, m})$ for all $i, j, m$.

\end{enumerate}
See Fig.~\ref{fig:main-overview}(a) and Fig.~\ref{fig: canonical-basis} for illustrations of these properties.
{As shown in Section~\ref{sec: main-LP-basis}, these conditions are mild: a typical randomized code with odd lift size $l$ and \(r_A, r_B^* \ge 1\) is found to be a canonical LP code}{, and the canonical basis generalizes to even lift sizes $l$ as well, with the same form and the same properties (i)--(v)}. For a symmetric LP code with base matrix \(A \in R^{(n_A-r_A)\times n_A}\), we use the shorthand \(\mr{LP}_l^{(n_A-r_A)\times n_A}\). Our main examples of canonical LP codes are \(\LP^{3\times5}_{33} = [[1122,148,\leq\!20]]\) (with \(r_A=2\)) and \(\LP^{3\times7}_{75} = [[4350,1224,\leq\!20]]\) (with \(r_A=4\)), both used in Ref.~\cite{cain2026shorsalgorithmpossible10000}. These codes encode hundreds to thousands of logical qubits at high encoding rates and distances---precisely the regime that has remained difficult to access with HGP code families. {The even case further unlocks several new canonical LP codes that we use throughout this work, such as \(\LP^{2\times3}_{36} = [[468,36,20]]\), and \(\LP^{3\times5}_{28} = [[952,112,17]]\).}

\subsection{Fault-tolerant logic}\label{sub: main-FT-logic}
The canonical logical basis unlocks a complete logical instruction set, summarized in Table~\ref{tab: main-instruction-set} and illustrated in Fig.~\ref{fig:main-overview}(b)--(e). Throughout, one \emph{logical cycle} denotes $\Theta(d)$ rounds of syndrome extraction for a distance-$d$ code.

\subsubsection{Automorphism and fold-transversal gates}
The symmetry of the canonical LP codes directly enables symmetry-derived logical gates, such as the automorphism and fold-transversal {ones}~\cite{Breuckmann_2024}. 
Here, utilizing the structure of the canonical logical basis, we can easily characterize their logical actions (Proposition~\ref{prop: main-fold-transversal-canonical-LP}). 

A physical cyclic shift, $\sigma_t: q_{L/R}(i, j, m) \rightarrow q_{L/R}(i, j, (m + t){\bmod} l)$, implements a logical cyclic shift $\bar{\sigma}_t$ acting on all logical fibres in parallel,
\begin{equation}\label{expr: main-logical-cyclic-shift}
\begin{aligned}
    \bar{\sigma}_t:\ &\bar{X}_{i, j, m} \mapsto \bar{X}_{i, j, (m + t)\bmod l}, \\
    &\bar{Z}_{i, j, m} \mapsto \bar{Z}_{i, j, (m + t)\bmod l},
\end{aligned}
\end{equation}
or equivalently, $\bar{q}_{i, j, m} \mapsto \bar{q}_{i, j, (m + t){\bmod} l} $, with no ancilla overhead. 

Exploiting the fold $ZX$-duality of symmetric LP codes \(\LP_l(A,A^*)\), under which the \(Z\)- and \(X\)-checks, as well as the canonical \(Z\)- and \(X\)-logical operators, are exchanged by the physical involution
$\tau:\; q_{L/R}(i,j,m)\leftrightarrow q_{L/R}(j,i,m)$,
corresponding to reflection across the diagonal of the physical grid (Fig.~\ref{fig:main-overview}(a)), we obtain fold-transversal Clifford gates per Ref.~\cite{Breuckmann_2024}. 
In particular, the fold-transversal \(\overline{\mathrm{H\text{-}SWAP}}\), consisting of physical transversal Hadamards followed by the fold permutation \(\tau\), implements global logical Hadamards followed by the same fold permutation on the logical qubits--- 
$\bar{X}_{i,j,m} \leftrightarrow \bar{Z}_{j,i,m}$.
Similarly, the fold-transversal \(\overline{\mathrm{S\text{-}CZ}}\), consisting of transversal \(S\) gates on the left-block fold-invariant qubits \(\{q_{L}(i,i,m)\}_{i,m}\), transversal \(S^{\dagger}\) gates on the right-block fold-invariant qubits \(\{q_{R}(i,i,m)\}_{i,m}\), and transversal \(CZ\) gates on each folded qubit pair \(\{(q_{L/R}(i,j,m),q_{L/R}(j,i,m))\}_{i\neq j,m}\), implements logical \(\bar{S}\) gates on the fold-invariant logical fibres and logical \(\overline{CZ}\) gates between folded logical pairs.

Note that these fold-transversal gates are in exact analogy with those for symmetric HGP codes~\cite{Quintavalle_2023, xu2025fast}. 

\subsubsection{Graph code surgery}
The rich symmetries of the canonical logical basis enable a highly modular graph-surgery toolkit for implementing an arbitrary logical Pauli-product measurement
\begin{equation}
    \mathcal{M}_{\bar{L}}, \qquad \bar{L}\in \bar{\mathcal{P}}_{k_c}.
\end{equation}
When \(\bar{L}\) has low \emph{logical weight} \(\mathrm{b}(\bar{L})\), i.e., the number of logical Pauli factors (\emph{bodies}) in the product, we realize \(\mathcal{M}_{\bar{L}}\) by bridging \(\mathrm{b}(\bar{L})\) copies drawn from a small set of ``seed'' surgery gadgets. 
Specifically, for a symmetric $\mr{LP}_l(A,A^{*})$ code with $A \in R^{(n_A-r_A)\times n_A}$, there exists a set $\mathfrak{S}_{\mathrm{seed}}$ of only $r_A$ seed surgery gadgets, independent of the lift size $l$, each of size $\tilde{O}(\omega)$, where $\omega$ is the maximum \emph{physical} weight of a canonical logical operator. Any $\mathcal{M}_{\bar{L}}$ can then be implemented by bridging $\mathrm{b}(\bar{L})$ such gadgets via the universal adapters of Ref.~\cite{he2025extractorsqldpcarchitecturesefficient}.
This reduction to only \(r_A\) distinct seed surgery gadgets relies crucially on the symmetry of the canonical basis: the same gadget that measures one representative operator \(\bar{Z}_{i,j,m}\) can, by cyclic shifts, the row/column parallel structure,  and the $ZX$-duality, be reused to measure any \(\bar{Z}_{i,j',m'}\) or \(\bar{X}_{j',i,m'}\) (Fig.~\ref{fig: seed-surgery}(a) and Fig.~\ref{fig: canonical-basis}). 
As a result, arbitrary low-weight logical measurements reduce to duplicating, shifting, and bridging a constant set of small gadgets, rather than designing a new surgery gadget for each logical operator. For instance, the \(\LP^{3\times7}_{75}=[[4350,1224,\leq 20]]\) code requires only \(4\) seed surgery gadgets, each using roughly \(200\) ancilla qubits (excluding checks; Table~\ref{tab: seed-surgery-gadget-resources}).

For high-logical-weight operators \(\bar{L}\), bridging seed surgery gadgets incurs a space overhead that grows linearly with \(\mathrm{b}(\bar{L})\). In this regime, we instead construct a single \emph{canonical extractor}: a fixed ancilla graph \(\mathcal{X}\) that measures arbitrary \(\bar{L}\) with ancilla size \(|\mathcal{X}|=\tilde{O}(n_A l)\), independent of \(\mathrm{b}(\bar{L})\)~\cite{he2025extractorsqldpcarchitecturesefficient}. While generic qLDPC extractors are typically much larger than the data block, the structure of the canonical basis lets us build a substantially smaller extractor for LP codes. The construction begins with a small extractor coupled to one column of physical fibres for measuring \(Z\)-type logical operators on a column of logical fibres, and then extends it to arbitrary logical measurements by exploiting, again, the symmetries of the canonical logical basis (Fig.~\mbox{\ref{fig: seed-surgery}}(b))
For instance, for the \(\LP^{3\times5}_{33}=[[1122,148,\leq 20]]\) code, the resulting canonical extractor is below half the size of the data block (Table~\ref{tab: LP-canonical-extractor-resources}). {In Table~\ref{tab: LP-canonical-extractor-numerics}, we construct canonical extractors for another three canonical LP codes, \(\LP^{2\times3}_{36}=[[468,36,20]]\), \(\LP^{3\times5}_{28}=[[952,112,17]]\), and \(\LP^{3\times7}_{75}=[[4350,1224,\leq 20]]\).} A key consequence of these symmetry-based reductions is that the full surgery toolkit in this work can be constructed explicitly, fully certified, and proved fault-tolerant.

\subsubsection{Parallel code surgery}

The row/column parallel structure of the canonical logical basis also enables highly parallel logical measurements. First, we develop \emph{parallel intra-column} (and, by symmetry, intra-row) code surgery, which simultaneously implements an arbitrary collection
\[
    \{\mathcal{M}_{\bar{L}^{(t)}}\}_t,
\]
of logically disjoint \(Z\)-type (respectively, \(X\)-type) PPMs supported on a single column (respectively, row) of logical fibres (Fig.~\ref{fig: parallel-surgery}(a)). Since these logical operators are localized to one column (resp. row) of physical fibres, the surgery reduces essentially to a surgery protocol for the underlying classical code \(A\). 
{We further extend this to arbitrary mixed-type or $Y$-type PPMs through generalized graph isomorphisms (Section~\ref{sec: parallel-surgery-arbitrary-addressable}), which we refer to as generalized bridging.} We derive algebraic conditions {-- soundness conditions --} that substantially reduce the required hypergraph thickening while preserving fault tolerance, leading to practical, low-overhead protocols with full addressability. 
For example, arbitrary single- and two-body \(Z\)-PPMs on the \(66\) logical qubits of {the} \(\mr{LP}^{3\times5}_{33}\) code require ancillas {of }only {about {\(2\!-\!3.5\times\)} the size of the data block} (Table~{\ref{tab: high-rate-two-body-addressable}}). Furthermore, we co-design a subfamily of canonical LP codes that achieve the provably optimal soundness-distance tradeoff, such as \(\LP^{2\times4}_{16}=[[320,64,10]]\), $\LP^{3\times 5}_{20} =[[680, 80, 13]]$, \(\LP^{2\times3}_{36}=[[468,36,20]]\), and \(\LP^{3\times5}_{28}=[[952,112,17]]\). These codes require hypergraph thickening of length only $3$, $3$, $4$, and $4$, respectively (Tables~\ref{tab: high-rate-two-body-addressable} and~\ref{tab: high-rate-arbitrary-resources}), with ancillas of roughly {\(1\!-\!1.5\times\)} the size of the data code for same-type single-body measurements{, \(3\!-\!4\times\) for arbitrary two-body measurements} and {\(3.5\!-\!5.5\times\)} for arbitrary logically disjoint {(any-body)} measurements.
We further introduce a highly efficient \emph{parallel inter-column} surgery that simultaneously measures
\[
    \{\mathcal{M}_{\bar{Z}_{i,{j},m}\bar{Z}_{i,{j^{\prime}},m}}\}_{i,m},
\]
between two columns of logical fibres (Fig.~\ref{fig: parallel-surgery}(b)). Unlike intra-column surgery, this protocol requires no hypergraph thickening: the merged code remains an LP code, allowing distance preservation to be certified directly. Consequently, the ancilla consists only of a single auxiliary column of physical fibres, yielding an extremely low-overhead protocol that generalizes the homomorphic measurement gadget for HGP codes~\cite{xu2025fast}.

\subsubsection{Parallel magic}
Finally, we introduce a highly parallel protocol for injecting all \(k_c=r_A r_B^*l\) logical magic states of a canonical $\mr{LP}_l(A, B)$ code. Motivated by transversal qLDPC distillation factories~\cite{xu2025fast,cain2026shorsalgorithmpossible10000}, which require an entire block of input magic states, our protocol injects \(r_A l\) logical \(\ket{T}\) states in parallel from \(r_A l\) distance-\(d_s\) surface codes prepared, for example, by magic-state cultivation~\cite{gidney2024magicstatecultivationgrowing}. The key ingredient is a \emph{transistor code}, another LP code that mediates the transfer from the surface codes to one column of logical fibres of $\mr{LP}_l(A, B)$ through two rounds of the parallel inter-column surgery introduced above (Fig.~\ref{fig: magic-state-injection}). Repeating this procedure for each logical-fibre column injects all \(k_c\) magic states in $r_B^*$ batches.

The protocol has extra space overhead dominated by the transistor code, while the additional surgery ancillas are comparatively small. Its runtime is \(O(r_B^* d_s)\), and its phenomenological fault distance is at least \(\min(d, d_s)\). For example, the \(\mr{LP}^{3\times5}_{33}\) code can be loaded with \(132\) logical \(\ket{T}\) states using distance-7 surface codes in two batches, with a transistor code approximately \(1.5\times\) the size of the data block.

\section{Canonical LP codes and logical basis\label{sec: main-LP-basis}}
Before defining canonical LP codes and their canonical logical basis in Definition~\ref{def: main-canonical-basis-LP}, we introduce the minimal notation and background needed for general (quasi-cyclic) LP codes. In particular, we derive a general formula (necessary and sufficient condition) for an algebraic logical basis of an arbitrary LP code that specializes to the canonical LP basis under mild algebraic conditions on the underlying classical codes.

Logical-basis constructions are known for several LP subfamilies, including bivariate bicycle (BB) codes~\cite{eberhardt2024logicaloperatorsfoldtransversalgates} and cluster cyclic (CC) codes~\cite{gu2026qgpuparallellogicquantum}. 
Here, we present a single general formula that unifies the previous constructions and applies to generic LP codes over arbitrary Abelian groups, by combining the K\"unneth formula~\cite{bott-tu-1982} with the algebraic structure of semisimple rings~\cite{roman1997field}~\footnote{It also gives an explicit, constructive version of the remarks in Appendix~B of Ref.~\cite{Panteleev_2022}}. In this work, we represent a CSS stabilizer code $\mathcal{Q}$ with $X$ (resp. $Z$) parity-check matrix $H_X$ (resp. $H_Z$) as a three-term (co)chain complex
\begin{equation}\label{eq: main-2-chain-complex}
\begin{tikzcd}
	{\mathcal{Q}:} & {\mathcal{Q}_2} & {\mathcal{Q}_1} & {\mathcal{Q}_0}
	\arrow["{\partial^Q_2}", from=1-2, to=1-3]
	\arrow["{\partial^Q_1}", from=1-3, to=1-4]
\end{tikzcd},
\end{equation}
where $\mc{Q}_2$, $\mc{Q}_1$, and $\mc{Q}_0$ are $\mbb{F}_2$-vector spaces whose bases are associated with the $Z$-checks, the qubits, and the $X$-checks of $\mc{Q}$, respectively. The boundary maps are given by $\partial^Q_2 = H_Z^T$ and $\partial^Q_1 = H_X$.
Then, the logical $Z$- and $X$-type operators are given by the first homology group $H_1(\mathcal{Q}) := \ker{\partial_1^Q}/\IM{\partial_2^Q}$ and the first cohomology group $H^1(\mathcal{Q}) := \ker{(\partial_2^Q)^T}/ \IM{(\partial_1^Q)^T}$, respectively.

To define a lifted-product code, we first consider chain complexes over a group-algebra ring $R = \mbb{F}_2[G]$ for an Abelian group $G$. 
In this work, we further specialize $G$ to the cyclic group $C_l$ of order $l$, where $R$ becomes isomorphic to the univariate quotient polynomial ring $R\simeq\mbb{F}_2[x]/(x^l + 1)$. 
Note that our formalism and techniques could also be generalized to other groups, including non-Abelian groups{,} and {to }more general balanced-product codes (see upcoming works~\cite{Bhardwaj2026MittenCodes, Hong2026RateOneFifthLDPC}).
An $R$-valued chain complex $\mc{C}$ is defined as a collection of free $R$-modules $\{\mc{C}_i = R^{m_i}\}_i$ and $R$-linear boundary maps $\partial_i \in R^{m_{i-1}\times m_i}$.
\begin{definition}[2D lifted-product code~\cite{Panteleev_2022, Panteleev_2021}]\label{def: main-2d-LPAB}
    Let $R$ be a finite univariate polynomial ring, and let $\mathcal{A}_{\bullet}: \mathcal{A}_1 \xrightarrow{A} \mathcal{A}_0$ and $\mathcal{B}_{\bullet}: \mathcal{B}_1 \xrightarrow{B} \mathcal{B}_0$ be two $1$-chains over $R$, with $A \in R^{m_A \times n_A}$ and $B \in R^{m_B \times n_B}$.
    The 2D lifted-product code, denoted $\LP_l(A, B)$, is the total complex~\cite{breuckmann2021quantum} $\mc{M}:
    \begin{tikzcd} 
    {\mathcal{M}_2} & {\mathcal{M}_1} & {\mathcal{M}_0}
	\arrow["{\partial^M_2}", from=1-1, to=1-2]
	\arrow["{\partial^M_1}", from=1-2, to=1-3]
\end{tikzcd}$, 
    of the product complex $\mathcal{A}_{\bullet} \otimes_R \mathcal{B}_{\bullet}$,
    where $\mc{M}_2 = \mc{A}_1 \otimes_R \mc{B}_1$, $\mc{M}_1 = \mc{A}_1 \otimes_R \mc{B}_0 \oplus \mc{A}_0 \otimes_R \mc{B}_1$, $\mc{M}_0 = \mc{A}_0 \otimes_R \mc{B}_0$, and the boundary maps are given by
\begin{equation}\label{eq: main-LP-boundary-maps-ring}
\begin{aligned}
(\partial_2^M)^* & = \left(I_{n_A} \otimes_R B^*, A^* \otimes_R I_{n_B}\right), \\
\partial_1^M & = \left(A \otimes_R I_{m_B}, I_{m_A} \otimes_R B\right), 
\end{aligned}
\end{equation}
where $A^*$ is the transpose of $A \in R^{m_A \times n_A}$ with each polynomial entry $f(x)$ replaced by $f(x^{l-1})$. 
\end{definition}
To obtain a valid $\mbb{F}_2$-chain complex $\mc{Q}$ corresponding to a qubit code, we take the \emph{binarization} of the $R$-chain complex $\mc{M}$, defined as follows.
First, utilizing the fact that $R$ is also an $l$-dimensional vector space with basis $\{x^i\}_{i = 0, 1, \cdots, l - 1}$, we can define an isomorphism $\mbb{B}: R^m \rightarrow \mbb{F}_2^{ml}$.
This induces a faithful binary representation of $R$-linear maps  $\mbb{B}: R^{m\times n} \rightarrow \mbb{F}_2^{ml \times nl}$.
Furthermore, for any $A \in R^{m \times n}$ and $u \in R^n$, we have 
$\mbb{B}(Au) = \mbb{B}(A) \mbb{B}(u)$.
Utilizing $\mbb{B}$, we can define the binarization of $\mc{M}$ as
\begin{equation}
    \mbb{B}(\mc{M}): \begin{tikzcd} 
    {\mbb{B}(\mathcal{M}_2)} & {\mbb{B}(\mathcal{M}_1)} & {\mbb{B}(\mathcal{M}_0})
	\arrow["{\mbb{B}(\partial^M_2)}", from=1-1, to=1-2]
	\arrow["{\mbb{B}(\partial^M_1)}", from=1-2, to=1-3]
    \end{tikzcd}.
\end{equation}
Taking $\mbb{B}(\mc{M})$ as $\mc{Q}$ then produces the standard definition of a 2D LP code. The interplay between binary chain{s} and chain{s} defined over {the} ring $R$ requires {distinct notation}. Throughout the work, we define row span ($\rs_{\mathbb{F}_2}$ vs. $\rs_{R}$), kernel ($\ker_{\mathbb{F}_2}$ vs. $\ker_{R}$), cokernel ($\coker_{\mathbb{F}_2}$ vs. $\coker_{R}$), tensor product ($\otimes_{\mathbb{F}_2}$ vs. $\otimes_{R}$), etc., with {a subscript} indicating {the ring }over which the arithmetic {is} defined.

Here, we are interested in understanding the algebraic properties, e.g., the homology and cohomology groups, of $\mc{Q}$ through $\mc{M}$. 
As such, it is useful to establish some formal connection between the homologies of $\mc{Q}$ and $\mc{M}$.
Since we can view $\mbb{B}$ as a bijective chain map $\mbb{B}: \mc{M} \rightarrow \mc{Q}$, it also induces an isomorphism $\mbb{B}: H_1(\mc{M}) \rightarrow H_1(\mc{Q})$.
Since $H_1(\mc{Q})$ is hard to calculate directly, we will first calculate $H_1(\mc{M})$ using algebraic tools on the ring $R$, and then obtain $H_1(\mc{Q})$ via $\mbb{B}(H_1(\mc{M}))$. 
Observe that $H_1(\mc{Q})$ is a quotient vector space, which is isomorphic to an $\mbb{F}_2$-subspace of $\mbb{F}_2^n$. This means that we can find a generator matrix $L_Z \in \mbb{F}_2^{k \times n}$ for the logical $Z$ operators such that $H_1(\mc{Q}) = \rs_{\mbb{F}_2}{L_Z}$, where $k$ denotes the dimension of 
$H_1(\mc{Q})$.
Then, according to the isomorphism, we know that there exists an $L^M_Z \in R^{k \times n_M}$ such that $H_1(\mc{M}) = \rs_{\mbb{F}_2}{L^M_Z}$. 
Now, we show how to obtain $L^M_Z$ through the K\"unneth formula on $\mc{M}$:

To save notation, unless stated otherwise, all vectors and matrices are $R$-valued. For odd $l$, every $R$-module is projective; see Definition~\ref{def: free-projective-module} and Proposition~\ref{proposition: semi-simple-iff-projective}. For example, $u \in \mathcal{A}_1 = R^{n_A}$ denotes a vector with entries $u_i \in R$ for $i \in [n_A]$.
\begin{theorem}[Informal: Basis characterization of lifted-product{ codes} over a {semisimple} ring]\label{thm: main-LP-basis-characterization}
    Let $R = \mathbb{F}_2[x]/(x^l + 1)$ be the finite, univariate polynomial ring, with odd $l$. Let the 2D lifted-product code $\LP_l(A, B)$ be given in Definition~\ref{def: main-2d-LPAB}. Then the $Z$-type logical operators corresponding to $H_1(\mc{M})$ have supports generated by
    \begin{align}\label{eq: main-LZ}
        L_{Z}^M =   G_{A} \otimes_R E_{B} \bigoplus E_{A} \otimes_R G_{B},
    \end{align}
    where $G_C$ and $E_C$ satisfy $\rs_R(G_C) = \ker_R C$ and $\rs_{\mbb{F}_2}(E_C) = \coker_R C$, respectively, for $C \in \{A, B\}$.
    Similarly, the $X$-type logical operators corresponding to $H^1(\mc{M})$ have supports generated by
    \begin{align}\label{eq: main-LX}
        L_{X}^M =  E_{A^*} \otimes_R G_{B^*} \bigoplus G_{A^*} \otimes_R E_{B^*},
    \end{align}
    where we take $\rs_{\mbb{F}_2} G_{C^*} = \ker_R C^*$ and $\rs_R E_{C^*} = \coker_R C^*$, respectively, for $C \in \{A, B\}$. Binarizing the nonzero rows of $L_Z^M$ and $L_X^M$ gives the physical supports of the $L_Z$ and $L_X$ logical bases, respectively.
\end{theorem}

The problem thus reduces to analyzing the classical base codes \(A\) and \(B\), specifically by computing their kernels and cokernels over the ring \(R\). Since carrying out these computations directly over \(R\) is generally difficult, we instead decompose \(R\) via the Chinese remainder theorem and perform the analysis over the resulting finite fields via algebraic root theory~\cite{roman1997field}. Let $R = \mathbb{F}_2[x]/(x^l + 1)$ be the finite, univariate polynomial ring with odd $l$. Then there exist unique irreducible polynomials $\{b^{(i)} \}^t_{i=1}$ such that $\prod^t_{i=1} b^{(i)} = x^l+1$ and $R_{(i)} := \mathbb{F}_2[x]/(b^{(i)}) \simeq \mbb{F}_{2^{\mr{deg}(b^{(i)})}}$, which is a binary extension field of degree $\deg(b^{(i)})$. Over these finite-field extensions, we can perform Gaussian elimination on any matrix $G$, which always has pivot columns. We denote these pivot column indices by the \emph{information set} (See Definition~\ref{def: information-set} and Definition~\ref{def: standard-form}). In short, up to column permutation, we can always row-reduce $G$ to $\begin{pmatrix}
    I &G_0
\end{pmatrix}$, where the column indices of the leftmost identity block define an information set. We refer to $b^{(i)}|x^l+1$ as an irreducible polynomial factor, and we obtain the following:
\begin{theorem}[Informal: Evaluations of $\ker_R A$ and $\coker_R A$]\label{thm: main-computation-kernel-cokernel} 
     Let $R = \mathbb{F}_2[x]/(x^l + 1)$ be the finite, univariate polynomial ring with odd $l$, and let $\{b^{(i)} \}^t_{i=1}$ be the set of irreducible polynomial factors. For each factor $b^{(i)}$, let $c_{b^{(i)}} \in R$ be a central idempotent polynomial such that $R = \bigoplus_{i} c_{b^{(i)}} R$. Under the identification of $c_{b^{(i)}}R$ with $R_{(i)}:=\mbb{F}_2[x]/(b^{(i)})$, we have $A_{(i)} \cong c_{b^{(i)}}A$ and
         \begin{align}
             \ker_R A = \bigoplus^t_{i=1} c_{b^{(i)}}\ker_{R_{(i)}} A_{(i)},
         \end{align}
         where each component is lifted back to $c_{b^{(i)}}R\subseteq R$ through this identification.
         Similarly, we have
         \begin{align}
             \coker_R A  =  \bigoplus^t_{i=1}c_{b^{(i)}}\coker_{R_{(i)}} A_{(i)}.
         \end{align}
         Let $G^{(i)}_{A} \in R^{(\dim_{R_{(i)}} \ker_{R_{(i)}} A_{(i)}) \times n_A}$ be the lifted matrix whose rows form an $R_{(i)}$-basis for $\ker_{R_{(i)}} A_{(i)}$, and let $I^{(i)}_A$ be its information set. Similarly, let $E^{(i)}_{A} \in R^{(\dim_{\mathbb{F}_2} \coker_{R_{(i)}} A_{(i)}) \times m_A}$ be the lifted matrix whose rows form an $\mbb{F}_2$-basis for $\coker_{R_{(i)}} A_{(i)}$, and let $I^{(i)}_{A^*}$ be {an information set of $\ker_{R_{(i)}} A^T_{(i)}$}. Then, we can take $G_{A} = \mathrm{vstack}(G^{(1)}_{A}, \cdots, G^{(t)}_{A})$ and $E_{A} = \mathrm{vstack}(E^{(1)}_{A}, \cdots, E^{(t)}_{A})$, vertically stacking these submatrices. 
\end{theorem}
\begin{remark}
    Constructed in this way, $G_C$ and $E_C$ (for $C \in \{A, B\}$) satisfy the conditions of Theorem~\ref{thm: main-LP-basis-characterization}, so that the nonzero rows of the product matrices give an $\mbb{F}_2$-linearly independent set of nontrivial logical operators of $\LP_l(A, B)$, although $G_C$ (resp. $E_C$) need not be a minimal $R$-linearly independent generating set of $\ker_R C$ (resp. $\coker_R C$). Equivalently, the balanced-product construction gives a {conceptually} simpler but less explicit description{. It is} stated in Theorem~\ref{thm: kunneth-theorem-semi-simple-with-balanced} and {illustrated by }the example in Section~\ref{subsec: general-formulation-homological-algebra}{. For the definition} of the balanced product{, see} Ref.~\cite{breuckmann2021quantum} and Section~\ref{app-sec: product-chains-complexes}. 
\end{remark}

In the following, we provide several code examples to illustrate Theorem~\ref{thm: main-LP-basis-characterization} and Theorem~\ref{thm: main-computation-kernel-cokernel}. Detailed statements and proofs are given in Theorem~\ref{thm: formal-LP-basis-characterization}, Theorem~\ref{thm: formal-computation-kernel-cokernel}, and Section~\ref{app-sec: lp-basis-characterization} in the Appendices.
Throughout the examples below, we write $\chi := \sum^{l-1}_{i=0} x^i$ for the central idempotent associated with the factor $b = 1 + x$; this identity holds for odd $l$.
\begin{example}[Twisted toric code]
    We consider $\LP_l(A, A)$ with $ A = \left( \begin{array}{cc}
    1 & 1\\
    1 & x
\end{array}\right)$ over $R = \mathbb{F}_2[x]/(x^l+1)$ (odd $l$). Let $x^l + 1 = \prod_{i = 1}^t b^{(i)}$ with irreducible $b^{(i)}$.
In this case, we have $\ker_{R_{(i)}} A = \coker_{R_{(i)}} A = 0$ for all $i$ except for $i=1$ with $b^{(1)} = 1 + x$ and $c_{b^{(1)}} = \chi$, where
\begin{equation}
\begin{aligned}
    \ker_{R_{(1)}} A & = \ker_{\mbb{F}_2[x]/(x + 1)} A = \rs_{\mbb{F}_2}((1, 1)), \\
    \coker_{R_{(1)}} A & = \coker_{\mbb{F}_2[x]/(x + 1)} A = \rs_{\mbb{F}_2}((1, 0)), 
\end{aligned}
\end{equation}

Here $\IM_{R_{(1)}} A_{(1)}=\rs_{\mbb{F}_2}((1,1))$, so $(1,0)$ represents the nonzero cokernel class, whereas $A(\chi,0)^T=(\chi,\chi)^T$. Therefore, we can choose
\begin{align}
    G_A=\left(\chi,\chi\right), \qquad E_A=\left(\chi,0\right).
\end{align}
The two direct summands $\ker_R A\otimes_R\coker_R A$ and $\coker_R A\otimes_R\ker_R A$ each contribute one logical operator, supported on the first and second sectors, respectively, yielding
\begin{align}
  L^M_Z
=\begin{pmatrix}
\chi & 0 & \chi & 0 & 0 & 0 & 0 & 0 \\
0 & 0 & 0 & 0 & \chi & \chi & 0 & 0
\end{pmatrix}.
\end{align}
\end{example}
\begin{example}[Cluster cyclic code Ref.~\cite{gu2026qgpuparallellogicquantum}]
    We next consider the cluster cyclic code of Ref.~\cite{gu2026qgpuparallellogicquantum}, defined by $\LP_5(A, B)$ with $l = 5$ and code parameters $[[40, 8, 5]]$, where
\begin{align}
    A = \left( \begin{array}{cc}
    x + x^3 & x + x^4\\
    x + x^4 & x + x^4
\end{array}\right); \quad B = \left( \begin{array}{cc}
     1 + x& x^3 + x^4\\
    1 + x^4 & 1 + x^4
\end{array}\right).
\end{align}
By Definition IV.1 of Ref.~\cite{gu2026qgpuparallellogicquantum}, the determinants of $A$ and $B$ over $R$ are nonzero: $ \mathrm{det}_R(A) =x + \chi  ; \quad \mathrm{det}_R(B) = x + x^2$, where $ \chi = 1 + x + x^2 + x^3 + x^4$. In this case, 
\begin{align}
    R  = c_{b^{(1)}}R \oplus c_{b^{(2)}} R \cong R/(b^{(1)}) \oplus R/(b^{(2)}),
\end{align}
where $b^{(1)} = 1+x$ and $b^{(2)} = 1 + x + x^2 + x^3 + x^4$, with corresponding central idempotents $c_{b^{(1)}} = \chi$ and $c_{b^{(2)}} = x + x^2 + x^3 + x^4 = \chi+1$. Note that $\chi^2=\chi$. Over the fields $R_{(i)}$ for $i=1, 2$,
\begin{align}
    \begin{aligned}
         &\mathrm{det}_{R_{(1)}}(A) = 0; \quad  \mathrm{det}_{R_{(2)}}(A) = x\\
          &\mathrm{det}_{R_{(1)}}(B) = 0; \quad  \mathrm{det}_{R_{(2)}}(B) = x + x^2.
    \end{aligned}
\end{align}
and every $1 \times 1$ submatrix of $A$ and $B$ has zero determinant over $R_{(1)}$ since its entry is a binomial. Hence, 
\begin{align}
    \ker_{R_{(1)}} A  = \coker_{R_{(1)}} A  = \rs_{\mathbb{F}_2}((1,0), (0, 1)).
\end{align}
The same statements hold for $B$. Hence, we can take
\begin{align}
    G_{A} = G_{B} =E_{A} =E_{B}  = \left( \begin{array}{cc}
         \chi &0  \\
         0& \chi 
    \end{array}\right).
\end{align}
 In this case, a generator matrix for the homology group according to Theorem~\ref{thm: main-LP-basis-characterization} is
\begin{align}
   L^M_Z = \left( \begin{array}{cc}
       \chi I_4 & 0 \\
       0 & \chi I_4 
   \end{array} \right),
\end{align}
which reproduces the key feature in Ref.~\cite{gu2026qgpuparallellogicquantum} that the eight logical qubits have disjoint logicals, each supported on a ``cluster" of $l$ qubits. Interestingly, $L^M_X = L^M_Z$.
 
\end{example}

\begin{figure}[!t]
    \centering
    \includegraphics[width=0.5\textwidth]{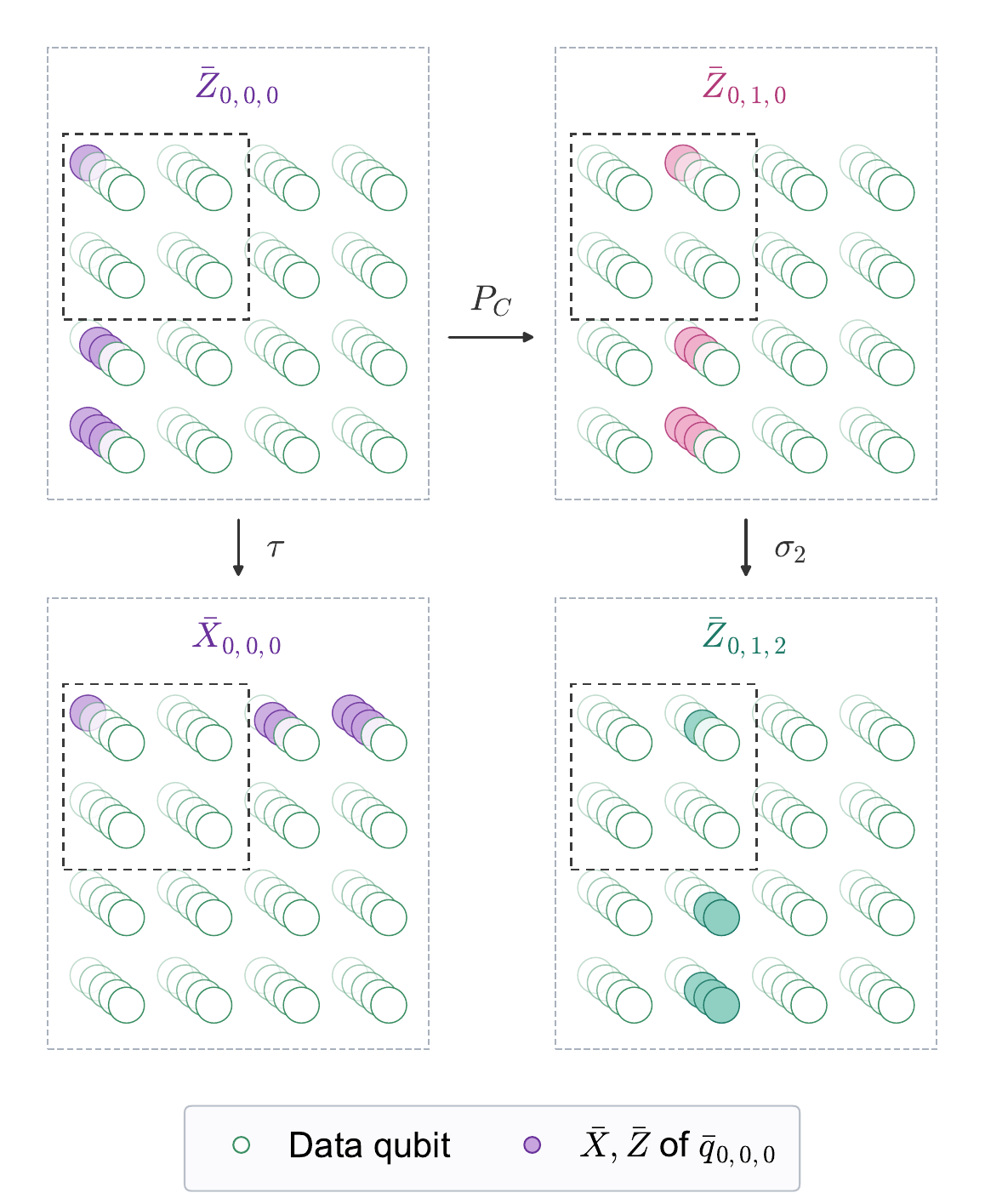}
   \caption{\textbf{Symmetries of the canonical logical basis:} (Proposition~\ref{prop: main-property-canonical-basis}), on the toy code of Fig.~\ref{fig:main-overview}. Each panel shows the left-block support of one canonical logical basis operator (one color per logical qubit; dark dashed boxes: the information sub-grid where the logicals intersect).
   The arrows generate all other logicals that are permutation-equivalent to $\bar Z_{0,0,0}$: the column permutation $P_C$ translates it to another column ($\bar Z_{0,j,0}$, drawn $j=1$); the cyclic shift $\sigma_t$ shifts by $t$ within each logical fibre (drawn $t=2$); and, for $B=A^{*}$, the folding involution $\tau$ transposes it into $\bar X_{0,0,0}$ of the \emph{same} logical qubit.}
    \label{fig: canonical-basis}
\end{figure}

Using the general formalism developed in Theorems~\ref{thm: main-LP-basis-characterization} and~\ref{thm: main-computation-kernel-cokernel}, we now define a family of canonical LP codes equipped with a canonical logical basis analogous to that of HGP codes~\cite{Quintavalle_2023, xu2025fast}. Recall that an HGP code with full-rank classical base check matrices admits a canonical logical basis \(\{(\bar{X}_{i,j},\bar{Z}_{i,j})\}\), in which each conjugate pair labels a logical qubit: up to qubit permutations, \(\bar{Z}_{i,j}\) and \(\bar{X}_{i,j}\) are supported on the \(j\)-th column and \(i\)-th row of the physical-qubit grid, respectively, and intersect uniquely on the \((i,j)\)-th physical qubit (see Fig.~1 of Ref.~\cite{xu2025fast}). Intuitively, the canonical logical basis for LP codes recovers the same HGP structure, except that each \((i,j)\) now labels a \emph{logical fibre} of \(l\) logical qubits that are equivalent up to cyclic permutation. 

\begin{definition}[Canonical LP basis and canonical LP code] \label{def: main-canonical-basis-LP}
    For an $\mathrm{LP}_l(A, B)$ code, we say that its logical bases $L^M_Z$ in Eq.~\eqref{eq: main-LZ} and $L^M_X$ in Eq.~\eqref{eq: main-LX} are canonical if, up to column permutations over $R$, they contain a subset of rows of the following form, where $r_A := n_A - m_A$ and $r_B^* := m_B - n_B$:
    \begin{equation}
    \begin{aligned}
        \tilde{L}_Z^M & = \mathrm{vstack}\{x^m \left( I_{r_A}, G_A^0\right)\otimes_R \left(I_{r_B^*}, 0\right) \oplus 0\}_{m = 0}^{l - 1}; \\
        \tilde{L}_X^M & = \mathrm{vstack}\{x^{m^{\prime}} \left( I_{r_A}, 0\right)\otimes_R \left(I_{r_B^*}, G_B^{*0}\right) \oplus 0\}_{m^{\prime} = 0}^{l - 1},
    \end{aligned}
    \label{eq:canonical_basis}
    \end{equation}
    where $G_A^0 \in R^{r_A \times m_A}$ and $G_B^{*0} \in R^{r_B^* \times n_B}$. We can also rewrite these matrices in terms of their indexed rows:
    \begin{equation}
    \begin{aligned}
        \tilde{L}_Z^M & = \mathrm{vstack}\{\supp \bar{Z}_{i, j, m}\}_{i \in [r_A], j \in [r_B^*], m \in [l]}; \\
        \tilde{L}_X^M & = \mathrm{vstack}\{\supp \bar{X}_{i, j, m}\}_{i \in [r_A], j \in [r_B^*], m \in [l]},
    \end{aligned}
    \label{eq:indexed_basis}
    \end{equation}
    where $\supp \bar{Z}_{i, j, m}$ and $\supp \bar{X}_{i, j, m}$ denote the following support vectors over $R$, respectively:
    \begin{equation}\label{eq: main-canonical-basis-support-vectors}
    \begin{aligned}
       \supp \bar{Z}_{i, j, m} &:=  x^m u^{(i)}_{\mathcal{A}_1} \otimes_R e^{(j)}_{\mathcal{B}_0}, \\
       \supp \bar{X}_{i, j, m} &:= x^m e^{(i)}_{\mathcal{A}_1} \otimes_R v^{(j)}_{\mathcal{B}_0}
    \end{aligned}
    \end{equation}
    where $i \in [r_A]$ and $j \in [r_B^*]$. The vectors $u^{(i)}_{\mathcal{A}_1}$ and $e^{(i)}_{\mathcal{A}_1}$ are the $i$-th rows of $\left( I_{n_A - m_A}, G_A^0\right)$ and $\left( I_{n_A - m_A}, 0\right)$, respectively. The vectors $e^{(j)}_{\mathcal{B}_0}$ and $v^{(j)}_{\mathcal{B}_0}$ are the $j$-th rows of $\left(I_{m_B - n_B}, 0\right)$ and $\left(I_{m_B - n_B}, G_B^{*0}\right)$, respectively.
     We say that an $\mathrm{LP}_l(A, B)$ code is canonical if it admits a canonical logical basis in the above form.
   
\end{definition}

\begin{proposition}[Properties of canonical LP basis]\label{prop: main-property-canonical-basis}
For a canonical $\LP_l(A, B)$ code, the two logical-basis families $\{\bar{Z}_{i, j, m}\}$ and $\{\bar{X}_{i, j, m}\}$ satisfy the following properties (see Fig.~\ref{fig:main-overview}(a) and Fig.~\ref{fig: canonical-basis}):
   \begin{enumerate}[label=(\roman*)]
       \item (Conjugate pairs).
       The supports of $\bar{Z}_{i,j,m}$ and $\bar{X}_{i',j',m'}$ intersect if and only if $i=i'$, $j=j'$, and $m=m'$. In that case, their supports intersect on exactly one physical qubit $q_L(i, j, m)$, the $m$-th physical qubit of the $(i,j)$-th physical fibre of the left block. Thus, we have  $[\bar{Z}_{i,j,m}, \bar{X}_{i',j',m'}] = 2\bar{Z}_{i,j,m} \bar{X}_{i',j',m'}\delta_{i, i^{\prime}}\delta_{j, j^{\prime}}\delta_{m, m^{\prime}}$.
       \item (Cyclic symmetry). The $l$ $Z$-type (resp. $X$-type) logical operators sharing the same $(i, j)$ index are equivalent up to cyclic shifts:
       \begin{equation}
       \begin{aligned}
           \sigma_t:\quad \supp \bar{Z}_{i,j,m} &\mapsto \supp \bar{Z}_{i,j,(m+t)\bmod l}, \\
          \supp  \bar{X}_{i,j,m} &\mapsto \supp \bar{X}_{i,j,(m+t)\bmod l},
       \end{aligned}
       \end{equation} 
       where the physical cyclic shift $\sigma_t$ is induced {by} multiplication {by} $x^t$ {on} the LP code space $\mathcal{M}_1$. 
       \item (Row/column parallel structure). For any column permutation $\sigma_C$ of $[r_B^*]$ and any row permutation $\sigma_R$ of $[r_A]$, there exist corresponding physical permutations $P_C$ and $P_R$~\eqref{expr: main-column-permutation-physical}, acting on the supports of the canonical basis operators such that $\supp \bar{Z}_{i,\sigma_C(j),m} = P_C(\supp \bar{Z}_{i,j,m})$ and $\supp \bar{X}_{\sigma_R(i),j,m} = P_R(\supp \bar{X}_{i,j,m})$.
       
       \item ($ZX$-duality). When the code is symmetric, i.e. $B = A^*$, the $Z$- and $X$-type logical operators also take the same form. Concretely, there exists a physical permutation $\tau$, in the form of an involution, such that $\supp \bar{X}_{i, j, m} = \tau (\supp \bar{Z}_{j, i, m})$ for all $i, j, m$. Its action on the physical supports of canonical logical basis operators is given by
       \begin{align}\label{expr: main-tensor-swap-involution}
        \tau: \supp \bar{Z}_{i, j, m} \mapsto \supp \bar{X}_{j, i, m}
       \end{align}
       These symmetries are illustrated in Fig.~\ref{fig: canonical-basis}. 
   \end{enumerate}
\end{proposition}
For odd $l$, the following gives a necessary and sufficient condition for an LP code being canonical:
\begin{theorem}[Informal: Canonical LP codes]\label{thm: main-canonical-LP-basis}
    Let $\LP_l(A, B)$ be a lifted-product (LP) code with odd $l$. For $C \in \{A, B^*\}$, let $I^{(i)}_{C}$ denote {an} information set of $G^{(i)}_C$, defined as in Theorem~\ref{thm: main-computation-kernel-cokernel} for each component $R_{(i)}$. Then $\LP_l(A, B)$ admits a canonical LP basis if and only if these information sets can be chosen to align for $i=1, \ldots, t${. That is, there exist column permutations of} $A$ and $B^*$ such that $[r_A] \subseteq I_A^{(i)}$ and $[r^*_B] \subseteq I_{B^*}^{(i)}$. The resulting canonical sector contains $k_c:=r_A r_B^*l\leq k$ logical qubits; any additional logical qubits are residual.
\end{theorem}

See the formal statement and proof in Theorem~\ref{thm: formal-condition-canonical-LP}. 

\begin{example}[A small canonical LP code $\mr{LP}_5^{2\times 4}$]
    Take $l = 5$, for which $x^5 + 1 = b^{(1)} b^{(2)}$ with irreducible factors $b^{(1)} = 1 + x$ and $b^{(2)} = 1 + x + x^2 + x^3 + x^4$
    (so $R_{(1)} \cong \mbb{F}_2$ and $R_{(2)} \cong \mbb{F}_{16}$), and the monomial base matrices
    \begin{align}
    \label{eq:toy_code}
        A = \left( \begin{array}{cccc}
             1 & 1 & 1 & 1\\
             1 & x & x^2 & {x^3}
        \end{array} \right),
    \end{align}
    with $r_A = n_A - m_A = 2$. This defines a canonical LP code $\LP^{2 \times 4}_{5} = \LP_5(A, A^*)$ with parameters $[[100, 26, 4]]$. Applying Theorem~\ref{thm: main-LP-basis-characterization}, Theorem~\ref{thm: main-computation-kernel-cokernel}, and Theorem~\ref{thm: main-canonical-LP-basis}{ over $R_{(2)} \cong \mbb{F}_{16}$}, {we obtain the canonical basis} 
    \begin{align}
\begin{aligned}
\widetilde L_Z^M
 &=\operatorname{vstack}
   \left\{x^m(G_A\otimes_R E_{A^*})\oplus0\right\}_{m=0}^{4},\\
\widetilde L_X^M
 &=\operatorname{vstack}
   \left\{x^m(E_{A^*}\otimes_R G_A)\oplus0\right\}_{m=0}^{4}.
\end{aligned}
\end{align}
where
\begin{align}
\begin{aligned}
G_A &={\begin{pmatrix}
1&0&1+x^3+x^4&x^3+x^4\\
0&1&1+x^4&x^4
\end{pmatrix}},\\
E_{A^*} &= \begin{pmatrix}
    1 & 0 & 0 & 0 \\
    0 & 1 & 0 & 0
\end{pmatrix}
\end{aligned}
\end{align}
    This defines $20$ canonical logical qubits, and there are 6 \emph{residual} logical qubits arising from the $b^{(1)} = 1+x$ component.
    This toy code {is used in the illustrations of} {Figs.~\ref{fig:main-overview}--\ref{fig: parallel-surgery}}.
    Note that the number of residual logical qubits  {for the canonical LP codes studied in this paper} does not scale with the {lift} size $l$. 
\end{example}

{In fact, the necessary and sufficient condition for an LP code to admit a canonical basis does not depend on whether $R$ is semisimple. It depends instead on the rank properties of the binarized base matrices $A$ and $B^*$.}

{\begin{proposition}[Canonical LP codes with arbitrary lift size]\label{prop: main-LP-even-lifting-basis}
    Let $\LP_l(A, B)$ be a lifted-product code such that $\mathbb{B}(A)$ and $\mathbb{B}(B^*) = \mathbb{B}(B)^T$ are full row-rank. Then it admits the canonical logical basis of Eq.~\eqref{eq: main-canonical-basis-support-vectors} if and only if, up to column permutations over $R$, the binary coordinates of the first $r_A$ columns of $A$ are contained in an information set of $\ker \mathbb{B}(A)$, and those of the first $r^*_B$ columns of $B^*$ are contained in an information set of $\ker \mathbb{B}(B^*)$. In this case, it encodes precisely $k=k_c := r_A r^*_B l$ logical qubits.
\end{proposition}}

{This is proved in Theorem~\ref{thm: basis-canonical-product-code} of Section~\ref{sec: LP-canonical-basis}; see also the related work~\cite{Bhardwaj2026MittenCodes}. The full-row-rank assumption is a matter of convenience. In the general case, there are residual logical qubits, so that $k > k_c$ (see Section~\ref{sec: LP-canonical-basis}). Proposition~\ref{prop: main-LP-even-lifting-basis} directly unlocks several promising LP canonical codes presented in Example~\ref{example: main-soundness-codesign} and in Table~\ref{tab: high-rate-two-body-addressable}, with rigorous logical parallel addressability with explicit low-overhead constructions.} The properties of the canonical LP basis in Proposition~\ref{prop: main-property-canonical-basis} directly yield automorphism gadgets~\cite{Breuckmann_2024, sayginel2025faulttolerantlogicalcliffordgates} and fold-transversal gadgets~\cite{Breuckmann_2024, Quintavalle_2023} that exploit the code symmetries.

\begin{proposition}[Automorphism and fold-transversal gadgets]\label{prop: main-fold-transversal-canonical-LP}
    Let $\mathcal{Q}$ be a canonical LP code and suppose further that $B = A^*$ so that $\mathcal{Q} = \LP_l(A, A^*)$. Then the following depth-$1$ circuits implement the corresponding logical gadgets:
    \begin{enumerate}[label=(\roman*)]
        \item (Cyclic-shift automorphism). The physical permutation $ \sigma_t:q_{L/R}(i,j,m)\mapsto q_{L/R}(i,j,(m+t)\bmod l)$~{\eqref{expr: main-cyclic-shift-physical}} is implemented on each physical fibre in parallel
        and it induces the logical action
        \begin{align}
           \begin{aligned}
             \bar{\sigma}_t:\quad \bar{Z}_{i,j,m}&\mapsto\bar{Z}_{i,j,(m+t)\bmod l}, \\ 
            \bar{X}_{i,j,m}&\mapsto\bar{X}_{i,j,(m+t)\bmod l},
           \end{aligned}
        \end{align}
        which corresponds to a logical cyclic shift across different logical fibres in parallel: $\bar{\sigma}_t: \bar{q}_{i, j, m} \mapsto \bar{q}_{i, j, (m+t){\bmod} l}$. Note that this automorphism {gadget} also applies to {non-symmetric} canonical LP codes $\mr{LP}_l(A, B)$.
        \item (Fold-transversal $\overline{\mathrm{H\text{-}SWAP}}$ {gadget}). The logical Hadamard-fold $\overline{\mathrm{H\text{-}SWAP}}$ is implemented by a fold-transversal circuit,
        \begin{align}
            \tau \circ H^{\otimes n} , 
        \end{align}
        where $\tau:q_{L/R}(i,j,m)\mapsto q_{L/R}(j,i,m)$ transposes the indices on both physical blocks, as in {\eqref{expr: main-XZ-duality-involution-physical}}. The logical action is given by 
        \begin{align}
            \bar{Z}_{i, j, m} \mapsto \bar{X}_{j, i, m}; \quad \bar{X}_{i, j, m} \mapsto \bar{Z}_{j, i, m},
        \end{align}
        which corresponds to applying a Hadamard gate to every logical qubit, followed by a fold permutation.
        
        \item (Fold-transversal $\overline{\mathrm{S\text{-}CZ}}$ gadget). The following depth-$1$ fold-transversal circuit{ implements an $\overline{\mathrm{S\text{-}CZ}}$ gadget:}
        \begin{align}
            \bigotimes_{q \in D_L} S_q \cdot \bigotimes_{q \in D_R} S^{\dagger}_q \cdot \bigotimes_{q<\tau(q)} CZ_{q,\tau(q)},
        \end{align}
        where $D_L = \{q_L(i,i,m)\}_{i,m}$ and $D_R = \{q_R(i,i,m)\}_{i,m}$ are the fold-invariant qubits of the left and right block{s, respectively}, and $<$ is any fixed ordering of the physical qubits, so {that }each two-cycle of $\tau$ appears once{. The induced logical action is}
        \begin{align}
            \begin{aligned}
            \bar{Z}_{i, j, m} &\mapsto \bar{Z}_{i, j, m} \\
                \bar{X}_{i, j, m} &\mapsto \bar{X}_{i,j,m}\bar{Z}_{ j,i, m}.
            \end{aligned}
        \end{align}
        For $i = j$, the map is understood with the phase made explicit: $\bar{X}_{i,i,m} \mapsto i\,\bar{X}_{i,i,m}\bar{Z}_{i,i,m} = \bar{Y}_{i,i,m}$, i.e., the logical $\bar{S}$ gate on the fold-invariant logical fibres; for $i \neq j$, it implements a logical CZ gate between the paired logical qubits.
    \end{enumerate}
\end{proposition}

\section{Graph surgery}\label{sec: main-graph-surgery}
\begin{figure*}[t]
    \centering
    \includegraphics[width=\textwidth]{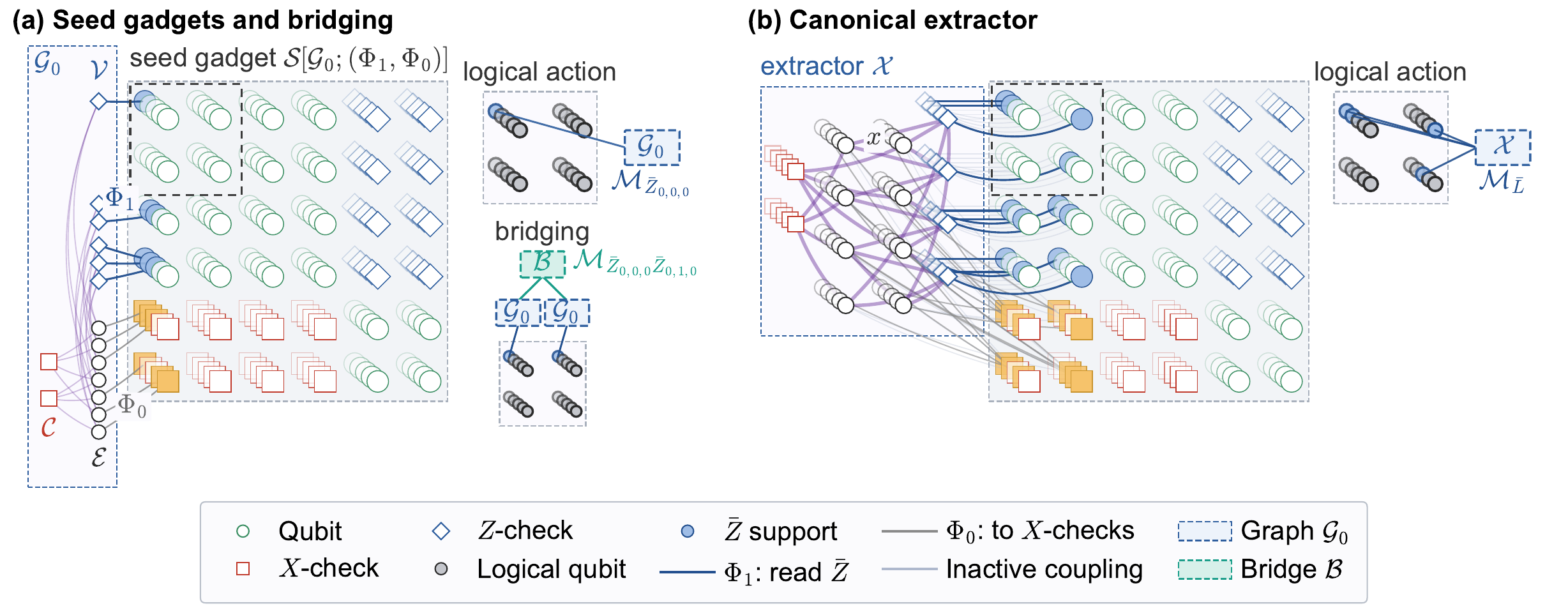}
    \caption{\textbf{Seed surgery gadgets and the canonical extractor}, drawn on the
    toy code $\LP^{2\times4}_{5}$ (same visual conventions).
    \textbf{(a) Seed surgery gadgets and bridging.} The seed surgery gadget
    $\mathcal{S}[\mathcal{G}_0;(\Phi_1,\Phi_0)]$
    (Theorem~\ref{thm: main-seed-logical-canonical-LPAA}) measures the seed logical
    $\bar Z_{0,0,0}$ (Definition~\ref{def: main-seed-logical-operators-gadget}). Its
    ancilla system is the graph $\mathcal{G}_0$ (blue dashed box), with vertex
    $Z$-checks $\mathcal{V}$ (diamonds), edge ancilla qubits $\mathcal{E}$ (circles),
    and cycle $X$-checks $\mathcal{C}$ (squares): the vertex checks read the
    operator's support (blue data qubits) transversally ($\Phi_1$, blue lines), while
    the ancilla qubits attach to the data $X$-checks ($\Phi_0$, gray lines; coupled
    checks in gold). 
    Re-wired by the basis symmetries of
    Fig.~\ref{fig: canonical-basis}, the same gadget measures many other canonical
    logicals; the upper inset shows the induced logical action, with the gadget
    abstracted to a box $\mathcal{G}_0$. Logical Pauli products are then measured by
    \emph{bridging} copies of the seed surgery gadgets (lower subpanel): a generalized bridge
    $\mathcal{B}$ (Section~\ref{subsec: main-LP-seed-logical}) couples two copies of $\mathcal{G}_0$ into the joint measurement
    $\mathcal{M}_{\bar Z_{0,0,0}\bar Z_{0,1,0}}$.
    \textbf{(b) Canonical extractor.}
    A single \emph{fixed} ancilla graph $\mathcal{X}$ measures an \emph{arbitrary}
    logical Pauli product $\bar L$, as in the inset's $\mathcal{M}_{\prod \bar Z}$: the
    full extractor couples to both the first two columns and the first two rows of
    physical fibres, addressing logicals of any type ($X$, $Z$, or mixed, e.g.\ $Y$);
    for simplicity, only the column couplings are drawn. Its
    ancilla graph is a cyclic (voltage) lift of a small base graph and inherits the
    code's symmetry: the vertices (ancilla $Z$-checks), edges (ancilla qubits), and
    cycles (ancilla $X$-checks) all come in cyclic fibres of $l$ copies, drawn as
    stacks, and each thickened edge stands for a whole fibre of couplings, weighted by
    a monomial (one example labeled: $x$). The faint lines show the static connectivity map{s} to the information columns,
    of which the adaptive masking activates only a target-dependent subset --- shown for a four-body logical $Z$-PPM supported on $12$ physical qubits spanning both
    columns.}
    \label{fig: seed-surgery}
\end{figure*}

The symmetry-based gates of the previous section, such as automorphism and fold-transversal gates, incur essentially no additional space--time overhead but offer limited addressability, as they cannot selectively target a subset of logical qubits. Code-surgery techniques~\cite{CohenKimBartlettBrown2022LongRangeConnectivity, cowtan2025fastfaulttolerantlogicalmeasurements, williamson2024lowoverheadfaulttolerantquantumcomputation, he2025extractorsqldpcarchitecturesefficient, cross2025improvedqldpcsurgerylogical}, generalized from lattice surgery for topological codes~\cite{horsman2012surface}, overcome this limitation by coupling the code to an ancilla system and realizing targeted logical Pauli-product measurements (PPMs). A particularly simple recent approach restricts the ancilla to a graph, yielding a fault-tolerant measurement of a logical operator \(\bar{L}\in\bar{\mc P}_k\) in \(O(d)\) QEC cycles. Using randomized graph-theoretic constructions~\cite{williamson2024lowoverheadfaulttolerantquantumcomputation, yuan2026parsimonious}, one can construct ancilla graphs of size \(O(|\bar{L}|\,\mathrm{polylog}(|\bar{L}|))\){.} 

Despite this near-linear asymptotic scaling, three practical challenges remain for large qLDPC codes: 
(i) {Potentially} large space costs in practice. The polylogarithmic factors and hidden constants could be substantial, requiring extensive graph augmentation and thickening to guarantee fault tolerance;
(ii) Lack of modularity. Different PPMs generally require distinct ancilla graphs. As logical circuits grow, the number of surgery gadgets that must be designed and constructed can become prohibitively large;
(iii) Difficult certification. The combination of large overheads and many distinct gadgets makes it challenging to obtain a logical instruction set whose surgery operations are comprehensively benchmarked and rigorously proven fault-tolerant.
    
In this section, we show that the symmetries of canonical LP codes help us overcome these obstacles. Rather than constructing a new ancilla graph for each logical measurement, {we assemble }arbitrary PPMs from a small collection of reusable graph-surgery primitives{. This yields} a modular, low-overhead, and fully certifiable logical toolkit.

We first fix a set of notations and provide a graph-surgery framework that is slightly more general than those in the literature. 
We refer the readers to Ref.~\cite{he2025extractorsqldpcarchitecturesefficient} for a more comprehensive review of the previous surgery frameworks.  
Let $\bar{L}$ be the logical operator that we wish to measure and $M_L$ be the corresponding restriction map to its physical qubit support. Then the surgery can be described as follows. 
Let $\mathcal{G}(\mathcal{V}, \mathcal{E}, \mathcal{C})$ be a graph with a cycle basis $\mathcal{C}$ as ancilla cycle checks, edges $\mathcal{E}$ as ancilla qubits, and vertices $\mathcal{V}$ as ancilla vertex checks. 
If we specify further that $\mathcal{G}(\mathcal{V}, \mathcal{E}, \mathcal{C})$ is a connected graph, then this defines a merged (deformed) code which is formed by attaching $\mc{G}$ to a data code $\mathcal{Q}$: 

\begin{definition}[Graph-surgery gadget]\label{def: main-graph-surgery-gadget}
Let $\mathcal{G}(\mathcal{V}, \mathcal{E}, \mathcal{C})$ be a sparse graph with a sparse cycle basis $\mathcal{C}$, edges $\mathcal{E}$, and vertices $\mathcal{V}$, forming an exact $2$-chain (graph chain Definition~\ref{def: graph-chains}): $\mathcal{G}: \mathbb{F}_2[\mathcal{V}] \xrightarrow{\partial_1} \mathbb{F}_2[\mathcal{E}]\xrightarrow{\partial_0} \mathbb{F}_2[\mathcal{C}]$ where $\partial_1$ and $\partial_0$ {denote, respectively,} the edge-vertex and cycle-edge incidence matrices. A surgery gadget, denoted $\mathcal{S}[\mathcal{G}; (\Phi_1, \Phi_0)]$, on a CSS \emph{data code}  $\mathcal{Q}: \mathcal{Q}_2 \xrightarrow{H^T_Z} \mathcal{Q}_1 \xrightarrow{H_X} \mathcal{Q}_0$ is constructed by connecting $\mathcal{V}$ to qubits $\mathcal{Q}_1$ via $\Phi_1 = \begin{pmatrix}
        \Phi_1^{(X)} \\ \Phi_1^{(Z)}
    \end{pmatrix}$ and connecting $\mathcal{E}$ to the $X$-checks $\mathcal{Q}_0$ and the $Z$-checks $\mathcal{Q}_2$ via, respectively, $\Phi_0^{(Z)}$ and $\Phi_0^{(X)}$, where we denote $\Phi_0 = \begin{pmatrix}
        \Phi_0^{(X)} \\ \Phi_0^{(Z)}
    \end{pmatrix}$. The \emph{merged code} is defined as the check matrix

    \begin{equation}\label{eq: main-merged-parity-check}
       H_{\text{merged}}=\quad
    \begin{NiceArray}{l@{\quad}cccc}[baseline=line-4]
    & \Block{1-2}{\mathcal{Q}_1} & & \Block{1-2}{\mathcal{E}} & \\
    \mathcal{Q}_0 & H_X & 0 & \Phi_0^{(Z)} & 0 \\
    \mathcal{Q}_2 & 0 & H_Z & \Phi_0^{(X)} &  0\\
    \mathcal{V}
    & (\Phi_1^{(X)})^T
    & (\Phi_1^{(Z)})^T
    & 0 & \partial_1^T \\
    \mathcal{C} & 0 & 0 & \partial_0 & 0
    \CodeAfter
    \SubMatrix({2-2}{5-5})[hlines=2,vlines=2]
    \end{NiceArray}\quad ,
    \end{equation}
where each row
$r=\left(r_X^{\mc{Q}_1},\, r_Z^{\mc{Q}_1}\mid r_X^{\mc{E}},\, r_Z^{\mc{E}}\right)$
represents a set of checks associated with the label on the left (e.g., the \(Z\)-type data checks \(\mathcal{V}\)). The corresponding Pauli \(X\) (resp. \(Z\)) operators are supported on the data qubits \(\mc{Q}_1\) according to \(r_X^{\mc{Q}_1}\) (respectively, \(r_Z^{\mc{Q}_1}\)), and on the ancilla qubits \(\mc{E}\) according to \(r_X^{\mc{E}}\) (respectively, \(r_Z^{\mc{E}}\)).
    In addition to $H_X H_Z^T = 0$, the commutativity of the merged-code stabilizers requires that
    \begin{equation}
    \label{eq: main-commutation_requirement}
    \begin{aligned}
        (\Phi_1^{(X)})^T \Phi_1^{(Z)} & = (\Phi_1^{(Z)})^T \Phi_1^{(X)}, \\
        H_Z \Phi_1^{(X)} + \Phi_0^{(X)} \partial_1 & = 0, \\
        H_X \Phi_1^{(Z)} + \Phi_0^{(Z)} \partial_1 & = 0.
    \end{aligned}
    \end{equation}
    We denote a surgery gadget (for a given code) with ancilla graph $\mathcal{G}$ and the \emph{connectivity map} $(\Phi_1, \Phi_0)$ by $\mathcal{S}[\mathcal{G}; (\Phi_1, \Phi_0)]$, and the measured logical operators are given by $\Phi_1(\ker \partial_1)$. See Fig.~\ref{fig: seed-surgery}(a) for an illustration. 
\end{definition}

    The notation  $\Phi_1 = \begin{pmatrix}
        \Phi_1^{(X)} \\ \Phi_1^{(Z)}
    \end{pmatrix}$ refers to the decomposition of the connectivity map into its $X$-type and $Z$-type components, where the upper-half denotes the $X$-type component and the lower-half denotes the $Z$-type component of the Pauli operator on $\mc{Q}_1$, connected to the same set of vertices in $\mc{V}$. Note that, contrary to $\Phi_1$, the components of $\Phi_0 = \begin{pmatrix}
        \Phi_0^{(X)} \\ \Phi_0^{(Z)}
    \end{pmatrix}$ are of the same $X$-type, which deform the original $H_Z$ and $H_X$, respectively. 
Hence, this should not be understood in the symplectic notation. As a result, the super-scripts on the components of the connectivity map refer to the type ($X$ or $Z$) of the logical operators to be measured. Throughout, both $\Phi_1$ and $\Phi_0$ are stacked with the $X$-superscript component on top; in the merged parity check Eq.~\eqref{eq: main-merged-parity-check}, $\Phi^{(Z)}_0$ enters the $X$-check rows and $\Phi^{(X)}_0$ the $Z$-check rows. 
The row blocks of Eq.~\eqref{eq: main-merged-parity-check} are the (deformed) $X$-checks and $Z$-checks (upper rows) and the added vertex and cycle checks (lower rows); the column blocks are, respectively, the $X$- and $Z$-components on the data qubits $\mathcal{Q}_1$, followed by the $X$- and $Z$-components on the ancilla qubits $\mathcal{E}$: the vertex checks act on the ancilla qubits by $Z$, while the deformations $\Phi_0$ and the cycle checks act by $X$. The commutation relation of the vertex checks requires that $(\Phi_1^{(X)})^T\Phi_1^{(Z)} = (\Phi_1^{(Z)})^T\Phi_1^{(X)}$, or, equivalently, $(\Phi_1^{(X)})^T\Phi_1^{(Z)} $ is symmetric over $\mathbb{F}_2$. 

For a matrix $H \in \mathbb{F}^{m \times n}_2$, we denote $\omega_{\mathrm
col}(H)$ and $\omega_{\mathrm
row}(H)$ to be, respectively, the maximum column weight and maximum row weight of $H$. We denote the \emph{degree} of a quantum code to be the maximum of its check (stabilizer generator) weights and, for each qubit, the number of $X$-type or $Z$-type checks acting on it (the qubit $X$-degree and qubit $Z$-degree). We refer to the \emph{added degree} of the surgery gadget as the increase in code degree of the merged code relative to the data code $\mathcal{Q}$. Note that this definition does not necessarily restrict us to CSS merged codes. 
To ensure fault tolerance, it is necessary to certify that the distance of the merged code is at least that of the data code $\mathcal{Q}$, which we denote by $d$ throughout this work. A typical sufficient condition is to certify that the (relative) Cheeger constant of the edge-vertex incidence matrix is at least $1$~\cite{Ide_2025, he2025extractorsqldpcarchitecturesefficient, williamson2024lowoverheadfaulttolerantquantumcomputation}, which works for connected graphs. The Cheeger constant can be understood as a form of (co)boundary expansion, and as a particular case of a more general quantity that we refer to as \emph{small-set soundness}: 

\begin{definition}[Small-set soundness]\label{def: main-small-set-soundness}
    We say that $\partial_1: \mathbb{F}^{n}_2 \rightarrow \mathbb{F}^m_2$ has a $(t, \rho_t)$-soundness if for all $u \in \mathbb{F}^{n}_2$ such that $|\partial_1 u| < t$,  
    \begin{align}
        |\partial_1 u| \geq \rho_t \, \mathrm{dist}(u, \ker \partial_1),
    \end{align}
    where $\mathrm{dist}(u, \ker \partial_1) = \min_{v \in \ker \partial_1} |u + v|$.
\end{definition}

If $t = m+1$ and $\partial_1$ specializes to the {edge-vertex incidence} matrix of a graph, we recover the standard {Cheeger constant. Beyond graph incidence matrices, the same quantity is more generally called soundness}. It turns out that, in many applications such as high-dimensional expanders~\cite{dinur2025expansionhigherdimensionalcubicalcomplexes}, the strict notion of soundness is not necessary for (proving) distance lower bounds. To this end, the small-set soundness (later we will simply say soundness) is a strict generalization of the Cheeger constant, as evidenced by Table~\ref{tab: LP-canonical-extractor-resources}, where, particularly for soundness larger than or equal to $2$ with $\omega_{\mathrm{col}}(\Phi_1) = 2$, the Cheeger constant is strictly less than $2$. Note that the soundness condition is also studied in the context of high-rate surgery~\cite{zheng2025highratesurgeryconstantoverheadlogical, CowtanHeWilliamsonYoder2025ParallelCodeSurgery}. For simplicity, we say that a surgery gadget $\mathcal{S}[\mathcal{G}; (\Phi_1, \Phi_0)]$ is $(d, \rho_d)$-sound if its graph edge-vertex incidence matrix $\partial_1$ is $(d, \rho_d)$-sound; such a gadget preserves the distance provided $\rho_d \geq \omega_{\mathrm{col}}(\Phi^{(X)}_1)$ and $\rho_d \geq \omega_{\mathrm{col}}(\Phi^{(Z)}_1)$ (see Lemma~\ref{lemma: distance-preserving-soundness}). 

A logical measurement cycle is completed by transversally measuring ancilla qubits $\mathcal{E}$ and performing residual corrections to the deformed checks by $\Phi_0$~\cite{williamson2024lowoverheadfaulttolerantquantumcomputation}, known as \emph{splitting}. We need to ensure the splitting can be designed appropriately so that fault-tolerant analysis of the logical cycle is reduced to that of the merged code. The detailed protocols and statements are summarized in Appendix~\ref{sec: code-surgery-general}. 
In this work, we simply assume that each surgery gadget takes $\Theta(d)$ rounds to ensure {a} reliable logical measurement{. Faster surgery has recently been developed}~\cite{baspin2025fastsurgeryquantumldpc, cowtan2025fastfaulttolerantlogicalmeasurements, chang2026constant}{. We leave its application to LP codes to future work}.

The ability of a surgery protocol to fault-tolerantly measure any given logical operator $\bar{L}$ can be summarized as follows: 

\begin{definition}[{Informal:} Surgery gadget desiderata]\label{def: main-surgery-gadgets-desiderata} 
    Let $ \mathcal{S}^{(s)}[\mathcal{G}; (\Phi_1, \Phi_0)]$ be a graph-surgery gadget for measuring a logical operator $\bar{L}$ of a data code $\mathcal{Q}$ with distance $d$.  We say that the graph-surgery gadget $ \mathcal{S}[\mathcal{G}; (\Phi_1, \Phi_0)]$ satisfies graph-surgery gadget desiderata with added degree $\delta$ and (small-set) soundness $\rho_d$ if:
    \begin{enumerate}[label=(\roman*)]
        \item {Its} connectivity map satisfies Eq.~\eqref{eq: main-commutation_requirement}, and $\Phi_1(\ker \partial_1) = \bar{L}$ ($\bar{L}$ is in symplectic form). 

        \item The graph edge-vertex incidence matrix $\partial_1$ of the surgery graph $\mc{G}$ is $(d, \rho_d)$-sound, with $\rho_d \geq \max(\omega_{\mathrm{col}}(\Phi_1^{(X)}), \omega_{\mathrm{col}}(\Phi_1^{(Z)}))$. 
       
        \item The degree of the merged code in Eq.~\eqref{eq: main-merged-parity-check} exceeds that of the data code by $\delta$. 
    \end{enumerate}
\end{definition}

{A formal and precise definition of surgery gadget(s) desiderata is given in Definition~\ref{def: formal-surgery-gadgets-desiderata}.} As shown in Section~\ref{sec: code-surgery-general}, these desiderata guarantee that the gadget in Definition~\ref{def: main-surgery-gadgets-desiderata}  is distance-preserving, and hence that the corresponding logical measurements are fault-tolerant.
Note that satisfying large small-set soundness is strictly easier than having a large Cheeger constant.
In addition, certifying (a lower bound of) the small-set soundness can be much more computationally efficient (Lemma~\ref{lemma: spectral-criterion-soundness}). 

Under the general graph-surgery framework described above, we are ready to present the co-designed surgery toolkit for the canonical LP codes.

\begin{table*}[!t]
\centering
\footnotesize
\setlength{\tabcolsep}{2.5pt}
\begin{tabular}{@{}l ccc cc c ccc@{}}
\toprule
 & \multicolumn{6}{c}{Surgery graph}
 & \multicolumn{3}{c}{Merged degree} \\
\cmidrule(lr){2-7} \cmidrule(l){8-10}
 & \multicolumn{3}{c}{Graph size} & \multicolumn{2}{c}{Cycles}
 & Soundness
 & \multirow{2}{*}{\shortstack{Stabilizer\\weight}}
 & \multirow{2}{*}{\shortstack{Qubit\\$X$-degree}}
 & \multirow{2}{*}{\shortstack{Qubit\\$Z$-degree}} \\
\cmidrule(lr){2-4} \cmidrule(lr){5-6} \cmidrule(lr){7-7}
Seed surgery gadget
 & $|\mathcal{V}|$ & $|\mathcal{E}|$ & $|\mathcal{C}|$
 & $\omega_{\mathrm{row}}(\mathcal{C})$ & $\omega_{\mathrm{col}}(\mathcal{C})$
 & $\rho_d$ & & & \\
\midrule
\multicolumn{10}{@{}l}{\emph{$\mathrm{LP}^{3\times5}_{33} = [[1122, 148, \le 20]]$}} \\
\addlinespace[2pt]
\quad $\bar{Z}_{0,0,0}$
 & $36$ & $77$ & $42$ & $5$ & $5$ & $\geq 1$
 & $9$\,{\scriptsize$(+1)$} & $5$\,{\scriptsize$(+0)$} & $5$\,{\scriptsize$(+0)$} \\
\addlinespace
\quad $\bar{Z}_{1,0,0}$
 & $44$ & $92$ & $49$ & $6$ & $6$ & $\geq 1$
 & $9$\,{\scriptsize$(+1)$} & $6$\,{\scriptsize$(+1)$} & $5$\,{\scriptsize$(+0)$} \\
\midrule
\multicolumn{10}{@{}l}{\emph{$\mathrm{LP}^{3\times7}_{75} = [[4350, 1224, \le 20]]$}} \\
\addlinespace[2pt]
\quad $\bar{Z}_{0,0,0}$
 & $104$ & $234$ & $131$ & $8$ & $7$ & $\geq 1$
 & $11$\,{\scriptsize$(+1)$} & $8$\,{\scriptsize$(+1)$} & $7$\,{\scriptsize$(+0)$} \\
\addlinespace
\quad $\bar{Z}_{1,0,0}$
 & $100$ & $231$ & $132$ & $7$ & $7$ & $\geq 1$
 & $12$\,{\scriptsize$(+2)$} & $8$\,{\scriptsize$(+1)$} & $7$\,{\scriptsize$(+0)$} \\
\addlinespace
\quad $\bar{Z}_{2,0,0}$
 & $92$ & $210$ & $119$ & $7$ & $7$ & $\geq 1$
 & $11$\,{\scriptsize$(+1)$} & $7$\,{\scriptsize$(+0)$} & $7$\,{\scriptsize$(+0)$} \\
\addlinespace
\quad $\bar{Z}_{3,0,0}$
 & $102$ & $234$ & $133$ & $8$ & $7$ & $\geq 1$
 & $12$\,{\scriptsize$(+2)$} & $8$\,{\scriptsize$(+1)$} & $7$\,{\scriptsize$(+0)$} \\
\bottomrule
\end{tabular}
\caption{\textbf{Seed surgery gadgets.} Explicit constructions of seed surgery gadgets for the $\LP^{3 \times 5}_{33}$ and $\LP^{3 \times 7}_{75}$ canonical LP codes, respectively.
According to Theorem~\ref{thm: main-seed-logical-canonical-LPAA}, we require only $2$ and $4$ seed logicals, respectively, whose ($1$-body) graph-surgery gadgets form the seed gadgets.
For each gadget, we report the numbers of vertices $\mc{V}$, edges $\mc{E}$, and cycles $\mc{C}$ in its surgery graph, and certify that its (small-set) soundness is at least $1$, guaranteeing distance preservation.
In addition, each surgery gadget results in a merged code with added degree $\delta=1$ for $\LP^{3 \times 5}_{33}$ and $\delta \leq 2$ for $\LP^{3 \times 7}_{75}$, respectively.}
\label{tab: seed-surgery-gadget-resources}
\end{table*}

\subsection{Bridging seed surgery gadgets}\label{subsec: main-LP-seed-logical}
Here, we present the first scheme for measuring an arbitrary logical operator $\bar{L}$ by bridging copies of $O(1)$ small ``seed surgery gadgets''. 
Bridging two surgery gadgets by gluing their ancilla graphs through a universal adapter~\cite{swaroop2024universal} implements a measurement of the product of the logical operators measured by the individual gadgets.
We refer to Refs.~\cite{swaroop2024universal, he2025extractorsqldpcarchitecturesefficient} for its general theory and the construction details. 

The idea of the seed surgery gadgets is the following: we partition a logical basis into several distinct \emph{orbits} by grouping basis operators related by \emph{transforming maps}. The \emph{seed} logical (basis) operators are representatives of the distinct orbits~\cite{webster2025explicitconstructionlowoverheadgadgets}.
Then, any surgery gadget $\mc{S}[\mc{G};(\Phi_1,\Phi_0)]$ that measures a seed logical operator and satisfies the desiderata of Def.~\ref{def: main-surgery-gadgets-desiderata} can be \emph{rewired} to measure any other logical operator in the same orbit: the graph $\mc{G}$ is kept fixed, while its {connectivity maps} $\Phi_1$ and $\Phi_0$ {to the data code }are composed with suitable {transforming maps}.

\begin{definition}[Informal: Seed logicals and seed surgery gadgets]\label{def: main-seed-logical-operators-gadget}
   Let $\mathcal{Q}$ be an $[[n,k,d]]$ quantum CSS code represented by the chain complex
   $\mathcal{Q}_2 \xrightarrow{H_Z^T} \mathcal{Q}_1 \xrightarrow{H_X} \mathcal{Q}_0$, and let $\mathfrak{S}^{(\mathrm{Z})} = \{\bar{Z}_i\}_{i \in [k]}$ and $\mathfrak{S}^{(\mathrm{X})} = \{\bar{X}_i\}_{i \in [k]}$ denote, respectively, its $Z$- and $X$-type logical basis sets~\footnote{Here, we require $\mathfrak{S}^{(\mathrm{Z})}$ and $\mathfrak{S}^{(\mathrm{X})}$ to form $k$ conjugate pairs of logical operators, thereby specifying the $k$ logical qubits. This requirement could be relaxed by requiring only that $\mathfrak{S}^{(\mathrm{Z})}$ (resp. $\mathfrak{S}^{(\mathrm{X})}$) be a minimal generating set for all logical $Z$-type (resp. $X$-type) operators}. Let
   $\mathfrak{S}_{\mathrm{seed}} \subseteq \mathfrak{S}^{(\mathrm{X})} \cup \mathfrak{S}^{(\mathrm{Z})}$
be a subset of these basis logical operators, and let
$
\mathfrak{X}_{\mathrm{seed}}
=
\left\{
\mathcal{S}^{(s)}
\bigl[
\mathcal{G}^{(s)};
(\Phi_1^{(s)},\Phi_0^{(s)})
\bigr]
\right\}_{\bar{L}^{(s)} \in \mathfrak{S}_{\mathrm{seed}}}
$
be a corresponding set of surgery gadgets, one for measuring each logical operator in $\mathfrak{S}_{\mathrm{seed}}$, such that every gadget satisfies Definition~\ref{def: main-surgery-gadgets-desiderata}. Here, we identify each logical Pauli operator $\bar{L}$ with its symplectic vector representation in $\mbb{F}_2^{2n}$. We say that $\mathfrak{S}_{\mathrm{seed}}$ is a complete set of \emph{seed logicals}, and that $\mathfrak{X}_{\mathrm{seed}}$ is the corresponding set of \emph{seed surgery gadgets}, if:
   for any basis logical $\bar{L}^{(i)} \in \mathfrak{S}^{(\mathrm{X})} \cup \mathfrak{S}^{(\mathrm{Z})}$, 
   \begin{enumerate}[label=(\roman*)]
       \item There exists some seed logical $\bar{L}^{(s)}$ and a transforming map $T^{(i, s)}= (T_1^{(i, s)}, T_0^{(i,s)})$, where $T_1^{(i, s)} \in \mbb{F}_2^{2n \times 2n}$ acts on the symplectic space of $n$-qubit Pauli operators, whereas $T_0^{(i, s)}$
       acts on the total check space, i.e., the direct sum of the $X$-check and $Z$-check spaces, whose logical action satisfies $\bar{L}^{(i)} = T_1^{(i, s)} \bar{L}^{(s)}$. 
       \item The transformation map $T^{(i,s)}$ is chosen so that the rewired surgery gadget $\mathcal{S}^{(i)}[\mathcal{G}^{(s)};(T_1^{(i,s)}\Phi_1^{(s)},T_0^{(i,s)}\Phi_0^{(s)})]$ fault-tolerantly measures $\bar{L}^{(i)}$.
   \end{enumerate}
 \end{definition}

{The motivation for constructing these seed surgery gadgets is as follows. With at most $O(k)$ seed surgery gadgets (one for each basis logical $Z$ or $X$ operator), one can measure an exponential number of logical Pauli operators. One simply duplicates and bridges the seed surgery gadgets and connects them to the code. Adaptive sequences of these measurements generate the $\exp(\tilde{\Theta}(k^2))$-size logical Clifford group.}

\begin{definition}[Body (logical weight) of a logical operator]\label{def: main-logical-weight-body}
    Let $\mathcal{Q}$ be a $[[n,k,d]]$ CSS quantum code equipped with logical basis sets
    $\mathfrak{S}^{(\mathrm{Z})}=\{\bar{Z}_i\}_{i\in[k]}$ and
    $\mathfrak{S}^{(\mathrm{X})}=\{\bar{X}_i\}_{i\in[k]}$.
    Every logical Pauli operator can be written, up to an overall phase, as
    \begin{align}\label{eq: generic-logical-operator}
        \bar{L}
        =
        \prod_{i\in J^{(Z)}} \bar{Z}_i
        \prod_{j\in J^{(X)}} \bar{X}_j,
    \end{align}
    for some index sets $J^{(Z)},J^{(X)}\subseteq[k]$.
    We define the \emph{body}, or \emph{logical weight}, of $\bar{L}$ with respect to this logical basis by
    \begin{align}
        b(\bar{L})
        :=
        |J^{(Z)}|+|J^{(X)}|.
    \end{align}
\end{definition}

\begin{proposition}[Informal: bridged surgery gadgets]\label{prop: bridge_get_full_Clifford-informal}
     Let $\mathcal{Q}$ be a (CSS) {LDPC} code {of degree at most a constant $c$ and} with distance $d$, with {a complete set of seed logicals $\mathfrak{S}_{\mathrm{seed}}$ (Definition~\ref{def: main-seed-logical-operators-gadget})}, and let $\mathfrak{X}_{\mathrm{seed}}$ be {the corresponding set of seed surgery gadgets,} each satisfying the surgery gadget desiderata Definition~\ref{def: formal-surgery-gadgets-desiderata}. Let $|\mathcal{V}|$, $|\mathcal{E}|$, and $|\mathcal{C}|$ denote the maximum numbers of vertices, edges, and cycles, respectively, among all seed surgery gadgets, and let $|\mathcal{E}_b|$ and $|\mathcal{C}_b|$ denote the maximum numbers of edges and cycles, respectively, over all bridges between pairs of seed gadgets.
    Then any logical operator $\bar{L} \in \bar{\mc{P}}_k$ of $\mc{Q}$ can be fault-tolerantly measured using at most $b(\bar{L})$ copies of the seed surgery gadgets, which are bridged together using at most $b(\bar{L}) - 1$ bridges, with the total size of the ancilla system being at most $b(\bar{L})\left( |\mathcal{V}| + |\mathcal{E}| + |\mathcal{C}|\right) + (b(\bar{L})-1)(|\mathcal{E}_b| + |\mathcal{C}_b|)$.
\end{proposition}

{The formal statement is Theorem~\ref{thm: bridge_get_full_Clifford} in Appendix~\ref{sec: FT-surgery-gadgets}. We first bridge the gadgets of the same type ($Z$-type or $X$-type) in Eq.~\eqref{eq: generic-logical-operator}. We then bridge the joint $X$-type gadget to the joint $Z$-type gadget where the two operators overlap, which covers $Y$-type and mixed-type measurements, using the (generalized) \emph{graph isomorphism} of Appendix~\ref{sec: bridging-techniques-graph-surgery}. Both steps preserve the $(d, \rho_d)$-soundness, and the merged code remains LDPC.}

Here, we utilize three classes of transforming maps for a $[[n, k, d]]$ symmetrical canonical LP code $\mr{LP}_l(A, A^*)$, induced by the symmetries of the canonical logical basis (Fig.~\ref{fig: canonical-basis}), to obtain a \emph{constant-size} seed gadget set: 
\begin{enumerate}[label=(\roman*)]
    \item \label{main-orbit-cyclicity} \textbf{Cyclic shifts.} 
    The cyclic-shift automorphisms (Eq.~\eqref{expr: main-cyclic-shift-physical}) naturally induce a class of transforming maps{,} as they {map} stabilizers to stabilizers and logicals to logicals. Specifically, let $\sigma_t \in \mbb{F}_2^{n\times n}$ be the cyclic permutation matrix on the physical qubits{. Let $\sigma_t^{\prime}$ be the cyclic permutation matrix on the $Z$- and $X$-checks, where $H_X \sigma_t = \sigma_t^{\prime}H_X$ and $H_Z \sigma_t = \sigma_t^{\prime}H_Z$. We take} 
    \begin{equation}
        T_1 = \sigma_t \oplus \sigma_t; \quad T_0 = \sigma_t^{\prime} \oplus \sigma_t^{\prime},      
    \end{equation}
    which maps $\bar{Z}_{i, j, m}$ to $\bar{Z}_{i, j, (m + t)\mod l}$.
    {Taken over all $l$ cyclic shifts, these maps identify} each fibre of $Z$ logicals (similarly for $X$ logicals) {as} a single orbit, i.e.{,} $\bar{Z}_{i, j, 0} \sim_{\sigma_t} \bar{Z}_{i, j, m}$ for any $m$. 
    \item\label{main-orbit-parallel-structure} \textbf{Column/row permutations.} 
    The column- (resp. row-) permutations induce another class of transforming maps on the $Z$- (resp. $X$-) logicals  {by permuting the logical supports and the check restrictions to those supports}.
    Specifically, for any column permutation $\sigma_C$ of $[r^*_B]$, the physical permutation $P_C$~\eqref{expr: main-column-permutation-physical}  {is extended by the identity outside the permuted left-block columns}. {For a $Z$-type seed gadget whose connectivity $\Phi_1^{(Z)}$ is supported on the left block, the corresponding permutation $P'_C$ of the $X$-checks satisfies}
    \[
        {H_X P_C\Phi_1^{(Z)}=P'_C H_X\Phi_1^{(Z)}.}
    \]
    {This restricted identity suffices for the gadget transformation, as in the remark following Lemma~\ref{lemma: classical-chain-maps}.} 
    {The permutation $P_C$} induces 
    \begin{equation}
        T_1 = P_C \oplus P_C; \quad 
        T_0 = I \oplus P_C^{\prime},       
    \end{equation}
    which maps $\bar{Z}_{i, j, m}$ to $\bar{Z}_{i, \sigma_C(j), m}$. 
    {Taken over} all such column permutations (similarly{,} row permutations){, these maps} identify {the }different columns (resp. rows) of $Z$- (resp. $X$-) logicals as a single orbit, i.e. $\bar{Z}_{i, 0, m} \sim \bar{Z}_{i, j, m}$ for any $j$ and $\bar{X}_{0, j, m} \sim \bar{X}_{i, j, m}$ for any $i$.
    Note that, although these column and row permutations are not code automorphisms, they remain valid transformation maps, going beyond the framework of Ref.~\cite{webster2025explicitconstructionlowoverheadgadgets}.
    \item \label{main-orbit-ZX-duality}\textbf{$ZX$-duality.} Finally, the $ZX$-duality $\tau$ (Eq.~\eqref{expr: main-tensor-swap-involution}) induces another class of transforming maps in between $Z$- and $X$-logicals.
    Specifically, let $\Psi_1 = \tau \in \mbb{F}_2^{n \times n}$ be the fold involution on the physical qubits, where {$H_Z \Psi_1 = \Psi_0 H_X$, with $\Psi_0$ mapping the $X$-checks to the $Z$-checks}{. Then $\Psi_1$ and $\Psi_0$ induce}
    \begin{equation}
        T_1 = \begin{pmatrix} 0 & \Psi_1 \\ \Psi_1 & 0 \end{pmatrix}; \quad 
        T_0 = \begin{pmatrix}
        0 & \Psi_0 \\
        {\Psi_0^T} & 0
    \end{pmatrix},      
    \end{equation}
    which maps $\bar{Z}_{i, j, m}$ to $\bar{X}_{j, i, m}$. This identifies each $X$-logical with the $Z$-logical of {its fold-paired} logical qubit as a single orbit, i.e.{,} $\bar{Z}_{i, j, m} \sim \bar{X}_{j,i,m}$ for any $i$, $j$, and $m$. 
\end{enumerate}

\begin{table*}[!ht]
\centering
\footnotesize
\setlength{\tabcolsep}{4pt}
\begin{tabular}{@{}l ccc c ccc ccc@{}}
\toprule
 & \multicolumn{4}{c}{Surgery graph}
 & \multicolumn{3}{c}{Merged size}
 & \multicolumn{3}{c}{Merged degree} \\
\cmidrule(lr){2-5} \cmidrule(lr){6-8} \cmidrule(l){9-11}
 & \multicolumn{3}{c}{Graph size} & Soundness
 & \multirow{2}{*}{Qubits} & \multirow{2}{*}{$X$-checks}
 & \multirow{2}{*}{$Z$-checks}
 & \multirow{2}{*}{\shortstack{Stabilizer\\weight}}
 & \multirow{2}{*}{\shortstack{Qubit\\$X$-degree}}
 & \multirow{2}{*}{\shortstack{Qubit\\$Z$-degree}} \\
\cmidrule(lr){2-4} \cmidrule(lr){5-5}
 & $|\mathcal{V}|$ & $|\mathcal{E}|$ & $|\mathcal{C}|$
 & $\rho_d$ & & & & & & \\
\midrule
$\mathrm{LP}^{3\times5}_{33}$
 & $165$ & $495$ & $363$ & $9/2$
 & $1617$\,{\scriptsize$(+495)$} & $858$\,{\scriptsize$(+363)$}
 & $660$\,{\scriptsize$(+165)$}
 & $11$\,{\scriptsize$(+3)$} & $9$\,{\scriptsize$(+4)$}
 & $5$\,{\scriptsize$(+0)$} \\
\bottomrule
\end{tabular}
\caption{\textbf{LP canonical extractors:} Constructed for $\LP^{3 \times 5}_{33}= [[1122, 148, \leq 20]]$ for fault-tolerantly measuring an \emph{arbitrary} logical operator, with certified (small-set) soundness $(d, \rho_d) = (20, 9/2)$ of the ancilla graph, and added degree $3$ in the worst case (asymmetric $Y$-type measurements). The parenthetical entries indicate the maximum increase of the stabilizer weight and the qubit $X$-/$Z$-degrees in the merged codes relative to those of the data codes. The reported merged degrees are in the worst-case; for a single-type ($\bar{Z}$-only or $\bar{X}$-only) measurement they reduce to $(10, 7, 5)$ for $\LP^{3 \times 5}_{33}$, i.e., added degree $2$. The soundness is certified by the spectral criterion of Lemma~\ref{lemma: spectral-criterion-soundness}; since $\rho_d = 9/2 \geq r_A = 2$, Theorem~\ref{thm: main-LP-extractor} applies and the extractor is distance-preserving. Small-set soundness is a finer criterion than expansion: the Cheeger constant of this graph is certified to lie in $[0.94, 1.27]$, below the level $r_A = 2$ that an expansion-based argument would require. Hence, LP canonical extractors stand as highly space-resource-optimal constructions with provable fault tolerance. {The full LP canonical extractors for all four codes, $\LP^{3\times5}_{33} = [[1122, 148, \leq 20]]$, $\LP^{2\times3}_{36} = [[468, 36, 20]]$, $\LP^{3\times5}_{28} = [[952, 112, 17]]$, and $\LP^{3\times7}_{75} = [[4350, 1224, \leq 20]]$, are reported in Table~\ref{tab: LP-canonical-extractor-numerics}.}}
\label{tab: LP-canonical-extractor-resources}
\end{table*}

The precise, general statements are stated and proved in Lemma~\ref{lemma: classical-chain-maps} in Section~\ref{sec: code-surgery-symmetry}, along with their fault-tolerance properties. 
Utilizing the above transforming maps, we can reduce the number of seed surgery gadgets to a surprisingly small set: 

\begin{theorem}[Seed logical operators and surgery gadgets for the canonical LP code $\LP_l(A, A^*)$]\label{thm: main-seed-logical-canonical-LPAA}
   Let $\mathcal{Q} = \LP_l(A, A^*)$ be a canonical LP code. Then, a set of seed logical operators (Definition~\ref{def: main-seed-logical-operators-gadget}) {consists} of 
    \begin{align}\label{eq: main-seed-logical-LPAA-canonical}
        \begin{aligned}
        \mathfrak{S}_{\mathrm{seed}} = \{&\bar{Z}_{i,0,0}\}_{i \in [r_A] }.
        \end{aligned}
    \end{align}
     
\end{theorem}

In particular, the number of seed logical operators is $r_A$, independently of the {lift} size. 

In short, the cyclic orbit~\ref{main-orbit-cyclicity} removes the degrees-of-freedom along the fibre dimension; the $ZX$-duality~\ref{main-orbit-ZX-duality} relates $Z$-type logical basis operators to $X$-type logical basis operators, and the column parallel structure~\ref{main-orbit-parallel-structure} reduces $r^2_A$ to $r_A$, obtaining Eq.~\eqref{eq: main-seed-logical-LPAA-canonical}.  

\begin{example}
We only require $2$ and $4$ seed logical basis operators for $\LP^{3 \times 5}_{33}: [[1122, 148(132), d\leq 20]]$ and $\LP^{3 \times 7}_{75}:[[4350, 1224(1200), d \leq 20]]$, respectively. A corresponding set of seed surgery gadgets for each code {is} given in Table~\ref{tab: seed-surgery-gadget-resources} (soundness $\rho_d \geq 1$).
These gadgets satisfy the surgery gadgets desiderata Definition~\ref{def: main-surgery-gadgets-desiderata} with added degree $\delta = 1$ for $\LP^{3 \times 5}_{33}$ and $\delta =2$ for $\LP^{3 \times 7}_{75}$.  

\end{example}

\subsection{Extractor}\label{subsec: main-LP-canonical-extractor}

Measuring high-weight logical Paulis can be more powerful than measuring low-weight Paulis.
As shown in the previous section, constructing a set of seed surgery gadgets for measuring single-qubit logicals, and then bridging them to measure multi-qubit logicals, is one way to achieve this.
However, such a method comes at the expense of a space overhead scaling linearly with the logical {weight (body)} of the {measured logical operator.} 
For {measuring arbitrary logical operators}, especially those of high logical weight, it is desirable to construct another type of surgery gadget, commonly referred to as {an extractor}~\cite{he2025extractorsqldpcarchitecturesefficient}, whose space overhead does not scale with {the logical weight of the measured operator}. Here, we further prefer an extractor that preserves the cyclic symmetry of the LP codes, meaning that it is itself an $R$-lifted graph {(Section~\ref{sec: spectral-voltage-graph-theory})}, {which is compatible with hardware implementations of the data code (Section~\ref{sec: LP-hardware}).}

\begin{definition}[Informal: Cyclic extractor]\label{def: main-extractor-desiderata} 
    Let $\mathcal{G}: \mathbb{F}_2[\mathcal{V}] \xrightarrow{\partial_1} \mathbb{F}_2[\mathcal{E}] \xrightarrow{\partial_0} \mathbb{F}_2[\mathcal{C}]$ be a graph chain and $\mathcal{Q}$ a {canonical} LP code with distance $d$. Then we say that $\mathcal{X}[\mathcal{G}; ({S}_1, {S}_0)]$ is a cyclic extractor with  \emph{static} connectivity map{s} $({S}_1, {S}_0)$ if: $(i)$ For every logical operator $\bar{L} = \bar{L}_X \bar{L}_Z$, there exists a connectivity map $(\Phi^{(L)}_1, \Phi^{(L)}_0) := (M^{(L)}_1 {S}_1,\;{S}_0{\odot} M^{(L)}_0)$ where
\begingroup\setlength{\arraycolsep}{3pt}
\begin{equation}\label{eq: main-restriction-map-extractor}
\resizebox{0.9\linewidth}{!}{$\displaystyle
    M^{(L)}_1 := \begin{pmatrix}
    M^{(L_X)}_1 & 0 \\
    0 & M^{(L_Z)}_1
\end{pmatrix}, \quad {S}_0{\odot} M^{(L)}_0 := \begin{pmatrix}
   {S}^{(X)}_0 {\odot} M^{({L_X})}_0 \\
   {S}^{(Z)}_0 {\odot} M^{({L_Z})}_0
\end{pmatrix}.
$}
\end{equation}
\endgroup
 {Here $M^{(L_X)}_1$ and $M^{(L_Z)}_1$ are diagonal projection matrices, while $M^{(L_X)}_0$ and $M^{(L_Z)}_0$ are binary masks of the same dimensions as $S^{(X)}_0$ and $S^{(Z)}_0$, respectively. The symbol $\odot$ denotes the entrywise Hadamard product of binary matrices.} $(ii)$ The {resulting} surgery gadget $\mathcal{S}[\mathcal{G}; (\Phi^{(L)}_1, \Phi^{(L)}_0)]$ measures $\bar{L}$, i.e., $\bar{L} = \Phi_1^{(L)}\ker(\partial_1)${ in the symplectic notation, where $M^{(L_X)}_1$ and $M^{(L_Z)}_1$ are the restrictions to the $X$- and $Z$-support components of $\bar{L}$, respectively. This gadget satisfies the surgery gadget desiderata of Definition~\ref{def: main-surgery-gadgets-desiderata}} with soundness $\rho_d$ and added degree at most $\delta$. $(iii)$ The edge-vertex and cycle-edge incidence matrices $\partial_1$ and $\partial_0$ are $R$-linear. 

\end{definition}

{This definition is formally stated by combining Definition~\ref{def: elementary-connectivity-surgery}, Definition~\ref{def: linear-extractor}, Definition~\ref{def: canonical-extractor-desiderata}, and the $R$-lifted graph of Section~\ref{sec: spectral-voltage-graph-theory}.}

Intuitively, an extractor utilizes a fixed ancilla graph to measure an arbitrary logical operator by only ``masking" (Eq.~\eqref{eq: main-restriction-map-extractor}) its connections to the data code. 
While directly constructing a full extractor that is attached to an entire qLDPC code according to the generic recipe in Ref.~\cite{he2025extractorsqldpcarchitecturesefficient} is costly, 
one can get more efficient constructions by exploiting the code structures, e.g. for HGP codes~\cite{blue2026full}. 

Here, we construct a cyclic extractor{,} tailored for the canonical LP codes{,} that (i) {preserves} the (cyclic) symmetry of the code{ and} (ii) has {a }small size and low degrees{,} by exploiting the row/column parallel structure and the {$ZX$-duality} of the canonical logical basis.
The key observation is that the $Z$ (resp. $X$) logical operators of a canonical LP code localize into different columns (resp. rows) and are all codewords of the classical code. Thus, we can start by constructing a small partial extractor for measuring arbitrary logical $Z$ operators within the first column of logical fibres (essentially a copy of the classical base code), which we refer to as {an} \emph{LP column-canonical extractor}.

\begin{definition}[{Informal: }Explicit construction of LP column-canonical extractor]\label{def: main-lp-column-canonical-extractor}
    Let $\mathcal{Q} = \LP_l(A, A^*)$ be a canonical LP code with $A \in R^{(n_A - r_A) \times n_A}$ and distance $d$, and fix a target (small-set) soundness $\rho_d \geq r_A$. For ease of presentation, we assume that all the entries of $A$ are nonzero monomials, but the construction generalizes readily.
    An LP column-canonical extractor $\mathcal{X}[\mathcal{G}; ({S}_1, {S}_0)]$ with  \emph{static} connectivity map{s} $({S}_1, {S}_0)$ is constructed by the following protocol.
    \begin{enumerate}
    \item \label{LP-extractor-vertices} (Add vertex checks). Take $\mc{V}$ to be the standard $R$-basis of $\mathcal{A}_1 = R^{n_A}$---one vertex fibre per bit (column) of $A$---so that $\mathbb{F}_2[\mathcal{V}] \cong \Bmat(\mathcal{A}_1)$ is in one-to-one correspondence with the first column of physical fibres of $\mc{Q}$.
    \item \label{LP-extractor-graphify} (Add edge qubits). Add edge fibres with $\mbb{F}_2[\mc{E}] \cong  \Bmat(R^{(n_A-r_A) n_A})$, and take $\partial_1 \in R^{(n_A-r_A) n_A \times n_A}$ to be
    \begin{equation}
        \partial_1 = \mr{vstack}\{(0, \cdots, A_{i, j}, A_{i, j+1}, \cdots, 0)\}_{i \in [n_A - r_A], j \in [n_A]},
        \label{eq: extractor_partial_1}
    \end{equation}
    with the cyclic convention $A_{i, n_A + 1} := A_{i, 1}$; that is, break each row $A_i$ of $A$ into its $n_A$ cyclically adjacent pairs of entries, and add each pair to the graph as an edge, so that every row closes into a cycle through its bits.
    Since every edge carries exactly two monomial entries, $\partial_1$ is $R$-linear and is the edge-vertex incidence matrix of an $R$-lift of a base graph on $n_A$ vertices, with vertex degree at most $2\omega_{\mathrm{col}}(A)$.
    \item \label{LP-extractor-boosting} (Soundness boosting).
    Add further monomial edges to $\mc{E}$ until $\partial_1$ is certified (Lemma~\ref{lemma: spectral-criterion-soundness}) $(d, \rho_d)$-sound. We additionally require the lifted graph $\Bmat(\partial_1)$ to be connected. 
    \item \label{LP-extractor-cycles} (Cycle checks). Choose a set of $R$-valued cycles $\mc{C}$---possibly over-complete, so as to remain $R$-linear---whose cycle-edge incidence matrix $\partial_0$ satisfies $\IM \partial_0^T = \ker \partial_1^T$. This completes the $R$-linear graph chain $\mathcal{G}: \mathbb{F}_2[\mathcal{V}] \xrightarrow{\partial_1} \mathbb{F}_2[\mathcal{E}] \xrightarrow{\partial_0} \mathbb{F}_2[\mathcal{C}]$.
    \item \label{LP-extractor-connectivity-map} (Static connectivity map{s}). Set ${S}^{(X)}_1 = 0$ and take ${S}^{(Z)}_1$ to be the transversal identification of the vertex fibres with the first column of physical fibres, so that $\omega_{\mathrm{col}}({S}_1) = \omega_{\mathrm{row}}({S}_1) = 1$. Set ${S}^{(X)}_0 = 0$ and take ${S}_0^{(Z)}$ such that each $X$-check in the first column (Fig.~\ref{fig: seed-surgery}(b)), corresponding to a check of $A$, is connected to all the edges (rows of $\partial_1$) decomposed from the same check of $A$ in Eq.~\eqref{eq: extractor_partial_1}; all other check fibres are unconnected.
    \end{enumerate}
\end{definition}

{The formal definition and construction are presented in Definition~\ref{def: lp-column-canonical-extractor}. }

\begin{example}[Column-canonical LP extractor for the toy code $\mr{LP}_{5}^{2\times 4}$]\label{example: column-canonical-toy-extractor}
We illustrate the protocol on the toy code $\LP^{2 \times 4}_5 = [[100, 26, 4]]$ of Eq.~\eqref{eq:toy_code}, for which $n_A = 4$, $r_A = 2$, and $R_5 = \mathbb{F}_2[x]/(x^5+1)$. Steps~\ref{LP-extractor-vertices} and~\ref{LP-extractor-graphify} give four vertex fibres, $\mathbb{F}_2[\mathcal{V}] \cong \Bmat(R_5^{4})$, and one $4$-cycle per check row of $A$: rows $1$--$4$ and $5$--$8$ of $\partial_1$ below are the cyclically adjacent pairs of $A_1$ and $A_2$, respectively, following Eq.~\eqref{eq: extractor_partial_1}. A single boost orbit (row $9$, Step~\ref{LP-extractor-boosting}) suffices for the spectral criterion Lemma~\ref{lemma: spectral-criterion-soundness} to certify $(4, \rho_4)$-soundness with $\rho_4 \geq 2 = r_A$:
\begin{align}
    \partial_1 = \left( \begin{array}{cccc}
        1 & 1 & 0 & 0 \\
        0 & 1 & 1 & 0 \\
        0 & 0 & 1 & 1 \\
        1 & 0 & 0 & 1 \\
        1 & x & 0 & 0 \\
        0 & x & x^{2} & 0 \\
        0 & 0 & x^{2} & x^{3} \\
        1 & 0 & 0 & x^{3} \\
        0 & 1 & x^{3} & 0
    \end{array}\right),
\end{align}
For Step~\ref{LP-extractor-cycles}, an over-complete $R$-linear cycle basis, preserving the cyclic symmetry, is given by the rows of
\begin{align}
    \partial_0 =
    \begin{pmatrix}
        1 & 1 & 1 & 1 & 0 & 0 & 0 & 0 & 0 \\
        1 & 0 & x^{3} & 0 & 0 & 0 & 0 & 1 & 1 \\
        1 & x & 0 & 0 & 1 & x^{4} & 0 & 0 & 0 \\
        0 & 1 & x & 0 & 0 & x^{4} & x^{3} & 0 & 0 \\
        0 & 1 & 0 & 0 & 0 & x & 0 & 0 & 1 + x^{2} \\
        0 & 0 & 0 & 1 & 1 & 0 & x^{2} & 0 & x
    \end{pmatrix}.
\end{align}
With the adaptive masking described below, this column extractor fault-tolerantly measures an arbitrary product of the $l r_A = 10$ canonical logical basis operators $\{\bar{Z}_{i, 0, m}\}_{i \in [r_A], m \in [l]}$, supported on the first column of the dashed $2 \times 2$ grid of logical fibres in Fig.~\ref{fig: seed-surgery}(b). The merged-code degree is $7$ for a single-body measurement and $8$ for an arbitrary $Z$-type product within the column, in both cases preserving the code distance exactly ($d_X = d_Z = 4$); the certified soundness is in fact $\rho_4 \geq 2.178$.
\end{example}

Having specified the fixed graph $\mc{G}$ and the static connectivity map{s} $({S}_1, {S}_0)$ in Definition~\ref{def: main-lp-column-canonical-extractor}, we now describe how to mask $({S}_1, {S}_0)$ adaptively into a valid surgery gadget for each logical $Z$ operator $\bar{L}_Z \in \langle \bar{Z}_{i, 0, m}\rangle_{i \in [r_A], m \in [l]}$ within the first column of logical fibres.
First, since $\supp \bar{L}_Z$ lies within the first column of physical fibres, the image of ${S}_1^{(Z)}$, we can take the mask $M_1^{(L_Z)}$ to be the diagonal projection onto $\supp \bar{L}_Z$, so that $M_1^{(L_Z)} {S}_1^{(Z)}(\ker{\partial_1}) = \bar{L}_Z$, i.e., the masked gadget measures $\bar{L}_Z$ and nothing else. Accordingly, we set
\begin{align}
    \Phi^{({L_Z})}_1 = \begin{pmatrix}
        0 & 0 \\
        0 & M^{(L_Z)}_1
    \end{pmatrix} \begin{pmatrix}
        0 \\
        {S}^{(Z)}_1
    \end{pmatrix}.
\end{align}
Second, since $\bar{L}_Z$ commutes with the stabilizers, every check of $A$ overlaps $\supp \bar{L}_Z$ on an even number of qubits; and since every check closes into a cycle through its bits in Eq.~\eqref{eq: extractor_partial_1}, its restriction to $\supp \bar{L}_Z$ can always be realized by an $R$-weighted alternating subset (a path matching) of its edges on a cycle, using at most $\lfloor n_A/2 \rfloor$ edges per check (Lemma~\ref{lemma: graphify-parity-check}). This path matching defines the mask $M_0^{(L_Z)}$ with
\begin{align}
   \Phi^{({L_Z})}_0 = \begin{pmatrix}
        0 \\
        {S}^{(Z)}_0 {\odot} M^{(L_Z)}_0
    \end{pmatrix},
\end{align}
under which the merged-code checks commute (Eq.~\eqref{eq: main-commutation_requirement}). The induced surgery gadget $\mc{S}[\mc{G}; (\Phi_1^{({L_Z})}, \Phi_0^{({L_Z})})]$ thus measures $\bar{L}_Z$, and does so fault-tolerantly by the certified $(d, \rho_d)$-soundness.

We now describe how to upgrade this partial extractor into a full extractor, capable of measuring an arbitrary logical operator $\bar{L}$ on all logical qubits, using the same graph $\mc{G}$ {(with some $R$-linear bridges)} and merely extending the {static} connectivity map{s} $({S}_1, {S}_0)$ to cover a larger subset of the code $\mc{Q}$.
First, the column extractor of Definition~\ref{def: main-lp-column-canonical-extractor} couples only to the first column of physical fibres. Since all columns share the same structure---each is checked identically by $A$---we can extend ${S}_1^{(Z)}$ by coupling the vertices $\mc{V}$ transversally to each of the first $r_A$ columns of physical fibres, and extend ${S}_0^{(Z)}$ by coupling the $X$-checks of each column to $\mc{E}$ exactly as in the first column; in other words, the extractor is rewired to couple to all $r_A$ columns simultaneously and identically {(Fig.~\ref{fig: seed-surgery}(b))}. The adaptive masking then proceeds column by column as above. 
This yields an extractor that measures an arbitrary logical $Z$ operator {$\bar{L}_Z$} on all $k_c$ logical qubits.
Finally, as detailed in Section~\ref{sec: LP-canonical-extractor} {of the} Appendix, we can further extend the connectivity map to measure logical operators of any type---$X$-type, mixed $XZ$-type, or $Y$-type---by exploiting the $ZX$-duality of the symmetric LP code $\mc{Q}$.
For instance, the logical $X$ operators {$\bar{L}_X$} within a row, together with the checks anti-commuting with them, are again governed by the structure of $A$; hence, {we attach} the same extractor graph $\mc{G}$ to that row, with the {connectivity maps} $\Phi_1^{({L_X})}$ and $\Phi_0^{({L_X})}$ isomorphic{ to} $\Phi_1^{({L_Z})}$ and $\Phi_0^{({L_Z})}$ {above via the involution $\tau$ in Eq.~\eqref{expr: main-tensor-swap-involution}}{. The resulting gadget} measures any logical $X$ operator {$\bar{L}_X$} within that row. With some additional care---certain vertex checks must measure $Y$ Paulis on the information qubits where the measured columns and rows intersect---we obtain a full extractor, coupled to both the first $r_A$ columns and the first $r_A$ rows of $\mc{Q}$, that measures an arbitrary logical operator:

\begin{theorem}[Informal: LP canonical extractor]\label{thm: main-LP-extractor}
      Let $\mathcal{Q} = \LP_l(A, A^*)$ be {a }symmetric canonical LP code and $\mathcal{X}[\mathcal{G}; ({S}_1, {S}_0)]$ {an} LP column-canonical extractor constructed according to Definition~\ref{def: main-lp-column-canonical-extractor}, which measures an arbitrary logical $Z$ operator within a column, with added degree at most $\delta$. 
      Then, by extending the {static} connectivity map{s} to both the first $r_A$ columns and the first $r_A$ rows, 
      the {resulting} full cyclic extractor {(}according to Definition~\ref{def: main-extractor-desiderata}) is capable of measuring an arbitrary $\bar{L} \in \bar{\mc{P}}_{k_c}${ and} is distance-preserving with added degree at most $\delta + r_A$. 
\end{theorem}

{The formal theorem is proven in Section~\ref{sec: LP-canonical-extractor} with Theorem~\ref{thm: formal-LP-extractor}. In theory, for a canonical LP code $\LP_l(A, A^*)$ with any constant-size prototype matrix $A$, i.e., with constant $r_A$, we can ensure sufficient soundness by applying Step~\ref{LP-extractor-boosting}.}

{Finally, we report the full LP canonical extractor for $\LP^{3 \times 5}_{33}$ in Table~\ref{tab: LP-canonical-extractor-resources}, and refer to Example~\ref{example: concrete-LP-extractor-degrees} in Appendix~\ref{sec: LP-canonical-extractor} for all four codes, $\LP^{3 \times 5}_{33}$, $\LP^{2\times3}_{36}$, $\LP^{3\times5}_{28}$, and $\LP^{3\times7}_{75}$ (see also Table~\ref{tab: LP-canonical-extractor-numerics}).}  {For all four codes, the soundness $\rho_d$ of $\partial_1$, certified by the spectral criterion (Lemma~\ref{lemma: spectral-criterion-soundness}), satisfies $\rho_d \geq r_A$, so that Theorem~\ref{thm: main-LP-extractor} applies and the extractors are distance-preserving. For $\LP^{3 \times 5}_{33}$, $\rho_d = 9/2$, and the added degree is $\delta = 2$ for single-type measurements and $\delta = 3$ in the worst case; for the other three codes, the worst-case added degrees are $\delta = 2$, $4$, and $6$, respectively.}

\begin{figure*}[!t]
    \centering
    \includegraphics[width=\textwidth]{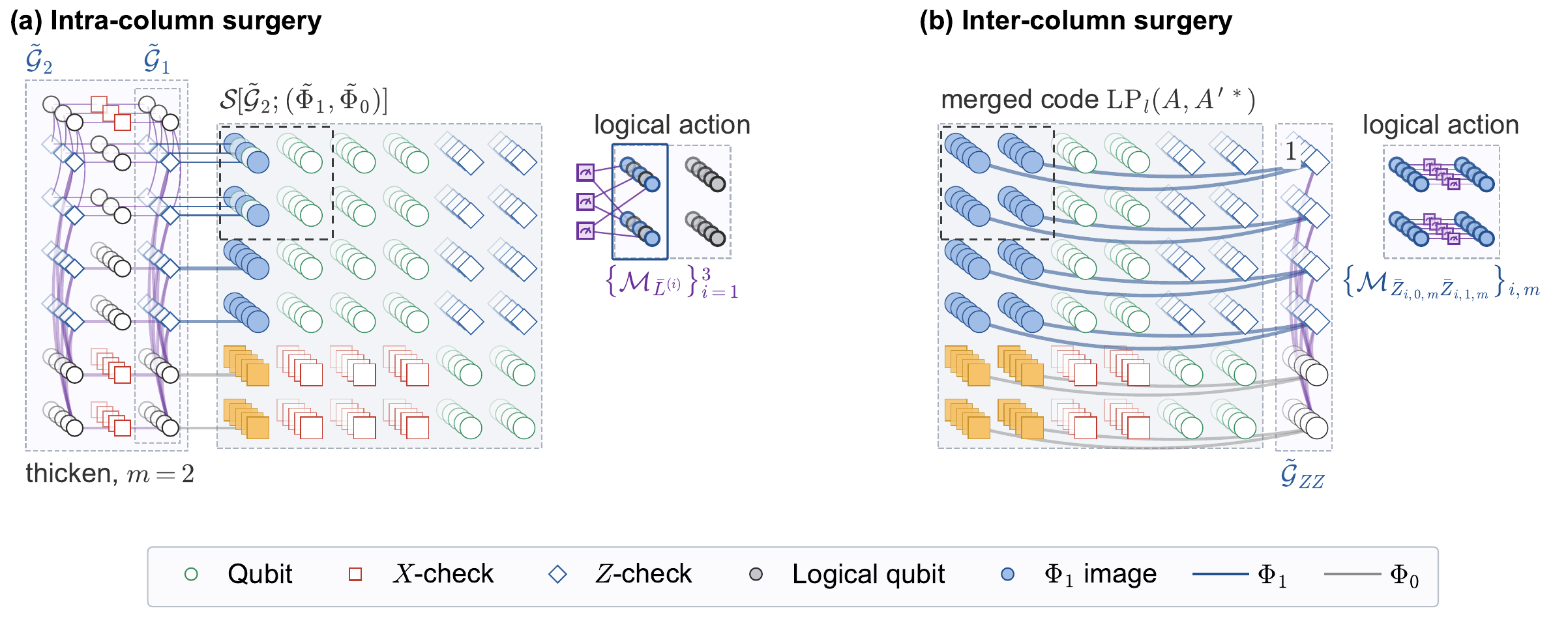}
    \caption{\textbf{Parallel hypergraph surgery on canonical LP codes}, drawn on the toy code $\LP^{2\times4}_{5}$. Visual conventions: data qubits highlighted blue form the image of $\Phi_1$, i.e.\ the qubits read by the ancilla vertex checks; each solid purple box in the insets is one logical PPM, supported on the logical qubits it is wired to; thick coupling edges denote $R$-linear (whole-fibre) couplings, thin ones individual binary couplings.
    \textbf{(a) Intra-column surgery} (Definition~\ref{def: main-LP-intracolumn-hyerpgraph-surgery}), measuring three logically disjoint two-body PPMs $\{\mathcal{M}_{\bar{L}^{(i)}}\}_{i=1}^{3}$ with $\bar{L}^{(i)}=\bar Z_{0,0,m}\bar Z_{1,0,m}$, $m\in\{0,2,4\}$, on column $j=0$ (inset: one PPM box per pair). The initial skeleton $\tilde{\mathcal{G}}_1$ (inner dashed box, adjacent to the code) reads the physical column $j=0$ via $\Phi_1$, its edge qubits attaching to the $X$-checks via $\Phi_0$. A targeted measurement set is not $R$-linear: the two information vertex fibres are binary-\emph{punctured} to the targeted layers, and one \emph{augmentation} qubit per pair couples (weight $2$, per layer) to the two surviving information vertex checks---removing the single-body operators from the measured set. The gadget is then \emph{thickened} (Definition~\ref{def: length-m-thickening}) to $\tilde{\mathcal{G}}_2$, a thickening of length $2$ chosen for the visualization here.
    \textbf{(b) Inter-column surgery} (Definition~\ref{def: main-LP-intercolumn-surgery}), measuring the pairs $\{\bar Z_{i,0,m}\bar Z_{i,1,m}\}_{i,m}$ between columns $j=0,1$ in parallel via the gadget $\tilde{\mathcal{G}}_{ZZ}$. The merged code is again an LP code, $\mathrm{LP}_l(A, A'^{\,*})$, where $A'$ is the classical code obtained by adding one extra check to $A$ (see the horizontal rows of (b))---the same mechanism as in Fig.~\ref{fig: magic-state-injection}.}
    \label{fig: parallel-surgery}
\end{figure*}
\section{Parallel logic \label{sec: main-LP-parallel-logic}}

While the previous section focused on reducing the \emph{space} cost and implementation complexity of logical operations, another important objective is to reduce their \emph{time} overhead. 
Here, we explore new routes to increase the logical \emph{throughput}, the number of logical operations executed per logical cycle, thereby enabling highly parallel logical circuits, which is particularly valuable for architectures with relatively slow clock speeds~\cite{Bluvstein2024LogicalProcessorAtomArrays}. 
For Pauli-based computation, where computation is driven by logical Pauli-product measurements (PPMs) and magic-state preparation~\cite{Litinski_2019, BravyiKitaev2005MagicStateDistillation}, throughput can be characterized by how many magic states can be prepared and how many PPMs can be implemented in parallel per logical cycle. 
Achieving such throughput while retaining high addressability---namely, the ability to execute flexible patterns of logical operations simultaneously---is particularly challenging for high-rate qLDPC codes. In this section, we exploit the structure of the canonical logical basis to develop highly parallel logical primitives, including parallel PPMs (Section~\ref{subsec: main-parallel-PPMs}) and parallel magic-state generation (Section~\ref{subsec: main-parallel-magic}), with provably fault-tolerant properties. 

\subsection{Parallel and addressable Pauli-product measurements \label{subsec: main-parallel-PPMs}}
Here, we present several logical instructions and techniques for implementing low-weight PPMs in parallel for the canonical LP code family. 
The most straightforward approach is to extend the graph-surgery techniques in the previous section, and
to use a disconnected set of graphs for measuring logical Pauli operators with disjoint/low-overlapping physical support. 

\begin{proposition}[Parallel surgery with disconnected graphs]\label{prop: main-parallel-low-rate}
      Let $\mathcal{Q}$ be a quantum code equipped with some set of graph-surgery gadgets $\{ \mathcal{S}^{(i)}[\mathcal{G}^{(i)}; (\Phi^{(i)}_1, \Phi^{(i)}_0)] \}$, fault-tolerantly measuring a set of logical Pauli operators $\mc{L} = \{\bar{L}^{(i)} \}$ correspondingly. 
  Then, we can measure $\mc{L}$ in parallel, using a surgery gadget {$\mc{S}[\mc{G}; (\Phi_1, \Phi_0)]$ whose graph is the} disjoint union of the individual graphs, i.e.,
   \begin{equation}\label{eq: main-union-of-graphs}
  \begin{aligned}
      \mc{G} & = \sqcup_i \mc{G}^{(i)}; \\
      \Phi_1 & = \mathrm{hstack}\{\Phi_1^{(i)}\};\\
      \Phi_0 &= \mathrm{hstack}\{\Phi_0^{(i)}\},
  \end{aligned} 
  \end{equation}
  provided that $\Phi_1$ and $\Phi_0$ satisfy the conditions in Eqs.~\eqref{eq: main-commutation_requirement} in Definition~\ref{def: main-graph-surgery-gadget}.
\end{proposition}
However, to ensure {the }LDPC property, we need to keep the row weights $\omega_{\mathrm{row}}(\Phi_1)$ and $\omega_{\mathrm{row}}(\Phi_0)$ {constant}. 
This requires the physical supports of the logical operators being measured, as well as those of their corresponding anticommuting checks, to have limited overlap. 
For canonical LP codes, the PPMs listed below are illustrative examples that, owing to the structure of the canonical logical basis, can be measured in parallel without increasing the degree at all.

\begin{itemize}
    \item (Mixed-type). If $\mathcal{L} = \{ \bar{Z}_{i, j, m}, \bar{X}_{i', j', m'} \}$ where either $i \neq i'$, $j \neq j'$, or $m \neq m'$, then both their supports and anticommuting checks are disjoint. 
    \item (Column/row parallel). If $\mathcal{L} =  \{ \bar{Z}_{i, j, m}, \bar{Z}_{i', j', m'} \}$ for $j \neq j'$, or, similarly, $\mathcal{L} =  \{ \bar{X}_{i, j, m}, \bar{X}_{i', j', m'} \}$ for $i \neq i'$, then their supports are disjoint.
\end{itemize}

More generally, see Lemma~\ref{lemma: distance-preserving-graph-parallel} for a systematic treatment. The above graph-surgery recipe is a minimal adaptation from Section~\ref{sec: main-graph-surgery}, where one increases the dimension of $\ker \partial_1$ by simply taking the disjoint union of individual graphs. 
Note that when these component graphs and gadgets overlap significantly, one could apply further techniques to ``branch" the code so that these gadgets end up being disjoint before sticking them together, allowing arbitrary (logically disjoint) logicals in principle~\cite{CowtanHeWilliamsonYoder2025ParallelCodeSurgery}.
However, a fundamental limitation of such generalized graph-surgery protocols is that the extra space cost grows with both the number of logical operators to measure in parallel and the logical weights of them, eventually blowing out the space costs for highly parallel and high-weight PPMs.

A complementary and more general approach is to allow hypergraphs as the ancilla system.

\begin{definition}[Hypergraph surgery]\label{def: main-hypergraph-surgery}
Let $\mathcal{Q}$ be a quantum (CSS) code. A hypergraph-surgery gadget $\mathcal{S}[\tilde{\mathcal{G}}; (\tilde{\Phi}_1, \tilde{\Phi}_0)]$ to $\mathcal{Q}$ is defined similarly to that in Definition~\ref{def: main-graph-surgery-gadget}, with the following additions/exceptions. 
    \begin{enumerate}[label=(\roman*)]
        \item $\tilde{\mathcal{G}}: \mathbb{F}_2[\mathcal{V}] \xrightarrow{\partial_1} \mathbb{F}_2[\mathcal{E}] \xrightarrow{\partial_0} \mathbb{F}_2[\mathcal{C}]$ is a generic $\mbb{F}_2$-chain, representing a hypergraph with vertices $\mathcal{V}$ (ancilla $Z$-type checks), hyperedges $\mathcal{E}$ (ancilla qubits), and faces/cycles $\mathcal{C}$ (ancilla $X$-type checks) .
        $\mc{G}$ is not necessarily an exact sequence. 
        \item $\dim \Phi_1(\ker \partial_1) \geq 1$ and every nonzero $ \bar{L} \in \Phi_1(\ker \partial_1)$ corresponds to a nontrivial logical operator of $\mathcal{Q}$. 
    \end{enumerate}
\end{definition}

When $\tilde{\mathcal{G}}$ fails to be exact, the ancilla chain supports nontrivial logical operators, which may remain nontrivial in the merged code. 

\begin{lemma}[Informal]\label{lemma: main-merged-code-dimension}
    Let $\mathcal{Q}$ encode $k(\mathcal{Q})$ logical qubits and $\mathcal{G}$ encode $k(\mathcal{G})$ logical qubits. Then the merged code Eq.~\eqref{eq: main-merged-parity-check} encodes at most $k(\mathcal{Q}) - \dim \Phi_1(\ker \partial_1)+ k(\mathcal{G})$ logical qubits. In particular, we can characterize any logical operator $\bar{L}$ in one of the following two types: 
\begin{enumerate}
    \item \label{main-surgery-logical-data} (Bare logical). $\bar{L}$ corresponds to a logical operator of {the }data code {that is neither measured nor anticommuting} with {the }measured logical operators in $\Phi_1(\ker \partial_1)$. 
    \item \label{main-surgery-logical-ancilla} (Gauge logical). {The support of $\bar{L}$} on {the }ancilla hyperedges $\mathbb{F}_2[\mathcal{E}]$ corresponds to a nontrivial logical operator of the ancilla hypergraph chain $\tilde{\mathcal{G}}$. 
\end{enumerate}
\end{lemma}

\begin{table*}[!t]
{
\centering
\footnotesize
\setlength{\tabcolsep}{1.0pt}
\resizebox{\textwidth}{!}{%
\begin{tabular}{@{}l l ccc c ccc@{}}
\toprule
 & & \multicolumn{4}{c}{Surgery hypergraph} & \multicolumn{3}{c}{Merged degree} \\
\cmidrule(lr){3-6} \cmidrule(l){7-9}
 & & \multicolumn{3}{c}{Ancilla size} & \multirow{2}{*}{Certificate} & \multirow{2}{*}{\shortstack{Stabilizer\\weight}} & \multirow{2}{*}{\shortstack{Qubit\\$X$-degree}} & \multirow{2}{*}{\shortstack{Qubit\\$Z$-degree}} \\
\cmidrule(lr){3-5}
\multicolumn{2}{@{}l}{LP intra-column gadgets} & $Z$-checks & Qubits & $X$-checks & & & & \\
\midrule
 \multirow{4}{*}{\shortstack[l]{Same-type\\single-body $\bar{Z}$}} & $\LP^{2\times4}_{16} = [[320, 64, 10]]$ & $384$ & $448$ & $128$ & $\rho_t = 1/3$ & $10$\,{\scriptsize$(+1)$} & $6$\,{\scriptsize$(+0)$} & $6$\,{\scriptsize$(+0)$} \\
  & $\LP^{2\times3}_{36} = [[468, 36, 20]]$ & $432$ & $612$ & $216$ & $\rho_t = 1/4$ & $11$\,{\scriptsize$(+1)$} & $5$\,{\scriptsize$(+0)$} & $6$\,{\scriptsize$(+1)$} \\
  & $\LP^{3\times5}_{20} = [[680, 80, 13]]$ & $600$ & $760$ & $240$ & $\rho_t = 1/3$ & $11$\,{\scriptsize$(+1)$} & $6$\,{\scriptsize$(+0)$} & $6$\,{\scriptsize$(+0)$} \\
  & $\LP^{3\times5}_{28} = [[952, 112, 17]]$ & $1120$ & $1512$ & $504$ & $\rho_t = 1/4$ & $10$\,{\scriptsize$(+1)$} & $5$\,{\scriptsize$(+0)$} & $5$\,{\scriptsize$(+0)$} \\
\midrule
 \multirow{4}{*}{\shortstack[l]{Same-type\\two-body}} & $\LP^{2\times4}_{16} = [[320, 64, 10]]$ & $384$ & $832$ & $672$ & $\rho_t = 1/3$ & $10$\,{\scriptsize$(+1)$} & $7$\,{\scriptsize$(+1)$} & $6$\,{\scriptsize$(+0)$} \\
  & $\LP^{2\times3}_{36} = [[468, 36, 20]]$ & $432$ & $972$ & $774$ & $\rho_t = 1/4$ & $12$\,{\scriptsize$(+2)$} & $10$\,{\scriptsize$(+5)$} & $6$\,{\scriptsize$(+1)$} \\
  & $\LP^{3\times5}_{20} = [[680, 80, 13]]$ & $600$ & $1360$ & $1120$ & $\rho_t = 1/3$ & $11$\,{\scriptsize$(+1)$} & $9$\,{\scriptsize$(+3)$} & $6$\,{\scriptsize$(+0)$} \\
  & $\LP^{3\times5}_{28} = [[952, 112, 17]]$ & $1120$ & $2632$ & $2240$ & $\rho_t = 1/4$ & $11$\,{\scriptsize$(+2)$} & $9$\,{\scriptsize$(+4)$} & $5$\,{\scriptsize$(+0)$} \\
\midrule
 \multirow{3}{*}{\shortstack[l]{Same-type\\one column\\($66$ logical qubits)}} & $\LP^{3\times5}_{33} = [[1122, 148, \leq 20]]$, single-body & $1630$ & $2457$ & $891$ & $\lambda(\partial_1) = 2$ & $9$\,{\scriptsize$(+1)$} & $5$\,{\scriptsize$(+0)$} & $5$\,{\scriptsize$(+0)$} \\
  & $\LP^{3\times5}_{33} = [[1122, 148, \leq 20]]$, joint & $1630$ & $3757$ & $2711$ & $\lambda(\partial_1) = 2$ & $9$\,{\scriptsize$(+1)$} & $7$\,{\scriptsize$(+2)$} & $5$\,{\scriptsize$(+0)$} \\
  & $\LP^{3\times5}_{33} = [[1122, 148, \leq 20]]$, single-body, numerical estimate & $910$ & $1473$ & $594$ & numerical & $9$\,{\scriptsize$(+1)$} & $5$\,{\scriptsize$(+0)$} & $5$\,{\scriptsize$(+0)$} \\
\midrule
 \multirow{4}{*}{\shortstack[l]{Two-body\\Clifford}} & $\LP^{2\times4}_{16} = [[320, 64, 10]]$ & $576$ & $1152$ & $800$ & $\rho_t = 1/3$ & $10$\,{\scriptsize$(+1)$} & $7$\,{\scriptsize$(+1)$} & $6$\,{\scriptsize$(+0)$} \\
  & $\LP^{2\times3}_{36} = [[468, 36, 20]]$ & $720$ & $1476$ & $990$ & $\rho_t = 1/4$ & $12$\,{\scriptsize$(+2)$} & $10$\,{\scriptsize$(+5)$} & $6$\,{\scriptsize$(+1)$} \\
  & $\LP^{3\times5}_{20} = [[680, 80, 13]]$ & $960$ & $1960$ & $1360$ & $\rho_t = 1/3$ & $11$\,{\scriptsize$(+1)$} & $9$\,{\scriptsize$(+3)$} & $6$\,{\scriptsize$(+0)$} \\
  & $\LP^{3\times5}_{28} = [[952, 112, 17]]$ & $1792$ & $3808$ & $2744$ & $\rho_t = 1/4$ & $11$\,{\scriptsize$(+2)$} & $9$\,{\scriptsize$(+4)$} & $5$\,{\scriptsize$(+0)$} \\
\bottomrule
\end{tabular}
}
\caption{\textbf{Worst-case resources of parallel hypergraph surgery.} \textbf{Same-type}: for the co-designed canonical LP codes of Example~\ref{example: main-soundness-codesign}, the worst case over all logically disjoint sets of single-body and of two-body same-type PPMs, measured in one logical cycle by the hypergraph surgery of Definition~\ref{def: main-LP-intracolumn-hyerpgraph-surgery} on all columns (resp.\ rows) of the logical fibres; for $\LP^{3\times5}_{33}$, one column, i.e.\ half ($66$) of its $k_c = 132$ canonical logical qubits, for single-body and for joint PPMs of any number of bodies (the last row of this block is estimated numerically). \textbf{Two-body Clifford}: the worst case over all logically disjoint sets of two-body Clifford PPMs (products on at most two logical qubits from the menu $\{\bar{X}, \bar{Z}, \bar{Y}, \bar{X}\bar{X}, \bar{Z}\bar{Z}, \bar{X}\bar{Z}\}$), measured by the bridged hypergraph surgery of Theorem~\ref{thm: main-bridged-hypergraph-surgery-informal} (formally Theorem~\ref{thm: bridged-hypergraph-parallel-surgery} and Lemma~\ref{lemma: hypergraph-symmetrization}); arbitrary 2-PPMs ($\bar{X}\bar{Y}$, $\bar{Z}\bar{Y}$, $\bar{Y}\bar{Y}$ included) and sets of any number of bodies are reported in Table~\ref{tab: high-rate-arbitrary-resources}. The ancilla size lists the numbers of vertex checks ($|\mathcal{V}|$, ancilla $Z$-checks), hyperedges ($|\mathcal{E}|$, ancilla qubits) and cycles ($|\mathcal{C}|$, ancilla $X$-checks) of the thickened hypergraph chain. Parentheses report the increase over the data code. Same-type PPMs cost roughly 1--3 times the code size in qubits, two-body Clifford PPMs roughly 3--4 times. The co-designed codes' prototype matrices (used for constructing intra-column (resp. row) hypergraph-surgery gadgets) are at least $(t, 1/\gamma_A)$-sound (Lemma~\ref{lemma: unconditional-soundness-canonical-matrix}) for $t {>} 1$, so a thickening of length $1/\rho_t = \gamma_A = 3, 4, 3, 4$--in the row order--with multiplicity $t = 4, 5, 5, 5$ preserves the distance (Theorem~\ref{thm: main-high-rate-surgery-distance-preserving}(ii)); $\LP^{3\times5}_{33}$ has $\lambda(\partial_1) = 2$ and is thickened to length $10$ (Theorem~\ref{thm: main-high-rate-surgery-distance-preserving}(iii)).}
\label{tab: high-rate-two-body-addressable}
}
\end{table*}

The formal statement and proof are given in Lemma~\ref{lemma: surgery-cone-code-logical-operators}. Guaranteeing distance preservation for a hypergraph-surgery gadget is much more demanding than that for a graph gadget. 
In Ref.~\cite{CohenKimBartlettBrown2022LongRangeConnectivity}, this is achieved by starting from {an} initial ``skeleton'' hypergraph and then \emph{thickening} it {to }$m \geq d$ layers{ via} a hypergraph product with a length-$m$ line graph. 
We refer to the surgery gadget thickened by a length-$m$ line graph as the \emph{$m$-thickened surgery gadget}. 
In practice, however, it is not necessary to take $m \geq d$; rather, we now prove that the thickening length depends on the properties of the initial skeleton hypergraph.  

\begin{theorem}[Distance-preserving conditions for hypergraph surgery]\label{thm: main-high-rate-surgery-distance-preserving}
  Let $\mathcal{S}[\tilde{\mathcal{G}}; (\tilde{\Phi}_1, \tilde{\Phi}_0)]$ be a (hypergraph) surgery gadget to a data code $\mathcal{Q}$ with distance $d${, such that} $\Phi_1$ has column-weight at most $1$, {and }with (hyper)graph chain $\tilde{\mathcal{G}}: \mathbb{F}_2[\mathcal{V}] \xrightarrow{\partial_1} \mathbb{F}_2[\mathcal{E}] \xrightarrow{\partial_0} \mathbb{F}_2[\mathcal{C}]$, and let the $Z$-distance of $\tilde{\mathcal{G}}$ be $d(\tilde{\mathcal{G}})$, with $d(\tilde{\mathcal{G}}) := \infty$ when $\ker \partial_0 \setminus \IM \partial_1 = \emptyset$ (in particular, when there is no logical operator of type~\ref{main-surgery-logical-ancilla}). Let the length-$m$ thickened hypergraph be $\tilde{\mathcal{G}}_m$, and the length-$m$ thickened hypergraph-surgery gadget $\mc{S}[\tilde{\mc{G}}_m; (\tilde{\Phi}_1, \tilde{\Phi}_0)]$, measuring $ \Phi_1(\ker \partial_1)$ in parallel. Then the following statements hold. 
    \begin{enumerate}[label=(\roman*)]
        \item The $m$-thickened surgery gadget $\mathcal{S}[\tilde{\mathcal{G}}_m; (\tilde{\Phi}_1, \tilde{\Phi}_0)]$ measures $ \Phi_1(\ker \partial_1)$ in parallel and is distance-preserving whenever $m \geq d$. 
        \item Suppose that $m\, d(\tilde{\mathcal{G}}) \geq d$ and that $\mathcal{S}[\tilde{\mathcal{G}}; (\Phi_1, \Phi_0)]$ is $(t, \rho_t)$-sound for some $t$; then the $m$-thickened surgery gadget $\mathcal{S}[\tilde{\mathcal{G}}_m; (\tilde{\Phi}_1, \tilde{\Phi}_0)]$ is distance-preserving whenever $mt \geq d$ and $m\rho_t \geq 1$.  
        
        \item Suppose $m\, d(\tilde{\mathcal{G}}) \geq d$.  
        Then $\mc{S}[\tilde{\mc{G}}_m; (\tilde{\Phi}_1, \tilde{\Phi}_0)]$ is distance-preserving if 
        \begin{align}
             m \lambda(\partial_1) \geq d
        \end{align}
        where $ \lambda(\partial_1) := \min_{u \notin \ker \partial_1} |\partial_1 u|$.
       
    \end{enumerate}
\end{theorem}
The above is formally stated in Theorem~\ref{thm: formal-surgery-distance-preserving} and proved in Section~\ref{sec: surgery-thickening}. Note that {Theorem}~\ref{thm: main-high-rate-surgery-distance-preserving} provides an explicit, finite-size criterion for preserving the code distance without the need to thicken exactly $d$ layers. 
According to (ii), the thickening length can be reduced below $d$ if $\partial_1$ of the skeleton hypergraph has large $(t, \rho_t)$ soundness; according to (iii), this can also be achieved whenever $\lambda(\partial_1)$ is larger than $1$.  

Based {on }hypergraph {surgery}, we first introduce a gadget that measures an arbitrary set of logically disjoint $Z$- (resp. $X$-) PPMs within a column (resp. row) of logical fibres in parallel. 

\begin{definition}[LP intra-column hypergraph surgery]\label{def: main-LP-intracolumn-hyerpgraph-surgery}
    Let $\mathcal{Q} = \LP_l(A, B)$ be a canonical LP code. We say a surgery gadget $\mathcal{S}[\tilde{\mathcal{G}}; (\tilde{\Phi}_1, \tilde{\Phi}_0)]$ is a $Z$-type LP intra-column hypergraph-surgery gadget (see Fig.~\ref{fig: parallel-surgery}(a)) if it is constructed as follows: 
    \begin{enumerate}
         \item \label{construction-initial-skeleton}\textbf{Constructing the initial skeleton.} Start with a classical $1$-chain-map condition:
        \begin{equation}\label{diagram: classical-chain-code-modification}
             \begin{tikzcd}
    	{\mbb{F}_2[\mathcal{V}]} & {\mbb{F}_2[\mathcal{E}]} \\
    	{\mathbb{B}(\mathcal{A}_1)} & {\mathbb{B}(\mathcal{A}_0)}
    	\arrow["{\partial_1}", from=1-1, to=1-2]
    	\arrow["{\phi_1}"', from=1-1, to=2-1]
    	\arrow["{\phi_0}", from=1-2, to=2-2]
    	\arrow["{\mathbb{B}(A)}"', from=2-1, to=2-2]
         \end{tikzcd}.
        \end{equation}
        Namely, the above forms a commuting diagram. 
        Construct a skeleton hypergraph $\mathcal{G}: \mathbb{F}_2[\mathcal{V}] \xrightarrow{\partial_1} \mathbb{F}_2[\mathcal{E}] \xrightarrow{\partial_0} \mathbb{F}_2[\mathcal{C}]$, with connectivity map $\{\Phi^{(Z)}_1, \Phi^{(Z)}_0\}$ {attached} to {the }$j$-th column for {$j \in [r^*_B]$}. {Restricted to the} physical qubits on the $j$-th column, $\Phi^{(Z)}_1 = \phi_1 \otimes e^{(j)}_{\mathcal{B}_1}$ and $\Phi^{(Z)}_0 = \phi_0 \otimes e^{(j)}_{\mathcal{B}_0}$, for unit vectors $e^{(j)}_{\mathcal{B}_1} \in \mc{B}_1$ and $e^{(j)}_{\mathcal{B}_0} \in \mc{B}_0$. We set $\partial_0 = 0$ and assume that there is no logical operator of type~\ref{main-surgery-logical-ancilla} (this holds for canonical LP codes with monomial prototype matrix $A$ or/and $B$, see Lemma~\ref{lemma: monomial-matrix-no-ancilla-logical}). This gives the initial surgery gadget $\mathcal{S}[\tilde{\mathcal{G}}_1; (\tilde{\Phi}_1, \tilde{\Phi}_0)]${. For bookkeeping, we write} $\tilde{\mc{G}}_1 \equiv \tilde{\mc{G}}${ for this} length-$1$ thickening.   
      
        \item \textbf{Thickening.} We thicken the initial skeleton to {an} $m$-thickened surgery gadget $\mathcal{S}[\tilde{\mathcal{G}}_m; (\tilde{\Phi}_1, \tilde{\Phi}_0)]$. See, for example, Ref.~\cite{he2025extractorsqldpcarchitecturesefficient}, in Definition~\ref{def: length-m-thickening} and Section~\ref{sec: surgery-thickening}.
    \end{enumerate}
    
\end{definition}

 {To construct the initial skeleton, which measures any targeted set of logical Pauli operators associated with the classical codewords of $\mathbb{B}(A)$, we use the following code modification techniques~\cite{xu2025fast, zheng2025highratesurgeryconstantoverheadlogical}:}

 We first take $\partial_1 = \mathbb{B}(A)$ so that the hypergraph is identical to the classical base code $A$. This, however, would measure all $r_A l$ single-qubit logical operators of $\mathbb{B}(A)$ according to $\phi_1 \ker{\partial_1}$. To perform addressable measurements, we apply the puncturing (removing columns of $\partial_1$) and augmenting (adding weight-$2$ rows to $\partial_1$) techniques to manipulate the kernel of $\partial_1$, so that it gives the target set of logical operators~\cite{xu2025fast, zheng2025highratesurgeryconstantoverheadlogical}. Note that, when the logical PPMs are logically disjoint, the post-augmentation matrix $\partial_1$ is still LDPC. Crucially, the code modifications above induce a strictly \emph{transversal} connectivity map; that is, the column (and row) {weights of} $\phi_1$ and $\phi_0$, and thus {of }the induced $\Phi_1$ and $\Phi_0$, {are} $1${. This satisfies the condition} outlined in Theorem~\ref{thm: main-high-rate-surgery-distance-preserving}. See details of {the }construction of {the }classical code {modifications} (punctures and augmentations) in Section~\ref{app-sec: classical-error-correcting-codes} in {the }Appendix. See Fig.~\ref{fig: parallel-surgery}(a) for an illustration. We can similarly construct {an} $X$-type hypergraph-surgery gadget that measures any {set of }logically disjoint $X${-type} operators within a row of the logical fibres by attaching the hypergraph in Def.~\ref{def: main-LP-intracolumn-hyerpgraph-surgery} to a row of physical fibres and $Z$ checks instead. 

\begin{figure*}[!t]
    \centering
    \includegraphics[width=\textwidth]{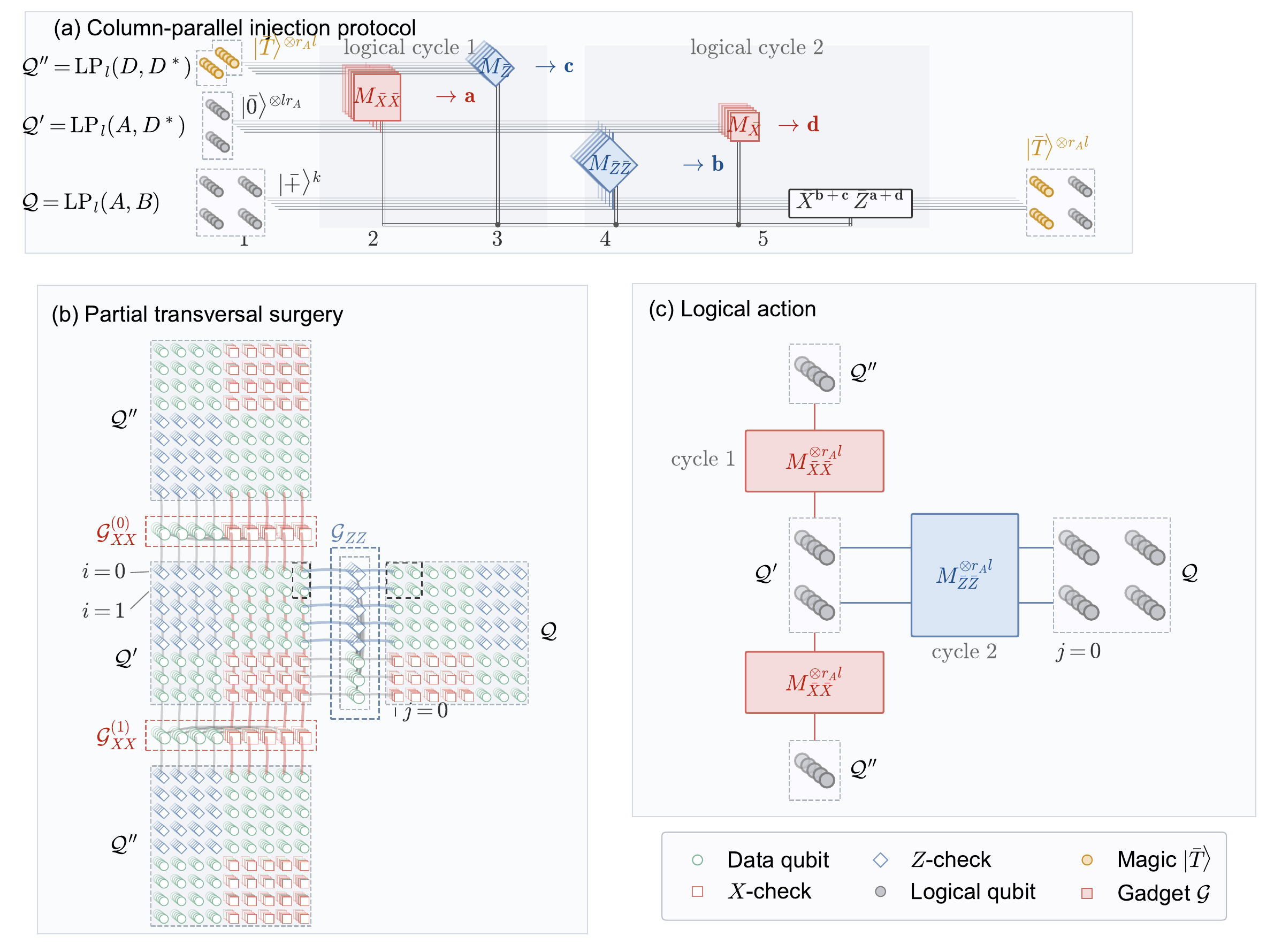}
    \caption{\textbf{Column-parallel magic-state injection}, teleporting $r_A l$ magic states in parallel into one information column of the data code $\mathcal{Q}=\LP_l(A,B)$ (gold), via $r_A$ stacked distance-$d_s$ surface codes $\mathcal{Q}''=\LP_l(D,D^*)$ and the transistor code $\mathcal{Q}'=\LP_l(A,D^*)$ (code parameters in the main text).
    \textbf{(a) Teleportation circuit.} The two-step protocol in Definition~\ref{def: main-magic-state-injection}, with measurement outcomes $\vec{a}$--$\vec{d}$ and the Pauli feedback correction $\bar X^{\vec{b}+\vec{c}}\bar Z^{\vec{a}+\vec{d}}$; the two joint measurements preserve the code distance (Theorem~\ref{thm: main-magic-injection-distance-preserving}), and each takes one logical cycle of $\Theta(d_s)$ syndrome-measurement rounds (shaded bands).
    \textbf{(b) Physical layout.} The two parallel joint measurements are implemented with the LP inter-column hypergraph surgery protocol (Definition~\ref{def: main-LP-intercolumn-surgery}; see also Fig.~\ref{fig: parallel-surgery}(b)), through the gadgets $\tilde{\mathcal{G}}_{XX}^{(i)}$ and $\tilde{\mathcal{G}}_{ZZ}$ (drawn: $l=5$, $d_s=5$).
    \textbf{(c) Logical action.} Each purple box is one joint PPM, transversally coupling the same-layer logical qubits of the two fibres it is wired to (conventions of Fig.~\ref{fig: parallel-surgery}): in cycle 1, $\mathcal{M}^{\otimes l}_{\bar X\bar X}$ per copy of $\mathcal{Q}''$; in cycle 2, $\mathcal{M}^{\otimes r_A l}_{\bar Z\bar Z}$ between $\mathcal{Q}'$ and column $j=0$ of $\mathcal{Q}$.}
    \label{fig: magic-state-injection}
\end{figure*}

{
\begin{example}[Soundness co-design with LP canonical codes]\label{example: main-soundness-codesign}
    From Theorem~\ref{thm: main-high-rate-surgery-distance-preserving}(ii), the thickening length of an LP intra-column hypergraph-surgery gadget is governed by the soundness of its skeleton (hyperedge-vertex incidence matrix) $\partial_1$. By Definition~\ref{def: main-LP-intracolumn-hyerpgraph-surgery}~\ref{construction-initial-skeleton}, the initial unmodified skeleton is $\partial_1 = \mathbb{B}(A)$. This suggests a route to co-design canonical LP codes whose classical prototype matrix might have large soundness. We propose a subfamily of symmetric canonical LP codes $\LP_l(A, A^*)$ whose base matrix has, up to a permutation of columns over $R$, the form
    \begin{align}
        A = \left(P \,|\, A' \right)
    \end{align}
    with $\mathbb{B}(A')$ invertible. The columns of $P$ are an information set, which corresponds to classical codewords of $\ker \mathbb{B}(A)$ and hence the canonical logical basis of $\LP_l(A, A^*)$ (Theorem~\ref{thm: basis-canonical-product-code}; this is the case $k = k_c$ of Proposition~\ref{prop: main-LP-even-lifting-basis}). For such code instances, the punctured skeleton $\mathbb{B}(A')$ is $(t, 1/\gamma_A)$-sound for every $t \geq 2$, where $\gamma_A = \omega_{\mathrm{col}}(\mathbb{B}(A')^{-1})$ (Lemma~\ref{lemma: unconditional-soundness-canonical-matrix}, Eq.~\eqref{eq: gamma-A-canonical}); hence, any modified skeleton $\partial_1$--whose modification leaves $\mathbb{B}(A')$ intact--can be made to have soundness at least $1/\gamma_A$. As a result of Theorem~\ref{thm: main-high-rate-surgery-distance-preserving}(ii), the hypergraph-surgery gadgets need a thickening of length only $m = \gamma_A$ for intra-column (resp.\ intra-row) surgery with arbitrary logically disjoint addressability. This family includes several new finite-size instances: $\LP^{2\times4}_{16} = [[320, 64, 10]]$, $\LP^{2\times3}_{36} = [[468, 36, 20]]$, $\LP^{3\times5}_{20} = [[680, 80, 13]]$ and $\LP^{3\times5}_{28} = [[952, 112, 17]]$, with $\gamma_A = 3, 4, 3$ and $4$, respectively (Table~\ref{tab: high-rate-two-body-addressable}). There is an inherent trade-off between distance and soundness in this family: $d \leq 1 + \gamma_A\, \omega_{\mathrm{col}}(A)$ (Theorem~\ref{thm: no-go-soundness-distance-canonical-product-code} in Appendix~\ref{sec: soundness-criterion-for-canonical-product}). All these code instances saturate their respective (more refined) bounds (Example~\ref{example: saturating-canonical-product-codes}). 
\end{example}
}

{Next, to measure in parallel an arbitrary logically disjoint set of logical Pauli operators, including $Y$-type and mixed-type ones, we bridge/augment an intra-column (or inter-column) gadget and an intra-row (or inter-row) gadget, which measure the $Z$- and $X$-type operators, respectively. The bridging technique is a generalized version of that presented in the graph surgery and systematically presented in Section~\ref{sec: parallel-surgery-arbitrary-addressable}.} 

{ 
\begin{theorem}[Informal: Parallel measurement of logically disjoint Pauli operators]\label{thm: main-bridged-hypergraph-surgery-informal}
    Let $\mathcal{S}[\tilde{\mathcal{G}}_Z; (\Phi^{(Z)}_1, \Phi^{(Z)}_0)]$ and $\mathcal{S}[\tilde{\mathcal{G}}_X; (\Phi^{(X)}_1, \Phi^{(X)}_0)]$ be a $Z$-type LP intra-column and an $X$-type LP intra-row hypergraph-surgery gadget, measuring the $Z$- and $X$-components, respectively, of 
    \begin{align}
        \mathcal{L} := \left\{ \bar{L}_\mu: \bar{L}_\mu \text{ are mutually logically disjoint} \right\}. 
    \end{align}
    Then there exists a constructive procedure, called hypergraph bridging, such that the bridged hypergraph-surgery gadget $\mathcal{S}[\tilde{\mathcal{G}}; (\tilde{\Phi}_1, \tilde{\Phi}_0)]$ measures $\mathcal{L}$ in parallel fault-tolerantly such that the degree of the merged code of the bridged gadget increases by at most a constant, which depends on $t$ and on the column and row weights of the hypergraph incidence matrices of the constituent $Z$-type and $X$-type gadgets.
\end{theorem}  

Here we can choose $t$ to be small, i.e. $t=4, 5$, for the co-designed code instances reported in Example~\ref{example: main-soundness-codesign} and Table~\ref{tab: high-rate-two-body-addressable}. The formal statement with a more refined definition of soundness is detailed in Theorem~\ref{thm: bridged-hypergraph-parallel-surgery} and Lemma~\ref{lemma: hypergraph-symmetrization}.
}

 {Table~\ref{tab: high-rate-two-body-addressable} reports the worst-case resources of these same-type gadgets, for single-body and for two-body PPMs, on the canonical LP codes of Example~\ref{example: main-soundness-codesign} and on one column of $\LP^{3\times5}_{33}$.}
{It also reports the worst case over all logically disjoint sets of two-body Clifford PPMs, i.e.\ products on at most two logical qubits from the menu $\{\bar{X}, \bar{Z}, \bar{Y}, \bar{X}\bar{X}, \bar{Z}\bar{Z}, \bar{X}\bar{Z}\}$, measured by the bridged hypergraph surgery of Theorem~\ref{thm: main-bridged-hypergraph-surgery-informal}. Table~\ref{tab: high-rate-arbitrary-resources} in Section~\ref{sec: parallel-surgery-arbitrary-addressable} reports those for any-body sets, i.e.\ arbitrary logically disjoint sets of products of any number of bodies, whose worst case coincides with that over all 2-PPMs (products on at most two logical qubits, $\bar{X}\bar{Y}$, $\bar{Z}\bar{Y}$, $\bar{Y}\bar{Y}$ included).}

While Definition~\ref{def: main-LP-intracolumn-hyerpgraph-surgery} only achieves parallel logical Pauli-product measurements within a column (resp. row) of logical $Z$- (resp. $X$-) operators, we present another gadget that can achieve inter-column (resp. inter-row) $Z$-type (resp. $X$-type) logical Pauli-product measurements, which is analogous to the homomorphic measurement gadgets for homological-product codes in Ref.~\cite{xu2025fast}.

\begin{definition}[Informal: LP inter-column hypergraph surgery]\label{def: main-LP-intercolumn-surgery}
    Let $\mathcal{Q} = \LP_l(A, B)$ be a canonical LP code. We construct a surgery gadget, denoted $\mathcal{S}[\tilde{\mathcal{G}}_{ZZ}; (\tilde{\Phi}_1, \tilde{\Phi}_0)]$, that measures $\{\bar{Z}_{i, j, m}\bar{Z}_{i, j^{\prime}, m}\}_{i \in [r_A], m \in [l]}$, i.e. parallel ZZ operators between two columns of logical fibres{. Its} hypergraph {is }identical to $A$, i.e. $\partial_1 = A${, and is} transversally connected to the $j$-th and $j^{\prime}$-th column{s} of physical fibres and checks (see Fig.~\ref{fig: parallel-surgery}(b)). 
    The merged code of this protocol is $\mr{LP}_l(A, B^{\prime})$, where $B^{\prime}$ is obtained by appending an extra column with two $1$s to $B$. 
\end{definition}

We can similarly construct {an} $X$-type inter-column surgery gadget,  $\mathcal{S}[\tilde{\mathcal{G}}_{XX}; (\Phi_1, \Phi_0)]$, that measures $\{\bar{X}_{i, j, m} \bar{X}_{i^{\prime}, j, m}\}_{j \in [r_B^{*}], m\in[l]}$ in parallel between two rows of logical fibres. In Section~\ref{sec: surgery-formulation} in the Appendix, we show that Definition~\ref{def: main-LP-intercolumn-surgery} can be seen as {an} LP code modification, transforming {an} LP code {into a variant of it}. Intuitively, {since the inter-column surgery transforms the data code into} another LP code, it inherits the distance protection of that code. In many cases, this ensures that the LP inter-column surgery gadget is distance-preserving. 

\subsection{Parallel magic \label{subsec: main-parallel-magic}}
\begin{table*}[t]
\centering
\begin{tabular}{lcccccc}
\toprule
 & \multicolumn{3}{c}{Size} & \multicolumn{3}{c}{Maximum degree} \\
\cmidrule(lr){2-4}\cmidrule(lr){5-7}
 & Qubits & $X$-checks & $Z$-checks & Stabilizer weight & Qubit $X$-degree & Qubit $Z$-degree \\
\midrule
Data LP code $\mathcal{Q}$ & 1122 & 495 & 495 & 8 & 5 & 5 \\
Transistor code $\mathcal{Q}'$ & 1749 & 693 & 990 & 7 & 3 & 5 \\
Stacked surface code $\mathcal{Q}''$ & 2805 & 1386 & 1386 & 4 & 2 & 2 \\
\midrule
LP inter-column $\mathcal{M}_{\bar{X}\bar{X}}$ gadget  & $+396$ & $+462$ & $+0$ & $+0$ & $+1$  & $+0$ \\
LP inter-column $\mathcal{M}_{\bar{Z}\bar{Z}}$ gadget  & $+99$ & $+0$ & $+165$ & $+1$ & $+0$ & $+0$ \\
\bottomrule
\end{tabular}
\caption{\textbf{Column-parallel magic-state injection.}
Space costs of the parallel magic-state injection protocol (Def.~\ref{def: main-magic-state-injection}) from the stacked distance-$7$ surface codes $\LP_{33}(D,D^*)=[[2805,33,7]]$ into the canonical LP code $\LP^{3 \times 5}_{33}=[[1122,148,\leq 20]]$, using a helper transistor code $\mathcal{Q}'=[[1749,68,d_X=7,d_Z\leq20]]$.
The key primitives are the $X$-type (resp. $Z$-type) LP inter-column surgery gadgets between $\mc{Q}^{\prime\prime}$ and $\mc{Q}^{\prime}$ (resp. between $\mc{Q}^{\prime}$ and $\mc{Q}$); see Fig.~\ref{fig: magic-state-injection}.
The qubits and checks added by the corresponding surgery hypergraphs (vertex checks of the measured type only, since $\partial_0 = 0$) and the added degrees of the merged codes are reported in the last two rows.
}
\label{tab: magic-state-injection}
\end{table*}

Here, we present a logical primitive for preparing $\Theta(k)$ magic states in parallel for a canonical LP code $\mc{Q} = \mr{LP}_l(A, B)$ encoding $k$ logical qubits.
First, we introduce a parallel magic-state injection protocol that injects $\Theta(k)$ $\ket{T}$ states into an LP code from $\Theta(k)$ copies of small distance-$d_s$ surface codes, where the $\ket{T}$ states could be prepared in parallel using, e.g. magic-state cultivation protocols~\cite{gidney2024magicstatecultivationgrowing}.
Crucially, by utilizing the inter-column hypergraph surgery protocol (Definition~\ref{def: main-LP-intercolumn-surgery}), we can implement this parallel injection with a (phenomenological) fault distance $\geq \min(d, d_s)$ while incurring small extra space costs in addition to the surface codes.
Then, combining these injected $\ket{T}$ states with the transversal distillation factory~\cite{xu2025fast, cain2026shorsalgorithmpossible10000, BravyiKitaev2005MagicStateDistillation, Litinski_2019}, which primarily uses transversal CNOTs between the LP codes, one could obtain higher-fidelity magic states in parallel. Our contribution here is the first injection step, which is the most challenging step in such a protocol since the 2D LP codes do not produce magic states natively.

The parallel injection is based on a key observation: a stack of $l$ copies of the distance-$d_s$ unrotated surface code can be viewed exactly as a symmetric LP code over $R$, $\mc{Q}'' = \mr{LP}_l(D, D^{*})$ with a base matrix $D \in R^{(d_s - 1)\times d_s}$, where $D$ is the standard repetition-code check matrix except that each element $1$ is now the unit element of $R$.
We refer to such $\mc{Q}^{\prime \prime}$ as a \emph{stacked} surface code.
{With the assistance of a ``transistor'' code $\mathcal{Q}' = \LP_l(A, D^*)$, we can then teleport logical states between (copies of) this stacked surface code and our target canonical LP code. We use the parallel inter-column hypergraph-surgery gadget in Definition~\ref{def: main-LP-intercolumn-surgery}.} 
Specifically,  each stacked surface code encodes a single fibre of $l$ logical qubits, whereas the transistor code $\mc{Q}^{\prime}$ encodes a column of $r_A$ fibres (see Fig.~\ref{fig: magic-state-injection}(c)).
Then, by viewing $\mc{Q}^{\prime \prime}$ and $\mc{Q}^{\prime}$ as a joint code $\mr{LP}_l(A \oplus D, D^*)$ with $r_A + 1$ rows (fibres) of logical qubits, we can apply the $X$-type inter-column surgery protocol (Def.~\ref{def: main-LP-intercolumn-surgery}) for parallel XX measurements, $\mc{M}_{\bar{X}\bar{X}}^{\otimes l}$, between the surface code logical qubits and a particular row of $\mc{Q}^{\prime}$. This can be {further parallelized by coupling} $r_A$ copies of $\mc{Q}^{\prime \prime}$ to different rows of the $\mc{Q}^{\prime}$ logical qubits (see Fig.~\ref{fig: magic-state-injection}(b)(c)). In addition, similar parallel ZZ measurements $\mc{M}_{\bar{Z}\bar{Z}}^{\otimes r_A l}$ can be applied between $\mc{Q}^{\prime}$ and a column of $\mc{Q}$ logical qubits by viewing them as $\mr{LP}_l(A, D^* \oplus B)$ and applying the $Z$-type inter-column surgery protocol (Def.~\ref{def: main-LP-intercolumn-surgery}).

Based on the transistor code and the parallel surgery primitives described above, we now present a scheme that injects $r_A l$ $| \bar{T}\rangle$ states from $\mc{Q}^{\prime \prime \otimes r_A}$ to a column of logical fibres of $\mc{Q}$ in parallel.

\begin{definition}[LP Column-parallel magic-state injection protocol]\label{def: main-magic-state-injection}
    Given a canonical LP code $\mc{Q} = \LP_l(A, B)$, we can inject $| \bar{T}\rangle$ states from $r_A$ copies of the stacked surface code $\mc{Q}^{\prime \prime} = \LP_l(D, D^*)$ using a helper ``transistor" canonical LP code $\mc{Q}^{\prime} = \mr{LP}_l(A, D^*)$ with $k^{\prime}_c = r_A l$, via the following protocol (see Fig.~\ref{fig: magic-state-injection})
\begin{enumerate}
    \item \label{magic-injection-preparation} \textbf{Preparation.} We prepare $r_A$-many copies of $\mathcal{Q}'' = \LP_l(D, D^*)$ with $|\bar{T} \rangle^{\otimes l}_{\mathcal{Q}''}$ (that is, $r_Al$-many disjoint copies of surface code patches, each supporting $| \bar{T}\rangle$). The transistor code $\mathcal{Q}' = \LP_l(A, D^*)$ is initialized at $|\bar{0}\rangle^{\otimes lr_A}_{\mathcal{Q}'}$ using $O(d_s)$ rounds of syndrome extraction.
    Furthermore, we assume that the data canonical LP code has been prepared at $|\bar{+}\rangle^{\otimes k_c}_{\mathcal{Q}}$ before the injection protocol.
  
    \item \label{magic-injection-Mxx} \textbf{Surface-to-transistor ${\mathcal{M}}_{\bar{X}\bar{X}}$}. 
    For each $i \in [r_A]$, we measure joint ${\mathcal{M}}^{\otimes l}_{\bar{X}\bar{X}}$ between a copy of $\mc{Q}^{\prime \prime}$ and the $i$-th row of the logical fibre of $\mc{Q}^{\prime}$ (see Fig.~\ref{fig: magic-state-injection}(b)(c)) using the $X$-type inter-column surgery gadget (Definition~\ref{def: main-LP-intercolumn-surgery}). These measurements are all performed in parallel.
    Record the parity measurement outcome vector $\vec{a}$.
    \item \label{magic-injection-Mzz} \textbf{Transistor-to-data ${\mathcal{M}}_{\bar{Z}\bar{Z}}$.} We measure joint $\mathcal{M}^{\otimes r_Al}_{\bar{Z}\bar{Z}}$ between the transistor code $\mathcal{Q}'$ and the $j$-th column of logical qubits in the data canonical LP code $\mathcal{Q}$ 
    similarly using the $Z$-type LP inter-column surgery gadget {(Definition}~\ref{def: main-LP-intercolumn-surgery}). Record the parity measurement outcome vector $\vec{b}$. 
    \item \label{magic-injection-transistor-Mx} \textbf{Pauli feedback correction.} We transversally measure in the $X$-basis on the transistor code $\mathcal{Q}'$ and record its parity measurement outcome $\vec{d}$, and we transversally measure in the $Z$-basis on $r_A$ copies of $\mathcal{Q}''$ and record their parity measurement outcome $\vec{c}$. Then we apply the Pauli feedback correction $\bar{X}^{\vec{b} + \vec{c}}\bar{Z}^{\vec{a} + \vec{d}}$ on $\mathcal{Q}$.
\end{enumerate}
\end{definition}

The inter-column surgery gadgets for {the }$\bar{X}\bar{X}$- and $\bar{Z}\bar{Z}$-joint measurements in {steps} 2 and 3 can be shown to preserve the code distances{ of} the surface code patches $\mathcal{Q}''$ and the data code $\mathcal{Q}${, respectively}. 

\begin{theorem}[LP Column-parallel magic-state injection is fault-tolerant]\label{thm: main-magic-injection-distance-preserving}
{Let $\mathcal{Q}$ have distance $d$, and let the stacked surface codes $\mathcal{Q}''$ have distance $d_s$. The following hold for the column-parallel magic-state injection protocol in Definition~\ref{def: main-magic-state-injection}.} 
\begin{itemize}
    \item The transistor code $\mathcal{Q}'$ has $X$-distance $d_s$ and $Z$-distance $\geq d$. 
    \item {The LP inter-column surgery gadgets for the $\mc{M}_{\bar{X}\bar{X}}$ and $\mc{M}_{\bar{Z}\bar{Z}}$ measurements have merged-code distances} at least $d_s$ and $\min(d_s, d)$, respectively.  
\end{itemize}
Consequently, the protocol has a phenomenological distance $\min(d, d_s)$. 
\end{theorem}

Hence, for {practical} purposes, the above parallel magic-state {injection protocol }is provably fault-tolerant{. Its} phenomenological distance is lower-bounded by the distance of the stacked copies of surface codes $\mc{Q}''$. 

\begin{example}
    In Table~\ref{tab: magic-state-injection}, we give an example of the high-rate, parallel magic-state injection to the headline code $\LP^{3 \times 5}_{33}$, where the protocol injects $66 = k_c/2$-many $|T\rangle$ states in parallel. Note that the surgery gadgets for LP inter-column joint measurements are distance-preserving and incur only minimal space resources ({last two} rows of Table~\ref{tab: magic-state-injection}). Specifically, the $\bar{X}\bar{X}$-joint measurements between the surface code patches (2 copies of $\mathcal{Q}''$) and the transistor code have merged distance $d_s=7$, and $\bar{Z}\bar{Z}$-joint measurements preserve the $Z$-distance $d \leq 20$ of $\mathcal{Q}$ (merged-code distance $\geq \min(d_s, d)$ by Theorem~\ref{thm: main-magic-injection-distance-preserving}). Since the connectivity maps of these two LP inter-column surgery gadgets are transversal, i.e., having maximum row or column weight $1$, we can conclude that this protocol is fault-tolerant. 
\end{example}
 
Using the column-parallel injection protocol in Def.~\ref{def: main-magic-state-injection}, we can further inject magic states into all $k_c$ logical qubits of a canonical LP code $\mc{Q}=\mr{LP}_l(A,B)$ by applying the protocol once to each of the $r_B^*$ columns. Consequently, when $r_B^*$ is constant, the total injection time is independent of $l$. Altogether, this yields a highly parallel, provably fault-tolerant magic-state teleportation protocol with low additional space overhead.

We refer to Appendix~\ref{sec: LP-parallel-magic-injection} for a detailed discussion of the parallel magic-state injection protocol, including the proof of Theorem~\ref{thm: main-magic-injection-distance-preserving}.

\section{Outlook}
By exploiting the algebraic structure and symmetries of Abelian LP codes, we obtain a large sub-family of canonical LP codes with a highly structured logical basis{. We equip these codes} with a flexible, low-overhead, and fully benchmarked fault-tolerant instruction set. We expect these constructions to push the frontier of fault-tolerant quantum computation on high-rate architectures~\cite{cain2026shorsalgorithmpossible10000}.

Looking forward, several directions remain open. The logical {cycle time}, currently set by the $\Theta(d)$ syndrome-extraction rounds of each surgery, could be reduced to $O(1)${ rounds} by incorporating \emph{single-shot} surgery gadgets~\cite{Campbell_2019, tan2025singleshotuniversalityquantumldpc, golowich2025constantoverheadaddressablegatessingleshot}. Moreover, we anticipate that several of our gadgets, most notably the parallel code surgery and the parallel magic-state injection, admit further optimizations that would reduce their space overhead. Finally, we leave detailed numerical simulations and optimizations, as well as the instantiation of our schemes on concrete hardware platforms, to future work.

\section{Acknowledgements}
We refer readers to the concurrent works~\cite{Bhardwaj2026MittenCodes, Hong2026RateOneFifthLDPC} for related developments on non-Abelian LP codes.
We thank Aditya Bhardwaj, Muzhou Ma, Harald Putterman, Robbie King, Dolev Bluvstein, John Preskill, Madelyn Cain, Hsin-Yuan Huang, Yifan Hong, and Zhiyang He for helpful discussions. 
We acknowledge support from the ARO(W911NF-23-1-0077), ARO MURI (W911NF-21-1-0325), AFOSR MURI (FA9550-21-1-0209, FA9550-23-1-0338), NSF (ERC-1941583, OMA-2137642, OSI-2326767, CCF-2312755, OSI-2426975). Q.X. acknowledges funding by the Walter
Burke Institute for Theoretical Physics at Caltech.

\tableofcontents %
\appendix

\section{Preliminaries}\label{app-sec: premliminaries}
In this section, we review the necessary background for the discussion that follows. 

\subsection{Basics on the commutative ring theory}\label{app-sec: commutative-ring-theory}
Throughout the texts, we exclusively consider polynomial rings of characteristic $2$. All rings are commutative with identity, and all modules are left modules unless the side of the action matters (Definition~\ref{def: R-module}), as in the balanced product. Let $\mathbb{F}_2[x_1, \cdots, x_m]$ be the space of $\mathbb{F}_2$-valued multivariate polynomials. Given polynomials $f_1, \cdots, f_s \in \mathbb{F}_2[x_1, \cdots, x_m]$, we define the quotient ring
\begin{align}
    R := \mathbb{F}_2[x_1, \cdots, x_m]/(f_1, \cdots, f_s).
\end{align}
A standard choice is given by $R_l := \mathbb{F}_2[x]/(x^l +1)$; when the context is clear, we write $R$ for $R_l$. In what follows, we work in some generality. 

\begin{definition}[Integral domain]
     A commutative ring $R$ with $1 \neq 0$ is said to be an integral domain if $ab = 0$ implies $a = 0$ or $b = 0$.
\end{definition}

Obviously $\mathbb{F}_2[x_1, \cdots, x_m]$ is an integral domain. Another important notion is that of an \emph{ideal} $I \subseteq R$: an additive subgroup of $R$ that is closed under multiplication by arbitrary elements of $R$. 

\begin{definition}[Principal ideal domain]
     An integral domain $R$ is said to be a principal ideal domain (PID) if every ideal $I \subset R$ is generated by a single element.
\end{definition}

\begin{example}
     It is easy to observe that $\mathbb{F}_2[x]$ is an integral domain. It remains to show that every ideal is generated by a single element. Suppose an ideal is generated by two elements, $(f(x), g(x))$; observe that it is then also generated by their greatest common divisor $\operatorname{gcd}(f, g)$. It remains to prove that coprime $f, g$ generate the entire ring $\mathbb{F}_2[x]$. Let $h \in (f, g)$ have minimal degree. Then $h$ divides every element of $(f, g)$: otherwise, the Euclidean algorithm produces an element $af + bg \bmod h$ of degree less than that of $h$, a contradiction. Since $f, g$ are coprime, it follows that $h = 1 \in (f, g)$, whence $(f, g) = (1) = \mathbb{F}_2[x]$. This is also called \emph{Bezout's identity}. 
\end{example}

\begin{example}
     Let us also consider a famous counter-example. Consider the ideal generated by $(2, x)$ in $\mathbb{Z}[x]$; we show that it cannot be generated by any single element. Suppose $(2, x) = (g)$ for some $g \in \mathbb{Z}[x]$. Then $g \mid 2$ forces $g$ to be a constant, $g \in \{1, 2\}$ up to sign. If $g = 2$, then $g \nmid x$; if $g = 1$, then $(2, x) = \mathbb{Z}[x]$, which fails because every element of $(2, x)$ has an even constant term.
\end{example}

Given an ideal $I \subset R$, we can form the \emph{quotient ring} $R/I$. A canonical example is $I = (x^n + 1)$ for some integer $n$. A quotient ring is typically not an integral domain when $x^n+1$ factorizes. Write $x^n + 1 = ab$ with both factors nonconstant. Then $ab = x^n+1 =0$ while $a\neq 0$, $b \neq 0$. We therefore need additional tools to handle the quotient structures, which are crucial to our study.

\begin{definition}[Principal ideal ring]
     A commutative ring $R$ is a principal ideal ring if for every ideal $I \subseteq R$, there exists an element $g \in R$ such that $I = (g)$.
\end{definition}

Note that every principal ideal domain is a principal ideal ring, but the converse fails: a principal ideal ring may contain zero divisors, as the next example shows. In particular, $R_l$ with reducible $x^l + 1$ is a principal ideal ring that is not a principal ideal domain.

\begin{lemma}
    If $R$ is a principal ideal domain (PID), then $R/I$ is a principal ideal ring for every ideal $I\subseteq R$. 
    \begin{proof}
        Let $f: R \rightarrow R/I$ be the quotient map. It is surjective and satisfies $f(x) = x + I = [x]$. For any ideal $J \subseteq R/I$, we have that $f^{-1}(J)$ is an ideal over $R$: it is closed under addition, and for any $r \in R$ and $a \in f^{-1}(J)$, we have $f(ra) = f(r)f(a) \in J$ since $J$ absorbs multiplication by elements of $R/I$; hence $ra \in f^{-1}(J)$. Since $R$ is a principal ideal domain, $f^{-1}(J)$ is generated by a single element, $f^{-1}(J) = (d)$. By surjectivity of $f$, it follows that $J = f(f^{-1}(J)) = f((d)) = ([d])$ is generated by a single element. 
    \end{proof}
\end{lemma}

\begin{example}[$\mathbb{F}_2 \times \mathbb{F}_2$]
     In general, a direct product of fields is a PIR. Consider the simplest case, the product of two binary fields. Since $(1, 0) \cdot (0, 1) = (0, 0)$, the ring has zero divisors and hence cannot be a PID. It is nevertheless a principal ideal ring, by the following fact. 
\end{example}

\begin{lemma}
   Any finite direct product of principal ideal rings is a principal ideal ring. 
\end{lemma}

 This is a standard result (see Ref.~\cite{AtiyahMacdonald1969IntroCA}). Note that the above statement is known to be false for infinite products of principal ideal rings. This product structure is ubiquitous when taking quotients.

\begin{lemma}[Chinese remainder theorem]\label{lemma: chinese-reminder}
      Consider a polynomial $h(x) = f^{e_1}_1 f^{e_2}_2 \cdots f^{e_s}_s(x)$, where $f_1, \cdots, f_s$ are pairwise distinct irreducible polynomials and $e_i \geq 1$ are integer exponents. We define $ R_h := \mathbb{F}_2[x] / ( h(x) )$. The ring $R_h$ is isomorphic to the direct product of rings, 
    \begin{align}
        R_h \cong  R_{(1)} \times \cdots \times R_{(s)}, 
    \end{align}
    where $R_{(i)} := \mathbb{F}_2[x]/( f^{e_i}_i(x) )$.  
\end{lemma}

An important observation relevant to our discussion is the following. 

\begin{corollary}
     For any prime power $q$, the quotient ring $\mathbb{F}_q[x]/(x^l-1)$ is a principal ideal ring, regardless of whether $l$ is odd or even. Over $\mathbb{F}_2$, $x^l - 1$ coincides with $x^l + 1$, and this specializes to $R_l$. 
\end{corollary}

\begin{remark}\label{remark: square-free-odd-l}
    For odd $l$, the polynomial $x^l + 1$ is square-free: in characteristic $2$, $\gcd(x^l+1, (x^l+1)') = \gcd(x^l+1, l x^{l-1}) = \gcd(x^l+1, x^{l-1}) = 1$, since $l$ is odd and $x \nmid x^l + 1$. Hence, every root is simple, the factorization $x^l + 1 = \prod^t_{i=1} b^{(i)}$ into irreducible polynomials is multiplicity-free, and $R_l \cong \prod^t_{i=1} \mathbb{F}_2[x]/(b^{(i)})$ is a finite direct product of fields. This square-freeness underlies the semi-simplicity of $R_l$ used throughout Appendix~\ref{app-sec: lp-basis-characterization}.
\end{remark}

An important technical reason for considering principal ideal rings is Bezout's identity. 

\begin{lemma}\label{lemma: bezout-lemma}
    Let $R$ be a principal ideal ring, let $f, g \in R$, and write $\operatorname{gcd}(f, g) = b$. Then there exist $r, s \in R$ such that 
    \begin{align}
        r(x)f(x) + s(x)g(x) = b(x).
    \end{align}
    \begin{proof}
        Note that in a principal ideal ring, every ideal is generated by a single element. In particular, consider the ideal generated by $(f, g) = (\operatorname{gcd}(f, g)) = (b)$. Hence there exist $r, s$ with $rf + sg = b$, as desired. 
    \end{proof}
\end{lemma}

\subsection{Basics on chain complexes}\label{app-sec: basic-chain-complex}
We also give a basic introduction to the theory of chain complexes over a ring $R$. Throughout this work, we can assume $R$ to be a finite commutative ring of characteristic $2$ and, even more, $R = R_l$ for concreteness. A chain (complex) $\mathcal{C}_{\bullet}$ over the ring $R$ of dimension $t$ is given by 
\begin{align}
    \mathcal{C}_t \xrightarrow{\partial^C_t} \cdots \rightarrow \mathcal{C}_1 \xrightarrow{\partial^C_1} \mathcal{C}_0, 
\end{align}
such that $\partial^C_{j}\partial^C_{j+1} = 0$ for $j=1, \cdots, t-1$. 

\begin{remark}
    Formally, we need to define \emph{chain complexes} by first specifying a topology such as simplexes. Here we do not make this distinction and use the terms ``(co)chain'' and ``(co)chain complex'' interchangeably. 
\end{remark}

To define the cochain complex, it is convenient to introduce an inner-product notation. We define this on $R$ though the extension to the multivariate case is straightforward. For $f = \sum^{l-1}_{i=0} a_i x^i \in R$, define its conjugate by
\begin{align}\label{eq: conjugate-polynomial}
    \bar{f} = \sum^{l-1}_{i=0} a_i x^{-i} = \sum^l_{i=1} a_{l-i} x^i,
\end{align}
 and we denote 
\begin{align}\label{eq: inner-product-ring}
    \langle f, g \rangle_R : = \sum^{l-1}_{i=0} a_i b_i =:  [x^0]f \bar{g},
\end{align}
for $g = \sum^{l-1}_{i=0}b_i x^i$. Treating $f, g$ as coefficient vectors $a, b \in  \mathbb{F}^l_2$, we have that $\langle f, g \rangle_R = \langle a, b \rangle_{\mathbb{F}_2}$: this pairing is precisely the standard binary inner product of the coefficient vectors. This identification defines the isomorphism $f \mapsto \langle \cdot, f \rangle$. Hence, letting $\mathcal{C}^i := \{ \varphi: \mathcal{C}_i \rightarrow \mathbb{F}_2 \} = \operatorname{Hom}(\mathcal{C}_i, \mathbb{F}_2)$ be the $i$-cochain, we have $\mathcal{C}^i \cong \mathcal{C}_i$ and the cochain complex 
\begin{align}
    \mathcal{C}^t \xleftarrow{\delta^C_{t-1}} \mathcal{C}^{t-1} \leftarrow \cdots \leftarrow \mathcal{C}^1 \xleftarrow{\delta^C_0} \mathcal{C}^0, 
\end{align}
where the coboundary map $\delta^C_i$ is defined via identification by 
\begin{align}
    \langle f, \delta^C_i(g) \rangle_R = \langle \partial^C_{i+1}(f), g \rangle_R. 
\end{align}
We can write the boundary and coboundary maps more explicitly, i.e., as matrices, which requires us to introduce more structures. 

\begin{definition}[Left $R$-module]\label{def: R-module}
     A left $R$-module is an Abelian group $(M, +)$ equipped with a left action by $R$,
    \begin{align}
        R \times M \ni (r, a) \mapsto r \cdot a \in M,
    \end{align}
    such that the following conditions are satisfied.
    \begin{itemize}
         \item  $r \cdot(a+b)=r \cdot a+r \cdot b$,
\item $r \cdot 0=0$,
\item  $(r+s) \cdot a=r \cdot a+s \cdot a$,
\item  $r \cdot(s \cdot a)=(r s) \cdot a$,
\item $1 \cdot a=a$.
    \end{itemize}
    A right $R$-module is defined analogously through a right action $M \times R \rightarrow M$. Since our $R$ is commutative, the two notions coincide, and we simply refer to $R$-modules unless the side of the action matters, as in the balanced product.
\end{definition}

\begin{lemma}\label{lemma: lifted-cochain}
    Let $\mathcal{C}_{\bullet}$ be a $t$-dimensional chain complex over $R$ whose boundary maps $\partial^C_i \in \operatorname{Hom}(\mathcal{C}_i, \mathcal{C}_{i-1})$ are assumed to be matrices over $R$, for $i = 1, \cdots, t$. Then its cochain complex $\mathcal{C}^{\bullet}$ is given by the coboundary maps $\delta^C_{i-1} = (\overline{\partial^C_i})^T \equiv (\partial^C_i)^*$. 
    \begin{proof}
        Extend the pairing entrywise, $\langle u, v \rangle_R = \sum_k [x^0]\, u_k \bar{v}_k$. Then we compute
        \begin{align}
           \begin{aligned}
            & \langle \partial^C_i u, v \rangle_R
            = \sum_k \sum_j  [x^0] \left((\partial^C_i)_{kj} u_j\right) \bar{v}_k \\
            &= \sum_j \sum_k[x^0]\, u_j  (\partial^C_i)_{kj} \bar{v}_k \\
            &= \sum_j \sum_k [x^0]\, u_j\, (\partial^C_i)^T_{jk} \bar{v}_k \\
            &=: \sum_j \sum_k [x^0]\, u_j\, \overline{\delta^C_{i-1}}_{jk} \bar{v}_k
           \end{aligned}
        \end{align}
        where the third equality uses that conjugation is a ring automorphism of $R$. Hence, the defining relation of the coboundary map gives $\delta^C_{i-1} = (\overline{\partial^C_i})^T$.
    \end{proof}
\end{lemma}

The chain-complex formalism provides a convenient language for stabilizer codes on qubits, whose checks are given by two consecutive (co)boundary maps. We define the $i$th (co)homology group $H_i(\mathcal{C})$ of a chain complex $\mathcal{C}_{\bullet}$ ($\mathcal{C}^{\bullet}$) by
\begin{align}
    H_i(\mathcal{C}) := \ker \partial^C_i / \IM \partial^C_{i+1}; \quad H^{i+1}(\mathcal{C}) := \ker \delta^C_{i+1} / \IM \delta^C_{i}, 
\end{align}
for $i =0, \cdots, t-1$, with the conventions $\partial^C_0 := 0$ and $\delta^C_t := 0$; in particular, $H_0(\mathcal{C}) = \mathcal{C}_0/\IM \partial^C_1$ and $H^t(\mathcal{C}) = \mathcal{C}^t/\IM \delta^C_{t-1}$. It is easy to see that the $i$th (co)homology groups are indeed $R$-modules. 

\begin{definition}[Chain maps]
    Let $\mathcal{C}_{\bullet}$ and $\mathcal{D}_{\bullet}$ be two chains over $R$. A chain map from $\mathcal{C}_{\bullet}$ to $\mathcal{D}_{\bullet}$ is a collection of maps $(\varphi_0, \varphi_1, \cdots)$ such that $\varphi_i \partial^C_{i+1} = \partial^{D}_{i+1}\varphi_{i+1}$. 
\begin{equation}
    \begin{tikzcd}
	{\cdots } & {\mathcal{C}_{i+1}} & {\mathcal{C}_i } & {\mathcal{C}_{i-1}} & {\cdots } \\
	{\cdots } & {\mathcal{D}_{i+1}} & {\mathcal{D}_i} & {\mathcal{D}_{i-1}} & {\cdots }
	\arrow[from=1-1, to=1-2]
	\arrow["{\partial^C_{i+1}}", from=1-2, to=1-3]
	\arrow["{\varphi_{i+1} }", from=1-2, to=2-2]
	\arrow["{\partial^C_{i}}", from=1-3, to=1-4]
	\arrow["{\varphi_{i} }", from=1-3, to=2-3]
	\arrow[from=1-4, to=1-5]
	\arrow["{\varphi_{i-1} }", from=1-4, to=2-4]
	\arrow[from=2-1, to=2-2]
	\arrow["{\partial^D_{i+1}}", from=2-2, to=2-3]
	\arrow["{\partial^D_{i}}", from=2-3, to=2-4]
	\arrow[from=2-4, to=2-5]
\end{tikzcd}
\end{equation}
\end{definition}
The chain maps naturally induce maps for the homology groups from $H(\mathcal{C})$ to $H(\mathcal{D})$. For $u\in\ker\partial^C_i$, we have $\partial^D_i\varphi_i(u)=0$, so $\varphi_i(u)\in\ker\partial^D_i$. Similarly, for $u \in \IM \partial^C_{i+1}$, writing $u = \partial^C_{i+1}(w)$ gives $\varphi_i(u) = \partial^D_{i+1}(\varphi_{i+1}(w)) \in \IM \partial^D_{i+1}$. This checks well-definedness; the explicit action on the $i$th homology groups typically requires a more detailed understanding of the homology groups themselves, which is the central subject of this work in the case of lifted-product codes. 

A central notion for chain complexes is the weight of elements of an $R$-module. For a polynomial $f = \sum_{(i_1, \cdots, i_m)} a_{i_1 \cdots i_m} x_1^{i_1} \cdots x_m^{i_m} \in R$, we define its weight as the Hamming weight of its coefficient vector. Formally, following the notation of Refs.~\cite{panteleev2022asymptoticallygoodquantumlocally, Panteleev_2021}, we consider 
\begin{align}
    f \mapsto \mathbb{B}(f) = \vec{a}, 
\end{align}
so that $|f| = |\vec{a}|$. More generally, for a vector $u \in R^n$, $\mathbb{B}(u) \in \mathbb{F}^{ln}_2$ stacks the coefficient vectors of the entries $u_j$, and $|u| := |\mathbb{B}(u)|$. For a matrix $A \in R^{m_A \times n_A}$, $\mathbb{B}(A) \in \mathbb{F}^{l m_A \times l n_A}_2$ replaces each entry by the $l \times l$ circulant matrix of multiplication by that entry, so that $\mathbb{B}(Au) = \mathbb{B}(A)\mathbb{B}(u)$ and $\mathbb{B}(\bar{A}^T) = \mathbb{B}(A)^T$; the multivariate case is analogous. 

\begin{definition}[(Co)systolic distance]
    Let $\mathcal{C}_i$ ($\mathcal{C}^i$) be the $i$th (co)chain over $R$. The $i$th (co)systolic distance is given by
    \begin{align}
        d_{\operatorname{sys}}(i)
        &:= \min\{ |u|: u \in \ker \partial^C_i \setminus \IM \partial^C_{i+1} \}, \\
        d_{\operatorname{cosys}}(i)
        &:= \min\{ |u|: u \in \ker \delta^C_i \setminus \IM \delta^C_{i-1} \}.
    \end{align}
\end{definition}

 We are primarily concerned with CSS codes, which are $2$-dimensional chain complexes; when the context is clear, we refer to the code distance simply as the (co)systolic distance, without specifying the degree. 

We now introduce the notion of (co)boundary expansion, together with the auxiliary reduced distance. 

\begin{definition}[Reduced distance]
    Let $\varphi: M \rightarrow M'$ be a linear map over $R$. The reduced distance of $u \in M$ with respect to $\varphi$ and $N \subseteq M$ is defined as 
    \begin{align}
        d_{\varphi}(u, N) = \min_{v \in N} |\varphi(u + v)|,
    \end{align}
    where $\varphi$ should be taken as the identity map when omitted. 
\end{definition}

\begin{definition}[(Co)boundary expanding]
    Let $\mathcal{C}_{\bullet}$ be a chain complex over $R$. We say that $\mathcal{C}_{\bullet}$ is $i$th (co)boundary expanding with (co)isoperimetric constant $h_i$ if, for all $u \in \mathcal{C}_i$,
    \begin{align}
        |\partial^C_i u| \geq h_i d(u, \IM \partial^C_{i+1}). 
    \end{align}
\end{definition}
Note that this definition is not normalized, in contrast to those used in standard high-dimensional-expander contexts~\cite{kalachev2023twosidedrobustlytestablecodes, dinur2025expansionhigherdimensionalcubicalcomplexes}. A famous example of (co)boundary expansion is the Cheeger constant of a graph. Graphs can be presented as chain complexes, with a slight abuse of notation, as follows. In our setting, we simply say that the chain $\mathcal{C}_{\bullet}$ is boundary expanding if it is $i$th boundary expanding, where $\mathcal{C}_i$ encodes the qubits. 

\begin{definition}[Graph chain]\label{def: graph-chains}
    Let $\mathcal{C}_{\bullet}$ be a $2$-dimensional chain, where we associate 
    \begin{itemize}
    \item (Binary field). $\mathcal{C}_2 = \mathbb{F}_2 = \{0, 1 \}$. 
    \item (Vertex space). $\mathcal{C}_1$ is the space of vertices. 
    \item (Edge space). $\mathcal{C}_0$ is the space of edges.
    \item (Cycle space). $\mathcal{C}_{-1}$ is the space of cycles (faces). 
\end{itemize}
\begin{equation}
    \begin{tikzcd}
	{\mathcal{C}_2} & {\mathcal{C}_1} & {\mathcal{C}_0} & {\mathcal{C}_{-1}}
	\arrow["{\partial^C_2}", from=1-1, to=1-2]
	\arrow["{\partial^C_1}", from=1-2, to=1-3]
	\arrow["{\partial^C_0}", from=1-3, to=1-4]
\end{tikzcd},
\end{equation}
where $\IM \partial^C_1 = \ker \partial^C_0$ by construction and $\partial^C_2(0) = 0$, $\partial^C_2(1) = \mathbf{1}_V$, the all-ones vector on $\mathcal{C}_1$. This definition applies purely to graphs. Throughout, we assume the graph is connected, so that the chain is exact at $\mathcal{C}_1$: $\ker \partial^C_1 = \IM \partial^C_2 = \{0, \mathbf{1}_V \}$. 
\end{definition}
It is then clear that the Cheeger constant is the first isoperimetric constant $h_1$. Note that, throughout, we use chains instead of cochains, which differs from the standard treatment. We can further extend the notion of (co)boundary expansion using chain maps: let $\varphi_i$ be the degree-$i$ component of a chain map from $\mathcal{C}_{\bullet}$ to $\mathcal{D}_{\bullet}$; then the boundary expansion with respect to $\varphi_i$ is defined by 
\begin{align}
     |\partial^C_i u| \geq h_i d_{\varphi_i}(u, \IM \partial^C_{i+1}). 
\end{align}

\subsection{Product chains and complexes}\label{app-sec: product-chains-complexes}

Higher-dimensional chain complexes can be constructed using products, such as the cubical complexes that underpin the code constructions of extensive families of quantum LDPC codes~\cite{bravyi2013homologicalproductcodes, TillichZemor2009HypergraphProduct, panteleev2022asymptoticallygoodquantumlocally, Panteleev_2021, dinur2025expansionhigherdimensionalcubicalcomplexes}. Another convenient and equivalent framework is that of balanced product~\cite{Breuckmann2021BalancedProductQuantumCodes}, which generalizes the homological product over finite fields and rings. 

\begin{definition}[Balanced product~\cite{Breuckmann2021BalancedProductQuantumCodes}]\label{def: balanced-product}
    Let $G$ be a finite group and $\mathbb{F}_2[G]$ be its group algebra. Let $M$ be a right and $N$ a left $\mathbb{F}_2[G]$-module; equivalently, $M$ and $N$ are $\mathbb{F}_2$-vector spaces carrying a right and a left $G$-representation, $\rho_M: G \rightarrow \mathrm{GL}(M, \mathbb{F}_2)$ and $\rho_N: G \rightarrow \mathrm{GL}(N, \mathbb{F}_2)$, where we write $u g := \rho_M(g)(u)$ and $g v := \rho_N(g)(v)$. Denote the $\mathbb{F}_2$-Kronecker product by $M \otimes_{\mathbb{F}_2} N$, and let
    \begin{align}
        \begin{aligned}
        T_G := \mathrm{span}_{\mathbb{F}_2}\{&u g \otimes_{\mathbb{F}_2} v + u \otimes_{\mathbb{F}_2} g v : \\
        &u \in M,\ v \in N,\ g \in G\}
        \end{aligned}
    \end{align}
    be the subspace of balancing relations. Then the balanced product between $M$ and $N$ is the quotient
    \begin{align}
        M \boxtimes_G N := (M \otimes_{\mathbb{F}_2} N) / T_G,
    \end{align}
    so that $u g \otimes_{\mathbb{F}_2} v \sim u \otimes_{\mathbb{F}_2} g v$ for all $g \in G$, where $[\cdot]$ denotes the equivalence class after the quotient. 
\end{definition}

Equivalently, the balancing relation can be written as $u g \otimes_{\mathbb{F}_2} g^{-1}v \sim u \otimes_{\mathbb{F}_2} v$. Modules over a field are vector spaces, and we refer to them as modules for consistency. 

\begin{example}[Balanced product of vector spaces]
    Let $V$ be a right and $W$ a left $G$-representation over $\mathbb{F}_2$. Then the balanced product is spanned by the classes of simple tensors,
    \begin{align}
        V \boxtimes_G W
        = \mathrm{span}_{\mathbb{F}_2}\{ [v \otimes_{\mathbb{F}_2} w] : v \in V,\ w \in W\},
    \end{align}
    with the identification
    \begin{align}
        [v g \otimes_{\mathbb{F}_2} w]
        = [v \otimes_{\mathbb{F}_2} gw].
    \end{align}
    Since $T_G$ is a linear subspace, the quotient $V \boxtimes_G W = (V \otimes_{\mathbb{F}_2} W)/T_G$ remains a vector space. 
\end{example}

\begin{definition}[Tensor product over $R$]\label{def: tensor-product-over-R}
    Let $R$ be any commutative ring. We define the tensor product $\otimes_R$ for any two $R$-modules, $M \otimes_R N$, by the following conditions. For all $r \in R$, $m \in M$, and $n \in N$
    \begin{enumerate}
        \item $(m_1 + m_2) \otimes_R n = m_1 \otimes_R n + m_2 \otimes_R n$,
        \item $m \otimes_R (n_1 + n_2) = m\otimes_R n_1 + m \otimes_R n_2$,
        \item $(rm) \otimes_R n = m \otimes_R (rn)$.
    \end{enumerate}
\end{definition}

\begin{lemma}[Tensor product over the group algebra as a balanced product]\label{lemma: balanced-product-isomorphism-R-linear-tensor}
    Let $G$ be a finite group and $R=\mbb{F}_2[G]$. Let $M$ and $N$ be a right and a left $R$-module, respectively. Then there is a natural $\mbb{F}_2$-linear isomorphism
    \begin{align}
        s:M\otimes_R N &\longrightarrow M\boxtimes_G N, &
        m\otimes_R n &\longmapsto [m\otimes_{\mbb{F}_2}n].
    \end{align}
    \begin{proof}
        By Definitions~\ref{def: balanced-product} and~\ref{def: tensor-product-over-R}, the map is $\mbb{F}_2$-linear and well-defined because
        \begin{align}
            s(mg\otimes_R n)
            =[mg\otimes_{\mbb{F}_2}n]
            =[m\otimes_{\mbb{F}_2}gn]
            =s(m\otimes_Rgn).
        \end{align}
        It is surjective by the definition of the balanced product. To prove injectivity, let $v\in M\otimes_RN$ satisfy $s(v)=0$, and choose a representative $\widetilde v\in M\otimes_{\mbb{F}_2}N$. Then $\widetilde v\in T_G$, where $T_G$ is the subspace of balancing relations defined above. Every generator of $T_G$ maps to zero in $M\otimes_RN$ because
        \begin{align}
            mg\otimes_Rn+m\otimes_Rgn=0.
        \end{align}
        Hence, $v=0$, so $s$ is injective and therefore an isomorphism.
    \end{proof}
\end{lemma}

\begin{definition}[Balanced product over chains]
    Let $\mathcal{C}_{\bullet}$ and $\mathcal{D}_{\bullet}$ be two chains with a linear right (respectively left) action of a group $G$. We define the balanced product of the two chains by 
    \begin{align}
        \begin{aligned}
            &(\mathcal{C} \boxtimes_{G} \mathcal{D})_{p,q} = \mathcal{C}_p \boxtimes_G \mathcal{D}_q, \\
            &\partial^v = \partial^C \boxtimes_G \mathrm{1}_{D}; \quad \partial^h = \mathrm{1}_{C}\boxtimes_G \partial^D
        \end{aligned}
    \end{align}
    which form the commuting square of the double complex 
    \begin{equation}
        \begin{tikzcd}
	{\mathcal{C}_p \boxtimes_{G}\mathcal{D}_q} & {\mathcal{C}_{p} \boxtimes_{G}\mathcal{D}_{q-1}} \\
	{\mathcal{C}_{p-1} \boxtimes_{G}\mathcal{D}_q} & {\mathcal{C}_{p-1} \boxtimes_{G}\mathcal{D}_{q-1}}
	\arrow["{\partial^h}", from=1-1, to=1-2]
	\arrow["{\partial^v}"', from=1-1, to=2-1]
	\arrow["{\partial^v}", from=1-2, to=2-2]
	\arrow["{\partial^h}"', from=2-1, to=2-2],
\end{tikzcd}
    \end{equation}
    The total complex is $\mathrm{Tot}_n(\mathcal{C} \boxtimes_G \mathcal{D}) := \bigoplus_{p+q=n} \mathcal{C}_p \boxtimes_G \mathcal{D}_q$ with boundary $\partial = \partial^v + \partial^h$, acting on classes as
    \begin{align}
        \partial[z \otimes z'] =  [\partial^C z \otimes z'] + [z \otimes \partial^D z'].
    \end{align}
\end{definition}

In this view, the lifted-product code is a special instance of the balanced product, as established by Lemmas~\ref{lemma: balanced-product-isomorphism-R-linear-tensor} and~\ref{lemma: binarized-tensor-balanced-chain}. 

\begin{lemma}\label{lemma: binarized-tensor-balanced-chain}
     Let $\mathcal{A}_{\bullet}$ and $\mathcal{B}_{\bullet}$ be two chain complexes over $R = \mathbb{F}_2[G]$ for some finite group $G$. Let $\mathcal{C}_{\bullet}:=\mbb B(\mathcal{A}_{\bullet})$ and $\mathcal{D}_{\bullet}:=\mbb B(\mathcal{B}_{\bullet})$ be the corresponding binarized chains, that is, $\mathcal{C}_i = \mbb B(\mathcal{A}_i)$ and $\mathcal{D}_i = \mbb B(\mathcal{B}_i)$ with the induced boundary maps
\begin{align}
    \mathcal{C}_{i+1} \xrightarrow{\mbb B(\partial^A_{i+1})} \mathcal{C}_{i}; \quad \mathcal{D}_{i+1} \xrightarrow{\mbb B(\partial^B_{i+1})} \mathcal{D}_{i}
\end{align}
for all $i \geq 0$. Then, we have 
\begin{align}
    H_k(\mathcal C \boxtimes_G \mathcal D) \cong H_k(\mathcal A \otimes_R \mathcal B).
\end{align}
\begin{proof}
    We have that $H_i(\mathcal{C}) \cong H_i(\mathcal{A})$, since the binarization $\mathbb{B}$ defines a chain isomorphism. Moreover, the degreewise isomorphisms of Lemma~\ref{lemma: balanced-product-isomorphism-R-linear-tensor} commute with the boundary maps $\partial^v$ and $\partial^h$, and hence assemble into an isomorphism of the total complexes. Taking homology yields the claim. 
\end{proof}
\end{lemma}

\subsection{Basics on classical error-correcting codes}\label{app-sec: classical-error-correcting-codes}

We introduce the notation used below. Throughout this subsection, $\mathcal{C}$ denotes a classical linear code; it is not to be confused with the chain complexes $\mathcal{C}_{\bullet}$ of Appendix~\ref{app-sec: basic-chain-complex} or the hypergraph cycle space appearing in the surgery constructions. We always work in characteristic $2$; that is, $q = 2^s$ is a power of $2$. 

\begin{definition}
   [Linear code] A linear code $\mathcal{C}$ of parameter $[n, k]_{\mathbb{F}_q}$ is a $k$-dimensional subspace of $\mathbb{F}^n_q$. 
\end{definition}

A central notion of a code $\mathcal{C}$ is the distance. We define the \emph{Hamming weight} of a vector $u \in \mathbb{F}^n_q$ as the number of positions at which it is nonzero. We say an $[n, k]_{\mathbb{F}_q}$ code $\mathcal{C}$ has distance $d$ (denoted $[n, k, d]_{\mathbb{F}_q}$) if the minimal Hamming weight of a nonzero codeword of $\mathcal{C}$ is $d$. We also introduce the following notation.

\begin{definition}[Standard form]\label{def: standard-form}
     For any $k$-dimensional linear code $\mathcal{C}$ over $\mathbb{F}^n_q$, generator matrix $G$ and parity check matrix $H$ can take the standard form
	\begin{align}\label{eq: classical-code-systematic-form}
		G & =\left(\begin{array}{ll}
			J^T & I_k
		\end{array}\right), \\
		H & =\left(\begin{array}{ll}
			I_{n-k} & J
		\end{array}\right),
	\end{align}
	where $J \in \mathbb{F}^{(n-k) \times k}_q$.
\end{definition}

\begin{definition}[Information set] \label{def: information-set}
   An information set $I \subseteq [n]$ of $\mathcal{C}$ is a set of cardinality $|I| = k$ on which the column restriction $c \mapsto c|_{I}$ is injective on $\mathcal{C}$; equivalently, it is a minimum-cardinality set such that $c|_{I} = 0$ implies $c = 0$. 
\end{definition}

Given an information set $I$ of $\mathcal{C}$, the \emph{systematic generator matrix with respect to $I$} is the unique generator matrix whose column restriction to $I$ is the identity $I_k$; Eq.~\eqref{eq: classical-code-systematic-form} is the special case where $I$ consists of the last $k$ coordinates.

\begin{lemma}\label{lemma: information-set}
    Let $\mathcal{C}$ be a classical code with distance $d$. For any index set $T \subseteq[n]$ with $|T| \geq n-d+1$, there exists an information set $I \subseteq T$ such that the following hold:

    \begin{enumerate}
        \item $|I|=k=\operatorname{dim} \mathcal{C}$.
        \item For any codeword $c \in \mathcal{C}$, if $\left.c\right|_I=\mathbf{0}$, then $c$ is the trivial codeword (i.e., the all-zeros vector of length $n$).
        \item The restriction of the parity check matrix $H$ to the complement of $I$ denoted $\left.H\right|_{\bar{I}}$ is full-rank. Hence, the information set serves as a $k$-puncture set for the parity check matrix $H$.
    \end{enumerate}

\begin{proof}
    For the first statement, consider the generator matrix $G \in \mathbb{F}_q^{k \times n}$. We claim that $\left.G\right|_T$ has rank $k$. Otherwise, there is a nonzero codeword $c$ with $c|_T=0$, supported entirely on $\bar T$. However, $|\bar{T}|<d$, so it cannot support any nonzero codeword. Hence, by row reduction and column swaps, we can always find $k$ pivots for $G$, so that $G$ is of the systematic form
\begin{align}
    G=\left(J^T \mid I_k\right), 
\end{align}
for a $k$-by- $k$ identity matrix and some $k$-by- $(n-k)$ matrix $J^T$. The last $k$ coordinates of the permuted code form an information set; since the pivot columns lie in $T$, undoing the column swaps yields an information set $I \subseteq T$. For the second statement, if $\left.c\right|_I=\mathbf{0}$, then clearly $c$ cannot be generated by $G$ unless $c$ is the all-zero vector. Lastly, denoting the complement of $I$ by $\bar{I}$, we show that $\left.H\right|_{\bar{I}}$ has full rank, i.e., $\operatorname{rk}(\left.H\right|_{\bar{I}}) = n-k$. This follows from the fact that any vector whose support is entirely contained in $\bar{I}$ must correspond to a distinct syndrome under $H$. Suppose, for contradiction, that two such vectors $u$ and $v$ satisfy $H u=H v$. Then their difference $u-v$ lies in the code $\mathcal{C}$, since $H(u-v)=0$. However, $u-v$ is also supported entirely on $\bar{I}$, which contradicts the fact that $I$ is an information set: no nonzero codeword can be zero on all positions in $I$. Therefore, all such vectors produce distinct syndromes, implying that the columns of $\left.H\right|_{\bar{I}}$ are linearly independent. Hence, $\left.H\right|_{\bar{I}}$ has full rank $n-k$, and $I$ is a $k$-puncture set for the parity check matrix $H$.
\end{proof}
\end{lemma}

In what follows, we consider two classical code modifications, called augmentation and puncture. 

\begin{definition}[Puncture and augmentation]\label{def: puncture-augmentation}
     Let $H \in \mathbb{F}^{(n-k) \times n}_q$ be the parity check matrix of the code $\mathcal{C}$ and $I \subseteq [n]$ be an information set. Let $\kappa \subseteq [n]$ be a column subset whose complement satisfies $\bar{\kappa} \subseteq I$. The \emph{$\bar{\kappa}$-puncture} of $\mathcal{C}$ is the code with parity check matrix $H|_{\kappa}$, the column restriction of $H$ onto $\kappa$. Let $H_{\eta} \in \mathbb{F}^{\eta \times n}_q$ be a matrix whose rows are mutually linearly independent and supported exclusively on the information set, $H_{\eta}|_{\bar{I}} = 0$. The \emph{$\eta$-augmentation} of $\mathcal{C}$ is the code with the row-augmented parity check matrix
    \begin{align}
      H' =  \begin{pmatrix}
        H \\
        H_{\eta}
    \end{pmatrix}.
    \end{align}
\end{definition}

These operations belong to a general class of code modifications (see, for example, the review~\cite{hall2015codingnotes}), and we show below that they are distance-preserving. Note that we depart from the notation of Ref.~\cite{hall2015codingnotes} in favor of a more intuitive presentation. 

\begin{lemma}
    Applied to an $[n, k, d]$ code $\mathcal{C}$, the augmentation and puncture operations result in a modified code $\mathcal{C}'$ with dimension $k'$ and distance $d'$ satisfying $k' \leq k$ and $d' \geq d$. 
    \begin{proof}
        For the $\eta$-augmentation, every codeword of $\mathcal{C}'$ satisfies the checks of $H$, so $\mathcal{C}' \subseteq \mathcal{C}$ is a subcode: $k' \leq k$, and every nonzero codeword of $\mathcal{C}'$ is a nonzero codeword of $\mathcal{C}$, so $d' \geq d$. For the $\bar{\kappa}$-puncture, a vector $u \in \mathbb{F}^{\kappa}_q$ satisfies $H|_{\kappa} u = 0$ if and only if its extension by zeros on $\bar{\kappa}$ is a codeword of $\mathcal{C}$; hence the punctured code is the restriction to $\kappa$ of the subcode $\{c \in \mathcal{C}: \operatorname{supp}(c) \subseteq \kappa \}$. In particular, every nonzero codeword of $\mathcal{C}'$ extends to a nonzero codeword of $\mathcal{C}$ of the same weight, so $d' \geq d$. Since $\bar{\kappa} \subseteq I$ and the column restriction onto $I$ is a bijection on $\mathcal{C}$, the constraints $c|_{\bar{\kappa}} = 0$ impose $|\bar{\kappa}|$ independent conditions, so $k' = k - |\bar{\kappa}| \leq k$.
    \end{proof}
\end{lemma}

Note that the $\bar{\kappa}$-puncture above coincides with \emph{shortening} at $\bar{\kappa}$ in the standard coding-theory terminology~\cite{hall2015codingnotes}; standard puncturing, which deletes columns without the information-set constraint, may decrease the distance.

\section{Basis characterization of lifted-product codes}\label{app-sec: lp-basis-characterization}
We detail explicitly the basis characterization for the lifted-product code. We now give an elementary treatment of finding the logical basis for a general class of lifted-product codes, defined on $R = \mathbb{F}_2[x]/(x^l + 1)$ (Appendix~\ref{app-sec: premliminaries}). 

\subsection{General formulation via homological algebra}\label{subsec: general-formulation-homological-algebra}
An $R$-module $M$ is \emph{finitely generated} if there is a finite subset $E\subseteq M$ such that every element of $M$ is an $R$-linear combination of elements of $E$. In this way, we denote $M = \mathrm{span}_R\{e \in E \}$. Note that $E$ might not be a basis set, which stands in a key contrast to the case of modules over a field (vector spaces).
\begin{definition}[Free and projective module]\label{def: free-projective-module}
    An $R$-module $M$ is free if there exist a finite index set $I \subseteq \mathbb{N}$ and a set $E = \{ e_i \}_{i\in I}$ such that, for every $m \in M$, there exist unique coefficients $\{r_i \in R \}_{i \in I}$ with $m = \sum_{i\in I} r_i e_i$. In this case, we refer to $E$ as the basis set. We say that an $R$-module $P$ is projective if it is a direct summand of a free $R$-module.

\end{definition}
\begin{definition}[Simple, semi-simple modules ]
      Let $R$ be any ring. An $R$-module $M$ is said to be \emph{simple} if there is no non-zero proper submodule. An $R$-module is called \emph{semi-simple} if $M$ is isomorphic to a direct sum of a family of simple modules.
\end{definition}
Furthermore, a ring $R$ is \emph{simple} if, viewed as a module over itself, it is a simple module. A ring $R$ is \emph{semi-simple} if it is a semi-simple module over itself. A useful result we will use is the following; see, for example, Ref.~\cite{weibel1994}.
\begin{proposition}\label{proposition: semi-simple-iff-projective}
A (left)-$R$ module $C$ is projective if and only if every short exact sequence
\begin{equation}
    \begin{tikzcd}
	0 & A & B & C & 0
	\arrow[from=1-1, to=1-2]
	\arrow["\psi", from=1-2, to=1-3]
	\arrow["\phi", from=1-3, to=1-4]
	\arrow[from=1-4, to=1-5]
\end{tikzcd}
\end{equation}
splits for (left)-$R$-modules $A, B$. That is, there exists an $R$-linear \emph{section} $s: C \rightarrow B$ with $\phi \circ s = \mathrm{1}_{C}$ such that $B \cong \psi(A) \oplus s(C)$. Equivalently, there exists an $R$-linear \emph{retraction} $r: B \rightarrow A$ such that $r \circ \psi = \mathrm{1}_{A}$. A ring $R$ is semi-simple if and only if every left $R$-module is projective.
\end{proposition}
\begin{theorem}
    [K{\"u}nneth theorem for a semi-simple ring]\label{thm: kunneth-theorem-semi-simple-with-balanced}
     Let $\mathcal{A}_{\bullet}$ and $\mathcal{B}_{\bullet}$ be chain complexes of right and left modules, respectively, over $R = \mathbb{F}_2[G]$ for any finite, Abelian group $G$. If $R$ is semi-simple or, equivalently, if the characteristic of $\mathbb{F}_2$ does not divide the order of $G$, then
\begin{equation}\label{expr: kunneth-semi-simple-ring}
    \begin{tikzcd}
		{\bigoplus_{p+q=k} H_p (\mathcal A) \otimes_R H_q(\mathcal B)} & {H_k(\mathcal A \otimes_R \mathcal B)}
		\arrow["\cong", from=1-1, to=1-2].
\end{tikzcd}
\end{equation}
Let $\mathcal{C}_{\bullet}:=\mbb B(\mathcal{A}_{\bullet})$ and $\mathcal{D}_{\bullet}:=\mbb B(\mathcal{B}_{\bullet})$ be the corresponding binarized chains, as in Lemma~\ref{lemma: binarized-tensor-balanced-chain}. Then
\begin{align}\label{expr: kunneth-semi-simple-balanced-product}
    \bigoplus_{p+q=k} H_p(\mathcal C) \boxtimes_G H_q(\mathcal D)
    \cong H_k(\mathcal C \boxtimes_G \mathcal D).
\end{align}
\begin{proof}
    Expression~\eqref{expr: kunneth-semi-simple-ring} follows from a standard fact when $R$ is semi-simple. The expression~\eqref{expr: kunneth-semi-simple-balanced-product} then follows from Lemma~\ref{lemma: binarized-tensor-balanced-chain}.
\end{proof}
\end{theorem}

\begin{remark}
    This can be proved beyond the Abelian group for any finite group $G$ of odd order; see Lemma~19 from Ref.~\cite{Breuckmann2021BalancedProductQuantumCodes}.
\end{remark}

We now give an explicit example illustrating the balanced-product route; the componentwise computation for the same code via algebraic roots appears in the corresponding example of Section~\ref{subsec: basis-algebraic-root-theory}.
\begin{example}[Generalized bicycle codes Ref.~\cite{webster2025explicitconstructionlowoverheadgadgets}]
    We look into the code examples $\LP_{31}(A, B)$ for $A = 1 + x^6 + x^{15}$ and $B =1 + x^5 + x^7 $ in Ref.~\cite{webster2025explicitconstructionlowoverheadgadgets}. In this case, denote 
\begin{align}
    \begin{aligned}
        c = &x^{30} + x^{29} + x^{27} + x^{26} + x^{23} + x^{22} + x^{21} \\
        &+ x^{16} + x^{15} + x^{13} + x^{11} + x^{8} + x^{4} + x^{2} + x + 1.
    \end{aligned}
\end{align}
Then we can compute that (note $c^2 =c$)
\begin{align}
    G_A=G_B=c, \qquad E_A=E_B=\begin{pmatrix}c&cx&cx^2&cx^3&cx^4\end{pmatrix}^T.
\end{align}
Hence, according to Theorem~\ref{thm: main-LP-basis-characterization}
\begin{align}
    \left\{\mbb B(c) \oplus 0, \ldots, \mbb B(cx^4) \oplus 0,
    0 \oplus \mbb B(c), \ldots, 0 \oplus \mbb B(cx^4)\right\}.
\end{align}
We now describe how to use the balanced product to find the basis set above. The set
\begin{align}
    \left\{\mbb B(c),\mbb B(cx),\mbb B(cx^2),\mbb B(cx^3),\mbb B(cx^4)\right\}
\end{align}
is a basis of both $\ker_{\mbb F_2}\mbb B(A)$ and $\ker_{\mbb F_2}\mbb B(B)$. Similarly, the chosen representatives $\mbb B(E_A)=\mbb B(E_B)$ give the basis
\begin{align}
    \left\{\mbb B(c),\mbb B(cx),\mbb B(cx^2),\mbb B(cx^3),\mbb B(cx^4)\right\}
\end{align}
of both $\coker_{\mbb F_2}\mbb B(A)$ and $\coker_{\mbb F_2}\mbb B(B)$. The $\mathbb{F}_2$-tensor product therefore gives $50$ logical operators before applying the balanced quotient relation
\begin{align}
    \mbb B(ux^s) \otimes_{\mathbb{F}_2} \mbb B(v) \sim \mbb B(u) \otimes_{\mathbb{F}_2} \mbb B(x^sv).
\end{align}
In this case, 
\begin{align}
    \mbb B(cx^i) \otimes_{\mathbb{F}_2} \mbb B(cx^j) \oplus 0
    \sim \mbb B(cx^{i+j}) \otimes_{\mathbb{F}_2} \mbb B(c) \oplus 0,
\end{align}
which reduces to $5$ linearly independent basis elements (To see this, note that $cx^{i+j}$ is always given by $\mathbb{F}_2$-linear combinations of $\{c, cx, cx^2, cx^3, cx^4\}$ , which is proved in Lemma~\ref{lemma: basis-dual-basis-from-finite-field-ring}), and similarly for the right sector. This reduces the number of logical basis operators to $10$, matching the earlier computation. 
\end{example}

The K{\"u}nneth theorem for a semi-simple ring indicates that if we wish to compute the first homology group, we need to characterize $H_1(A) = \ker_R A$ and $H_0(A) = \coker_R A$ and similarly for $\mathcal{B}_{\bullet}$.
Recall from Lemma~\ref{lemma: balanced-product-isomorphism-R-linear-tensor} that $M \otimes_R N \cong M \boxtimes_G N$; for odd $|G|$, this isomorphism admits the canonical $\mathbb{F}_2$-linear section $s([m \otimes_{\mathbb{F}_2} n]) = \sum_{g \in G} mg \otimes_R g^{-1}n$.
\begin{remark}[$\ker_R A$ and $\coker_R A$ as projective modules]
     Let $A: R^n \rightarrow R^m $. In general $\ker_R A$ and $\coker_R A$ are not free modules since there could exist some element $a \in R$ which annihilates $\ker A$ or $\coker A$. Our main concern is whether $\ker A$ and $\coker A$ are projective. For $\ker A$, we have the following exact sequence:
    \begin{equation}\label{exseq: ker-A}
        \begin{tikzcd}
	0 & {\ker_R A} & {R^n} & {\IM_R A} & 0
	\arrow[from=1-1, to=1-2]
	\arrow["{\iota }", from=1-2, to=1-3]
	\arrow["A", from=1-3, to=1-4]
	\arrow[from=1-4, to=1-5]
\end{tikzcd},
    \end{equation}
     where $\iota$ is the embedding map. Suppose there exists a retraction $r: R^n \rightarrow \ker_R A$, then $R^n = \ker_R r \oplus  \ker_R A$. Recall that $\coker_R A := R^m / \IM_R A$, which can be represented by the short exact sequence:
\begin{equation}\label{exseq: coker-A}
    \begin{tikzcd}
	0 & {\IM_R A} & {R^m} & {\coker_R A} & 0
	\arrow[from=1-1, to=1-2]
	\arrow["\psi", from=1-2, to=1-3]
	\arrow["\phi", from=1-3, to=1-4]
	\arrow[from=1-4, to=1-5].
\end{tikzcd}
\end{equation}
Suppose the exact sequence~\eqref{exseq: coker-A} splits so that there exists a section $s$ and let $\psi$ be inclusion map. It is clear that $R^m = \IM_R A \oplus \IM_R s$ and $\IM_R s \cong \coker_R A$. In this case, $\coker_R A$ is a projective $R$-module.
\end{remark}
\begin{example} Let $A = (1+x)$ and $R = \mathbb{F}_2[x]/(x^l+1)$. We systematically study the following examples.
\begin{enumerate}
        \item ($\coker_R A$). Suppose we choose $s([1]) =1$ and $s([0]) =0$. This is not a valid $R$-linear map since $s((1+x)[1]) = (1+x) s([1]) = 1+x = s([0]) =0$ which leads to a contradiction. The only consistent choice is that $s([1]) = e$ such that $(1+x) e =0$ so that $e = 1+x + \cdots + x^{l-1}$. If $l$ is odd, then $\mathrm{gcd}(1+x, e) =1$, hence, it defines a valid section. If $l$ is even, by the above, we can only choose $s([1]) = e$. But we have that $\phi \circ s = 1_{\coker_R A}$ so that $\phi \circ s([1]) = \phi(e) =[0] \neq [1]$, since in this case $x+1 | e$. Hence, we conclude that there is no splitting in the even case.
        \item ($\ker_R A$). For odd $l$, we can define an $R$-linear map $r: R \rightarrow \ker_R A$ by $r(1) = e \cdot 1$, where $e = 1+ x+ \cdots + x^{l-1}$. Note that in this case, $\ker_R A= (e)$ and $r$ is clearly a retraction; hence $R = (e) \oplus (1+x)$, which can always be given by the Bezout's identity Lemma~\ref{lemma: bezout-lemma}. However, suppose $l$ is even and say $l=2$. Then $\ker_R A = (1+x)$ which requires us to define that $r(1) = (1+x) \cdot 1$ but in this case $r(1+x) = 0$ which is a contradiction.
    \end{enumerate}

\end{example}
Since $\ker_R A$ is an $R$-submodule of $R^n$, we can also denote a set of generators $G_{A}$ such that $\mathrm{span}_R(G_{A}) = \ker_R A$. A subtle issue is that $\coker_R A$ is not directly an $R$-submodule of $R^m$, hence, we need the following characterization for finding its generating set in $R^m$.
\begin{lemma}\label{lemma: section-maps-coker-R}
    Let $\mathcal{A}_{\bullet}: \mathcal{A}_1 = R^{n_A} \xrightarrow{A} \mathcal{A}_0 = R^{m_A}$ be the base chain of Definition~\ref{def: main-2d-LPAB}. The following are equivalent.
    \begin{enumerate}
        \item There exists a finite set $E_{A} = \{ e^{(i)}_{\mathcal{A}_0} \} \subset R^{m_A}$ such that any element $u \in R^{m_A}$ can be written as
        \begin{align}\label{eq: splitting-condition-1}
            \sum_{i} r_i e^{(i)}_{\mathcal{A}_0}  + Ah = u
        \end{align}
        for some coefficients $r_i \in R$ and some $h \in R^{n_A}$. Furthermore, any formal sum satisfies
        \begin{align}\label{eq: splitting-condition-2}
            \sum_{i} r_i e^{(i)}_{\mathcal{A}_0} \in \IM_R A \iff  \sum_{i} r_i e^{(i)}_{\mathcal{A}_0}  = 0.
        \end{align}
        Let $M_E:=\rs_R(E_{A})\subseteq R^{m_A}$ be the space of formal sums. Then $M_E$ is an $R$-module and
        \begin{align}
            R^{m_A}=\IM A\oplus M_E, \qquad M_E\cong\coker_R A.
        \end{align}
        \item The exact sequence~\eqref{exseq: coker-A} splits; that is, there exists an $R$-linear section $s: \coker_R A \rightarrow R^{m_A}$ with $\phi \circ s = \mathrm{1}_{\coker_R A}$.
    \end{enumerate}
    \begin{proof}
        It is easy to see from $(1)$ to $(2)$. We just need to show $M_E$ is a module, i.e. closure. With respect to  $r \in R$, $rm \in M_E $ if $m \in M_E$, and $m, m' \in M_E$, $m + m' \in M_E$, since if $m + m' \in \IM_R A$, then $m +m' =0$ satisfying the constraint Eq.~\eqref{eq: splitting-condition-2}. From $(2)$ to $(1)$, choose a generating set $\{[v_i]\}$ of $\coker_R A$ and set $e^{(i)}_{\mathcal{A}_0} := s([v_i])$, with $E_{A} := \{ e^{(i)}_{\mathcal{A}_0} \}$. Then it is clear that $M_E := \IM_R s$ is an $R$-submodule (since $\coker_R A$ is an $R$-module and the map $s$ is an $R$-linear homomorphism). Furthermore by $\phi \circ s = \mathrm{1}_{\coker_R A}$, we note that
        \begin{align}
            \phi |_{M_E}: M_E \rightarrow \coker_R A
        \end{align}
        is an isomorphism. By surjection, Eq.~\eqref{eq: splitting-condition-1} is satisfied. By injection, we have that
        \begin{align}
            \phi \left( \sum_{i } r_i e^{(i)}_{\mathcal{A}_0} \right) = [0]  \implies \sum_i r_i e^{(i)}_{\mathcal{A}_0} =0.
        \end{align}
        That satisfies the condition Eq.~\eqref{eq: splitting-condition-2}.
    \end{proof}
\end{lemma}

With this understanding, we can give a first, intuitive basis characterization of the lifted-product codes, given by Theorem~\ref{thm: main-LP-basis-characterization}. 

\begin{lemma}[First basis characterization of LP code]\label{lemma: first-LP-basis-characterization}
    Let $R = \mathbb{F}_2[x]/(x^l + 1)$ with odd $l$, and let the $2D$ lifted-product code $\LP_l(A, B)$ be given in Definition~\ref{def: main-2d-LPAB}. Then the first homology group $H_1(\mc{M})$, corresponding to $Z$-type logical operators, admits a minimal generator matrix of the form
    \begin{align}\label{eq: generator-LZ-LP}
        L_{Z}^M =   G_{A} \otimes_R E_{B} \bigoplus E_{A} \otimes_R G_{B},
    \end{align}
    where, for $C \in \{A, B\}$:
    \begin{enumerate}
        \item $\rs_R(G_C) = \ker_R C$ and $\rs_{\mathbb{F}_2}(E_C) \cong \coker_R C$ in the sense of Lemma~\ref{lemma: section-maps-coker-R}.
        \item The nonzero rows of $L^M_Z$ generate $H_1(\mc{M})$ over $\mathbb{F}_2$.
    \end{enumerate}
    Similarly, the first cohomology group $H^1(\mc{M})$, corresponding to $X$-type logical operators, admits a minimal generator matrix of the form
    \begin{align}\label{eq: generator-LX-LP}
        L_{X}^M =  E_{A^*} \otimes_R G_{B^*} \bigoplus G_{A^*} \otimes_R E_{B^*},
    \end{align}
    where we take $\rs_{\mathbb{F}_2}(G_{C^*}) = \ker_R C^*$ and $\rs_R(E_{C^*}) \cong \coker_R C^*$ in the sense of Lemma~\ref{lemma: section-maps-coker-R}, for $C \in \{A, B\}$, and the nonzero rows of $L^M_X$ generate $H^1(\mc{M})$ over $\mathbb{F}_2$. Consequently, the nonzero rows of $\mathbb{B}(L^M_Z)$ and $\mathbb{B}(L^M_X)$ generate the binary $Z$- and $X$-type logical operators.
\end{lemma}
\begin{proof}
    We prove the claim for the summand $G_A \otimes_R E_B$; the other summands follow analogously. Let $M_E := \rs_{\mathbb{F}_2}(E_B)$, an $R$-submodule with $M_E \cong \coker_R B$ by Lemma~\ref{lemma: section-maps-coker-R}, and denote by $g^{(i)}_A$ and $e^{(j)}_B$ the rows. First, $\operatorname{span}_{\mathbb{F}_2}\{ g^{(i)}_A \otimes_R e^{(j)}_B \} \subseteq \ker_R A \otimes_R M_E$. Conversely, it suffices to consider a simple tensor $m \otimes_R n$ with $m = \sum_i u_i g^{(i)}_A$, $u_i \in R$, and $n \in M_E$; then
    \begin{align}
        m \otimes_R n = \sum_i g^{(i)}_A \otimes_R u_i n,
    \end{align}
    and $u_i n \in M_E$ since $M_E$ is an $R$-submodule, so that $m \otimes_R n \in \operatorname{span}_{\mathbb{F}_2}\{ g^{(i)}_A \otimes_R e^{(j)}_B \}$. Hence $\rs_{\mathbb{F}_2}(G_A \otimes_R E_B) = \ker_R A \otimes_R M_E \cong \ker_R A \otimes_R \coker_R B$, and the K\"unneth theorem, Lemma~\ref{thm: kunneth-theorem-semi-simple}, gives $\rs_{\mathbb{F}_2}(L^M_Z) = H_1(\mc{M})$ for any such choice of generating sets.
\end{proof}

Though this presentation is clean, it does not capture the inherent physics of the logical operators.
In particular, we cannot count the logical dimension from the above characterization. The following example illustrates that the nonzero rows of $G_A \otimes_R E_B$ need not be $\mathbb{F}_2$-linearly independent, even if $G_A$ and $E_B$ are minimal generating sets.

\begin{counterexample}[Minimality]\label{counterexample: first-LP-basis-minimality}
    Let $l = 3$ and $R = \mathbb{F}_2[x]/(x^3+1)$. Take
    \begin{align}
        A = \left( 1+x+x^2, \; 1+x+x^2 \right), \quad B = \left( \begin{array}{c} 1+x \\ 1 \end{array} \right).
    \end{align}
    Then $G_A$ with rows $u_1 = (1,\, 1+x+x^2)$ and $u_2 = (0,\, x+x^2)$ satisfies $\rs_R(G_A) = \ker_R A$ with the minimal number of generators, and $E_B$ with rows $\{(1,0), (x,0), (x^2,0)\}$ is a minimal $\mathbb{F}_2$-basis of $M_E = R \oplus 0 \cong \coker_R B$, where $R^2 = \IM_R B \oplus M_E$. All six rows of $G_A \otimes_R E_B$ are nonzero, yet
    \begin{align}
        \sum^{2}_{m=0} u_2 \otimes_R (x^m, 0) = 0,
    \end{align}
    since $(x+x^2)(1+x+x^2) = x(1+x^3) = 0$. Hence the nonzero rows are $\mathbb{F}_2$-linearly dependent: minimality of the generating sets alone does not imply linear independence of the nonzero rows.
\end{counterexample}

\subsection{Field decompositions}\label{sec: field-decompositions}
  By the rank-nullity theorem, the exact sequence also splits over finite fields. Furthermore, such $E_{A}$ can be constructed by choosing information sets, and $M_E$ is naturally a vector space. By proposition~\ref{proposition: semi-simple-iff-projective}, we can always find a set $E_{A}$ satisfying conditions Eq.~\eqref{eq: splitting-condition-1} and Eq.~\eqref{eq: splitting-condition-2}, which is of primary interest in subsequent sections (as well as $G_{A}$). The choices of $G_{A}$ and $E_{A}$ need not to be unique; they are given by $R$-linear combinations of each other. It is, however, desirable to find such $G_A$ and $E_{A}$ which are minimal. In what follows, we introduce the main mathematical machinery in computing explicitly the minimal generating set $G_A$ and $E_A$ through the \emph{field decomposition}. Recall that we can express
\begin{align}
    x^l +1 = \prod^t_{i=1} (b^{(i)})^{m_i},
\end{align}
for some integer multiplicities $m_i$. Application of the Chinese remainder theorem (Lemma~\ref{lemma: chinese-reminder})
\begin{align}
    R = \mathbb{F}_2[x]/(x^l +1) \cong \frac{\mathbb{F}_2[x]}{((b^{(1)})^{m_1})} \times \cdots \times \frac{\mathbb{F}_2[x]}{((b^{(t)})^{m_t})}.
\end{align}
Suppose that $m_i =1$. Then we have that $\mathbb{F}_2[x]/(b^{(i)}) \cong \mathbb{F}_{q_i}$ for $q_i = 2^{\deg b^{(i)}}$.
\begin{example}
    For odd $l$, each irreducible polynomial factor is multiplicity-free.  Hence, $R$ is semi-simple. When $l$ is even, at least one irreducible polynomial factor has multiplicity greater than one.
\end{example}
Let us consider $x^l +1 = \prod^t_{i=1} (b^{(i)})^{m_i}$ where $b^{(i)}|x^l +1$ is an irreducible polynomial. By Bezout's identity and the fact that $R$ is a principal ideal ring (PIR) (Lemma~\ref{lemma: bezout-lemma}), there exist polynomials $r_i(x)$ and $s_i(x)$ such that
\begin{align}
    r_i(x)b^{(i)}(x) + \prod_{j \neq i} b^{(j)}(x)s_i(x) = \operatorname{gcd}(b^{(i)}, \prod_{j \neq i} b^{(j)}).
\end{align}
We work in the case that $\operatorname{gcd}(b^{(i)}, \prod_{j \neq i} b^{(j)}) = 1$, which holds automatically for odd $l$ by Remark~\ref{remark: square-free-odd-l}. In this case, we say that the central idempotent $c_{b^{(i)}}$ associated with $b^{(i)}$ is
\begin{align}\label{eq: central-idempotent}
    c_{b^{(i)}} := \prod_{j \neq i} b^{(j)}(x)s_i(x).
\end{align}
In particular, $c_{b^{(i)}} \equiv 1 \pmod{b^{(i)}}$ and $c_{b^{(i)}} \equiv 0 \pmod{b^{(j)}}$ for $j \neq i$. For the factor $1+x$, the central idempotent is $\chi := \sum_{m=0}^{l-1} x^m$, as in Section~\ref{sec: main-LP-basis}.
The central idempotents behave like projection operators:
\begin{lemma}\label{lemma: central-idempotent-property}
    Let $c_{b^{(i)}}$ be a central idempotent for the irreducible factor $b^{(i)}$. Then the following hold over $R$:
    \begin{itemize}
     \item $c_{b^{(i)}} b^{(i)} = 0$,
        \item $(c_{b^{(i)}})^2 = c_{b^{(i)}}$,
        \item $c_{b^{(i)}} c_{b^{(j)}} = 0$ for $i\neq j$,
        \item $\sum^t_{i=1} c_{b^{(i)}} = 1$ if $l$ is odd, or equivalently, $m_i =1$ for all $i=1, \cdots, t$.
    \end{itemize}
    Furthermore, there exists a natural isomorphism
    \begin{align}
         ( c_{b^{(i)}} ) \cong \mathbb{F}_2[x]/(b^{(i)}).
    \end{align}

    \begin{proof}
    The first condition is straightforward by inspection. By Eq.~\eqref{eq: central-idempotent}, let $g_i = \prod_{j \neq i} b^{(j)}(x)$.
       \begin{align}
           (c_{b^{(i)}})^2 = s^2_i g^2_i = s_i g_i (1 + r b^{(i)})  = s_i g_i = c_{b^{(i)}},
       \end{align}
       where we used the fact that $b^{(i)} c_{b^{(i)}} = x^l +1 = 0$. For two different indices, $c_{b^{(i)}} c_{b^{(j)}} = s_i g_i s_jg_j$. Note that $b^{(i)} | g_j$ and $b^{(j)} | g_i$, so that $x^l +1 | c_{b^{(i)}} c_{b^{(j)}}$, which implies the result. Finally, we show that $\sum^t_{i=1} c_{b^{(i)}} = 1$. Notice that $1 -\sum^t_{i=1} c_{b^{(i)}}$ must lie in every prime ideal generated by $b^{(i)}$, $i=1, \cdots, t$. However, since they are coprime, it necessarily follows that $1 -\sum^t_i c_{b^{(i)}}$ must be of the form $a (x^l +1) =0$. Finally, we show the natural isomorphism $\phi_i: \langle c_{b^{(i)}} \rangle_{R} \cong \mathbb{F}_{q_i}$, given by
       \begin{align}
           \phi_i: a c_{b^{(i)}} \mod x^l +1 \mapsto a  \mod b^{(i)},
       \end{align}
       which naturally preserves the equivalence class since $b^{(i)} | x^l +1$. Note that $\ker_R \phi_i = \{a \in c_{b^{(i)}} R : a \in ( b^{(i)} ) \} $. Due to the coprime assumption, it follows that $\ker_R \phi_i = \{0 \}$. The surjectivity follows from the definition.
    \end{proof}
\end{lemma}
\begin{remark}
    Note that for general $l$, there might be a case where $b^{(i)}$ has multiplicity $m_i > 1$. In this case, we could still define the central idempotent $c_{b^{(i)}}$ with respect to $(b^{(i)})^{m_i}$. The main caveat is that the Chinese remainder theorem (CRT) gives $\mathbb{F}_2[x]/((b^{(i)})^{m_i})$, which is a local ring rather than a field.
\end{remark}
\begin{remark}
    We could show the isomorphism using techniques from homological algebra, which will be important in the following section. To see this, let us denote the ideal $I_i=(1-c_{b^{(i)}})R=(b^{(i)})$. Then we claim that we can build the following short exact sequences:
\begin{equation}
    \begin{tikzcd}
	0 & {I_i} & {R} & {c_{b^{(i)}}R} & 0
	\arrow[from=1-1, to=1-2]
	\arrow["{\varphi_i}", from=1-2, to=1-3]
	\arrow["{\psi_i}", from=1-3, to=1-4]
	\arrow[from=1-4, to=1-5].
\end{tikzcd}
\end{equation}
To see that this is indeed an exact sequence, first $\varphi_i$, given by the embedding, is naturally injective, and consequently $\psi_i$ is surjective. Furthermore, we have that $\ker_R \psi_i = \{h \in R: h c_{b^{(i)}} =0 \mod x^l +1 \}$. This holds if and only if $h$ contains the factor $b^{(i)}$, which implies that $h \in \IM_R \varphi_i$. The converse direction is straightforward to verify. The key observation is that the above exact sequence splits. In particular, we can construct a \emph{retraction} $r_i: R \rightarrow I_i$ with $r_i(h)$ given by the canonical projection. Then the standard result in homological algebra implies that the sequence splits as
\begin{align}
    R \cong I_i \oplus c_{b^{(i)}} R.
\end{align}
Note that we can also apply the canonical short exact sequence given by
\begin{equation}
    \begin{tikzcd}
	0 & {I_i} & {R} & {R/I_i} & 0
	\arrow[from=1-1, to=1-2]
	\arrow["{\alpha_i}", from=1-2, to=1-3]
	\arrow["{\gamma_i}", from=1-3, to=1-4]
	\arrow[from=1-4, to=1-5].
\end{tikzcd}
\end{equation}
Similarly, the short exact sequence splits and $R \cong I_i \oplus R / I_i$, which also concludes the result.
\end{remark}

\begin{lemma}\label{lemma: module-decomposition-odd-l}
Let $l$ be odd, and let $\pi_i: R \rightarrow c_{b^{(i)}} R$ be the orthogonal projection map, such that $\sum_i \pi_i=\operatorname{id}_{R}$ by the fourth property of Lemma~\ref{lemma: central-idempotent-property}. Let $M$ be any $R$-module. This induces the isomorphism
\begin{align}
    M \cong \bigoplus^t_{i=1} c_{b^{(i)}} M.
\end{align}
\begin{proof}
    Let $I_i=(1-c_{b^{(i)}})R=(b^{(i)})$. We can extend $\pi_i$ to the projection $\pi^M_i:M\rightarrow c_{b^{(i)}}M$ defined by $\pi^M_i(m)=c_{b^{(i)}}m$. This gives the short exact sequence
\[\begin{tikzcd}
	0 & {I_iM} & M & {c_{b^{(i)}}M} & 0
	\arrow[from=1-1, to=1-2]
	\arrow["{\alpha_i}", from=1-2, to=1-3]
	\arrow["{\pi_i^M}", from=1-3, to=1-4]
	\arrow[from=1-4, to=1-5].
\end{tikzcd}\]
Indeed, $\pi_i^M$ is surjective and $\ker_R\pi_i^M=I_iM$, since $c_{b^{(i)}}m=0$ if and only if $m=(1-c_{b^{(i)}})m$. Hence,
\begin{align}
    c_{b^{(i)}}M\cong M/I_iM.
\end{align}
Finally, $\sum_i c_{b^{(i)}}=1$ implies that every $m\in M$ can be written as $m=\sum_i c_{b^{(i)}}m$. Since $c_{b^{(i)}}c_{b^{(j)}}=0$ for $i\neq j$, this decomposition is direct, which proves the result.
\end{proof}
\end{lemma}
\begin{example}
    An important example is the kernel: let $A: R^{n_A} \rightarrow R^{m_A}$ and consider $M = \ker_R A$, an $R$-submodule of $R^{n_A}$. Then there exists an isomorphism
    \begin{align}
        \ker_R A \cong \bigoplus^t_{i=1} c_{b^{(i)}} \ker_R A,
    \end{align}
    which underlies the componentwise computation in Lemma~\ref{lemma: kernel-A-odd} and Lemma~\ref{lemma: cokernel-A-odd}.
\end{example}
\begin{example}
    Let $M = R^n_l$ and $N = R^m_l$, and let us denote by $\operatorname{Hom}_R(M, N)$ the space of $R$-module homomorphisms between $M$ and $N$. Then $\operatorname{Hom}_R(M, N)$ is also an $R$-module. Then there exists an isomorphism
    \begin{align}
        \operatorname{Hom}_R(M, N) \cong \bigoplus^t_{i=1} c_{b^{(i)}} \operatorname{Hom}_R(M, N).
    \end{align}
\end{example}
Recall the tensor product over $R$ from Definition~\ref{def: tensor-product-over-R}. We now compute it componentwise along the idempotent decomposition.
\begin{lemma}
     Let $M$ be any $R$-module, and let $c_{b^{(i)}}$ be the central idempotent with respect to a (square-free) irreducible factor $b^{(i)} |x^l +1$. Then we have that
    \begin{align}
            c_{b^{(i)}} M \cong M \otimes_R R/I_i \cong M / I_i M,
    \end{align}
    where $I_{i} = (b^{(i)})$ is the ideal generated by the (square-free) irreducible factor.
    \begin{proof}
        Let us work through these conditions. Denote the map
        \begin{align}
            \psi: c_{b^{(i)}} M \rightarrow M \otimes_R R/I_i; \quad \psi(c_{b^{(i)}}m) = m \otimes _R\bar{1},
        \end{align}
        where $\bar{1}$ denotes the identity of $R/I_i$. We first show that this map is well-defined. Suppose that if $c_{b^{(i)}} m = c_{b^{(i)}}m' $, then we have that $m -m' \in (1-c_{b^{(i)}})M = I_iM$, since by the construction of the central idempotent Eq.~\eqref{eq: central-idempotent} and Lemma~\ref{lemma: central-idempotent-property}, there exist some $u,u' \in R$ such that $m = b^{(i)} u + c_{b^{(i)}}m$ and $m' = b^{(i)} u' + c_{b^{(i)}}m$.  Denote $h = u - u'$, and  $\psi(c_{b^{(i)}}(m -m')) = (b^{(i)}h) \otimes _R1_{I_i} = h \otimes _R b^{(i)} = 0$. We then show it constitutes an isomorphism by explicitly characterizing its inverse. Define
        \begin{align}
            \phi: M \otimes_R R/I_i \rightarrow c_{b^{(i)}} M; \quad \quad \phi(m \otimes _R\bar{r}) = c_{b^{(i)}} m r,
        \end{align}
        where $\bar{r}$ is defined on $R/I_i$. This map is clearly well-defined and preserves the equivalence class. Furthermore, we have that
        \begin{align}
            \begin{aligned}
                &\phi \circ \psi(c_{b^{(i)}}m) = \phi(m \otimes _R\bar{1}) = c_{b^{(i)}}m = id_{c_{b^{(i)}}M}(c_{b^{(i)}}m), \\
                & \psi \circ \phi(m \otimes _R\bar{r}) = \psi(c_{b^{(i)}} m r) = c_{b^{(i)}} m \otimes _R\bar{r} \\
                &= m \otimes _R \overline{c_{b^{(i)}} r} = m \otimes_R \bar{r} = \operatorname{id}_{M \otimes_R R/I_i}(m \otimes _R\bar{r}),
            \end{aligned}
        \end{align}
        where the last line uses $c_{b^{(i)}} \equiv 1 \pmod{b^{(i)}}$.
        Furthermore, it is straightforward to see that $M \otimes_R R/I_i \cong M/I_i M$ via the isomorphism $m \otimes _R\bar{r} = \bar{r} m \otimes _R\bar{1}  \mapsto \bar{r}m$.
    \end{proof}
\end{lemma}
The above statement indicates the following.
\begin{lemma}
    Let $l$ be odd so that $R$ is semi-simple. Let $M$ and $N$ be two $R$-modules. Then there exists a canonical decomposition
    \begin{align}
        M \otimes_R N \cong \bigoplus^t_{i=1} c_{b^{(i)}} (M \otimes_R N).
    \end{align}
\begin{proof}
    Apply Lemma~\ref{lemma: module-decomposition-odd-l} to the $R$-module $M \otimes_R N$.
\end{proof}
\end{lemma}
Given $A: R^{n_a}_{l} \rightarrow R^{m_a}_l$ and $B: R^{n_b}_l \rightarrow R^{m_b}_l$, in what follows we compute the tensor product $A \otimes_R B$, viewed as a tensor product of $1$-chain complexes. We show how the fact that $R$ is semi-simple implies a drastically simplified version of the K\"unneth spectral sequence.
\begin{lemma}[K\"unneth theorem for semi-simple ring]\label{thm: kunneth-theorem-semi-simple}
     Let $l$ be odd so that $R$ is semi-simple. Let $\mathcal{C}_{\bullet}$ and $\mathcal{D}_{\bullet}$ be two chain complexes over $R$. Then the K\"unneth spectral sequence collapses to the following isomorphism at the $k$th homology:
\begin{align}
    \bigoplus_{p+q=k} H_p (\mathcal{C}) \otimes_R H_q(\mathcal{D}) \cong H_k(\mathcal{C} \otimes_R \mathcal{D}).
\end{align}
Furthermore, we have an isomorphism
\begin{align}
    H_p(\mathcal{C}) \otimes_R H_q(\mathcal{D}) \cong  \bigoplus^t_{i=1} c_{b^{(i)}} H_p(\mathcal{C})\otimes_{R} H_q(\mathcal{D}),
\end{align}
where $c_{b^{(i)}}$ is the central idempotent with respect to the irreducible factor $b^{(i)}|x^l +1$.
\end{lemma}
\begin{remark}
    For a general ring $R$, the Kunneth theorem is not typically given by a short exact sequence. Rather, it is given by the spectral sequences of the torsion groups. In the special case where the ring is semi-simple, the torsion group always vanishes. We investigate precisely the sources from which the torsion arises for even lifting size $l$. For even $l$, there will be cases of repeated irreducible factors $b^{(i)}$, where the number of occurrences is referred to as the multiplicity $m_i$. Following the Chinese remainder theorem, Lemma~\ref{lemma: chinese-reminder}, we can always associate the central idempotent $c_{(b^{(i)})^{m_i}}$ to $(b^{(i)})^{m_i}$. By Lemma~\ref{lemma: module-decomposition-odd-l}, if $\mathcal{C}_{\bullet}$ and $\mathcal{D}_{\bullet}$ are two chain complexes over $R$, the $k$th homology group $H_k(\mathcal{C} \otimes_{R} \mathcal{D})$ admits the decomposition
 \begin{align}
     H_k(\mathcal{C} \otimes_R \mathcal{D}) \cong \bigoplus^t_{i=1} c_{(b^{(i)})^{m_i}}   H_k(\mathcal{C} \otimes_R \mathcal{D}),
 \end{align}
with $x^l +1 = \prod^t_{i=1} (b^{(i)})^{m_i}$ and multiplicities $m_i$. A full treatment of the even-$l$ case via the K\"unneth spectral sequence is beyond the scope of this work.
\end{remark}

\subsection{Basis characterization using algebraic root theory}\label{subsec: basis-algebraic-root-theory}

The use of homological algebra is useful in understanding the dimensions of the relevant homology groups. In what follows, we utilize algebraic root theory to derive the basis construction, and we refer to Ref.~\cite{roman1997field} for mathematical background. We can write $x^l + 1 = (x - \beta_1) \cdots (x - \beta_{l})$, where we associate a splitting field $E$ in which these algebraic roots $\{\beta_1, \cdots, \beta_l \}$ are defined~\cite{roman1997field}. The polynomial $x^l + 1$ is called \emph{separable} over $\mathbb{F}_2$ if these roots are all distinct in $E$. Hence, $x^l + 1$ is separable if and only if $l$ is odd (Remark~\ref{remark: square-free-odd-l}). It is known that every finite family of polynomials has a splitting field. The structure of the splitting field is quite involved and largely unnecessary to state in our case. Instead, we focus on an irreducible polynomial factor $b^{(i)}$ with $m_i=1$ so that we can associate a central idempotent $c_{b^{(i)}}$. Note that we have the Galois field extension
\begin{align}\label{galois-central-idempotents-isometry}
    \mathbb{F}_{q_i} \cong \mathbb{F}_2[x]/(b^{(i)}) \cong c_{b^{(i)}}R,
\end{align}
 for $q_i = 2^{\deg b^{(i)}}$ and, hence, by the Chinese remainder theorem, we can decompose $R$ as a direct sum of finite fields given by the central idempotents. Furthermore, by Lemma~\ref{lemma: module-decomposition-odd-l}, any module over $R$ can be factorized in this way. In what follows, we explicitly construct the isomorphism~\eqref{galois-central-idempotents-isometry} --- an isometry once the trace pairing is introduced --- and provide the necessary mathematical background. 

\begin{proposition}\label{prop: roots-splitting-field-properties}
    Let $l$ be odd and let $E$ be the splitting field of $x^l+1$ with algebraic roots $G = \{ \beta_1, \beta_2, \cdots, \beta_l\}$. Then the following hold:
    \begin{itemize}
    \item $G$ is a cyclic group within $E$. 
    \item (Frobenius automorphism). Let $\sigma(\gamma) = \gamma^2$ denote the Frobenius squaring map for $\gamma \in G$. Then $\sigma$ decomposes $G$ into $t$ orbits $G_1, \cdots, G_t$, where $t$ is the number of irreducible polynomial factors of $x^l +1$. 
    \item The roots contained in each orbit $G_i$ for $i=1, \cdots, t$ are precisely the algebraic roots of each irreducible polynomial factor $b^{(i)}$. 
    \end{itemize}
   \begin{proof}
Since \(\mathrm{char}(\mathbb F_{2})=2\), we have \(x^{l}+1=x^{l}-1\). Because \(l\) is odd,
\begin{align}
    \gcd(x^{l}+1,(x^{l}+1)')
    &:= \gcd(x^{l}+1,lx^{l-1}) \\
    &= \gcd(x^{l}+1,x^{l-1}) = 1.
\end{align}
So \(x^{l}+1\) has \(l\) distinct roots in \(E\) (separable: see Ref.~\cite{roman1997field}). Thus \(G\) is precisely the group of \(l\)-th roots of unity in \(E\). As \(E\) is a finite field, \(E^\times\) is cyclic; hence every finite subgroup of \(E^\times\), in particular \(G\), is cyclic. Now let \(\sigma(\gamma)=\gamma^{2}\), then
\begin{align}
    \sigma(\gamma)^{l}=(\gamma^{2})^{l}=(\gamma^{l})^{2}=1,
\end{align}
so \(\sigma(G)=G\). Hence \(G\) is partitioned into Frobenius orbits. Fix \(i\), and let \(\beta\) be a root of \(b^{(i)}\). Since \(b^{(i)}\in\mathbb F_{2}[x]\), from \(b^{(i)}(\beta)=0\) we obtain
\begin{align}
    b^{(i)}(\beta^{2^{m}})=b^{(i)}(\beta)^{2^{m}}=0
\qquad (m\ge 0),
\end{align}
so that the Frobenius orbit of \(\beta\) is contained in the root set of \(b^{(i)}\). Let \(r\) be the size of this orbit, and consider the \emph{orbit polynomial}
\begin{align}
    m_\beta(x)=\prod_{j=0}^{r-1}\bigl(x+\beta^{2^{j}}\bigr).
\end{align}
The coefficient of each power of $x$ in $m_\beta$ is an elementary symmetric polynomial. Since $\sigma$ respects addition and multiplication, applying $\sigma$ to a coefficient evaluates the same symmetric polynomial on the squared orbit elements; and since $\beta^{2^{r}} = \beta$, squaring merely permutes the orbit cyclically, so every coefficient $c$ of $m_\beta$ satisfies $\sigma(c) = c$. Thus $c^2 = c$, i.e., $c(c+1) = 0$, and since $E$ is a field, $c \in \{0, 1\} = \mathbb{F}_2$. Hence, $m_\beta \in \mathbb{F}_2[x]$, and $\beta$ is a root of $m_\beta$. Since \(b^{(i)}\) is the minimal polynomial of \(\beta\) over \(\mathbb F_{2}\), we must have $\deg b^{(i)}\le r$. On the other hand, the orbit is contained in the root set of \(b^{(i)}\), so \(r\le \deg b^{(i)}\). Therefore $ r=\deg b^{(i)}$, and the Frobenius orbit of \(\beta\) is exactly the full set of roots of \(b^{(i)}\). Since the factors \(b^{(i)}\) are distinct and pairwise coprime, their root sets are disjoint, and since
\begin{align}
    x^{l}+1=\prod_{i=1}^{t} b^{(i)}(x),
\end{align}
these root sets exhaust all roots of \(x^{l}+1\). Hence the Frobenius orbits are in bijection with the irreducible factors \(b^{(1)},\dots,b^{(t)}\), proving the claim.
\end{proof}
\end{proposition}

With the use of the splitting field $E$, we introduce the notion of Lagrange interpolation, which is central to the study of algebraic geometry (AG) codes~\cite{hall2015codingnotes, stichtenoth2009algebraic}. Write the following \emph{Vandermonde} matrix, 
\begin{align}
    V = \left(\begin{array}{cccc}
1 & \beta_1 & \cdots & \beta_1^{l-1} \\
\vdots & \vdots & \cdots & \vdots \\
1 & \beta_l & \cdots & \beta_l^{l-1}
\end{array}\right).
\end{align}
Then this matrix has determinant $\operatorname{det}(V) = \prod_{i < j} (\beta_j + \beta_i)$, which is non-zero in the splitting field $E$. A famous result states the following.

\begin{theorem}[Lagrange interpolation]\label{thm: lagrange-interpolation}
     Write $p(x) := x^l +1$ and let $l$ be odd. Any $h \in R$ is uniquely determined by its values at the roots and can be represented as
    \begin{align}
        h(x) = \sum_{\gamma \in G} h(\gamma) L_\gamma(x),
    \end{align}
    where we have that 
    \begin{align}\label{eq: lagrange-root-function}
        L_\gamma(x) = \prod_{\beta \neq \gamma} \frac{x+\beta}{\gamma + \beta} = \frac{p(x)}{(x+\gamma) p'(\gamma)} = \sum^{l-1}_{t=0}x^t \gamma^{-t}.
    \end{align}
    Note that the interpolation coefficients live in the splitting field $E$; for $h \in R$, the combination on the right-hand side again has coefficients in $\mathbb{F}_2$.
\end{theorem}

The root theory provides a convenient characterization of elements in $R$. 

\begin{lemma}[Idempotents evaluated at roots]\label{lemma: idempotent-root-evaluation}
     An element $h \in R$ is the central idempotent $c_{b^{(i)}}$ if and only if
    \begin{enumerate}
        \item for any root $\gamma$ of $b^{(i)}$ ($\gamma \in G_i$), $h(\gamma) = 1$,
        \item for any root $\gamma$ of $b^{(j)}$ ($\gamma \in G_j$) with $j \neq i$, $h(\gamma) = 0$.
    \end{enumerate}
    \begin{proof}
        We first show that $c_{b^{(i)}}$ satisfies the two conditions. By the congruence characterization of the construction~\eqref{eq: central-idempotent}, $c_{b^{(i)}} + b^{(i)}u = 1$ for some $u \in R$. Evaluating at a root $\gamma \in G_i$, where $b^{(i)}(\gamma) = 0$, gives $c_{b^{(i)}}(\gamma) = 1$. For the second condition, by the orthogonality of central idempotents, $c_{b^{(i)}}c_{b^{(j)}} = 0$; evaluating at $\gamma \in G_j$ and using $c_{b^{(j)}}(\gamma) = 1$ forces $c_{b^{(i)}}(\gamma) = 0$.
        Conversely, suppose $h \in R$ satisfies the two conditions. Then $h$ and $c_{b^{(i)}}$ agree at all $l$ roots of $x^l+1$, since every root lies in exactly one orbit $G_j$. As the roots are distinct for odd $l$, Theorem~\ref{thm: lagrange-interpolation} shows that an element of $R$ is uniquely determined by its values at the roots; hence $h = c_{b^{(i)}}$. 
    \end{proof}
\end{lemma}

\begin{lemma}[Coefficient formula for central idempotents]\label{lemma: central-idempotent-coeff-root}
    Let $G_i$ be the set of roots of $b^{(i)} | x^l +1$, then we have that $c_{b^{(i)}}(x) = \sum_{\gamma \in G_i} L_\gamma(x)$. Furthermore, the coefficient of $x^{l-k}$ for $1\leq k \leq l$ in $c_{b^{(i)}}(x)$ is given by 
   \begin{align}
       [x^{l-k}] c_{b^{(i)}}(x) = \sum_{\gamma \in G_i} \gamma^k.
   \end{align}
   \begin{proof}
      Let us prove the first assertion. By the Lagrange interpolation Theorem~\ref{thm: lagrange-interpolation} and Lemma~\ref{lemma: idempotent-root-evaluation},
      \begin{align}
          c_{b^{(i)}}(x) = \sum_{\gamma} c_{b^{(i)}}(\gamma) L_\gamma(x) = \sum_{\gamma \in G_i} L_\gamma(x).
      \end{align}
      The second equality follows from the fact that $c_{b^{(i)}}(\gamma) = 0$ for all $\gamma \notin G_i$ and $c_{b^{(i)}}(\gamma) = 1$ for $\gamma \in G_i$. Recall that we have shown that 
      \begin{align}
          L_\gamma(x) = \frac{p(x)}{(x + \gamma) p'(\gamma)} = \frac{x^l +1}{(x +\gamma) \gamma^{l-1}} = \sum^{l-1}_{t=0} x^t \gamma^{-t}.
      \end{align}
      Hence, equating the $x^{l-k}$ term, we obtain 
      \begin{align}
          [x^{l-k}]c_{b^{(i)}}(x) = \sum_{\gamma \in G_i} \gamma^k. 
      \end{align}
      
   \end{proof}
\end{lemma}

The final ingredient concerns the field $\mathbb{F}_q$ for $q = p^{s}$ for some integer $s$ and prime $p$. Then we can always prescribe a basis set in which any $z \in \mathbb{F}_q$ can be written
\begin{align}
    z = \sum^{s-1}_{i=0} c_i \alpha_i, 
\end{align}
for coefficients $c_i \in \mathbb{F}_p$. 

\begin{proposition}[Trace map~\cite{roman1997field}]\label{prop: trace-map-properties}
    For $q = p^s$, we define the trace map 
    \begin{align}
        \tr_{\mathbb{F}_q/\mathbb{F}_p}(x) = \sum^{s-1}_{i =0} x^{p^{i}}
    \end{align}
    with the following properties:
    \begin{enumerate}
        \item $\tr_{\mathbb{F}_q/\mathbb{F}_p}: \mathbb{F}_q \rightarrow \mathbb{F}_p$ is a linear map with $\tr_{\mathbb{F}_q/\mathbb{F}_p}(x+y) = \tr_{\mathbb{F}_q/\mathbb{F}_p}(x) + \tr_{\mathbb{F}_q/\mathbb{F}_p}(y)$ and $\tr_{\mathbb{F}_q/\mathbb{F}_p}(\alpha x) = \alpha \tr_{\mathbb{F}_q/\mathbb{F}_p}(x)$ for $\alpha \in \mathbb{F}_p$
        \item 
        $\tr_{\mathbb{F}_q/\mathbb{F}_p}(x^{p})$ = $\tr_{\mathbb{F}_q/\mathbb{F}_p}(x)$
        \item Surjectivity. Precisely, $p^{s-1}$ distinct elements in $\mathbb{F}_q$ will be mapped to the same element of $\mathbb{F}_p$ by the trace. 
    \end{enumerate}
\end{proposition}

 An important concept here is the dual basis set. Let $\{\alpha_j\}^{s-1}_{j=0}$ be a basis set of $\mathbb{F}_q$; its dual basis set $\{\beta_{k}\}^{s-1}_{k=0}$ satisfies the property $\tr_{\mathbb{F}_q/\mathbb{F}_p}(\alpha_j \beta_k) = \delta_{jk}$. A basis $\{\alpha_j\}^{s-1}_{j=0}$ is called \emph{self-dual} if $\tr_{\mathbb{F}_q / \mathbb{F}_2}(\alpha_i \alpha_j) = \delta_{ij}$. A self-dual basis always exists for $\mathbb{F}_q$ with $q=2^s$~\cite{BayerFluckiger1989}. The trace map is a direct generalization of the familiar $\mathbb{F}_2$ inner product in the following sense, 

\begin{lemma}
    Let $x, y \in \mathbb{F}_q$ and denote a self-dual basis $\{\alpha_0, \alpha_1, \cdots, \alpha_{s-1}\}$ where $q = 2^s$. Then we have that 
    \begin{align}
        \tr_{\mathbb{F}_q/\mathbb{F}_2}(xy) = \langle a, b \rangle_{\mathbb{F}_2},
    \end{align}
    where $x= \sum^{s-1}_{i=0} a_i \alpha_i$ and $y = \sum^{s-1}_{i=0} b_i \alpha_i$. 
    \begin{proof}
        Writing $x, y$ with the self-dual basis, we compute 
        \begin{align}
            \begin{aligned}
\tr_{\mathbb{F}_q/\mathbb{F}_2} (xy) &= \sum^{s-1}_{i=0} \left(\sum^{s-1}_{k, l=0} a_k b_l \alpha_k \alpha_l \right)^{2^i}, \\
                &= \sum^{s-1}_{k, l=0} a_k b_l\tr_{\mathbb{F}_q/ \mathbb{F}_2}(\alpha_k \alpha_l), \\
                &= \sum^{s-1}_{k=0} a_k b_k,
            \end{aligned}
        \end{align}
        where the second equality uses the linearity of the trace map, noting that $a_k, b_l \in \mathbb{F}_2$. The final equality is a consequence of the self-dual basis. 
    \end{proof}
\end{lemma}

We can also express this in another basis set. 

\begin{lemma}\label{lemma: galois-field-root-basis}
    Let $b \in \mathbb{F}_p[x]$ be an irreducible polynomial of degree $s$ and denote $\mathbb{F}_q = \mathbb{F}_p[x]/(b)$ to be the Galois field extension. Let $\beta$ be an algebraic root with $b(\beta) =0$; then any element $h \in \mathbb{F}_{q}$ can be written as a polynomial $h = \sum^{s-1}_{j=0}c_j \beta^j$.
\end{lemma}

In what follows, we simply choose an irreducible, multiplicity-free polynomial factor $b |x^l+1$. In this case, we need not necessarily set $l$ to be odd, as long as the multiplicity associated with $b$ is $1$. 
Specific to our case, where $p=2$ and $d_b := \operatorname{deg} b$, we consider the trace map $\tr_{\mathbb{F}_{q}/\mathbb{F}_2}: \mathbb{F}_{q} \cong \mathbb{F}_2[x]/(b) \rightarrow \mathbb{F}_2$. We now revisit the isomorphism $c_bR \cong \mathbb{F}_{q}$ via the algebraic roots.

\begin{theorem}\label{thm: isometry-lifted-product}
    Let $b |x^l +1$ be an irreducible, multiplicity-free factor, and let $c_{b}$ be its central idempotent. Let $\beta$ be an algebraic root with $b(\beta) =0$. Let the map $\mathrm{ev}_{\beta}: c_bR \rightarrow \mathbb{F}_{q}$ with $q = 2^{d_b}$ be given by $\mathrm{ev}_{\beta}(h) = h(\beta)$. Then $\mathrm{ev}_{\beta}$ is an isomorphism and, furthermore, an isometry with respect to the bilinear form $[x^0]hw$ and the trace form $\tr_{\mathbb{F}_q/\mathbb{F}_2}$.
    \begin{proof}
         We first prove the isomorphism part and construct the inverse map explicitly. We first show that this map is a well-defined ring homomorphism. Let $h_1, h_2 \in c_bR$. Then it is clear that $\mathrm{ev}_\beta(h_1 h_2) = (h_1 h_2) (\beta) = h_1(\beta) h_2(\beta) = \mathrm{ev}_\beta(h_1) \mathrm{ev}_{\beta}(h_2)$. To see surjectivity, let $h \in R$. Bezout's identity in Lemma~\ref{lemma: bezout-lemma} gives $h = u b + c_{b} v$ for some $u, v \in R$. Note that by construction $\mathbb{F}_{q} \cong R/(b)$, so that every class has a representative in $c_bR$.
         Then $h(\beta) = v(\beta) c_{b}(\beta)= v(\beta)$ by construction, which ensures surjectivity. We now consider injectivity, where we aim to show that $h_1(\beta) = h_2(\beta) \iff c_{b}h_1 = c_{b}h_2$. Recall from the Frobenius automorphism, Proposition~\ref{prop: roots-splitting-field-properties}, that we can prescribe a root orbit
        \begin{align}
            G_b = \{\beta, \beta^2, \cdots, \beta^{2^i}, \cdots, \beta^{2^{d_b -1}} \},
        \end{align}
        where $h(\beta^{2^j}) = (h(\beta))^{2^j}$, so equality of values at $\beta$ propagates to all roots in $G_b$. It thus suffices to show that the values on $G_b$ determine $c_b h$. By the Lagrange interpolation Theorem~\ref{thm: lagrange-interpolation}, we can write
    \begin{align}
        c_bh_1 = \sum_{\beta \in G_b} h_1(\beta) L_\beta(x),
    \end{align}
    where, collecting the coefficients at $x^j$, we have $ [x^j]c_{b} h_1 = \sum_{\beta \in G_b} h_1(\beta) \beta^{l -j}$, using the fact that $c_{b}(\beta) = 1$ if $\beta$ is a root of $b$, and otherwise it is zero. Hence, 
    \begin{align}
       \begin{pmatrix}
           [x^0]c_b h_1 \\
           \vdots \\
           [x^{l-1}]c_bh_1
       \end{pmatrix} = \begin{pmatrix}
           1 & \cdots & 1 \\
           \vdots & \cdots & \vdots \\
           \beta & \cdots & \beta^{2^{d_b-1}}
       \end{pmatrix} \begin{pmatrix}
           h_1(\beta) \\
           \vdots \\
           h_1(\beta^{2^{d_b-1}})
       \end{pmatrix}.
    \end{align}
    The matrix is the (transpose of) Vandermonde matrix. In this case, this matrix has full column rank, so the result must be unique. This proves injectivity. We finally show that this map is an isometry with respect to the bilinear form $[x^0]hw$ and $\tr_{\mathbb{F}_{q}/\mathbb{F}_2}$; note that this form is related to the pairing of Appendix~\ref{app-sec: basic-chain-complex} by $\langle h, w \rangle_R = [x^0] h \bar{w}$. Denote $h = c_{b} \sum^{l-1}_{j=0}a_j x^j$ and $w = c_{b} \sum^{l-1}_{j=0}b_j x^j$. Compute $ hw = c_{b} \sum_{j, k}a_j b_k x^{j+k}$ and note that, upon equating at $x^0$, 
        \begin{align}
            [x^0]c_{b} hw
            &= \sum_{jk} a_j b_k [x^{l-j-k}]c_{b} \\
            &= \sum_{\beta \in G_b} \sum_{jk} a_j b_k \beta^{j+k}
             = \sum_{\beta \in G_b} h(\beta) w(\beta).
        \end{align}
        Recall that the root orbit is $G_b=\{\beta,\beta^2,\ldots,\beta^{2^{d_b-1}}\}$. Hence, we write 
        \begin{align}
            \begin{aligned}
                [x^0] h w
                &= \sum^{d_b-1}_{j=0} h(\beta^{2^j})w(\beta^{2^j}) \\
                &= \sum^{d_b-1}_{j=0} (h(\beta) w(\beta))^{2^j} \\
                &= \tr_{\mathbb{F}_{q}/ \mathbb{F}_2}
                   (\mathrm{ev}_{\beta}(h) \mathrm{ev}_{\beta}(w)).
            \end{aligned}
        \end{align}
        The second equality follows from the Frobenius automorphism in characteristic $2$, and the final equality follows from the trace-map formula for a Galois field extension. 
\end{proof}
\end{theorem}

Equipped with the isometry of Theorem~\ref{thm: isometry-lifted-product}, we can now restate Theorem~\ref{thm: main-computation-kernel-cokernel} in terms of the field decomposition. First we investigate the kernel and cokernel of a matrix over $R$ in terms of its components over the finite fields $\mathbb{F}_{q_i}$. Throughout, $d_i := \deg b^{(i)}$ and $q_i = 2^{d_i}$; we identify $c_{b^{(i)}}R$ with $R_{(i)} := \mathbb{F}_2[x]/(b^{(i)}) \cong \mathbb{F}_{q_i}$ via the isomorphism $\mathrm{ev}_{\beta_i}: c_{b^{(i)}}R \rightarrow \mathbb{F}_{q_i}$ of Theorem~\ref{thm: isometry-lifted-product}, where $\beta_i$ is a root of $b^{(i)}$, applied entrywise to vectors and matrices; in particular, $A_{(i)} := A(\beta_i) = \mathrm{ev}_{\beta_i}(c_{b^{(i)}}A)$.

\begin{lemma}[Structure of kernel of $A$]\label{lemma: kernel-A-odd}
    Let $l$ be odd, $R = \mathbb{F}_2[x]/(x^l+1)$, and $x^l+1=\prod^t_{i=1}b^{(i)}$ with corresponding central idempotents $c_{b^{(i)}}$. Let $A \in R^{m_A \times n_A}$, and let $A_{(i)} = \mathrm{ev}_{\beta_i}(c_{b^{(i)}}A) = A(\beta_i) \in \mathbb{F}_{q_i}^{m_A \times n_A}$ be the component matrix. Then the following hold: 
    \begin{enumerate}
        \item We have $u \in \ker_R A$ if and only if $u \in \ker_R c_{b^{(i)}}A$ for all $i=1, \cdots, t$.
        \item Let $I^{(i)}_A$ be an information set of $\ker_{R_{(i)}} A_{(i)}$, and let $G_{A_{(i)}} \in \mathbb{F}_{q_i}^{k_i \times n_A}$ with $k_i := \dim_{R_{(i)}} \ker_{R_{(i)}} A_{(i)}$ be the systematic generator matrix with respect to $I^{(i)}_A$ (Definition~\ref{def: standard-form}), whose rows are $g^{(i,j)}_{\mathcal{A}_1} \in \mathbb{F}_{q_i}^{n_A}$ for $j \in I^{(i)}_A$. Then, for every $u \in \ker_R A$, there exist unique coefficients $\alpha_j \in \mathbb{F}_{q_i}$ such that
        \begin{align}
            \mathrm{ev}_{\beta_i}(c_{b^{(i)}} u) = \sum_{j \in I^{(i)}_A} \alpha_j g^{(i,j)}_{\mathcal{A}_1},
        \end{align}
        namely, $\alpha_j$ is the $j$-th coordinate of $\mathrm{ev}_{\beta_i}(c_{b^{(i)}} u)$. Consequently, the lifted matrix $G^{(i)}_A := \mathrm{ev}^{-1}_{\beta_i}(G_{A_{(i)}})$ satisfies $c_{b^{(i)}} u \in \rs_R G^{(i)}_A$, and $\ker_R A = \bigoplus^t_{i=1} \rs_R G^{(i)}_A$.
    \end{enumerate}
\end{lemma}
\begin{proof}
    (1) Note that $u \in \ker_R c_{b^{(i)}}A$ means $c_{b^{(i)}}Au = 0$. If $Au = 0$, then $c_{b^{(i)}}Au = 0$ for every $i$. Conversely, if $c_{b^{(i)}}Au = 0$ for all $i$, then, since $1 = \sum_i c_{b^{(i)}}$ by properties of central idempotents Lemma~\ref{lemma: central-idempotent-property}, $Au = \sum_i c_{b^{(i)}}Au = 0$. (2) This follows from (1) and the isomorphism of Theorem~\ref{thm: isometry-lifted-product}: since $\mathrm{ev}_{\beta_i}: c_{b^{(i)}}R \rightarrow \mathbb{F}_{q_i}$ is a ring isomorphism applied entrywise, $\mathrm{ev}_{\beta_i}(c_{b^{(i)}}Aw) = A_{(i)}\, \mathrm{ev}_{\beta_i}(c_{b^{(i)}}w)$ for any $w \in R^{n_A}$, and hence, explicitly, $c_{b^{(i)}}w \in \ker_R c_{b^{(i)}}A \iff \mathrm{ev}_{\beta_i}(c_{b^{(i)}}w) \in \ker_{R_{(i)}} A_{(i)}$. Now fix $u \in \ker_R A$. By (1), $c_{b^{(i)}}u \in \ker_R c_{b^{(i)}}A$, so $v := \mathrm{ev}_{\beta_i}(c_{b^{(i)}}u) \in \ker_{R_{(i)}} A_{(i)}$. The rows of $G_{A_{(i)}}$ form a basis of $\ker_{R_{(i)}} A_{(i)}$, so $v = \sum_{j \in I^{(i)}_A} \alpha_j g^{(i,j)}_{\mathcal{A}_1}$ for unique $\alpha_j \in \mathbb{F}_{q_i}$; restricting this identity to the columns in $I^{(i)}_A$ and using the systematic form (Definition~\ref{def: standard-form}) gives $\alpha_j = v_j$. Applying $\mathrm{ev}^{-1}_{\beta_i}$ entrywise yields $c_{b^{(i)}}u = \sum_j \mathrm{ev}^{-1}_{\beta_i}(\alpha_j)\, \mathrm{ev}^{-1}_{\beta_i}(g^{(i,j)}_{\mathcal{A}_1}) \in \rs_R G^{(i)}_A$. Conversely, each row $\tilde{g} = \mathrm{ev}^{-1}_{\beta_i}(g^{(i,j)}_{\mathcal{A}_1})$ of $G^{(i)}_A$ satisfies $g^{(i,j)}_{\mathcal{A}_1} \in \ker_{R_{(i)}} A_{(i)}$, so the above equivalence gives $c_{b^{(i)}}A\tilde{g} = 0$; for $i' \neq i$, $c_{b^{(i')}}A\tilde{g} = A\, c_{b^{(i')}}c_{b^{(i)}}\tilde{g} = 0$ by orthogonality, and hence $\tilde{g} \in \ker_R A$ by (1). Since $u = \sum_i c_{b^{(i)}}u$ and $\rs_R G^{(i)}_A \subseteq c_{b^{(i)}}R^{n_A}$ with orthogonal sectors, $\ker_R A = \bigoplus^t_{i=1} \rs_R G^{(i)}_A$.
\end{proof}
Similarly, we state the result for $\coker_R A$. To this end, we first state a simple result on the cokernel of a matrix over a finite field.
\begin{lemma}[Cokernel basis from an information set]\label{lemma: coker-basis-ker-transpose}
    For any $H \in \mathbb{F}^{m \times n}_q$, let $I^*$ be an information set of $\ker_{\mathbb{F}_q} H^T \subseteq \mathbb{F}^m_q$. Then the classes of the unit vectors $\{[e^{(j)}] : j \in I^*\}$ form a basis of $\coker_{\mathbb{F}_q} H$.
\end{lemma}
\begin{proof}
    First, $|I^*| = \dim \ker_{\mathbb{F}_q} H^T = m - \operatorname{rank} H = \dim \coker_{\mathbb{F}_q} H$. For linear independence, let $c = \sum_{j \in I^*} \alpha_j e^{(j)} \in \IM_{\mathbb{F}_q} H$, say $c = Hz$. Let $\{w^{(j)} : j \in I^*\}$ be the systematic basis of $\ker_{\mathbb{F}_q} H^T$ with respect to $I^*$ (Definition~\ref{def: standard-form}), so that $w^{(j)}_{j'} = \delta_{jj'}$ for $j' \in I^*$. Then, for each $j \in I^*$, the standard bilinear form gives $0 = \langle z, H^T w^{(j)} \rangle = \langle Hz, w^{(j)} \rangle = \sum_{j' \in I^*} \alpha_{j'} w^{(j)}_{j'} = \alpha_j$. Hence the classes are linearly independent in $\coker_{\mathbb{F}_q} H$, and by the dimension count they form a basis.
\end{proof}

\begin{lemma}[Structure of cokernel of $A$]\label{lemma: cokernel-A-odd}
    Similarly under the conditions of Lemma~\ref{lemma: kernel-A-odd}, let $I^{(i)}_{A^*}$ be an information set of $\ker_{R_{(i)}} A^T_{(i)}$, and let $E_{A_{(i)}} = \{ e^{(i,j)}_{\mathcal{A}_0} \in \mathbb{F}^{m_A}_{q_i}\}^{\dim_{R_{(i)}} \coker_{R_{(i)}} A_{(i)}}_{j=1}$, where $e^{(i,j)}_{\mathcal{A}_0}$ is the unit vector supported on the $j$-th index of $I^{(i)}_{A^*}$; by Lemma~\ref{lemma: coker-basis-ker-transpose}, their classes form a basis of $\coker_{R_{(i)}} A_{(i)}$. Then the following hold:
    \begin{enumerate}
        \item Let $E_A = \cup^t_{i=1} \mathrm{ev}^{-1}_{\beta_i}(E_{A_{(i)}})$, where $\mathrm{ev}^{-1}_{\beta_i}(E_{A_{(i)}}) = \{ \mathrm{ev}^{-1}_{\beta_i}(e^{(i,j)}_{\mathcal{A}_0}) = c_{b^{(i)}} e^{(i,j)}_{\mathcal{A}_0} \in c_{b^{(i)}}R^{m_A}\}^{\dim_{R_{(i)}} \coker_{R_{(i)}} A_{(i)}}_{j=1}$. Then $E_A$ satisfies conditions Eq.~\eqref{eq: splitting-condition-1} and Eq.~\eqref{eq: splitting-condition-2} of Lemma~\ref{lemma: section-maps-coker-R}: any element $u \in R^{m_A}$ can be written as
        \begin{align}
            u = \sum_{i,j} r_{ij} c_{b^{(i)}} e^{(i,j)}_{\mathcal{A}_0} + Ah,
        \end{align}
        for some $h \in R^{n_A}$ and $r_{ij} \in R$. Furthermore, 
        \begin{align}
            \sum_{i,j} r_{ij} c_{b^{(i)}} e^{(i,j)}_{\mathcal{A}_0} \in \IM_R A \iff  \sum_{i,j} r_{ij} c_{b^{(i)}} e^{(i,j)}_{\mathcal{A}_0}=0.
        \end{align}
        \item By Lemma~\ref{lemma: section-maps-coker-R}, $E_A$ defines an $R$-linear section $s: \coker_R A \rightarrow R^{m_A}$ with $s(\coker_R A) = \rs_R E_A =: M_E$, giving the splitting
        \begin{align}
            R^{m_A} = \IM_R A \oplus M_E; \quad M_E \cong \coker_R A.
        \end{align}
    \end{enumerate}
\end{lemma}
\begin{proof}
    We prove the first point. The second follows from Lemma~\ref{lemma: section-maps-coker-R}. Recall that $\sum^t_{i=1} c_{b^{(i)}} =1$ and by properties of the central idempotents, we may prove component-wise for $c_{b^{(i)}}u$. Let $v:= \mathrm{ev}_{\beta_i}(c_{b^{(i)}}u) \in \mathbb{F}^{m_A}_{q_i}$. By Lemma~\ref{lemma: coker-basis-ker-transpose}, the classes of $E_{A_{(i)}}$ form a basis of $\coker_{R_{(i)}} A_{(i)}$; equivalently,
    \begin{align}
        v = \sum_j \alpha_j e^{(i,j)}_{\mathcal{A}_0} + A_{(i)} w,
    \end{align}
    for some $\alpha_j \in \mathbb{F}_{q_i}$ and $w \in \mathbb{F}_{q_i}^{n_A}$, and
    \begin{align}
        \sum_j \alpha_j e^{(i,j)}_{\mathcal{A}_0} \in \IM_{R_{(i)}} A_{(i)} \iff \sum_j \alpha_j e^{(i,j)}_{\mathcal{A}_0} = 0.
    \end{align}
    By Theorem~\ref{thm: isometry-lifted-product}, we conclude that the conditions are met for $c_{b^{(i)}}u$ so that 
    \begin{align}
        c_{b^{(i)}}u = \sum_j r_{ij} c_{b^{(i)}} e^{(i,j)}_{\mathcal{A}_0} + Ah^{(i)},
    \end{align}
    for $h^{(i)} \in c_{b^{(i)}}R^{n_A}$, and 
    \begin{align}
        \sum_j r_{ij} c_{b^{(i)}} e^{(i,j)}_{\mathcal{A}_0} \in \IM_R c_{b^{(i)}}A \iff  \sum_j r_{ij} c_{b^{(i)}} e^{(i,j)}_{\mathcal{A}_0}=0.
    \end{align}
    Note that this also means that 
    \begin{align}
        \sum_{i,j} r_{ij} c_{b^{(i)}} e^{(i,j)}_{\mathcal{A}_0} \in \IM_R A \iff  \sum_{i,j} r_{ij} c_{b^{(i)}} e^{(i,j)}_{\mathcal{A}_0}=0.
    \end{align}
    Suppose not: $0 \neq x := \sum_{i,j} r_{ij} c_{b^{(i)}} e^{(i,j)}_{\mathcal{A}_0} = Ah$ for some $h \in R^{n_A}$. Multiplying by $c_{b^{(i)}}$, the other central idempotents vanish and $c^2_{b^{(i)}} = c_{b^{(i)}}$, so $\sum_{j} r_{ij} c_{b^{(i)}} e^{(i,j)}_{\mathcal{A}_0} = c_{b^{(i)}} A h \in \IM_R c_{b^{(i)}}A$; by the component condition above, every sector part vanishes, so $x = 0$ --- a contradiction.
    Summing the component decompositions over $i$ and using $\sum_i c_{b^{(i)}} = 1$ gives Eq.~\eqref{eq: splitting-condition-1}, proving the first assertion.
    By Lemma~\ref{lemma: section-maps-coker-R}, the second assertion follows.
\end{proof}

Another part we consider is conjugation, which we need to define basis elements for both homology and cohomology. For $f\in R$, write $\overline{f}(x):=f(x^{-1})$. If $b\mid x^l+1$ is irreducible with $d_b = \deg b$, define its reciprocal polynomial by $b^\vee(x):=x^{d_b} b(x^{-1})$.

\begin{lemma}\label{lemma: conjugate-central-idempotent}
    Let $b\mid x^l+1$ be irreducible and let $c_b$ be its central idempotent. Then $b^\vee$ is an irreducible factor of $x^l+1$ and
    \begin{align}
        \overline{c_b}=c_{b^\vee}.
    \end{align}
    \begin{proof}
Let \(d_b=\deg b\). Since \(b\mid x^l+1\), we may write
\begin{align}
    x^l+1=b(x)q(x).
\end{align}
Applying the reciprocal operation gives
\begin{align}
    x^l+1
    =\bigl(x^{d_b} b(x^{-1})\bigr)\bigl(x^{l-d_b}q(x^{-1})\bigr).
\end{align}
Hence, $b^\vee$ divides $x^l+1$. To see that $b^\vee$ is irreducible, suppose
\begin{align}
    b^\vee(x)=u(x)v(x)
\end{align}
with \(\deg u,\deg v>0\). Then
\begin{align}
    b(x)=x^{d_b} b^\vee(x^{-1})
     =x^{\deg u}u(x^{-1})\cdot x^{\deg v}v(x^{-1}),
\end{align}
which contradicts the irreducibility of $b$. Thus, $b^\vee$ is irreducible. By Lemma~\ref{lemma: idempotent-root-evaluation}, the central idempotent $c_b$ is uniquely characterized by
\begin{align}
    c_b &\equiv 1 \pmod b, \\
    c_b &\equiv 0 \pmod a,\qquad a\neq b,
\end{align}
for every other irreducible factor $a$ of $x^l+1$. Applying conjugation and using that $x$ is a unit in $R$ gives $\overline{c_b}\equiv1\pmod{b^\vee}$. Similarly, for $a^\vee(x):=x^{\deg a}a(x^{-1})$, we have $\overline{c_b}\equiv0\pmod{a^\vee}$. Since $a\mapsto a^\vee$ permutes the irreducible factors of $x^l+1$, $\overline{c_b}$ satisfies the congruences characterizing the central idempotent $c_{b^\vee}$. As central idempotents are unique from the construction by Bezout's identity Lemma~\ref{lemma: bezout-lemma}, we conclude $\overline{c_b}=c_{b^\vee}$.
\end{proof}
\end{lemma}
\begin{example}
    Let $l = 15$, with $x^{15}+1 = (1+x)(1+x+x^2)(1+x+x^4)(1+x^3+x^4)(1+x+x^2+x^3+x^4)$. Take $b = 1+x+x^4$; then $b^\vee = x^4 b(x^{-1}) = 1+x^3+x^4$, with central idempotents
    \begin{align}
        c_b &= x+x^2+x^3+x^4+x^6+x^8+x^9+x^{12}, \\
        c_{b^\vee} &= x^3+x^6+x^7+x^9+x^{11}+x^{12}+x^{13}+x^{14}.
    \end{align}
    One verifies directly that $\overline{c_b} = c_{b^\vee}$. The remaining three factors are self-reciprocal, with $\overline{c_b} = c_b$.
\end{example}

\begin{lemma}\label{lemma: conjugate-root}
        Let $b\mid x^l+1$ be irreducible with $d_b=\deg b$, let $h\in R$, and let $\beta$ be a root of $b$. Then $\beta^{-1}$ is a root of $b^\vee$, and
    \begin{align}
        h(\beta)=\overline{h}(\beta^{-1}).
    \end{align}
    \begin{proof}
        First, $b^\vee(\beta^{-1})=\beta^{-d_b}b(\beta)=0$, so $\beta^{-1}$ is a root of $b^\vee$. By the Lagrange interpolation Theorem~\ref{thm: lagrange-interpolation},
        \begin{align}
            h(x) = \sum_{\gamma \in G} h(\gamma) L_{\gamma}(x),
        \end{align}
        where $G$ denotes all the $l$ distinct roots of $x^l+1$. Directly from Eq.~\eqref{eq: lagrange-root-function}, we observe
        \begin{align}
            \overline{L_{\gamma}}(x) = L_{\gamma}(x^{-1}) = \sum^{l-1}_{t=0} x^{-t}\gamma^{-t} = \sum^{l-1}_{t=0} x^{t}\gamma^{t} = L_{\gamma^{-1}}(x),
        \end{align}
        so that
        \begin{align}
            \overline{h}(x) = \sum_{\gamma \in G} h(\gamma) L_{\gamma^{-1}}(x).
        \end{align}
        We now evaluate at $x = \beta^{-1}$. By Eq.~\eqref{eq: lagrange-root-function}, $L_{\gamma^{-1}}(\beta^{-1}) = 0$ unless $\gamma^{-1} = \beta^{-1}$, while $L_{\beta^{-1}}(\beta^{-1}) = l = 1$ since $l$ is odd. Hence $\overline{h}(\beta^{-1}) = h(\beta)$.
    \end{proof}
\end{lemma}
A final ingredient we need is the following.
\begin{lemma}\label{lemma: basis-dual-basis-from-finite-field-ring}
    Let $b$ be an irreducible polynomial with roots $G_b$ and $b |x^l+1$ with odd $l$. Let $c_b$ be the associated central idempotent, and let $\beta \in G_b$. Then for any set of distinct powers $\{ j_1, \cdots, j_{d_b}\}$ where $0 \leq j_i \leq l-1 $, there exist $g_{j_1}, \cdots, g_{j_{d_b}} \in R$ such that $[x^{j_i}] c_bg_{j_k} = \delta_{i, k}$ if and only if $\{\beta^{-j_1}, \cdots, \beta^{-j_{d_b}}\}$ is an $\mathbb{F}_2$-basis of $\mathbb{F}_q$; equivalently, the images of $x^{j_1}, \cdots, x^{j_{d_b}}$ span $\mathbb{F}_2[x]/(b^\vee)$. 
    
    \begin{proof}
    By the Lagrange interpolation Theorem~\ref{thm: lagrange-interpolation}, $c_b g = \sum_{\beta' \in G_b} g(\beta') L_{\beta'}(x)$ for any $g \in R$, so that, by Eq.~\eqref{eq: lagrange-root-function},
        \begin{align}
            [x^{j_i}]c_bg_{j_k} =  \sum_{\beta' \in G_b} g_{j_k}(\beta') \beta'^{-j_i} = \tr_{\mathbb{F}_q / \mathbb{F}_2}\bigl(g_{j_k}(\beta)\, \beta^{-j_i}\bigr),
        \end{align}
        where the second equality follows from $G_b = \{\beta^{2^m}\}^{d_b-1}_{m=0}$. Since $\mathrm{ev}_{\beta}: c_bR \rightarrow \mathbb{F}_q$ is an isomorphism and the trace form is nondegenerate, the dual system with $\tr_{\mathbb{F}_q/\mathbb{F}_2}(g_{j_k}(\beta)\beta^{-j_i}) = \delta_{i,k}$ exists if and only if $\{\beta^{-j_1}, \cdots, \beta^{-j_{d_b}}\}$ is an $\mathbb{F}_2$-basis of $\mathbb{F}_q$. Finally, since $\beta^{-j_i} = x^{j_i}(\beta^{-1})$ and $\beta^{-1}$ is a root of $b^\vee$ (Lemma~\ref{lemma: conjugate-root}), $\{\beta^{-j_i}\}$ is an $\mathbb{F}_2$-basis of $\mathbb{F}_q$ if and only if the images of $x^{j_1}, \cdots, x^{j_{d_b}}$ span $\mathbb{F}_2[x]/(b^\vee)$.
\end{proof}
\end{lemma}

\begin{remark}
For simplicity, we can set $j_i = i$ for $i=0, \cdots, d_b-1$. For two different irreducible polynomials $b, b'|x^l +1$, the above lemma applies equally to different choices of dual basis elements. 
\end{remark}
\begin{example}
    Let $l = 15$ and $b = 1+x^3+x^4$, so that $b^\vee = 1+x+x^4$, and set $j_i = i$ for $i = 0, \cdots, 3$. Applying Lemma~\ref{lemma: basis-dual-basis-from-finite-field-ring} to $b^\vee$ yields the dual basis $\hat{g}_0, \cdots, \hat{g}_3 \in c_{b^\vee}R$ to the primal basis $\{c_b, c_b x, c_b x^2, c_b x^3\}$ of $c_b R$ with respect to the pairing of Eq.~\eqref{eq: inner-product-ring}:
    \begin{align}
        \hat{g}_0 &= 1+x^4+x^5+x^6+x^7+x^9+x^{11}+x^{12}, \\
        \hat{g}_1 &= x+x^5+x^6+x^7+x^8+x^{10}+x^{12}+x^{13}, \\
        \hat{g}_2 &= x^2+x^6+x^7+x^8+x^9+x^{11}+x^{13}+x^{14}, \\
        \hat{g}_3 &= x^3+x^4+x^5+x^6+x^8+x^{10}+x^{11}+x^{14},
    \end{align}
    which satisfy $\langle c_b x^{j_i}, \hat{g}_k \rangle_R = \delta_{i,k}$. Since $c_b \overline{c_b} = c_b c_{b^\vee} = 0$, the pairing vanishes identically on $c_b R$, and the dual basis lies in the $c_{b^\vee}$ sector. Moreover, $\hat{g}_1 = x\hat{g}_0$ and $\hat{g}_2 = x\hat{g}_1$, whereas $\hat{g}_3 \neq x\hat{g}_2$; see Lemma~\ref{lemma: lp-cyclic-shift-logical-basis}.
\end{example}

We also consider cyclic symmetry using the basis characterization. Note that the basis elements given by $c_b x^{j_i}$, $i = 0, \cdots, d_b-1$, are already in cyclic form. At the boundary term, we have $c_b x^{j_{d_b-1}} x^{j_0 -j_{d_b-1}} = c_b x^{j_0}$. The natural question is whether the dual basis elements $c_b g_{j_i}$ also have cyclic symmetry, and it turns out that they typically do not.

\begin{lemma}[No cyclic shift symmetry]\label{lemma: lp-cyclic-shift-logical-basis}
Let $b |x^l+1$ be an irreducible polynomial with root $\beta$. Denote by $g_{0}, \cdots, g_{d_b-1}$ {elements of $c_b R$} satisfying $[x^{i}] c_bg_{j} = \delta_{i, j}$ for $i, j = 0, \cdots, d_b-1${. They exist by Lemma~\ref{lemma: basis-dual-basis-from-finite-field-ring} and form a basis of $c_b R$}. Then we have the following relations, 
\begin{equation}\label{eq: dual-basis-recurrence-relation}
  \begin{aligned}
        &x g_j = g_{j+1} + b_{d_b-1-j} g_0; \quad j=0, \cdots, d_b-2 \\
        & xg_{d_b-1} = b_0 g_0 = g_0,
    \end{aligned}  
\end{equation}
where $b_k = [x^k]b$ and $b_0 = 1$ since $x \nmid b$.
\begin{proof}
   Recall that $g_{0}, \cdots, g_{d_b-1}$ form a basis of $c_b R$, so that $xg_j = \sum^{d_b-1}_{i=0} p_{ji} g_i$ with coefficients $p_{ji} = [x^i]c_b(xg_j) = \tr_{\mathbb{F}_q/\mathbb{F}_2}(g_j(\beta) \beta^{1-i})$, as in the proof of Lemma~\ref{lemma: basis-dual-basis-from-finite-field-ring}. For $i = 1, \cdots, d_b-1$, duality gives $p_{ji} = \delta_{j, i-1}$. For $i = 0$, $p_{j0} = \tr_{\mathbb{F}_q/\mathbb{F}_2}(g_j(\beta)\beta)$, and dividing $b(\beta) = 0$ by $\beta^{d_b-1}$ yields
   \begin{align}
       \beta = \sum^{d_b-1}_{k=0} b_k \beta^{-(d_b-1-k)},
   \end{align}
   so that, by duality, $p_{j0} = b_{d_b-1-j}$. Collecting terms, $xg_j = g_{j+1} + b_{d_b-1-j} g_0$ for $j = 0, \cdots, d_b-2$, and $xg_{d_b-1} = b_0 g_0 = g_0$.
\end{proof}
\end{lemma}

Hence, in general, the dual basis might not be related by a cyclic group element, and numerical calculation indicates that the weights for the dual basis elements are not uniform. However, we could say that these dual basis polynomials are related by the "dual" cyclic group. 

\begin{lemma}\label{lemma: dual-basis-transitive-action}
Let the dual basis be $\{ g_0, \cdots, g_{d_b-1}\}$ associated with $c_b $ and $d_b = \deg b$. There exists, for $0 \leq j \leq d_b-1$, a coset
\begin{align}
    \{u_j + qb: q \in R\} 
\end{align}
such that $ g_{j+1} = (u_j + x) g_{j}$ for $0 \leq j \leq d_b-1$, where the indices satisfy $g_{d_b} \equiv g_0$. Furthermore, any dual basis element can be generated by $g_0$:
\begin{align}
    g_j = f_j g_0,
\end{align}
for some $f_j\in R$ such that $1\leq |f_j|\leq |b|-1$.

    \begin{proof}
        We can solve the recurrence relation in Eqs.~\eqref{eq: dual-basis-recurrence-relation} so that all the dual basis elements are generated by $g_0$:
        \begin{align}\label{eq: recursed-dual-basis-boundary}
            g_j = \left(x^{j} + b_{d_b-1}x^{j-1} +\cdots+ b_{d_b-j}  \right)g_0, \quad 0 \leq j \leq d_b-1. 
        \end{align}
        This is consistent with the boundary relation in Eqs.~\eqref{eq: dual-basis-recurrence-relation}, since 
        \begin{align}
            xg_{d_b-1} = (b +1)g_0 = g_0, 
        \end{align}
        where $b g_0 = 0$ since $g_0 \in c_bR$. Denote $f_{j} g_0 := g_j$. Note that $c_bR \cong \mathbb{F}_q$ for $q = 2^{d_b}$, so that every nonzero element admits an inverse:
        \begin{align}
            c_b v_{j}\left(x^{j} + b_{d_b-1}x^{j-1} + \cdots + b_{d_b-j}\right) = c_b , 
        \end{align}
        for a unique $v_{j} \in c_bR$. Setting $u_j := b_{d_b-1-j} v_j$ for $0 \leq j \leq d_b-2$, the recurrence gives $g_{j+1} = xg_j + b_{d_b-1-j} g_0 = (u_j + x)g_j$; at the boundary, $g_0 = xg_{d_b-1}$ gives $u_{d_b-1} = 0$. Any representative of $\{u_j + qb: q \in R\}$ serves equally, since $b g_j = 0$.
    \end{proof}
\end{lemma}

\begin{remark}
The ideal $(b)$ forms a cyclic code. Hence, we could choose the sparsest representation from $\{u_j + qb: q \in R\} $, which is a typical syndrome-decoding problem. 
\end{remark}

\begin{example}[Generalized bicycle codes Ref.~\cite{webster2025explicitconstructionlowoverheadgadgets}]
    We look into the code example $\LP_{31}(A, B)$ for $A = 1 + x^6 + x^{15}$ and $B =1 + x^5 + x^7 $ in Ref.~\cite{webster2025explicitconstructionlowoverheadgadgets}; the balanced-product view of the same code appears in the example of Section~\ref{subsec: general-formulation-homological-algebra}. According to Theorem~\ref{thm: main-LP-basis-characterization}, it suffices to evaluate $\ker_R A, \coker_R A, \ker_R B, \coker_R B$. By Theorem~\ref{thm: main-computation-kernel-cokernel}, we can evaluate this explicitly through decomposing into different algebraic rules of irreducible polynomial factors of $x^l+1$: 
    \begin{align}
        \begin{aligned}
            b^{(1)} &= 1 +x, \\
            b^{(2)} &= 1 +x^2 + x^5, \\
            b^{(3)} &= 1  + x^3 + x^5, \\
            b^{(4)} &= 1 + x+ x^2 + x^3 + x^5, \\
            b^{(5)} &= 1 +x+ x^2 + x^4 + x^5, \\
            b^{(6)} &= 1 + x^2 + x^3 + x^4+ x^5, \\
            b^{(7)} &= 1 + x+ x^3 + x^4+ x^5.
        \end{aligned}
    \end{align}
      For $i = 1, \ldots, 6$, the reductions $A^{(i)}$ and $B^{(i)}$ are nonzero, hence invertible in the field $R_{(i)}$, so that $\ker_{R_{(i)}} A = \coker_{R_{(i)}} A = \ker_{R_{(i)}} B = \coker_{R_{(i)}} B = 0$; only $b^{(7)}$ contributes. Under the algebraic rule of $b^{(7)}=0$, $A$ reduces to $A^{(7)} = 0$. Hence, $\ker_{R_{(7)}} A \cong cR$ and $\coker_{R_{(7)}} A \cong cR$. Similarly, $B$ reduces to $B^{(7)} = 0$ so that $\ker_{R_{(7)}} B \cong cR$ and $\coker_{R_{(7)}} B \cong cR$, for central idempotent $c$ associated with $b^{(7)}$:
      \begin{align}
    \begin{aligned}
        c = &x^{30} + x^{29} + x^{27} + x^{26} + x^{23} + x^{22} + x^{21} \\
        &+ x^{16} + x^{15} + x^{13} + x^{11} + x^{8} + x^{4} + x^{2} + x + 1.
    \end{aligned}
\end{align}
        
Then we can compute that (note $c^2 =c$)
\begin{align}
   G_A = G_B = c
\end{align}
and 
\begin{align}
    E_A = E_B = \begin{pmatrix}
        c \\
        cx \\
        cx^2 \\
        cx^3 \\
        cx^4
    \end{pmatrix}.
\end{align}
An $R$-valued minimal generator matrix $L^M_Z$ is given by
\begin{align}\label{eq: main-webster-gb-code-basis}
    L^M_Z = \begin{pmatrix}
        c \\
        cx \\
        cx^2 \\
        cx^3 \\
        cx^4
    \end{pmatrix} \bigoplus \begin{pmatrix}
        c \\
        cx \\
        cx^2 \\
        cx^3 \\
        cx^4
    \end{pmatrix}
\end{align}
The corresponding binary generator matrix is $\mbb B(L_Z^M)$.
\end{example}
\begin{example}[$\LP_l(A, b)$] Let $b |x^l+1$ be an irreducible factor, $\LP_l(A, b)$ defines code families studied in~\cite{Panteleev_2021, Panteleev_2022}; in particular for $\LP_l(A, 1+x)$ which has asymptotic distance $n/\log n$ for an expanding $A$. Whenever $R$ is semi-simple (or, more generally, $b$ is the unique irreducible factor), we can analytically find the basis. Let $l = 9$,
\begin{align}
     A = \left( \begin{array}{ccc}
    1 & 1 & 1   \\
    1 & x & x^2
\end{array}\right),
\end{align}
and $\LP_l(A, b=x^2 + x+ 1)$. This code encodes only two logical qubits given by 
\begin{align}
    L^M_Z = G_A \otimes_R E_b \oplus 0
\end{align}
with $ \rs_R G_A = \ker_R A$ and $\rs_{\mathbb{F}_2} E_b \cong \coker_R b = R/(b) \cong c_bR$. The second summand $E_A \otimes_R G_b$ vanishes: it is supported in the $c_b$ sector, where $A(\beta)$ has full rank and hence $c_b \coker_R A = 0$. In this case, we can compute $c_b$ from the coefficient formula Lemma~\ref{lemma: central-idempotent-coeff-root}:
\begin{align}
    c_b = x + x^2 + x^4 + x^5 + x^7 + x^8
\end{align}
so that 
\begin{align}
    E_b = \left( c_b, c_b x\right)^T.
\end{align}
One can compute that $c_b\ker_R A = R(c_bu)$ for 
\begin{align}
   u = \left(x, x+1, 1 \right)^T.
\end{align}

This would give two $Z$-type basis operators: 
\begin{align}
    L^M_Z = ( c_b u, c_b x u)^T \oplus 0. 
\end{align}
\end{example}

\subsection{Proof of Theorem~\ref{thm: main-LP-basis-characterization} and Theorem~\ref{thm: main-computation-kernel-cokernel}}\label{subsec: lp-basis-characterization}

The preceding analysis using algebraic root theory provides a fully explicit method for proving the theorems stated in Section~\ref{sec: main}. 

Lemma~\ref{lemma: first-LP-basis-characterization} overlooks two caveats: $(i)$ We are unable to count analytically the number of logical operators for the lifted-product (LP) codes, as evident from Counter-Example~\ref{counterexample: first-LP-basis-minimality}. $(ii)$ There is so far no guarantee for $Z$-basis $L_Z$ and $X$-type basis $L_X$ to be conjugate bases.

\begin{theorem}[Restatement of Theorem~\ref{thm: main-computation-kernel-cokernel}]\label{thm: formal-computation-kernel-cokernel}
    Let $R = \mathbb{F}_2[x]/(x^l + 1)$ with odd $l$, let $x^l+1 = \prod^t_{i=1} b^{(i)}$ with central idempotents $c_{b^{(i)}}$ and $d_i := \deg b^{(i)}$, and let $A \in R^{m_A \times n_A}$ with components $A_{(i)} = \mathrm{ev}_{\beta_i}(c_{b^{(i)}}A) = A(\beta_i)$. For each $i = 1, \cdots, t$: 
    \begin{itemize}
        \item Let $I^{(i)}_A$ be an information set of $\ker_{R_{(i)}} A_{(i)}$ with systematic generator matrix $G_{A_{(i)}}$ whose rows are $g^{(i,j)}_{\mathcal{A}_1}$ for $j \in I^{(i)}_A$, and let $e^{(i,j)}_{\mathcal{A}_1}$ be the unit vector supported on the $j$-th index of $I^{(i)}_A$. Denote their row-stacked matrix $E_{A^*_{(i)}}$, which forms a basis of $\coker_{R_{(i)}} A^T_{(i)}$. 
        
        \item Let $I^{(i)}_{A^*}$ be an information set of $\ker_{R_{(i)}} A^T_{(i)}$ with systematic generator matrix $G_{A^*_{(i)}}$ whose rows are $g^{(i,j)}_{\mathcal{A}_0}$ for $j \in I^{(i)}_{A^*}$, and let $e^{(i,j)}_{\mathcal{A}_0}$ be the unit vector supported on the $j$-th index of $I^{(i)}_{A^*}$. Denote their row-stacked matrix $E_{A_{(i)}}$, which forms a basis of $\coker_{R_{(i)}} A_{(i)}$.
        
        \item Let $G^{(i)}_{A^*}$ be constructed as follows. For each row $g^{(i, j)}_{\mathcal{A}_0}$ of $G_{A^*_{(i)}}$, denote the conjugated lift $\bar{c}_{b^{(i)}}\bar{\xi}^{(i, j)}_{\mathcal{A}_0} := \overline{\mathrm{ev}^{-1}_{\beta_i}(g^{(i, j)}_{\mathcal{A}_0})}$. Let $m^{(i)}_0, \cdots, m^{(i)}_{d_i-1}$ be distinct integers from $[l]$ such that $\{\beta_i^{-m^{(i)}_k}\}^{d_i-1}_{k=0}$ is an $\mathbb{F}_2$-basis of $\mathbb{F}_{q_i}$ --- equivalently, such that the dual basis of Lemma~\ref{lemma: basis-dual-basis-from-finite-field-ring} exists; for instance, $m^{(i)}_k = k$ --- and denote the primal basis: 
        \begin{align}
            x^{m^{(i)}_0} \bar{c}_{b^{(i)}}\bar{\xi}^{(i, j)}_{\mathcal{A}_0}, \cdots, x^{m^{(i)}_{d_i-1}} \bar{c}_{b^{(i)}}\bar{\xi}^{(i, j)}_{\mathcal{A}_0}.
        \end{align}

        \item Similarly, let $E^{(i)}_{A}$ be constructed as follows. For each row $e^{(i,j)}_{\mathcal{A}_0}$ from $E_{A_{(i)}}$, denote $c_{b^{(i)}}e^{(i,j)}_{\mathcal{A}_0} = \mathrm{ev}^{-1}_{\beta_i}(e^{(i, j)}_{\mathcal{A}_0})$. Taking the integers $m^{(i)}_0, \cdots, m^{(i)}_{d_i-1}$ to be the same as in the previous item, we denote the dual basis:
        \begin{align}
            h^{(i)}_{m^{(i)}_0} c_{b^{(i)}} e^{(i,j)}_{\mathcal{A}_0}, \cdots, h^{(i)}_{m^{(i)}_{d_i-1}} c_{b^{(i)}} e^{(i,j)}_{\mathcal{A}_0},
        \end{align}
        for $h^{(i)}_{m^{(i)}_0}, \cdots, h^{(i)}_{m^{(i)}_{d_i-1}}$ the dual basis associated with $c_{b^{(i)}}$ as in Lemma~\ref{lemma: basis-dual-basis-from-finite-field-ring}.
    \end{itemize}
    Then let $G_{A^*} = \mathrm{vstack}(G^{(1)}_{A^*}, \cdots, G^{(t)}_{A^*})$, and $E_{A} = \mathrm{vstack}(E^{(1)}_{A}, \cdots, E^{(t)}_{A})$. Let $G_A = \mathrm{vstack}(\mathrm{ev}^{-1}_{\beta_{1}}(G_{A_{(1)}}), \cdots,\mathrm{ev}^{-1}_{\beta_{t}}(G_{A_{(t)}}) )$ and $E_{A^*} = \mathrm{vstack}(\overline{\mathrm{ev}^{-1}_{\beta_{1}}(E_{A^*_{(1)}})}, \cdots, \overline{\mathrm{ev}^{-1}_{\beta_{t}}(E_{A^*_{(t)}})})$, whose rows are $\bar{c}_{b^{(i)}} e^{(i,j)}_{\mathcal{A}_1}$. Then we have 
       \begin{align}
       \begin{aligned}
            &\rs_R G_{A} = \ker_R A; \quad \rs_{\mathbb{F}_2} E_{A} \cong  \coker_R A, \\
            &\rs_{\mathbb{F}_2} G_{A^*} = \ker_R A^*; \quad \rs_{R} E_{A^*} \cong  \coker_R A^*,
       \end{aligned}
    \end{align}
    and $G_{A^*}$, $E_{A}$ are minimal (over $\mathbb{F}_2$) generating matrices of $\ker_R A^*$ and $\coker_R A$, respectively.
\end{theorem}
\begin{proof}
        We prove the four displayed identities in turn, and then the minimality claim.
        \begin{itemize}
            \item ($\rs_R G_A = \ker_R A$). This is Lemma~\ref{lemma: kernel-A-odd}.

            \item ($\rs_{\mathbb{F}_2} E_A \cong \coker_R A$). By Lemmas~\ref{lemma: coker-basis-ker-transpose} and~\ref{lemma: cokernel-A-odd}, the lifted unit vectors $c_{b^{(i)}} e^{(i,j)}_{\mathcal{A}_0}$ satisfy conditions Eq.~\eqref{eq: splitting-condition-1} and Eq.~\eqref{eq: splitting-condition-2} of Lemma~\ref{lemma: section-maps-coker-R}, and the duals $\{h^{(i)}_{m^{(i)}_k}\}$ of Lemma~\ref{lemma: basis-dual-basis-from-finite-field-ring} form an $\mathbb{F}_2$-basis of $c_{b^{(i)}}R$, since $[x^{m^{(i)}_{k'}}]\big(c_{b^{(i)}}h^{(i)}_{m^{(i)}_k}\big)=\delta_{kk'}$.

 Hence the rows of $E_A$ are $\mathbb{F}_2$-linearly independent --- rows with distinct $j$ have disjoint supports, and for fixed $j$ the coefficients $h^{(i)}_{m^{(i)}_k}$ are $\mathbb{F}_2$-independent --- and $\rs_{\mathbb{F}_2} E_A = \rs_R E_A \cong \coker_R A$ by Lemma~\ref{lemma: section-maps-coker-R}.

            \item ($\rs_{\mathbb{F}_2} G_{A^*} = \ker_R A^*$). Conjugation is a ring automorphism with $\overline{c_{b^{(i)}}} = c_{b^{(i)\vee}}$ (Lemma~\ref{lemma: conjugate-central-idempotent}); applying it to $A^T c_{b^{(i)}}\xi^{(i,j)}_{\mathcal{A}_0} = 0$ gives $A^* \bar{c}_{b^{(i)}}\bar{\xi}^{(i,j)}_{\mathcal{A}_0} = 0$. Moreover, since $\overline{A^T u} = A^* \bar{u}$ for every $u \in R^{m_A}$, we have $u \in \ker_R A^T$ if and only if $\bar{u} \in \ker_R A^*$, and $v \in \IM_R A^T$ if and only if $\bar{v} \in \IM_R A^*$; applying this to Lemmas~\ref{lemma: kernel-A-odd} and~\ref{lemma: cokernel-A-odd} for $A^T$ gives the corresponding statements for $A^*$, with $c_{b^{(i)}}$ replaced by $\bar{c}_{b^{(i)}}$. By Lemma~\ref{lemma: conjugate-root}, $\beta_i^{-1}$ is a root of $b^{(i)\vee}$, and $\mathrm{ev}_{\beta_i^{-1}}(\bar{c}_{b^{(i)}} x^{m^{(i)}_k}) = \beta_i^{-m^{(i)}_k}$; since $\{\beta_i^{-m^{(i)}_k}\}^{d_i-1}_{k=0}$ is an $\mathbb{F}_2$-basis of $\mathbb{F}_{q_i}$ by construction and $\mathrm{ev}_{\beta_i^{-1}}$ is an isomorphism (Theorem~\ref{thm: isometry-lifted-product}), the set $\{\bar{c}_{b^{(i)}} x^{m^{(i)}_k}\}^{d_i-1}_{k=0}$ is an $\mathbb{F}_2$-basis of $\bar{c}_{b^{(i)}}R$. Hence, for each $j$, $\operatorname{span}_{\mathbb{F}_2}\{x^{m^{(i)}_k}\bar{c}_{b^{(i)}}\bar{\xi}^{(i,j)}_{\mathcal{A}_0}\}^{d_i-1}_{k=0} = R\,\bar{c}_{b^{(i)}}\bar{\xi}^{(i,j)}_{\mathcal{A}_0}$, so the rows of $G_{A^*}$ form an $\mathbb{F}_2$-basis of $\ker_R A^*$.

            \item ($\rs_R E_{A^*} \cong \coker_R A^*$). By conjugating Lemma~\ref{lemma: cokernel-A-odd} applied to $A^T$, the vectors $\bar{c}_{b^{(i)}} e^{(i,j)}_{\mathcal{A}_1}$ satisfy conditions Eq.~\eqref{eq: splitting-condition-1} and Eq.~\eqref{eq: splitting-condition-2} of Lemma~\ref{lemma: section-maps-coker-R} with $A$ replaced by $A^*$; hence $\rs_R E_{A^*} \cong \coker_R A^*$.

            \item (Minimality). $G_{A^*}$ has $\sum^t_{i=1} d_i |I^{(i)}_{A^*}|$ rows and $\dim_{\mathbb{F}_2} \ker_R A^* = \sum^t_{i=1} d_i \dim_{R_{(i)}} \ker_{R_{(i)}} A^T_{(i)} = \sum^t_{i=1} d_i |I^{(i)}_{A^*}|$; $E_A$ has the same number of rows, and $\dim_{\mathbb{F}_2} \coker_R A = \sum^t_{i=1} d_i \dim_{R_{(i)}} \coker_{R_{(i)}} A_{(i)}$ equals it by Lemma~\ref{lemma: coker-basis-ker-transpose}. A basis has the minimal number of rows among generating sets, which proves the last claim.
        \end{itemize}
    \end{proof}

    Constructed in this way, we can finally state a formal, precise version of Theorem~\ref{thm: main-LP-basis-characterization}, which addresses Counter-Example~\ref{counterexample: first-LP-basis-minimality} and the caveats in Lemma~\ref{lemma: first-LP-basis-characterization}.

    \begin{theorem}[Restatement of Theorem~\ref{thm: main-LP-basis-characterization}]\label{thm: formal-LP-basis-characterization}
          Let $R = \mathbb{F}_2[x]/(x^l + 1)$ with odd $l$, and let the 2D lifted-product (LP) code $\LP_l(A, B)$ be given in Definition~\ref{def: main-2d-LPAB}. Then the $Z$-type logical operators corresponding to $H_1(\mc{M})$ have supports generated by
    \begin{align}
        L_{Z}^M =   G_{A} \otimes_R E_{B} \bigoplus E_{A} \otimes_R G_{B},
    \end{align}
    where $G_C$ and $E_C$ are constructed in Theorem~\ref{thm: formal-computation-kernel-cokernel} for $C \in \{A, B\}$. Similarly, the $X$-type logical operators corresponding to $H^1(\mc{M})$ have supports generated by
    \begin{align}
        L_{X}^M =  E_{A^*} \otimes_R G_{B^*} \bigoplus G_{A^*} \otimes_R E_{B^*},
    \end{align}
    where we take $G_{C^*}$ and $E_{C^*}$ constructed in Theorem~\ref{thm: formal-computation-kernel-cokernel} for $C \in \{A, B\}$. Explicitly, for each $i$ and $m \in \{m^{(i)}_0, \cdots, m^{(i)}_{d_i-1}\}$, the nonzero rows of $L^M_Z$ and $L^M_X$ are, respectively,
    \begin{align}
        \supp \bar{Z}^{(L,i)}_{p, q, m} &:= c_{b^{(i)}} h^{(i)}_{m} \zeta^{(p)}_{\mathcal{A}_1} \otimes_R e^{(q)}_{\mathcal{B}_0}, \label{eq: Z-basis-left-LP}\\
        \supp \bar{X}^{(L,i)}_{p, q, m} &:= \bar{c}_{b^{(i)}}x^{m} e^{(p)}_{\mathcal{A}_1} \otimes_R  \bar{\xi}^{(q)}_{\mathcal{B}_0}, \label{eq: X-basis-left-LP}
    \end{align}
    for the left sector and $p \in I^{(i)}_A$ and $q \in I^{(i)}_{B^*}$, and $\zeta^{(p)}_{\mathcal{A}_1} := \mathrm{ev}^{-1}_{\beta_i}(g^{(i,p)}_{\mathcal{A}_1})$, $\bar{\xi}^{(q)}_{\mathcal{B}_0} := \overline{\mathrm{ev}^{-1}_{\beta_i}(g^{(i,q)}_{\mathcal{B}_0})}$, with $\zeta^{(q)}_{\mathcal{B}_1}$ and $\bar{\xi}^{(p)}_{\mathcal{A}_0}$ defined in the same way. For the right sector, they are
    \begin{align}
        \supp \bar{Z}^{(R,i)}_{p, q, m} &:= h^{(i)}_{m} c_{b^{(i)}} e^{(p)}_{\mathcal{A}_0} \otimes_R \zeta^{(q)}_{\mathcal{B}_1}, \label{eq: Z-basis-right-LP}\\
        \supp \bar{X}^{(R,i)}_{p, q, m} &:= x^{m} \bar{c}_{b^{(i)}} \bar{\xi}^{(p)}_{\mathcal{A}_0} \otimes_R  e^{(q)}_{\mathcal{B}_1}, \label{eq: X-basis-right-LP}
    \end{align}
    for $p \in I^{(i)}_{A^*}$ and $q \in I^{(i)}_B$. Then the $\bar{Z}^{(i)}_{p,q,m}$ and $\bar{X}^{(i)}_{p,q,m}$ form $\mathbb{F}_2$-bases of $H_1(\mc{M})$ and $H^1(\mc{M})$, respectively, and they are conjugate with respect to $\langle \cdot, \cdot \rangle_R$:
    \begin{align}\label{eq: conjugate-pairing-formal-LP-basis}
        \langle \supp \bar{Z}^{(i)}_{p, q, m},\; \supp \bar{X}^{(i')}_{p', q', m'} \rangle_R = \delta_{ii'}\,\delta_{pp'}\,\delta_{qq'}\,\delta_{mm'}.
    \end{align}
    \end{theorem}
\begin{proof}
    First, $\rs_{\mathbb{F}_2} L^M_Z = H_1(\mc{M})$ and $\rs_{\mathbb{F}_2} L^M_X = H^1(\mc{M})$ by Lemma~\ref{lemma: first-LP-basis-characterization}, whose conditions hold by Theorem~\ref{thm: formal-computation-kernel-cokernel}. Next, we show the linear independence. The rows $\bar{Z}^{(L,i)}_{p,q,m}$ and $\bar{Z}^{(R,i)}_{p,q,m}$ are supported on disjoint components of the physical space, so the two cases separate. Given a vanishing $\mathbb{F}_2$-combination $\sum_{i,p,q,m} c^{(i)}_{pqm} \bar{Z}^{(L,i)}_{p,q,m} = 0$, applying $c_{b^{(i)}}$ and $\mathrm{ev}_{\beta_i}$ entrywise gives, for each $i$,
    \begin{align}
        \sum_{p,q,m} c^{(i)}_{pqm}\, \mathrm{ev}_{\beta_i}(h^{(i)}_{m})\, g^{(i,p)}_{\mathcal{A}_1} \otimes e^{(q)}_{\mathcal{B}_0} = 0
    \end{align}
    over $\mathbb{F}_{q_i}$; the tensors $g^{(i,p)}_{\mathcal{A}_1} \otimes e^{(q)}_{\mathcal{B}_0}$ are $\mathbb{F}_{q_i}$-linearly independent, and $\{\mathrm{ev}_{\beta_i}(h^{(i)}_{m})\}_{m}$ is an $\mathbb{F}_2$-basis of $\mathbb{F}_{q_i}$, so all $c^{(i)}_{pqm} = 0$. The rows $\bar{Z}^{(R,i)}_{p,q,m}$ and the rows of $L^M_X$ are analogous. Hence the nonzero rows form $\mathbb{F}_2$-bases. Finally, we show the conjugacy Eq.~\eqref{eq: conjugate-pairing-formal-LP-basis}. For tensor rows,
    \begin{align}
       \begin{aligned}
        &\langle u_1 \otimes_R u_2, v_1 \otimes_R v_2 \rangle_R \\
        &= [x^0]\Big( \sum_a (u_1)_a \overline{(v_1)}_a \Big)\Big( \sum_b (u_2)_b \overline{(v_2)}_b \Big).
       \end{aligned}
    \end{align}
    Rows with different labels $L$ and $R$ pair to zero, since they are supported on disjoint components of the physical space. For $\bar{Z}^{(L,i)}_{p,q,m}$ and $\bar{X}^{(L,i')}_{p',q',m'}$, the two factors evaluate, by the systematic forms, to $\delta_{ii'}\delta_{pp'}\, c_{b^{(i)}} h^{(i)}_{m} x^{-m'}$ and $\delta_{ii'}\delta_{qq'}\, c_{b^{(i)}}$, so that
    \begin{align}
        \langle \supp \bar{Z}^{(L,i)}_{p,q,m}, \supp \bar{X}^{(L,i')}_{p',q',m'} \rangle_R &= \delta_{ii'}\delta_{pp'}\delta_{qq'}\, [x^{m'}]\big(c_{b^{(i)}} h^{(i)}_{m}\big) \\
        &= \delta_{ii'}\delta_{pp'}\delta_{qq'}\delta_{mm'};
    \end{align}
    the pairing of $\bar{Z}^{(R,i)}_{p,q,m}$ with $\bar{X}^{(R,i')}_{p',q',m'}$ is analogous, which establishes Eq.~\eqref{eq: conjugate-pairing-formal-LP-basis}.
\end{proof}

\begin{example}[Correct minimal basis construction to Counter-Example~\ref{counterexample: first-LP-basis-minimality}]
        We revisit Counter-Example~\ref{counterexample: first-LP-basis-minimality} with the construction of Theorem~\ref{thm: formal-computation-kernel-cokernel}. Here $x^3 + 1 = b^{(1)} b^{(2)}$ with $b^{(1)} = 1+x$ and $b^{(2)} = 1+x+x^2$, so that $d_1 = 1$, $d_2 = 2$, and $c_{b^{(1)}} = 1+x+x^2$, $c_{b^{(2)}} = x+x^2$. The information sets are $I^{(1)}_A = \{1\}$, $I^{(2)}_A = \{1, 2\}$, and $I^{(1)}_{B^*} = I^{(2)}_{B^*} = \{1\}$ with $e^{(1)}_{\mathcal{B}_0} = (1, 0)$; since $\ker_R B = 0$, there are no rows with label $R$. Taking $m^{(i)}_k = k$, the dual bases are $h^{(1)}_0 = 1+x+x^2$ and $h^{(2)}_0 = 1+x^2$, $h^{(2)}_1 = x+x^2$, and the lifted systematic rows are $\zeta^{(1)}_{\mathcal{A}_1} = (1+x+x^2,\, 1+x+x^2)$ for $b^{(1)}$, and $(x+x^2,\, 0)$, $(0,\, x+x^2)$ for $b^{(2)}$. The five nonzero rows of $L^M_Z$ are
        \begin{align}
            \supp \bar{Z}^{(L,1)}_{1,1,0} &= (1+x+x^2,\; 1+x+x^2) \otimes_R e^{(1)}_{\mathcal{B}_0}, \\
            \supp \bar{Z}^{(L,2)}_{1,1,0} &= (1+x^2,\; 0) \otimes_R e^{(1)}_{\mathcal{B}_0}, \\
            \supp \bar{Z}^{(L,2)}_{1,1,1} &= (x+x^2,\; 0) \otimes_R e^{(1)}_{\mathcal{B}_0}, \\
            \supp \bar{Z}^{(L,2)}_{2,1,0} &= (0,\; 1+x^2) \otimes_R e^{(1)}_{\mathcal{B}_0}, \\
            \supp \bar{Z}^{(L,2)}_{2,1,1} &= (0,\; x+x^2) \otimes_R e^{(1)}_{\mathcal{B}_0},
        \end{align}
        which are $\mathbb{F}_2$-linearly independent, matching $k = 5$. In contrast to Counter-Example~\ref{counterexample: first-LP-basis-minimality}, each row is supported in a single $c_{b^{(i)}}$-component, and the six-row dependence does not arise.
    \end{example}

\begin{corollary}[Dimension counting of lifted-product code]\label{cor: dimension-LP-semi-simple}
      When $l$ is odd, the lifted-product code $\LP_l(A, B)$ encodes
    \begin{align}
        k = \sum^t_{i=1} d_i \left( |I^{(i)}_A|\, |I^{(i)}_{B^*}| + |I^{(i)}_{A^*}|\, |I^{(i)}_B| \right)
    \end{align}
    logical qubits, where $|I^{(i)}_A| = \dim_{R_{(i)}} \ker_{R_{(i)}} A_{(i)}$ and $|I^{(i)}_{A^*}| = \dim_{R_{(i)}} \coker_{R_{(i)}} A_{(i)}$, and similarly for $B$.
\end{corollary}
Hence, both caveats of Lemma~\ref{lemma: first-LP-basis-characterization} are resolved: the number of logical operators is counted analytically by Corollary~\ref{cor: dimension-LP-semi-simple}, addressing $(i)$, and the constructed bases are conjugate with respect to $\langle \cdot, \cdot \rangle_R$ by Theorem~\ref{thm: formal-LP-basis-characterization}, addressing $(ii)$.

\subsection{Conditions of canonical LP codes and properties of canonical LP basis}\label{sec: LP-canonical-basis}

We now discuss the canonical LP codes and canonical LP basis. Recall the definition of canonical LP codes in Definition~\ref{def: main-canonical-basis-LP}, which we restate here for convenience.

\begin{definition}[Canonical LP basis and canonical LP code]\label{def: formal-canonical-LP-basis}
      For an $\mathrm{LP}_l(A, B)$ code, we say that its logical bases $L^M_Z$ in Eq.~\eqref{eq: main-LZ} and $L^M_X$ in Eq.~\eqref{eq: main-LX} are canonical if, up to column permutations over $R$, they contain a subset of rows of the following form, where $r_A := n_A - m_A$ and $r_B^* := m_B - n_B$:
    \begin{equation}
    \begin{aligned}
        \tilde{L}_Z^M & = \mathrm{vstack}\{x^m \left( I_{r_A}, G_A^0\right)\otimes_R \left(I_{r_B^*}, 0\right) \oplus 0\}_{m = 0}^{l - 1}; \\
        \tilde{L}_X^M & = \mathrm{vstack}\{x^{m^{\prime}} \left( I_{r_A}, 0\right)\otimes_R \left(I_{r_B^*}, G_B^{*0}\right) \oplus 0\}_{m^{\prime} = 0}^{l - 1},
    \end{aligned}
    \label{eq: canonical-basis-formal}
    \end{equation}
    where $G_A^0 \in R^{r_A \times m_A}$ and $G_B^{*0} \in R^{r_B^* \times n_B}$.
    We can also rewrite these matrices in terms of their indexed rows:
    \begin{equation}
    \begin{aligned}
        \tilde{L}_Z^M & = \mathrm{vstack}\{\supp \bar{Z}_{i, j, m}\}_{i \in [r_A], j \in [r_B^*], m \in [l]}; \\
        \tilde{L}_X^M & = \mathrm{vstack}\{\supp \bar{X}_{i, j, m}\}_{i \in [r_A], j \in [r_B^*], m \in [l]},
    \end{aligned}
    \label{eq: indexed-basis-formal}
    \end{equation}
    where $\supp \bar{Z}_{i, j, m}$ and $\supp \bar{X}_{i, j, m}$ denote the following support vectors over $R$, respectively:
    \begin{align}\label{eq: canonical-basis-support-vectors-formal}
       &\supp \bar{Z}_{i, j, m} :=  x^m u^{(i)}_{\mathcal{A}_1} \otimes_R e^{(j)}_{\mathcal{B}_0}, \\
       &\supp \bar{X}_{i, j, m} := x^m e^{(i)}_{\mathcal{A}_1} \otimes_R v^{(j)}_{\mathcal{B}_0}
    \end{align}
    where $i \in [r_A]$ and $j \in [r_B^*]$. The vectors $u^{(i)}_{\mathcal{A}_1}$ and $e^{(i)}_{\mathcal{A}_1}$ are the $i$-th rows of $\left( I_{n_A - m_A}, G_A^0\right)$ and $\left( I_{n_A - m_A}, 0\right)$, respectively. The vectors $e^{(j)}_{\mathcal{B}_0}$ and $v^{(j)}_{\mathcal{B}_0}$ are the $j$-th rows of $\left(I_{m_B - n_B}, 0\right)$ and $\left(I_{m_B - n_B}, G_B^{*0}\right)$, respectively.
     We say that an $\mathrm{LP}_l(A, B)$ code is canonical if it admits a canonical logical basis in the above form.
\end{definition}

This form has an enriched cyclicity structure and is primarily used for all later logical constructions. We now provide a necessary and sufficient condition for a lifted-product LP code to be canonical. 

\begin{theorem}[Conditions for canonical LP codes: Restatement of Theorem~\ref{thm: main-canonical-LP-basis}]\label{thm: formal-condition-canonical-LP}
     Let $\LP_l(A, B)$ be a lifted-product (LP) code with odd $l$. For $C \in \{A, B^*\}$, let $I^{(i)}_{C}$ denote the information set of $G^{(i)}_C$, defined as in Theorem~\ref{thm: formal-computation-kernel-cokernel} for each component $R_{(i)}$. Then $\LP_l(A, B)$ admits a canonical LP basis in the sense of Definition~\ref{def: formal-canonical-LP-basis} if and only if these information sets can be chosen such that
     \begin{align}\label{expr: information-set-overlaps}
        |\bigcap_{i=1}^t I^{(i)}_A| \geq n_A - m_A =: r_A, \quad |\bigcap_{i=1}^t I^{(i)}_{B^*}| \geq m_B - n_B =: r_B^*,
     \end{align}
     where $I^{(i)}_A \subseteq [n_A]$ and $I^{(i)}_{B^*} \subseteq [m_B]$ are the information sets of $G^{(i)}_A$ and $G^{(i)}_{B^*}$, respectively. The resulting canonical sector contains $k_c := r_A r_B^* l \leq k$ logical qubits; any additional logical qubits are residual.
\end{theorem}
\begin{proof}
The ``only if'' direction is straightforward. Suppose that a canonical LP basis exists after a column permutation. Apply the decomposition by central idempotents, using $\sum^{t}_{i=1} c_{b^{(i)}}=1$. Hence, by Theorem~\ref{thm: isometry-lifted-product}, the systematic forms $\left(I_{r_A}, G_A^0\right)$ (resp. $\left(I_{r_B^*}, G_B^{*0}\right)$) are defined in every finite field decomposition $\mathbb{F}_{q_i}$ with $q_i = 2 ^{\deg b^{(i)}}$, which implies expressions~\eqref{expr: information-set-overlaps}. We prove the ``if'' direction. After column permutations, assume $[r_A] \subseteq \bigcap^t_{i=1} I^{(i)}_A$ and $[r_B^*] \subseteq \bigcap^t_{i=1} I^{(i)}_{B^*}$. For $p \in [r_A]$, define the glued rows
    \begin{align}
        u^{(p)}_{\mathcal{A}_1} := \sum^t_{i=1} c_{b^{(i)}} \zeta^{(p)}_{\mathcal{A}_1},
    \end{align}
    where $\zeta^{(p)}_{\mathcal{A}_1}$ are the lifted systematic rows of Theorem~\ref{thm: formal-computation-kernel-cokernel}. Since $p \in I^{(i)}_A$ for every $i$, the systematic entries give $u^{(p)}_{\mathcal{A}_1}|_{[r_A]} = e^{(p)}$, so the rows $u^{(p)}_{\mathcal{A}_1}$ form the block $\left(I_{r_A}, G_A^0\right)$; similarly, $v^{(q)}_{\mathcal{B}_0} := \sum^t_{i=1} \bar{c}_{b^{(i)}} \bar{\xi}^{(q)}_{\mathcal{B}_0}$ for $q \in [r_B^*]$ forms $\left(I_{r_B^*}, G_B^{*0}\right)$. Moreover, since $q \in I^{(i)}_{B^*}$ for every $i$, the unit vector $e^{(q)}_{\mathcal{B}_0}$ represents a cokernel class in every component by Lemmas~\ref{lemma: coker-basis-ker-transpose} and~\ref{lemma: cokernel-A-odd}. Hence the families
    \begin{align}
        x^m u^{(p)}_{\mathcal{A}_1} \otimes_R e^{(q)}_{\mathcal{B}_0}, \quad x^{m'} e^{(p)}_{\mathcal{A}_1} \otimes_R v^{(q)}_{\mathcal{B}_0},
    \end{align}
    for $m, m' \in \{0, \cdots, l-1\}$, $p \in [r_A]$, and $q \in [r_B^*]$, lie in $H_1(\mc{M})$ and $H^1(\mc{M})$ by Lemma~\ref{lemma: first-LP-basis-characterization}, and constitute precisely the canonical form Eq.~\eqref{eq: canonical-basis-formal} of Definition~\ref{def: formal-canonical-LP-basis}. The converse direction follows straightforwardly. 
\end{proof}

\begin{counterexample}[Non-canonical LP code]\label{counterexample: non-canonical-LP}
    Let $l = 3$ and take $B = (1+x,\, 1)^T$ as in Counter-Example~\ref{counterexample: first-LP-basis-minimality}, with
    \begin{align}
        A = \left( 1+x, \; 1+x+x^2 \right),
    \end{align}
    so that $r_A = 1$ and $\LP_3(A, B)$ is a $[[15, 3]]$ code. The component kernels are $\ker_{R_{(1)}} A_{(1)} = \operatorname{span}\{(1, 0)\}$ and $\ker_{R_{(2)}} A_{(2)} = \operatorname{span}\{(0, 1)\}$, each spanned by a unit vector, so the information sets are unique: $I^{(1)}_A = \{1\}$ and $I^{(2)}_A = \{2\}$. Hence $|\bigcap^2_{i=1} I^{(i)}_A| = 0 < r_A$, and by Theorem~\ref{thm: formal-condition-canonical-LP}, $\LP_3(A, B)$ admits no canonical LP basis. In contrast, the code of Counter-Example~\ref{counterexample: first-LP-basis-minimality} satisfies the condition, with a canonical sector of $k_c = 3$ among its $k = 5$ logical qubits.
\end{counterexample}

\subsubsection*{Extending to non-semisimple rings}

We consider some reduced cases, and a similar conclusion of canonical LP basis holds for a general (not necessarily Abelian) non-semisimple ring $R$. For our interests, we present the case with a non-semisimple Abelian ring $R:= \mathbb{F}_2[x]/(x^l+1)$ for $l$ an even number. We refer to the case with non-Abelian rings to the inspired, companion paper~\cite{Bhardwaj2026MittenCodes}. 

\begin{definition}[Canonical matrices over $R$]\label{def: canonical-matrix}
    Let $G$ be a (not necessarily Abelian) group and $R := \mathbb{F}_2[G]$ be its group algebra. A matrix $A \in R^{m_A \times n_A}$ is called \emph{canonical} if up to column permutations, the square blocks $A' \in R^{m_A \times m_A}$  formed by the last $m_A$ columns of $A$  have invertible binarizations $\mathbb{B}(A')$.
\end{definition}

\begin{example}[Canonical matrices]\label{example: canonical-matrices}
    We give {four} examples of canonical matrices.
    \begin{itemize}
        \item (Canonical matrix for $\LP^{2 \times 3}_{36} = [[468, 36, 20]]$).
        Let $G = C_{36}$, so that $R = \mathbb{F}_2[x]/(x^{36} + 1)$, and let
	\begin{align}\label{eq: canonical-matrix-C36}
	A = \big(\, A' \;\big|\; q_1 \,\big) \in R^{2 \times 3}, \quad q_1 = \begin{pmatrix} 1 \\ x^{8} + x^{13} + x^{16} + x^{17} \end{pmatrix},
	\end{align}
    with $B = A^*$. Here $m_A = 2$, $r_A = 1$, and $B^* = A$. The pivot is the block $A'$ of the first two columns (the third column, $q_1$, is the information column), with $\det A' = 1 + x^{11} + x^{29}$ a unit, $(1 + x^{11} + x^{29})^2 = 1$ in $R$, so that
	\begin{align}\label{eq: canonical-pivot-C36}
        \begin{aligned}
	A' &= \begin{pmatrix} 1 & 1 + x^{11} + x^{29} \\ 1 & 0 \end{pmatrix}, \\
    A'^{-1} &= \begin{pmatrix} 0 & 1 \\ 1 + x^{11} + x^{29} & 1 + x^{11} + x^{29} \end{pmatrix},
        \end{aligned}
	\end{align}
    and $\mathbb{B}(A') = \mathbb{B}((B^*)')$ is invertible. This gives $\mathrm{LP}(A, A^*) = [[468, 36, 20]]$ (Definition~\ref{def: canonical-product-code} below), with $n = (3 \cdot 3 + 2 \cdot 2) \cdot 36 = 468$ and $k = r_A r^*_B |G| = 36$.
        \item (Canonical matrix for $\LP^{3 \times 5}_{28} = [[952, 112, 17]]$). The explicit matrix is given in Eq.~\eqref{eq: saturating-C28-matrix} of Example~\ref{example: saturating-canonical-product-codes} below: its pivot block $A'$ is upper triangular with unit diagonal entries, hence with invertible binarization $\mathbb{B}(A')$.
{
        \item (Canonical matrix for $\LP^{2 \times 4}_{16} = [[320, 64, 10]]$).
        Let $G = C_{16}$, so that $R = \mathbb{F}_2[x]/(x^{16}+1)$, and let
	\begin{align}\label{eq: canonical-matrix-C16}
	A = \begin{pmatrix} x^{11} & x^{4}{+}x^{7} & x^{3}{+}x^{9} & 1 \\ x^{8}{+}x^{11} & x & x^{7} & x^{5}{+}x^{15} \end{pmatrix} \in R^{2 \times 4},
	\end{align}
    with $B = A^*$, written as $A = (\, A' \mid q_1 \mid q_2 \,)$ with the pivot block $A'$ of the first two columns and the information columns $q_1$, $q_2$. Here $m_A = 2$, $r_A = 2$, and $B^* = A$. Here $\det A' = x^{11} \cdot x + (x^{4} + x^{7})(x^{8} + x^{11}) = x^{2}$, a monomial unit with $(\det A')^{-1} = x^{14}$, so that
	\begin{align}\label{eq: canonical-pivot-C16}
        \begin{aligned}
	A' &= \begin{pmatrix} x^{11} & x^{4}{+}x^{7} \\ x^{8}{+}x^{11} & x \end{pmatrix}, \\
    A'^{-1} &= \begin{pmatrix} x^{15} & x^{2}{+}x^{5} \\ x^{6}{+}x^{9} & x^{9} \end{pmatrix},
        \end{aligned}
	\end{align}
    and $\mathbb{B}(A') = \mathbb{B}((B^*)')$ is invertible. This gives $\mathrm{LP}(A, A^*) = [[320, 64, 10]]$, with $n = (4 \cdot 4 + 2 \cdot 2) \cdot 16 = 320$ and $k = r_A r^*_B |G| = 64$.

        \item (Canonical matrix for $\LP^{3 \times 5}_{20} = [[680, 80, 13]]$).
        Let $G = C_{20}$, so that $R = \mathbb{F}_2[x]/(x^{20}+1)$, and let $A \in R^{3 \times 5}$ be
	{\small\setlength{\arraycolsep}{3pt}	\begin{align}\label{eq: canonical-matrix-C20}
	A = \begin{pmatrix} x^{5}{+}x^{6}{+}x^{16} & 0 & 0 & x^{7}{+}x^{19} & x \\ 0 & x^{7}{+}x^{8}{+}x^{17} & 0 & x^{15} & x{+}x^{4} \\ 0 & 0 & x^{9}{+}x^{13}{+}x^{17} & x^{19} & x^{10} \end{pmatrix},
	\end{align}}
    with $B = A^*$, written as $A = (\, A' \mid q_1 \mid q_2 \,)$ with the pivot block $A'$ of the first three columns and the information columns $q_1$, $q_2$. Here $m_A = 3$, $r_A = 2$, and $B^* = A$. The three diagonal entries are units, with inverses $(x^{5} + x^{6} + x^{16})^{-1} = x^{6} + x^{15} + x^{16}$, $(x^{7} + x^{8} + x^{17})^{-1} = x + x^{11} + x^{12}$, $(x^{9} + x^{13} + x^{17})^{-1} = x^{7} + x^{15} + x^{19}$, and $\det A' = x^{2} + x^{3} + x^{5} + x^{6} + x^{7} + x^{9} + x^{10} + x^{13} + x^{15} + x^{17} + x^{19}$ is a unit with $(\det A')^{-1} = x + x^{2} + x^{6} + x^{7} + x^{11} + x^{14} + x^{17}$, so that
	\begin{align}\label{eq: canonical-pivot-C20}
        \begin{aligned}
	A' &= \begin{pmatrix} x^{5}{+}x^{6}{+}x^{16} & 0 & 0 \\ 0 & x^{7}{+}x^{8}{+}x^{17} & 0 \\ 0 & 0 & x^{9}{+}x^{13}{+}x^{17} \end{pmatrix}, \\
    A'^{-1} &= \begin{pmatrix} x^{6}{+}x^{15}{+}x^{16} & 0 & 0 \\ 0 & x{+}x^{11}{+}x^{12} & 0 \\ 0 & 0 & x^{7}{+}x^{15}{+}x^{19} \end{pmatrix},
        \end{aligned}
	\end{align}
    and $\mathbb{B}(A') = \mathbb{B}((B^*)')$ is invertible. This gives $\mathrm{LP}(A, A^*) = [[680, 80, 13]]$, with $n = (5 \cdot 5 + 3 \cdot 3) \cdot 20 = 680$ and $k = r_A r^*_B |G| = 80$.
}
    \end{itemize}
\end{example}

Recall that we define the conjugate transpose by $A^*$ to denote the transpose of $A$ with each entry $\sum_g c_g g$ replaced by $\sum_g c_g g^{-1}$.

\begin{definition}[Canonical-product codes over $R$]\label{def: canonical-product-code}
    Let $R := \mathbb{F}_2[G]$ for a finite group $G$, and let $\mathcal{A}_{\bullet}: \mathcal{A}_1 \xrightarrow{A} \mathcal{A}_0$ and $\mathcal{B}_{\bullet}: \mathcal{B}_1 \xrightarrow{B} \mathcal{B}_0$ be two $1$-chains over $R$, with $A \in R^{m_A \times n_A}$ and $B \in R^{m_B \times n_B}$, such that $A$ and $B^*$ are canonical matrices over $R$ (Definition~\ref{def: canonical-matrix}). Write $n_A = m_A + r_A$ and, for $B^* \in R^{m^*_B \times n^*_B}$ with $m^*_B := n_B$ and $n^*_B := m_B$, $n^*_B = m^*_B + r^*_B$. The canonical product code $\mathrm{LP}(A, B)$ is the CSS code obtained by binarizing the total complex $\mathcal{M}: \mathcal{M}_2 \xrightarrow{\partial^M_2} \mathcal{M}_1 \xrightarrow{\partial^M_1} \mathcal{M}_0$ of $\mathcal{A}_{\bullet} \otimes_R \mathcal{B}_{\bullet}$, where
	\begin{align}\label{eq: canonical-product-total-complex}
        \begin{aligned}
            &\mathcal{M}_2 = \mathcal{A}_1 \otimes_R \mathcal{B}_1, \quad \mathcal{M}_1 = \mathcal{A}_1 \otimes_R \mathcal{B}_0 \oplus \mathcal{A}_0 \otimes_R \mathcal{B}_1, \\&\mathcal{M}_0 = \mathcal{A}_0 \otimes_R \mathcal{B}_0,
        \end{aligned}
    \end{align}
    and the boundary maps are given by
	\begin{align}\label{eq: canonical-product-boundary-maps}
        \begin{aligned}
	(\partial^M_2)^* &= \big( I_{n_A} \otimes_R B^*,\, A^* \otimes_R I_{n_B} \big), \\
	\partial^M_1 &= \big( A \otimes_R I_{m_B},\, I_{m_A} \otimes_R B \big).
	\end{aligned}
    \end{align}

\end{definition}

A key statement that we state is that:

\begin{theorem}[Basis of canonical product code]\label{thm: basis-canonical-product-code}
    Consider the product $\mathcal A_{\bullet}\otimes_R\mathcal B_{\bullet}$. 
    Then we have that
	\begin{align}\label{eq: basis-canonical-homology}
        \begin{aligned}
	H_1(\mathcal{A}_{\bullet} \otimes_R \mathcal{B}_{\bullet}) &= H_1(\mathcal{A}_{\bullet}) \otimes_R H_0(\mathcal{B}_{\bullet}) \\
    &\cong \mathrm{rs}_R\{G_A \otimes_R E_{B} \},
        \end{aligned}
	\end{align}
	with
	\begin{align}\label{eq: basis-GA-EB}
    G_A = \begin{pmatrix}
	I_{r_A} \: G^{0}_A
	 \end{pmatrix}; \quad E_{B} = \left(I_{r^*_B} \: 0\right),
    \end{align}
    for some matrix $G^{0}_A$, if and only if with $A$ and $B^*$ canonical, as in Definition~\ref{def: canonical-product-code}. Furthermore, the row span over $\mathbb{F}_2$ of the binarization
    \begin{align}
        {\mathbb{B}(G_A\otimes_R E_B)},
    \end{align}
    extending to full space with $\mathcal{A}_0 \otimes_R \mathcal{B}_1$ by $0$, gives a minimal and complete basis for the homology group $\mathbb{B}(H_1(\mathcal{A}_{\bullet} \otimes_R \mathcal{B}_{\bullet}))$. Similarly, the cohomology:
	 \begin{align}\label{eq: basis-canonical-cohomology}
        \begin{aligned}
	 H^1(\mathcal{A}_{\bullet} \otimes_R \mathcal{B}_{\bullet}) &= H^1(\mathcal{A}_{\bullet}) \otimes_R H^0(\mathcal{B}_{\bullet}) \\
     &\cong \mathrm{rs}_R\{E_{A^*} \otimes_R G_{B^*} \},
        \end{aligned}
	 \end{align}
	 with
	 \begin{align}\label{eq: basis-GBstar-EAstar}
	 G_{B^*} = \begin{pmatrix}
	I_{r^*_B} \: G^{0}_{B^*}
	 \end{pmatrix}; \quad E_{A^*} = \left(I_{r_A} \: 0\right),
	 \end{align}
     for some matrix $G^{0}_{B^*}$, and similarly for the row span over $\mathbb{F}_2$ of the binarization $ {\mathbb{B}(G_A\otimes_R E_B)}$.
     \begin{proof}
        We prove the homology case. The cohomology case follows similarly. We first show that any nontrivial element $u \in \ker \partial^M_1$ and $u \notin \IM_R \partial^M_2$ can be expressed as an $R$-linear combination of the rows of $G_A \otimes_R E_{B}$. To see this, note that a generic element could have support on $\mathcal{A}_1 \otimes_R \mathcal{B}_0$ and $\mathcal{A}_0 \otimes_R \mathcal{B}_1$. Since $\mathbb{B}(A)$ has full row rank, stabilizer reduction gives a representative supported entirely on $\mathcal{A}_1 \otimes_R \mathcal{B}_0$. Write it as $h = \sum_{j=1}^{m_B} u^{(j)}_{\mathcal{A}_1} \otimes e^{(j)}_{\mathcal{B}_0}$, where $u_j \in \ker_R A$ and $e^{(j)}_{\mathcal{B}_0}$ is the $j$th standard basis vector of $\mathcal{B}_0$. {Since $B^*$ is a canonical matrix (recall that $\mathbb{B}(B^*) = \mathbb{B}(B)^T$), every $e^{(j)}_{\mathcal{B}_0}$ with $j \notin [r^*_B]$ is congruent modulo $\IM_R B$: 
        \begin{align}
            B = \left( \begin{array}{c}
                P_B \\ B'
            \end{array}\right),
        \end{align}
        where $\mathbb{B}(B')$, by definition, is full rank. Hence, for $j\notin [r^*_B]$,  we can also find some stabilizer element which deforms $u^{(j)}_{\mathcal{A}_1} \otimes_R e^{(j)}_{\mathcal{B}_0}$ into the first $r^*_B$ coordinates.
        Since $u_j \in \ker_R A$, this does not incur additional support over $\mathcal{A}_0 \otimes_R \mathcal{B}_1$.} We then prove that rows of $ G_A \otimes_R E_{B}$ are an $R$-generating set to nontrivial homology classes in $H_1(\mathcal{A}_{\bullet} \otimes_R \mathcal{B}_{\bullet})$: it does not lie in $\IM_R \partial^M_2$, and are minimal to $\mathrm{rs}_R\{G_A \otimes_R E_{B} \}$. By the tensor product structure, each row of $G_A \otimes_R E_{B}$ can be written as 
        \begin{align}\label{expr: nontrivial-canonical-product-basis}
            g^{(i)}_{\mathcal{A}_1} \otimes_R e^{(j)}_{\mathcal{B}_0},
        \end{align}
        for the $i$-th and $j$-th rows of $G_A$ and $E_B$ ($[r^*_B]$), respectively. Then by the above argument such a row \eqref{expr: nontrivial-canonical-product-basis} lies in $\ker \partial^M_1$ and not in $\IM_R \partial^M_2$, hence, non-trivial in homology. Furthermore, by the tensor-product structure and the forms of $G_A$ and $E_B$, any such row has a unique pivot entry such that no other row is supported on it, hence it is minimal in $\mathrm{rs}_R\{G_A \otimes_R E_{B} \}$. Similarly, this also shows that ${\mathbb{B}\!\left((G_A\otimes_R E_B)^T\right)^T}$ is also minimal whose $\mathbb{F}_2$ row span gives (is isomorphic to) $\mathbb{B}(H_1(\mathcal{A}_{\bullet} \otimes_R \mathcal{B}_{\bullet}))$, which is a vector basis. 

        {For the converse, the kernel-generator convention immediately preceding Lemma~\ref{lemma: section-maps-coker-R} gives $\operatorname{rs}_R(G_A)=\ker_R A$ and $\operatorname{rs}_R(G_{B^*})=\ker_R B^*$.
Write $A=(P_A\mid A')$, where $A'\in R^{m_A\times m_A}$ since $n_A=m_A+r_A$ by Definition~\ref{def: canonical-product-code}.
If $v\in R^{m_A}$ satisfies $A'v=0$, then $ \begin{pmatrix}
    0\\v
\end{pmatrix}\in\ker_R A$.
Completeness and the systematic form $G_A=(I_{r_A}\mid G_A^0)$ in Eq.~\eqref{eq: basis-GA-EB} therefore give $u\in R^{r_A}$ such that
\[
\begin{pmatrix}0\\v\end{pmatrix}
=G_A^T u
=\begin{pmatrix}u\\(G_A^0)^T u\end{pmatrix}.
\]
The first block forces $u=0$, hence $v=0$.
The bijectivity of vector binarization and the identity $\mathbb B(A'v)=\mathbb B(A')\mathbb B(v)$ imply that $\mathbb B(A')$ is injective.
It is square and therefore invertible.
Since $n_B^*=m_B^*+r_B^*$, the same argument with the complete generator $G_{B^*}=(I_{r_B^*}\mid G_{B^*}^0)$ in Eq.~\eqref{eq: basis-GBstar-EAstar} shows that the square block consisting of the last $m_B^*=n_B$ columns of $B^*$ has invertible binarization.
Thus $A$ and $B^*$ are canonical by Definition~\ref{def: canonical-matrix}.}
         Hence, the result follows. 
     \end{proof}
\end{theorem}

\begin{remark}
{More generally, suppose that information sets of $\ker\mathbb B(A)$ and $\ker\mathbb B(B^*)$ can be chosen to contain all binary coordinates in the first $r_A$ columns of $A$ and the first $r_B^*$ columns of $B^*$, respectively. The corresponding partial systematic kernel families yield $k_c=r_A r_B^*l$ conjugate logical pairs of the form in Eq.~\eqref{eq: canonical-basis-formal}, even when the constituent matrices fail to have full row rank. These pairs form bases of subspaces of binary homology and cohomology. They can be retained and extended to complete conjugate logical bases by Gaussian elimination over $\mathbb F_2$. The remaining logical operators can be chosen to have zero pairing with every canonical logical operator of the opposite type.}
\end{remark}

Hence, for the canonical-product codes, their distance is always upper-bounded by the classical distance of the constituent classical codes. Later in Section~\ref{sec: soundness-criterion-for-canonical-product}, we will see that there is a fundamental trade-off between the distance and the soundness of the code, which has direct implications on efficient and distance-preserving parallel (hypergraph) surgery. 

\subsection{Application to the multivariate cases}

It is of interest to study multivariate polynomial rings and their basis characterization. Here, with a slight abuse of notation, we denote 
\begin{align}
    R := \mathbb{F}_2[x, y]/(x^l+1, y^m+1), 
\end{align}
for some integers $l, m$. The Chinese remainder theorem still applies in this case: 
\begin{align}
    \prod^{t}_{i=1} (a^{(i)})^{\mu_i} = x^l +1, \quad  \prod^{s}_{j=1} (b^{(j)})^{\nu_j} = y^m +1,
\end{align}
where $a^{(i)}$ are the irreducible factors of $x^l+1$ and $b^{(j)}$ those of $y^m+1$,
and $\mu_i, \nu_j$ are multiplicities. Then $R$ admits the following factorization 
\begin{align}
    R \cong \prod^{t}_{i=1} \prod^{s}_{j=1} R / \left((a^{(i)})^{\mu_i},(b^{(j)})^{\nu_j}\right).
\end{align}
Hence, we can still find the corresponding central idempotent $c_{ij} := c_{a^{(i)}}c_{b^{(j)}}$, where $c_{a^{(i)}}$ and $c_{b^{(j)}}$ are, respectively, the central idempotents for $(a^{(i)})^{\mu_i}$ and $(b^{(j)})^{\nu_j}$. Hence, in the multiplicity-free case with $\mu_i, \nu_j=1$, the isomorphism in Theorem~\ref{thm: isometry-lifted-product} extends naturally with 

\begin{align}
    \begin{aligned}
    c_{ij}R
    &\cong \mathbb{F}_{q_i} \otimes_{\mathbb{F}_2} \mathbb{F}_{q_j} \\
    &\cong \prod^{\operatorname{gcd}(\deg a^{(i)}, \deg b^{(j)})}_{u=1}
    \mathbb{F}_{2^{\operatorname{lcm}(\deg a^{(i)}, \deg b^{(j)})}},
    \end{aligned}
\end{align}
for $q_i = 2^{\deg a^{(i)}}$ and $q_j = 2^{\deg b^{(j)}}$, where $\operatorname{lcm}(\deg a^{(i)}, \deg b^{(j)})$ denotes the least common multiple. In particular, when $\gcd(\deg a^{(i)}, \deg b^{(j)}) > 1$, the component $c_{ij}R$ is a product of fields --- semisimple but not a field --- in contrast to the univariate case.
This can be generalized further to more than two variables. Similarly, for $A \in R^{m_A \times n_A}$, 
\begin{align}
    A = \sum_{i, j} c_{a^{(i)}} c_{b^{(j)}} A,
\end{align}
since $\sum_{i, j} c_{a^{(i)}} c_{b^{(j)}} = 1$. In the multiplicity-free case, $R$ is semisimple, and the constructions of Theorems~\ref{thm: formal-computation-kernel-cokernel}--\ref{thm: formal-condition-canonical-LP} apply componentwise, with each component a finite product of fields.

\section{Code surgeries from chain complexes}\label{sec: code-surgery-general}
In this section, we systematically discuss the code surgery techniques used for canonical LP codes. This discussion supports a modular construction. We will assume the following notation hierarchies: 
\begin{enumerate}
    \item \textbf{Abstract logical operators.} We denote $\bar{L} \in \bar{\mathcal{P}}_k$ for a generic, abstract logical operator, and $\bar{Z}_{i, j, m}$ (resp. $\bar{X}_{i, j, m}$) by the basis set of $\bar{\mathcal{P}}_k$. 
    \item \textbf{Symplectic notation/Physical representatives.} We denote \emph{upper-case letters} without bars to be physical Pauli operators $L$ for $\bar{L}$, which can be thought of as symplectic notation. Though, when the context is clear, we will use Pauli (product) notation and symplectic (addition) notation interchangeably. 
    \item \textbf{Binary support vectors.} We denote \emph{lower-case letters} for binary support vectors, which can be used to represent the support of a physical Pauli representative such as 
    \begin{align}
        L:= L(u, v) = X(u) Z(v).
    \end{align}
\end{enumerate}

\subsection{General formulation}\label{sec: surgery-formulation}

A clean and systematic formulation of code surgery can be given in terms of chain complexes, initiated by the cone code formalism~\cite{Ide_2025}. 

 \begin{definition}[Graph surgery: Reformulations of Definition~\ref{def: main-graph-surgery-gadget}]\label{def: surgery-cone-code}
Let $\mathcal{G}(\mathcal{V}, \mathcal{E}, \mathcal{C})$ be a sparse graph with cycles $\mathcal{C}$, edges $\mathcal{E}$, and vertices $\mathcal{V}$ --- identified, respectively, with the added $X$-type checks, the ancilla qubits, and the added $Z$-type checks --- forming an exact $2$-chain $ \mathcal{G}: \mathbb{F}_2[\mathcal{V}] \xrightarrow{\partial_1} \mathbb{F}_2[\mathcal{E}]\xrightarrow{\partial_0} \mathbb{F}_2[\mathcal{C}]$, where $\partial_1$ and $\partial_0$ are, respectively, the edge-vertex and cycle-edge incidence matrices. Consider the chain maps $(\Phi^{(X)}_1, \Phi^{(Z)}_1)$ and $(\Phi^{(X)}_0, \Phi^{(Z)}_0)$ such that the following diagram commutes: 
\begin{equation}\label{diagram: surgery-cone-code}
   \begin{tikzcd}
	{\mathcal{Q}_0} & {\mathcal{Q}_1} & {\mathcal{Q}_2} \\
	{\mathcal{A}_1} & {\mathcal{A}_0} & {\mathcal{A}_{-1}} \\
	{\mathcal{Q}_2} & {\mathcal{Q}_1} & {\mathcal{Q}_0} \\
	{\textcolor{blue}{Z}} & {\textcolor{green}{Q}} & {\textcolor{red}{X}}
	\arrow[from=1-1, to=1-2]
	\arrow["{H_Z}", from=1-2, to=1-3]
	\arrow[shift right=4, draw={rgb,255:red,117;green,117;blue,117}, curve={height=-18pt}, dotted, tail reversed, from=1-2, to=3-2]
	\arrow["{\Phi^{(X)}_1}", from=2-1, to=1-2]
	\arrow["{\partial_1}", from=2-1, to=2-2]
	\arrow["{\Phi^{(Z)}_1}"{description}, from=2-1, to=3-2]
	\arrow["{\Phi^{(X)}_0}"{description}, from=2-2, to=1-3]
	\arrow["{\partial_0}", from=2-2, to=2-3]
	\arrow["{\Phi^{(Z)}_0}", from=2-2, to=3-3]
	\arrow[from=3-1, to=3-2]
	\arrow["{H_X}", from=3-2, to=3-3]
	\arrow[curve={height=-24pt}, dotted, from=4-1, to=2-1]
	\arrow[dotted, from=4-1, to=3-1]
	\arrow[dotted, from=4-2, to=3-2]
	\arrow[curve={height=24pt}, dotted, from=4-3, to=2-3]
	\arrow[dotted, from=4-3, to=3-3]
\end{tikzcd}
\end{equation}
    The dashed gray line identifies the upper and bottom $\mathcal{Q}_1$ as the physical qubits of the same code $\mathcal{Q}$. We denote the \emph{surgery gadget} to be $\mathcal{S}[\mathcal{G};  (\Phi_1, \Phi_0)]$ and the \emph{merged code} to be $\mathrm{cone}(\mathcal{Q}, \mathcal{S}[\mathcal{G};  (\Phi_1, \Phi_0)])$ for $\Phi_1 = (\Phi^{(X)}_1, \Phi^{(Z)}_1)$ and $\Phi_0 = (\Phi^{(X)}_0, \Phi^{(Z)}_0)$, with the parity-check matrix: 
    \begin{equation}
        H_{\text {merged }}=\quad
        \begin{NiceArray}{cccc}[baseline=line-5]
        \Block{1-2}{\mathcal{Q}} & & \Block{1-2}{\mathcal{G}} & \\
        \textcolor{red}{X} & \textcolor{blue}{Z} & \textcolor{red}{X} & \textcolor{blue}{Z} \\
        H_X & 0 & \Phi_0^{(Z)} & 0\\
        0 & H_Z &  \Phi_0^{(X)}& 0 \\
        (\Phi_1^{(X)})^T & (\Phi_1^{(Z)})^T & 0 & \partial_1^T \\
        0 & 0 & \partial_0 & 0
        \CodeAfter
        \SubMatrix({3-1}{6-4})[hlines=2,vlines=2]
        \end{NiceArray}\quad .
    \end{equation}
     In addition to $H_X H_Z^T = 0$, the commutativity of the merged code stabilizers requires that
    \begin{equation}
    \begin{aligned}
        (\Phi_1^{(X)})^T \Phi_1^{(Z)} & = (\Phi_1^{(Z)})^T \Phi_1^{(X)}, \\
        H_Z \Phi_1^{(X)} + \Phi_0^{(X)} \partial_1 & = 0, \\
        H_X \Phi_1^{(Z)} + \Phi_0^{(Z)} \partial_1 & = 0.
    \end{aligned}
    \end{equation}
    Here $\Phi^T_1 = ((\Phi^{(X)}_1)^T| (\Phi^{(Z)}_1)^T)$ is written in the symplectic notation. We say that the surgery gadget $\mathcal{S}[\mathcal{G};  (\Phi_1, \Phi_0)]$ measures a set of logical operators $\{ \bar{L}^{(i)} \}_{i}$ if they admit physical representatives $L^{(i)} \in \Phi_1(\ker \partial_1)$ for each $i$.
\end{definition}

\begin{remark}
    Note that technically, the set of operators measured by the surgery gadget $\mathcal{S}[\mathcal{G};  (\Phi_1, \Phi_0)]$ is given by $\Phi_1(\ker \partial_1)$, which may not necessarily correspond to logical operators of the code $\mathcal{Q}$. The commutativity only ensures that these operators commute with the stabilizers of $\mathcal{Q}$. It is, then, customary to impose such an assumption. In this case, we denote $\Phi_1(\ker \partial_1)$ to be the subspace of measured logical operators and $\dim \Phi_1(\ker \partial_1)$ to be the dimension of logical qubits (as opposed to the symplectic dimension). 
\end{remark}

The above definition presents a generalization of the cone code formalism~\cite{Ide_2025} to the case of measuring a set of logical operators, which might not be low-rate. Here, we say that a surgery measurement is \emph{low-rate} if $\dim \ker \partial_1=1$, i.e., a connected graph. We say that a surgery measurement is \emph{high-rate} if $\dim \ker \partial_1>1$, i.e., a disconnected graph or, more generally, a hypergraph. 

\begin{definition}[Hypergraph-surgery gadget: Restatement of Definition~\ref{def: main-hypergraph-surgery}]\label{def: hypergraph-surgery}
    A hypergraph-surgery gadget $\mathcal{S}[\mathcal{G}; (\Phi_1, \Phi_0)]$ to $\mathcal{Q}$ and the merged code $\mathrm{cone}(\mathcal{Q}, \mathcal{S}[\mathcal{G}; (\Phi_1, \Phi_0)])$ are defined as in Definition~\ref{def: surgery-cone-code}, with the following additions/exceptions. 
    \begin{itemize}
        \item $\mathcal{G}: \mathbb{F}_2[\mathcal{V}] \xrightarrow{\partial_1} \mathbb{F}_2[\mathcal{E}] \xrightarrow{\partial_0} \mathbb{F}_2[\mathcal{C}]$ is a generic $\mbb{F}_2$-valued hypergraph chain complex, representing a hypergraph with vertices $\mathcal{V}$, hyperedges $\mathcal{E}$, and hyperfaces $\mathcal{C}$, which is not necessarily exact. 
        \item $\dim \Phi_1(\ker \partial_1) \geq 1$, and every nonzero $L \in \Phi_1(\ker \partial_1)$ corresponds to a nontrivial logical operator of $\mathcal{Q}$. 
    \end{itemize}
\end{definition}

It turns out that the fault tolerance of this measurement relies on analyzing the merged cone code $\mathrm{cone}(\mathcal{Q}, \mathcal{S}[\mathcal{G}; (\Phi_1, \Phi_0)])$ for $\Phi_1 = (\Phi^{(X)}_1, \Phi^{(Z)}_1)$. We first understand its logical structure. 

\begin{lemma}[Classification of logical structures of cone code Definition~\ref{def: surgery-cone-code}]\label{lemma: surgery-cone-code-logical-operators}
Let $\mathcal{Q}$ have logical dimension $k(\mathcal{Q})$ and $\mathcal{G}$ have logical dimension $k(\mathcal{G})$. Then the cone code defined in~\eqref{diagram: surgery-cone-code} has logical dimension at most $k(\mathcal{Q}) - \dim \Phi_1(\ker \partial_1) + k(\mathcal{G})$. In particular, we can characterize any logical operator $\bar{L}$ which admits a physical representative $L=(L_{\mathcal{Q}}, L_{\mathcal{E}})$ in one of the following two types: 
\begin{enumerate}
    \item \label{surgery-logical-data} There exists a representative $L = (L^{\mathcal{Q}}_X L^{\mathcal{Q}}_Z, E_X)$ such that $L^{\mathcal{Q}}_X$ and $L^{\mathcal{Q}}_Z$ are supported on $\mathcal{Q}$, at least one of which is a nontrivial logical operator of $\mathcal{Q}$.  
    \item \label{surgery-logical-ancilla} Either $L$ admits a representative $L = (P^{\mathcal{Q}}_X P^{\mathcal{Q}}_Z,\, L^{\mathcal{G}}_Z E_X)$ with $L^{\mathcal{G}}_Z$ a nontrivial $Z$-type logical operator of $\mathcal{G}$, or $L$ admits a representative $L = (0, L^{\mathcal{G}}_X)$ with $L^{\mathcal{G}}_X$ a nontrivial $X$-type logical operator of $\mathcal{G}$. 
\end{enumerate}
\begin{proof}
    Following the remark after Definition~\ref{def: surgery-cone-code}, we assume throughout that every nonzero element of $\Phi_1(\ker \partial_1)$ is a nontrivial logical operator of $\mathcal{Q}$. To be consistent with notation, we will use upper-case letters for Pauli operators and lower-case letters for their support vectors. We can always denote a logical Pauli operator by
    \begin{align}
        L = (P^{\mathcal{Q}}_XP^{\mathcal{Q}}_Z, E_X E_Z)
    \end{align}
    for data qubits and ancillary qubits, respectively. By construction, we have that 
    \begin{align}
        \quad e_Z \in \ker \partial_0. 
    \end{align}
    Suppose that $e_Z \in \IM \partial_1$; then there exists a deformation by added (vertex) checks: 
    \begin{align}
     (P^{\mathcal{Q}}_XP^{\mathcal{Q}}_Z, E_X E_Z) \mapsto (P^{\mathcal{Q}}_X X(\Phi^{(X)}_1(v)) P^{\mathcal{Q}}_Z Z(\Phi^{(Z)}_1(v)), E_X).
    \end{align}
    Then it is clear that $p_X + \Phi^{(X)}_1(v) \in \ker H_Z$ and similarly $p_Z + \Phi^{(Z)}_1(v) \in \ker H_X$. Suppose that neither  $P^{\mathcal{Q}}_X X(\Phi^{(X)}_1(v))$ nor $P^{\mathcal{Q}}_Z Z(\Phi^{(Z)}_1(v))$ is a nontrivial logical operator in $\mathcal{Q}$. There exist $u \in \mathcal{Q}_0$ and $u' \in \mathcal{Q}_2$ such that the above can be deformed into 
    \begin{align}
          (0, E_X X(\Phi^{(X)T}_0(u'))X(\Phi^{(Z)T}_0(u))),
    \end{align}
    and $e_X +\Phi^{(X)T}_0(u') + \Phi^{(Z)T}_0(u) \in \ker \partial^T_1$, so that it is only nontrivial if $e_X +\Phi^{(X)T}_0(u') + \Phi^{(Z)T}_0(u) \notin \IM \partial^T_0$; that is, it is a nontrivial $X$-type logical operator of $\mathcal{G}$. If $e_Z \notin \IM \partial_1$, then $E_Z$ is a nontrivial $Z$-type logical operator of $\mathcal{G}$. Denote $L^{\mathcal{Q}}_X$ and $L^{\mathcal{Q}}_Z$ to be nontrivial logical operators of $\mathcal{Q}$, and $L^{\mathcal{G}}_X$ and $L^{\mathcal{G}}_Z$ to be nontrivial logical operators of $\mathcal{G}$. Hence, in summary, we can divide into the following two classes: 
    \begin{enumerate}
        \item  $(L^{\mathcal{Q}}_X L^{\mathcal{Q}}_Z, E_X)$ such that $L^{\mathcal{Q}}_X$ and $L^{\mathcal{Q}}_Z$ are supported on $\mathcal{Q}$, at least one of which is a nontrivial logical operator of $\mathcal{Q}$, such that there does not exist any $z \in \ker \partial_1$ such that $L^{\mathcal{Q}}_X X(\Phi^{(X)}_1(z)) L^{\mathcal{Q}}_Z Z(\Phi^{(Z)}_1(z))$ becomes a trivial logical operator in $\mathcal{Q}$. 
        \item The ancillary ``$Z$-type'' logical operators which admit representatives of the form $(P^{\mathcal{Q}}_XP^{\mathcal{Q}}_Z, L^{\mathcal{G}}_Z E_X)$ and ancillary $X$-type logical operators which admit representatives of the form $(0, L^{\mathcal{G}}_X)$. 
    \end{enumerate}
    Hence, the logical dimension of the merged code is at most $k(\mathcal{Q}) - \dim \Phi_1(\ker \partial_1) + k(\mathcal{G})$, as desired. 
    \end{proof}
\end{lemma}

\begin{remark}
    Note that it is in general not true that the logical dimension of the merged code is equal to $k(\mathcal{Q}) - \dim \Phi_1(\ker \partial_1) + k(\mathcal{G})$. In particular, there might exist some vector $u \in \mathcal{Q}_0$ such that $H^T_X u=0$ (resp. $u' \in \mathcal{Q}_2$ with $H^T_Z u'=0$), and $l^{\mathcal{G}}_X + \Phi^{(Z)T}_0(u)=0$ (resp. $l^{\mathcal{G}}_X + \Phi^{(X)T}_0(u')=0$). This is not the case if $H_X$ and $H_Z$ are full row rank, in which case the dimension of the merged (cone) code saturates the dimension upper bound. Furthermore, logical operators of type~\ref{surgery-logical-ancilla} vanish whenever the graph chain is exact. 
\end{remark}

{Hence, we can choose a basis of the logical operators of the merged code of the following bipartite form. The elements of the first part have trivial ancillary homology, so that each restricts to a nontrivial logical operator of the data code $\mathcal{Q}$ only. We call them type~\ref{surgery-logical-data}. The second part consists of one fixed representative for each ancillary logical class of $\mathcal{G}$, that is, for each element of a basis of $\ker \partial_0 / \IM \partial_1$ and of $\ker \partial^T_1 / \IM \partial^T_0$. We call them type~\ref{surgery-logical-ancilla}, as they support nontrivial logical operators of $\mathcal{G}$ only. In what follows, with a slight abuse of notation, we refer to the products of type-\ref{surgery-logical-data} basis elements and stabilizers as logical operators of type~\ref{surgery-logical-data}. Likewise, we refer to the products of type-\ref{surgery-logical-ancilla} basis elements and stabilizers as logical operators of type~\ref{surgery-logical-ancilla}.}

A key property that we are interested in is given below: 

\begin{definition}[Dressed distance of the merged code]
     Let $\mathcal{Q}$ be a CSS code of distance $d$. Let $\mathcal{G}(\mathcal{V}, \mathcal{E}, \mathcal{C})$ be a (hyper)graph and a chain $\mathcal{G}: \mathbb{F}_2[\mathcal{V}] \xrightarrow{\partial_1} \mathbb{F}_2[\mathcal{E}] \xrightarrow{\partial_0} \mathbb{F}_2[\mathcal{C}]$, given above Definition~\ref{def: main-hypergraph-surgery} (resp. Definition~\ref{def: main-graph-surgery-gadget}). Denote $d(\mathcal{G})$ to be the \emph{$Z$-distance of the graph chain} $\mathcal{G}$: the minimal Hamming weight of a vector $w \in \ker \partial_0 \setminus \IM \partial_1$ (equivalently, of the $Z$-type Pauli operator $Z(w)$). Consider the merged code as a subsystem code by treating the logical operators of type~\ref{surgery-logical-ancilla} as gauge operators. Its \emph{dressed distance} $d(\mathcal{S})$ is the minimum Hamming weight of a physical representative of any logical operator of type~\ref{surgery-logical-data} that remains nontrivial modulo the gauge group, where representatives may be multiplied by arbitrary stabilizers and gauge operators; we set $d(\mathcal{S}) := \infty$ when no such logical operator exists. 
\end{definition}

\begin{definition}[Measurement protocols]\label{def: surgery-measurement-protocol}
   Let $\mathcal{Q}$ be a CSS code of distance $d$. Let $\mathcal{G}(\mathcal{V}, \mathcal{E}, \mathcal{C})$ be a (hyper)graph and a chain $\mathcal{G}: \mathbb{F}_2[\mathcal{V}] \xrightarrow{\partial_1} \mathbb{F}_2[\mathcal{E}] \xrightarrow{\partial_0} \mathbb{F}_2[\mathcal{C}]$, given above Definition~\ref{def: main-hypergraph-surgery} (resp. Definition~\ref{def: main-graph-surgery-gadget}). The measurement protocols of the surgery gadget $\mathcal{S}[\mathcal{G}; (\Phi_1, \Phi_0)]$ onto $\mathcal{Q}$ are given as follows. 
    \begin{enumerate}
        \item \label{surgery-protocol-initialization} \textbf{Initialization:} Prepare the ancilla qubits in $\Lket{+}^{\otimes |\mathcal{E}|}_{\mathcal{E}}$. 
        \item \label{surgery-protocol-merge} \textbf{Merge:} Measure the stabilizers of the merged code Eq.~\eqref{eq: main-merged-parity-check}, which include the added vertex checks and cycle checks from the second-half rows as well as deformed $Z$- and $X$-type checks. 
    \item \label{surgery-protocol-split-correct} \textbf{Split-Correction:} Transversally measure the ancilla qubits in $\mathcal{E}$ in the $X$ basis, followed by $Z$-correction and/or $X$-correction, depending on the split measurement outcomes. 
    \end{enumerate}
    Before Step~\ref{surgery-protocol-initialization}, we perform $d$ rounds of syndrome extraction on $\mathcal{Q}$. In Step~\ref{surgery-protocol-merge}, the stabilizers of the merged code are measured for $d$ consecutive rounds. After Step~\ref{surgery-protocol-split-correct}, we again perform $d$ rounds of syndrome extraction on $\mathcal{Q}$. 
\end{definition}

For every $z \in \ker \partial_1$, the ancillary factors cancel in the product of vertex checks indexed by $z$, and the corresponding measured operator is
    \begin{align}\label{eq: main-logicals-measured}
        L(z)
        &= \prod_{v:z_v=1}X(\Phi^{(X)}_1(v))Z(\Phi^{(Z)}_1(v)) \notag\\
        &= X(\Phi^{(X)}_1(z))Z(\Phi^{(Z)}_1(z))
        \in \Phi_1(\ker \partial_1).
    \end{align}

 In the phenomenological noise model, each round of syndrome extraction suffers independent single-qubit Pauli errors and measurement-outcome flips, each counted as one fault; the phenomenological fault distance of a protocol is the minimum number of faults that produces a logical error---a corrupted surviving logical operator or a flipped inferred measurement outcome---while triggering no detection event. A useful result we will state is the following:

\begin{theorem}[Phenomenological distance of surgery protocol, Ref.~\cite{williamson2024lowoverheadfaulttolerantquantumcomputation}]\label{thm: phenomenological-distance-surgery}
    The protocol given in Definition~\ref{def: surgery-measurement-protocol} has phenomenological fault distance at least $\min(d, d(\mathcal{S}))$, where $d(\mathcal{S})$ is the dressed distance of the merged code.
\end{theorem}

Note that the $d$-round merge window caps the fault distance at $d$ in general: flipping the outcome of a single vertex check in all $d$ merged rounds triggers no detection event yet flips the inferred measurement outcome. The bound remains $\min(d,d(\mathcal S))$. In particular, $d(\mathcal S)\geq d$ suffices to preserve the data-code distance.

 Hence, it is sufficient to bound the dressed distance of the merged (cone) code, which motivates the following definition:

\begin{definition}[Distance-preserving of a surgery gadget]\label{def: distance-preserving-surgery}
   Let $\mathcal{Q}$ be a (CSS) code with distance $d$, and $\mathcal{S}[\mathcal{G}; (\Phi_1, \Phi_0)]$ be a surgery gadget. Denote $d(\mathcal{S})$ to be the dressed distance. We say that the surgery gadget $\mathcal{S}[\mathcal{G}; (\Phi_1, \Phi_0)]$ is \emph{distance-preserving} if $d(\mathcal{S}) \geq d$.
\end{definition}

Now we discuss how to ensure a surgery gadget is distance-preserving. We first introduce a class of connectivity maps that are relevant in our study. 

\begin{definition}[Elementary connectivity maps for graph surgery]\label{def: elementary-connectivity-surgery}
     Let $\mathcal{G}$ be a graph, $\mathcal{Q}$ a data code, and $\mathcal{S}[\mathcal{G}; (\Phi_1, \Phi_0)]$ their graph-surgery gadget. We say that $\Phi = (\Phi_1, \Phi_0)$ is a connectivity map if it satisfies Definition~\ref{def: surgery-cone-code}. In addition, $\Phi$ is an $r$-elementary connectivity map (or, simply, an elementary connectivity map if $r=1$) if there exist a tuple of static maps $(S_1, S_0)$ and a partition of edges
     \begin{align}
        \mathcal{E} = \bigsqcup_{a} \mathcal{E}_a
     \end{align}
     such that
     \begin{align}
       \begin{aligned}
         &S_1 = \mathrm{vstack} \{\underbrace{I_{\mathcal{V}}, \cdots, I_{\mathcal{V}}}_{r\: \text{many}} \}, \\
        &S_0 = \mathrm{vstack} \{\underbrace{\oplus_a\mathbf{1}_{\mathcal{E}_a}, \cdots, \oplus_a\mathbf{1}_{\mathcal{E}_a}}_{r\: \text{many}} \},
       \end{aligned}
     \end{align}
     where $I_{\mathcal{V}}$ is an identity matrix and $\mathbf{1}_{\mathcal{E}_a}$ is the all-ones vector over edge partition $\mathcal{E}_a$, such that
     \begin{align}
        \begin{aligned}
            \Phi^{(Z)}_1 &= \mathrm{diag}\{M^{(Z)}_{1, 1}, \cdots, M^{(Z)}_{1, r} \} S_1 \\
            \Phi^{(Z)}_0 &= \mathrm{vstack}\{ \oplus_a\mathbf{1}_{\mathcal{E}_a} M^{(Z)}_{0, 1}, \cdots, \oplus_a\mathbf{1}_{\mathcal{E}_a} M^{(Z)}_{0, r} \},
        \end{aligned}
     \end{align}
     where $M^{(Z)}_{1, i}: \mathbb{F}_2[\mathcal{V}] \rightarrow \mathcal{Q}_1$ is an (embedding) matrix transversally assigning vertices to physical qubits and $M^{(Z)}_{0, i} = \oplus_a {\Pi^{\mathcal{E}_a}}$ with ${\Pi^{\mathcal{E}_a}}$ the diagonal projection matrix over edge partition $\mathcal{E}_a$. Similarly, this holds for the $X$-components.
\end{definition}

This definition of connectivity maps captures all essential graph-surgery gadgets constructed for our cases.

\begin{example}[Elementary connectivity maps to the seed surgery gadgets]
 Let $\mathcal{S}[\mathcal{G}; (\Phi_1, \Phi_0)]$ be a surgery gadget to some (CSS) code $\mathcal{Q}$. If we wish to measure a $Z$-type logical operator $\bar{L}$ with a physical representative $L$, then we can always identify the ancilla vertices $\mbb{F}_2[\mathcal{V}] \cong \mathbb{F}_2^{|\supp(L)|}$. In this case, we can simply take $\Phi_1 = (0, M_L)$ for the embedding/restriction map $M_L$ onto the support of $L$. If the graph is connected, that is, $\dim \ker \partial_1 =1$, then there exists a guaranteed choice of $\Phi_0 = (0, \Phi^{(Z)}_0)$ to ensure the commutativity conditions in Eq.~\eqref{eq: main-commutation_requirement}. Connectivity guarantees a solution $\Phi^{(Z)}_0$ to the commutation equations. To obtain an elementary map with column weight at most $1$, use the private-edge graphification of Lemma~\ref{lemma: graphify-parity-check}.

\end{example}

In the case of hypergraph surgery or parallel surgery, we say, simply, that the connectivity map $(\Phi_1, \Phi_0)$ is \emph{elementary} if both $\Phi_1$ and $\Phi_0$ are transversal, i.e., $\omega_{\mathrm{col}}(\Phi_1) = \omega_{\mathrm{col}}(\Phi_0) = 1$, and $\omega_{\mathrm{row}}(\Phi_1) = \omega_{\mathrm{row}}(\Phi_0) = 1$.

\begin{example}[Connectivity maps to hypergraph surgery Definition~\ref{def: main-LP-intracolumn-hyerpgraph-surgery} and Definition~\ref{def: main-LP-intercolumn-surgery}]
    We first give a simple example with $\Phi_1 = (0, \Phi^{(Z)}_1)$, associated with the LP inter-column surgery gadget $\mathcal{S}[\mathcal{G}; (\Phi_1, \Phi_0)]$ to a canonical LP code $\LP_l(A, B)$ (see Fig.~\ref{fig: parallel-surgery}(b)). We can draw the LP inter-column surgery gadget 
\begin{equation}\label{diagram: main-LP-intercolumn-surgery}
    \begin{tikzcd}
	&& {\mathcal{A}_1 \otimes_R \mathcal{B}_1} & \\
	& {\mathcal{A}_1 \otimes_R \mathcal{B}_0} && {\mathcal{A}_0 \otimes_R \mathcal{B}_1} \\
	{\mathcal{A}_1 } && {\mathcal{A}_0 \otimes_R \mathcal{B}_0} \\
	& {\mathcal{A}_0 }
	\arrow["{I_{n_A} \otimes_R B}"', from=1-3, to=2-2]
	\arrow["{A \otimes_R I_{n_B}}", from=1-3, to=2-4]
	\arrow["{A \otimes_R I_{m_B}}", from=2-2, to=3-3]
	\arrow["{I_{m_A} \otimes_R B}", from=2-4, to=3-3]
	\arrow["{\Phi_1}"', from=3-1, to=2-2]
	\arrow["{\partial_1=A}"{description}, from=3-1, to=4-2]
	\arrow["{\Phi_0}"', from=4-2, to=3-3]
\end{tikzcd}
\end{equation}
for trivial embedding maps $\Phi_1 = M_{\mathcal{A}_1}$ and $\Phi_0 = M_{\mathcal{A}_0}$. The hypergraph is given by the identification $\mathcal{V} = \mathcal{A}_1$, $\mathcal{E} = \mathcal{A}_0$, and $\mathcal{C}=0$. This way, the connectivity map is transversal, i.e., column and row weights equal to $1$. This transversality can be explained as an LP code modification: the merged code is the LP code $\LP_l(A, B')$ of the column-augmented $B' = \left(B \mid e^{(j)}_{\mathcal{B}_0}\right)$, whose extra ring dimension of $\mathcal{B}_1$ carries the ancillary chain ($\mathcal{V} = \mathcal{A}_1$ as the added $Z$-checks and $\mathcal{E} = \mathcal{A}_0$ as the added qubits). The modified boundary-map components read

\begin{align}
    \begin{aligned}
        I_{n_A} \otimes_R \left(B \mid e^{(j)}_{\mathcal{B}_0}\right)
     &= \left(I_{n_A} \otimes_R B \mid \Phi_1 \right); \\
        I_{m_A} \otimes_R \left(B \mid e^{(j)}_{\mathcal{B}_0}\right)
     &= \left(I_{m_A} \otimes_R B \mid \Phi_0 \right); \\
        A \otimes_R I_{n_B+1} &= \left(\begin{array}{cc}
             A \otimes_R I_{n_B} & 0 \\
             0 & \partial_1
        \end{array} \right).
\end{aligned}
\end{align}
This way, it connects the homomorphic (generalized Steane-type) measurements~\cite{xu2025fast} with the LP inter-column surgery.
\end{example}

A useful result is the following:

\begin{lemma}[Monotonicity of elementary connectivity maps]\label{lemma: monotocity-elementary-connvectivity}
    For any vector $x \in \mathbb{F}_2^n$ and an $r$-elementary connectivity map $\Phi_1$, we have $|\Phi^{(X)}_1(x)| \leq r|x|$ and $|\Phi^{(Z)}_1(x)| \leq r|x|$.
    \begin{proof}
        This follows from the fact that $\omega_{\mathrm{col}}(\Phi^{(X)}_1)$ (resp. $\omega_{\mathrm{col}}(\Phi^{(Z)}_1)$) is less than or equal to $r$ for any $r$-elementary connectivity map.
    \end{proof}
\end{lemma}

A sufficient criterion for bounding the code distance is given by the small-set soundness in Definition~\ref{def: main-small-set-soundness}: 

\begin{definition}[Small-set soundness]
    We say that $\partial_1: \mathbb{F}^{n}_2 \rightarrow \mathbb{F}^m_2$ is $(t, \rho_t)$-sound or has $(t, \rho_t)$-soundness relative to $\Phi_1$ if for all $u \in \mathbb{F}^{n}_2$ such that $|\partial_1 u| < t$,  
    \begin{align}
        |\partial_1 u| \geq \rho_t \, \mathrm{dist}_{\Phi_1}(u, \ker \partial_1),
    \end{align}
    where $\mathrm{dist}_{\Phi_1}(u, \ker \partial_1) = \min_{v \in \ker \partial_1} |\Phi_1(u + v)|$ for a linear map $\Phi_1$. We say simply that $\partial_1$ is $(t, \rho_t)$-sound or has $(t, \rho_t)$-soundness if $\Phi_1$ is an identity matrix. 
\end{definition}

We introduce the following concept:

\begin{definition}[Soundness of a surgery gadget]\label{def: soundness-surgery-gadget}
    Let $\mathcal{S}[\mathcal{G}; (\Phi_1, \Phi_0)]$ be a surgery gadget associated with some quantum code $\mathcal{Q}$. We say that the surgery gadget is $(t, \rho_t)$-sound or has $(t, \rho_t)$-soundness if $\partial_1: \mathbb{F}_2[\mathcal{V}] \rightarrow \mathbb{F}_2[\mathcal{E}]$ is $(t, r\rho_t)$-sound, where $r = \max(\omega_{\mathrm{col}}(\Phi^{(X)}_1), \omega_{\mathrm{col}}(\Phi^{(Z)}_1))$. Hence, $\partial_1$ is at least $(t, \rho_t)$-sound relative to $\Phi^{(X)}_1$ and to $\Phi^{(Z)}_1$, respectively.
\end{definition}

Hence, the distance-preserving property of a given surgery gadget can be ensured with sufficient soundness and $Z$-distance of the (hyper)graph chain $\mathcal{G}$, defined by the minimal Hamming weight of a vector supported on $\mathcal{E}$ that lies in $\ker \partial_0 \setminus \IM \partial_1$. 

\begin{lemma}[Distance-preserving via small-set soundness]\label{lemma: distance-preserving-soundness}
    Let $\mathcal{Q}$ be the data code with distance $d$, and let its surgery gadget $\mathcal{S}[\mathcal{G}; (\Phi_1, \Phi_0)]$ be $(t, \rho_t)$-sound. Furthermore, let the $Z$-distance of $\mathcal{G}$ be $d(\mathcal{G})$. If $t\geq d$ and $\rho_t > 0$, set $\rho_d := \rho_t$. Then the dressed distance $d(\mathcal{S})$ is at least $\min(d, \rho_d d, d(\mathcal{G}))$.

    \begin{proof}
        Let $y \in \mathbb{F}_2[\mathcal{V}]$ be a vector and let $L = (L^{\mathcal
        Q}_X L^{\mathcal{Q}}_Z, E_X)$ be a representative of some logical operator of type~\ref{surgery-logical-data}. Since at least one of $L^{\mathcal{Q}}_X$ or $L^{\mathcal{Q}}_Z$ is nontrivial, the only ways to reduce the weight below $d$ are multiplication by stabilizers of $\mathcal{Q}$ (absorbed into the representative), the added vertex checks, and the type~\ref{surgery-logical-ancilla} dressings treated below; the cycle checks act only on the ancilla qubits as $X$-type operators. In this case, the deformation is applied to $L$:
        \begin{align}
           (L^{\mathcal{Q}}_X L^{\mathcal{Q}}_Z, E_X) \mapsto (L^{\mathcal{Q}}_X X(\Phi^{(X)}_1(y))Z(\Phi^{(Z)}_1(y)) L^{\mathcal{Q}}_Z , E_XZ(\partial_1(y))).  
        \end{align}
        If $|Z(\partial_1(y))| \geq d$, then the distance is preserved. Hence, it remains to prove the case when $|Z(\partial_1(y))| < d$. Without loss of generality, we assume that $L^{\mathcal{Q}}_Z$ is nontrivial so that we can only focus on the deformation map given by $\Phi^{(Z)}_1$, and subsequently we simply write $\Phi^{(Z)}_1$ as $\Phi_1$. Working back into binary representation, we have 
        \begin{align}
            |\partial_1(y)| \geq r\rho_t \,\mathrm{dist}(y, \ker \partial_1) \geq \rho_t \,\mathrm{dist}_{\Phi_1}(y, \ker \partial_1),
        \end{align}

        where the second inequality follows from linearity and the monotonicity of elementary connectivity maps from Lemma~\ref{lemma: monotocity-elementary-connvectivity} with a slight notation abuse $\Phi_1 = \Phi^{(Z)}_1$ in the subscript. Since $|\partial_1(y)|<d\leq t$, the soundness bound applies with $\rho_d=\rho_t$. Choose $z\in\ker\partial_1$ minimizing $|y+z|$. First suppose that $L^{\mathcal{Q}}_Z$ and $Z(\Phi_1(z))$ lie in two different logical classes (including the case $\Phi_1(z) = 0$), so that $|l^{\mathcal{Q}}_Z + \Phi^{(Z)}_1(z)| \geq d$. Hence, we have
        \begin{align}
             \begin{aligned}
                &|(L^{\mathcal{Q}}_X X(\Phi^{(X)}_1(y))Z(\Phi^{(Z)}_1(y)) L^{\mathcal{Q}}_Z , E_XZ(\partial_1(y)))| \\
                & \geq |l^{\mathcal{Q}}_Z  + \Phi^{(Z)}_1(y)| + \rho_d|\Phi^{(Z)}_1(y) + \Phi^{(Z)}_1(z)|.  \\
             \end{aligned}
        \end{align}

        If $\rho_d > 1$, then we have that 
        \begin{align}
           \begin{aligned}
            &|l^{\mathcal{Q}}_Z  + \Phi^{(Z)}_1(y)| + \rho_d|\Phi^{(Z)}_1(y) + \Phi^{(Z)}_1(z)| \\
            \geq& |l^{\mathcal{Q}}_Z  + \Phi^{(Z)}_1(z)| + (\rho_d - 1)|\Phi^{(Z)}_1(y) + \Phi^{(Z)}_1(z)| \\
            \geq & d. 
           \end{aligned}
        \end{align}
        If $\rho_d \leq 1$, then we reverse the argument: 
        \begin{align}
            \begin{aligned}
                &|l^{\mathcal{Q}}_Z  + \Phi^{(Z)}_1(y)| + \rho_d|\Phi^{(Z)}_1(y) + \Phi^{(Z)}_1(z)| \\
            \geq& (1-\rho_d)|l^{\mathcal{Q}}_Z  + \Phi^{(Z)}_1(y)| + \rho_d|l^{\mathcal{Q}}_Z  + \Phi^{(Z)}_1(z)| \\
            \geq & \rho_d d,
            \end{aligned}
        \end{align}
        where, in both cases, we use the fact that $|l^{\mathcal{Q}}_Z  + \Phi^{(Z)}_1(z)| \geq d$, established above. Next, suppose that $L^{\mathcal{Q}}_Z$ and $Z(\Phi_1(z))$ represent the same logical class. Then $L^{\mathcal{Q}}_X$ and $X(\Phi^{(X)}_1(z))$ must represent different classes: otherwise deforming by the vertex checks $z$ would render the operator trivial, contradicting type~\ref{surgery-logical-data}. In this case, apply the same argument to the $X$-type component using the fixed $z$. The common estimate above gives $|\partial_1(y)|\geq\rho_d|\Phi_1^{(X)}(y+z)|$. Separate soundness relative to $\Phi_1^{(X)}$ and $\Phi_1^{(Z)}$ is insufficient for this argument: their minima may require different kernel vectors.

 To determine the dressed distance, we need to apply the deformation of logical operators of type~\ref{surgery-logical-ancilla}. First, an $X$-type logical operator of type~\ref{surgery-logical-ancilla} admits a representative 
        \begin{align}
            (0, L^{\mathcal{G}}_X) = (0, X(x))
        \end{align}
        for $x \in \ker \partial^T_1 \setminus \IM \partial^T_0$. Note that applying this alone does not decrease the dressed distance in the above analysis, nor does dressing by any $X$-type stabilizers, since it would only induce a stabilizer deformation of the original data code $\mathcal{Q}$ onto a nontrivial logical operator of $\mathcal{Q}$, which preserves the distance up to $d$. Any $Z$-type logical operator of type~\ref{surgery-logical-ancilla} admits a representative: 
        \begin{align}
            (P^{\mathcal{Q}}_XP^{\mathcal{Q}}_Z, L^{\mathcal{G}}_Z E_X)
        \end{align}
        for $L^{\mathcal{G}} = (0, Z(w))$ a nontrivial $Z$-type logical operator of $\mathcal{G}$: $w \in \ker \partial_0 \setminus \IM \partial_1$. Furthermore, deforming by any vertex checks $y$ leaves the ancillary $Z$-support in the same coset: $|w + \partial_1 y| \geq d(\mathcal{G})$, since $w + \partial_1 y \in \ker \partial_0 \setminus \IM \partial_1$. For a product of several type~\ref{surgery-logical-ancilla} factors, $w$ is replaced by the sum of the corresponding vectors, which either remains in $\ker \partial_0 \setminus \IM \partial_1$ (the same bound applies) or lies in $\IM \partial_1${. In the latter case, every type-\ref{surgery-logical-ancilla} basis element appears an even number of times. The product is then a stabilizer, and the preceding analysis applies}. This implies the final lower bound:
        \begin{align}
            d(\mathcal{S}) \geq \min(d, \rho_d d, d(\mathcal{G})).
        \end{align}

    \end{proof}
\end{lemma}

\begin{remark}
    As stated, the $X$-distance of $\mathcal{G}$ is irrelevant to ensuring that the surgery gadget is fault-tolerant, which underpins the thickening protocol. 
\end{remark}

\begin{example}
    We give a concrete example where a small $d(\mathcal{G})$ could seriously affect the distance-preserving property of the surgery gadget. Suppose
    \begin{align}
        \partial_1 = \left( \begin{array}{cc}
            1 &  1\\
           1 & 1 \\
            1 & 1 \\
            1 & 1
        \end{array} \right); \quad \partial_0 = 0.
    \end{align}
    The map $\partial_1$ is $(5, 4)$-sound; yet $d(\mathcal{G})=1$ given by the binary support vector $\begin{pmatrix}
        1 & 0 &0 & 0
    \end{pmatrix}^T$. Subsequently, the low-weight $Z$-type logical operator of $\mathcal{G}$ would reduce the dressed distance $d(\mathcal{S})$ to $2$ ($1$ from the data sector), since the vertex stabilizers can never completely remove the support of an unmeasured data logical operator. 
\end{example}

\subsection{Code symmetries for surgery}\label{sec: code-surgery-symmetry}

The above Section~\ref{sec: surgery-formulation} establishes the basic principles of surgery techniques. Here we give an algebraic view of surgery gadgets, especially of their connectivity maps. 
A key characterization is that there are potentially many possible chain maps which satisfy Definition~\ref{def: surgery-cone-code} above. Hence, a surgery ancilla graph may measure many logical operators by rewiring its connectivity map. In what follows, we restrict ourselves to measuring $Z$-type logical operators, where $\Phi_1 = (0, \Phi^{(Z)}_1)$ and $\Phi_0 = (0, \Phi^{(Z)}_0)$; we then simply write $\Phi_1 = \Phi^{(Z)}_1$ and $\Phi_0 = \Phi^{(Z)}_0$, and default to the binary notation. The exception is the mixed-type setting of Remark~\ref{remark: mixed-type-surgery} and the second example below, where the full symplectic pairs reappear. 

\begin{lemma}\label{lemma: classical-chain-maps}
    Let $\mathcal{S}[\mathcal{G}; (\Phi_1, \Phi_0)]$ be a surgery gadget to a (CSS) code $\mathcal{Q}$ with distance $d$. Suppose that the surgery gadget is $(d, \rho_d)$-sound. Then the following hold.
    \begin{enumerate}
        \item \label{local-chain-map-XX} (Chain map for same-type logical operators). Let $\bar{L}$ and $\bar{L}'$ be two $Z$-type (resp. $X$-type) logical operators, and suppose that the gadget measures $\bar{L}$, i.e., $l \in \Phi_1(\ker \partial_1)$ for a physical, binary representative $l$ of $\bar{L}$. Suppose there exists a chain map $\Gamma = (\Gamma_1, \Gamma_0)$:
\begin{equation}\label{diagram: classical-chain-map-same-type}
   \begin{tikzcd}
	{\mathcal{Q}_{1}} & {\mathcal{Q}_0} \\
	{\mathcal{Q}_{1}} & {\mathcal{Q}_0}
	\arrow["{H_X}", from=1-1, to=1-2]
	\arrow["{\Gamma_1}"', from=1-1, to=2-1]
	\arrow["{\Gamma_0}", from=1-2, to=2-2]
	\arrow["{H_X}"', from=2-1, to=2-2]
\end{tikzcd}
\end{equation}
such that $\omega_{\mathrm{col}}(\Gamma_1) \leq r$, $\Gamma_1(l) = l'$ for a physical, binary representative $l'$ of $\bar{L}'$, and $\Gamma_1$ carries every nonzero element of $\Phi_1(\ker \partial_1)$ to a nontrivial logical operator (automatic when $\dim \ker \partial_1 = 1$). Then the transformed surgery gadget $\mathcal{S}[\mathcal{G};  (\Gamma_1 \Phi_1 , \Gamma_0 \Phi_0)]$ measures $\bar{L}'$, and is at least $(d, \rho_d/r)$-sound. For $X$-type logical operators, the same holds with $H_Z$ in place of $H_X$ in Diagram~\eqref{diagram: classical-chain-map-same-type}.

        \item \label{local-chain-map-XZ} (Chain map for mixed-type logical operators). Let, without loss of generality, $\bar{L}_Z$ be a $Z$-type logical operator and $\bar{L}_X$ be an $X$-type logical operator. Suppose that there exists a chain map $\Psi = (\Psi_1, \Psi_0)$: 
\begin{equation}\label{diagram: ZX-duality}
   \begin{tikzcd}
	{\mathcal{Q}_1} & {\mathcal{Q}_0} \\
	{\mathcal{Q}_1} & {\mathcal{Q}_2}
	\arrow["{H_X}", from=1-1, to=1-2]
	\arrow["{\Psi_1}"', from=1-1, to=2-1]
	\arrow["{\Psi_0}", from=1-2, to=2-2]
	\arrow["{H_Z}"', from=2-1, to=2-2]
\end{tikzcd}
\end{equation}
such that $\omega_{\mathrm{col}}(\Psi_1) \leq r$ and $\Psi_1(l_Z) = l_X$ with $l_Z$ and $l_X$ being physical, binary representatives of $\bar{L}_Z$ and $\bar{L}_X$, respectively, where the gadget measures $\bar{L}_Z$, i.e., $l_Z \in \Phi_1(\ker \partial_1)$, and $\Psi_1$ carries every nonzero element of $\Phi_1(\ker \partial_1)$ to a nontrivial logical operator.
Then the mapped surgery gadget $\mathcal{S}[\mathcal{G};  (\Psi_1 \Phi_1 , \Psi_0 \Phi_0)]$ measures $\bar{L}_X$ and is at least $(d, \rho_d/r)$-sound.
    \end{enumerate}
    \begin{proof}
        We can prove these through the chain diagrams respectively. For the first assertion, we have that 
        \begin{equation}\label{diagram: classical-code-transformation-surgery-Z-type}
          \begin{tikzcd}
	{\mathcal{G}} & {\mathbb{F}_2[\mathcal{V}]} & {\mathbb{F}_2[\mathcal{E}]} & {\mathbb{F}_2[\mathcal{C}]} \\
	{\Gamma } && {\mathcal{Q}_1} & {\mathcal{Q}_0} \\
	{\mathcal{Q}} & {\mathcal{Q}_2} & {\mathcal{Q}_1} & {\mathcal{Q}_0}
	\arrow["{\partial_1}", from=1-2, to=1-3]
	\arrow["{\Phi_1}", from=1-2, to=2-3]
	\arrow["{\partial_0}", from=1-3, to=1-4]
	\arrow["{\Phi_0}", from=1-3, to=2-4]
	\arrow["{H_X}", from=2-3, to=2-4]
	\arrow["{\Gamma_1}", from=2-3, to=3-3]
	\arrow["{\Gamma_0}", from=2-4, to=3-4]
	\arrow["{H^T_Z}", from=3-2, to=3-3]
	\arrow["{H_X}", from=3-3, to=3-4]
\end{tikzcd}. 
        \end{equation}
        By assumption, Diagram~\ref{diagram: classical-code-transformation-surgery-Z-type} commutes, which defines a proper (same-type) surgery gadget according to Definition~\ref{def: surgery-cone-code} (resp. Definition~\ref{def: hypergraph-surgery}). By assumption $l \in \Phi_1(\ker \partial_1)$ and $\Gamma_1(l) = l'$, so that $l' \in \Gamma_1 \Phi_1(\ker \partial_1)$; moreover, the transformed measured operators $\Gamma_1\Phi_1(\ker \partial_1)$ remain nontrivial logical operators by hypothesis, consistent with the standing assumption. Furthermore, we have $|\Gamma_1(x)| \leq r|x|$ for any $x$ since $\omega_{\mathrm{col}}(\Gamma_1) \leq r$; hence $\mathrm{dist}_{\Gamma_1\Phi_1}(u, \ker \partial_1) \leq r\,\mathrm{dist}_{\Phi_1}(u, \ker \partial_1)$ for all $u$, and the transformed gadget is at least $(d, \rho_d/r)$-sound. For the second assertion, the proof is analogous, with the commutative diagram: 
        \begin{equation}\label{diagram: ZX-duality-surgery}
            \begin{tikzcd}
	{\mathcal{G}} & {\mathbb{F}_2[\mathcal{V}]} & {\mathbb{F}_2[\mathcal{E}]} & {\mathbb{F}_2[\mathcal{C}]} \\
	\Psi && {\mathcal{Q}_1} & {\mathcal{Q}_0} \\
	{\mathcal{Q}} & {\mathcal{Q}_0} & {\mathcal{Q}_1} & {\mathcal{Q}_2}
	\arrow["{\partial_1}", from=1-2, to=1-3]
	\arrow["{\Phi_1}", from=1-2, to=2-3]
	\arrow["{\partial_0}", from=1-3, to=1-4]
	\arrow["{\Phi_0}", from=1-3, to=2-4]
	\arrow["{H_X}", from=2-3, to=2-4]
	\arrow["{\Psi_1}", from=2-3, to=3-3]
	\arrow["{\Psi_0}", from=2-4, to=3-4]
	\arrow["{H^T_X}", from=3-2, to=3-3]
	\arrow["{H_Z}", from=3-3, to=3-4]
\end{tikzcd}.
        \end{equation}
    \end{proof}

\end{lemma}

\begin{remark}
For full generality, we can define the chain maps restricted onto the supports of $\Phi_1$. For example, let $\Phi_1$ be an elementary connectivity map which measures a single $Z$-type logical operator $\bar{L}$ so that $\Phi_1 = M_L$, where $M_L$ is the restriction map onto the support of a physical representative $L$ of $\bar{L}$. Then it suffices to define $\Gamma$ on this support, i.e., as a chain map from the restricted complex $H_XM_L$ to $H_XM_{L'}$ for another $Z$-type logical operator $\bar{L}'$ whose physical representative is $L'$, together with $\Gamma_1(\mathbf{1}_{\supp L}) = \mathbf{1}_{\supp L'}$ and $\omega_{\mathrm{col}}(\Gamma_1) \leq r$, under which Lemma~\ref{lemma: classical-chain-maps} applies with the same constant.
\end{remark}

The symmetry characterization shows that a fixed surgery gadget can measure many different logical operators by combining the two symmetry actions fault-tolerantly. 
\begin{remark}\label{remark: mixed-type-surgery}
Let $\mathcal{S}[\mathcal{G}; (\Phi^{(Z)}_1, \Phi^{(Z)}_0)]$ and $\mathcal{S}[\mathcal{G}; (\Phi^{(X)}_1, \Phi^{(X)}_0)]$ be surgery gadgets to $\mathcal{Q}$ measuring $\bar{L}_Z$ and $\bar{L}_X$, respectively, such that $\Phi^{(X)T}_1 \Phi^{(Z)}_1 = \Phi^{(Z)T}_1 \Phi^{(X)}_1$. Then the pair $((\Phi^{(X)}_1, \Phi^{(Z)}_1), (\Phi^{(X)}_0, \Phi^{(Z)}_0))$ satisfies Definition~\ref{def: surgery-cone-code} (resp. Definition~\ref{def: hypergraph-surgery}), and the resulting mixed-type surgery gadget measures the joint operators $X(\Phi^{(X)}_1(z))Z(\Phi^{(Z)}_1(z))$ for $z \in \ker \partial_1$; in particular, when the representatives of $\bar{L}_X$ and $\bar{L}_Z$ arise from a common $z$ (as for a connected graph, where $\dim \ker \partial_1 = 1$), it measures $\bar{L}_X\bar{L}_Z$. By Definition~\ref{def: soundness-surgery-gadget}, its soundness is inherited componentwise: it is $(d, \rho_d)$-sound whenever both single-type gadgets are. Its dressed distance is then bounded by Lemma~\ref{lemma: distance-preserving-soundness}.
\end{remark}

We give a few examples relevant to our study. 
\begin{example}
    Let $\mathcal{S}[\mathcal{G}; (\Phi_1, \Phi_0)]$ be a surgery gadget, and $\Gamma$ be the chain map according to Diagram~\ref{diagram: classical-code-transformation-surgery-Z-type}, and $\Phi_1 = \Phi^{(Z)}_1$ and $\Phi_0 = \Phi^{(Z)}_0$. Then the linear combination of connectivity maps $(\Gamma_1 + I) \Phi_1$ and $(\Gamma_0+I) \Phi_0$ defines a new surgery gadget measuring the joint logical operator $\bar{L} \bar{L}'$, in the setting of item~\ref{local-chain-map-XX} of Lemma~\ref{lemma: classical-chain-maps}.
\end{example}

\begin{example}
    Let $\mathcal{S}[\mathcal{G}; (\Phi_1, \Phi_0)]$ be a surgery gadget, and let $\Psi$ be the chain map according to Diagram~\ref{diagram: ZX-duality-surgery}, and $\Phi_1 = \Phi^{(Z)}_1$ and $\Phi_0 = \Phi^{(Z)}_0$. Then the pair of connectivity maps $(\Psi_1\Phi^{(Z)}_1, \Phi^{(Z)}_1)$ and $(\Psi_0\Phi^{(Z)}_0, \Phi^{(Z)}_0)$ defines a new surgery gadget measuring the joint logical operator $\bar{L}_X \bar{L}_Z$ by Remark~\ref{remark: mixed-type-surgery}, if $\Psi_1$ is symmetric, which ensures that the commutation relation $\Phi^{(X)T}_1 \Phi^{(Z)}_1 = \Phi^{(Z)T}_1 \Psi^T_1 \Phi^{(Z)}_1 = \Phi^{(Z)T}_1 \Psi_1 \Phi^{(Z)}_1$ is symmetric.
\end{example}

We also give a stronger notion of $ZX$-duality useful for discussions in the following sections.
\begin{definition}[Canonical $ZX$-duality]\label{def: canonical-ZX-duality}
    Let $\mathcal{Q}$ be a CSS code with a $Z$-type logical basis $\mathfrak{S}^{(Z)}$ and an $X$-type logical basis $\mathfrak{S}^{(X)}$. We say that a chain map $(\Psi_1, \Psi_0)$ induces a canonical $ZX$-duality for a (CSS) code $\mathcal{Q}$ if it induces the commutative Diagram~\ref{diagram: ZX-duality}, and in addition, the following hold.
   \begin{itemize}
    \item $\Psi_1, \Psi_0$ are representations of involutions on the physical qubits and checks, respectively, according to Diagram~\ref{diagram: ZX-duality}; the latter requires $\dim \mathcal{Q}_0 = \dim \mathcal{Q}_2$.
    \item Let the action of $\Psi_1$ be represented by an involution on the physical qubits. We require that it maps the logical supports as
        \begin{align}
            \supp \bar{Z}_{\mu} \mapsto \supp \bar{X}_{\tau(\mu)}, \qquad
            \supp \bar{X}_{\mu} \mapsto \supp \bar{Z}_{\tau(\mu)},
        \end{align}
        for every basis index $\mu$, where $\tau$ is an involution on the basis indices.
    \item We have that
    \begin{align}
        \supp \bar{Z}_\mu \cap \supp \bar{X}_{\nu} = \emptyset
    \end{align}
    for $\mu \neq \nu$, while $\supp \bar{Z}_\mu \cap \supp \bar{X}_{\mu}$ consists of a single qubit.
   \end{itemize}
\end{definition}

Evidently, the symmetric canonical LP code $\LP_l(A, A^*)$ satisfies Definition~\ref{def: canonical-ZX-duality}, as do symmetric hypergraph-product (HGP) codes. In the next section, we assume that our codes satisfy the above definition, which leads to a space-efficient surgery design.

\subsection{Generalized bridges for graph-surgery gadgets}\label{sec: bridging-techniques-graph-surgery}
In this section, we introduce a unified, low-overhead, graph-theoretical framework for bridging surgery gadgets to measure an arbitrary logical operator. 

\subsubsection*{Basic graph theories}
In this section, we only consider the \emph{graph surgery} according to Definition~\ref{def: surgery-cone-code}, which is a special case of the hypergraph surgery according to Definition~\ref{def: hypergraph-surgery}. The use of a connected graph chain ensures the existence of connectivity maps, which are otherwise hard to construct. We make these precise as follows. 

\begin{definition}[Graphs, directed graphs, multi-edge graphs; Ref.~\cite{golowich2024quantumldpccodestransversal}]\label{def: formal-graph-definition}
    A (multi)graph $\mathcal{G}(\mathcal{V}, \mathcal{E}, \partial)$ consists of  a set of vertices $\mathcal{V}$, a set of edges $\mathcal{E}$ with two projections $\mathrm{v}_0, \mathrm{v}_1: \mathcal{E} \rightarrow \mathcal{V}$ with a tuple 
	$$
	 (\mathrm{v}_0(e), \mathrm{v}_1(e)).  
	$$
	We say that a graph is a directed (multi)graph if $\mathcal{E}$ is associated with $\mathrm{v}^{\mathrm{dir}}: \mathcal{E} \rightarrow \mathcal{V} \times \mathcal{V}$ given by $\mathrm{v}^{\mathrm{dir}}(e) = (\mathrm{v}_0(e), \mathrm{v}_1(e))$. We say a graph is an undirected (multi)graph if $\mathcal{E}$ is associated with $\mathrm{v}: \mathcal{E} \rightarrow \mathcal{V} \times \mathcal{V} / \sim$ given by $\mathrm{v}(e) = \{\mathrm{v}_0(e), \mathrm{v}_1(e)\}$ (so $\partial = \mathrm{v}$), where the equivalence relation denotes unordered orientation: $(\mathrm{v}_0, \mathrm{v}_1) \sim (\mathrm{v}_1, \mathrm{v}_0)$. Furthermore, we say that two edges are multiple to each other for an undirected graph if $\mathrm{v}(e) = \mathrm{v}(e')$, denoted by the equivalence relation $e \sim e'$. Finally, we say that an edge $e$ is a loop (self-loop) at $u \in \mathcal{V}$ if $\mathrm{v}_0(e) = \mathrm{v}_1(e) = u$: for a directed graph $\mathrm{v}^{\mathrm{dir}}(e) = (u, u)$, while for an undirected graph $\mathrm{v}(e) = \{u, u\} = \{u\}$ is a singleton, i.e. a fixed point of the flip $(\mathrm{v}_0, \mathrm{v}_1) \sim (\mathrm{v}_1, \mathrm{v}_0)$.
\end{definition}

\begin{definition}[Paths and cycles]\label{def: path-cycle}
    Let $\mathcal G(\mathcal V,\mathcal E, \mathrm{v})$ be a graph. A \emph{path} is a sequence: 
    \begin{align}
        (v_0, e_1, v_1, e_2, v_2, \ldots, e_r, v_r)
    \end{align}
    for a sequence of edges $\{e_1, \cdots, e_r\}$ such that $\mathrm{v}(e_i) = \{v_{i-1}, v_{i}\}$ for $i = 1, \ldots, r$, with $v_0, \ldots, v_r$ pairwise distinct for a \emph{path}; a \emph{cycle} is such a sequence with $v_0, \ldots, v_{r-1}$ pairwise distinct and $v_r = v_0$.
\end{definition}

In this section, we only consider undirected graphs \emph{without self-loops}, which simplifies the discussion by working in the binary field. All graphs and graph chains below, including those obtained by bridging, are assumed to satisfy this convention. Let $\mathcal{G}: \mathbb{F}_2[\mathcal{V}] \xrightarrow{\partial_1} \mathbb{F}_2[\mathcal{E}] \xrightarrow{\partial_0} \mathbb{F}_2[\mathcal{C}]$ be the graph chain associated with a graph $\mathcal{G}(\mathcal{V},\mathcal{E}, \mathcal{C})$. For any subset $L \subseteq \mathcal{V}$, let $x_L \in \mathbb{F}_2[\mathcal{V}]$ be its indicator vector, namely $x_L(v)=1$ if and only if $v \in L$. Throughout, the vertex subset $L$ is identified with the support of the measured logical operator, matching the restriction map $M_L$. If $|L|$ is even, we write $P(L) \subseteq \mathcal{E}$ for an edge set with boundary $L$, that is, its indicator vector $y_L \in \mathbb{F}_2[\mathcal{E}]$ satisfies
\begin{align}
    (\partial_1)^T y_L = x_L.
\end{align}

For a connected graph, such an edge set exists for every even-cardinality subset $L \subseteq \mathcal{V}$. A key question is to ensure the existence of sparse $\Phi_0$. 

\begin{lemma}\label{lemma: sparse-Phi0-existence}
Let $\mathcal{G}: \mathbb{F}_2[\mathcal{V}] \xrightarrow{\partial_1} \mathbb{F}_2[\mathcal{E}] \xrightarrow{\partial_0} \mathbb{F}_2[\mathcal{C}]$ be a \emph{connected} graph chain with graph vertex degree at most $\Delta$, $\Delta \geq 2$. Then the following statements hold.
    \begin{itemize}
        \item Let $H : \mathbb{F}_2[\mathcal{V}] \to \mathbb{F}_2^{m_H}$ be a linear map, represented by a matrix $H \in \mathbb{F}_2^{m_H \times |\mathcal{V}|}$. There exists a linear map $\Phi_0 : \mathbb{F}_2[\mathcal{E}] \to \mathbb{F}_2^{m_H}$ such that
        \begin{align}\label{eq: row-linear-comb-all-even-weight}
            H = \Phi_0 \circ \partial_1
        \end{align}
        if and only if every row of $H$ has even Hamming weight.

        \item There exists an even-weight row vector $H \in \mathbb{F}_2^{1 \times |\mathcal{V}|}$ such that every $\Phi_0$ satisfying
        \begin{align}
            H = \Phi_0 \circ \partial_1
        \end{align}
        has row Hamming weight at least $\Omega(\log_{\Delta} |\mathcal{V}|)$.
    \end{itemize}
    \begin{proof}
        Since $\mathcal{G}$ is connected, the kernel of $\partial_1$ is generated by the all-ones vector on $\mathcal{V}$. Hence
        \begin{align}
            \dim \rs \partial_1 = |\mathcal{V}|-1.
        \end{align}
        Moreover, every row of $\partial_1$ has Hamming weight $2$, so every vector in $\rs \partial_1$ has even Hamming weight. Conversely, the subspace of even-weight vectors in $\mathbb{F}_2[\mathcal{V}]$ also has dimension $|\mathcal{V}|-1$. Therefore $\rs \partial_1$ is precisely the even-weight subspace of $\mathbb{F}_2[\mathcal{V}]$. This proves the first statement.

        For the second statement, choose two vertices $u,v \in \mathcal{V}$ at graph distance $\mathrm{dist}_{\mathcal{G}}(u,v)=\operatorname{diam}(\mathcal{G})$, and let $H = (x_u + x_v)^T$, writing $x_w := x_{\{w\}}$ for the singleton indicator. If $H = \Phi_0 \circ \partial_1$, then $\Phi_0$ is a row vector indexed by edges, and its support defines an edge set $P(\{u,v\})$ satisfying
        \begin{align}
            (\partial_1)^T \Phi_0^T = x_u + x_v.
        \end{align}
        In particular, the support of $\Phi_0$ contains a $u$--$v$ path, so its row Hamming weight is at least $\mathrm{dist}_{\mathcal{G}}(u,v)=\operatorname{diam}(\mathcal{G})$. For bounded-degree graphs, the Moore bound~\cite{HoffmanSingleton1960} gives $\operatorname{diam}(\mathcal{G}) = \Omega(\log_{\Delta}|\mathcal{V}|)$, and the claim follows.
    \end{proof}
\end{lemma}

There is no general statement guaranteeing the sparsity of $\Phi_0$. We now state a useful result, motivated by the extractor construction of Ref.~\cite{he2025extractorsqldpcarchitecturesefficient}: when $H$ (the relevant part of the parity-check matrix) is LDPC, a sparse $P_L$ can be constructed explicitly by attaching one cycle per row of $H$.

\begin{definition}[Graphification of a parity-check matrix]\label{def: graphify-parity-check}
     Let $H \in \mathbb{F}^{m_H \times n_H}_2$ be a matrix. For the $a$-th row $H_a$ of $H$, write $r_a := |\supp H_a|$ and let $v_{a_0}, \cdots, v_{a_{r_a-1}}$ be vertices corresponding one-to-one to its nonzero column indices, with $v_{a_i}$ corresponding to the $i$-th such index. A \emph{cycle-graphification} of $H$ is given by a graph $\mathcal{G}(\mathcal{V}, \mathcal{E}, \mathcal{C})$ in which every row of $H$ is represented as a cycle:
     \begin{align}\label{expr: cycle-graphification-parity-check}
        (v_{a_0}, e_{a_1}, v_{a_1}, e_{a_2}, v_{a_2}, \ldots, e_{a_{r_a-1}}, v_{a_{r_a-1}}, e_{a_{r_a}})
    \end{align}
    with the $r_a$ edges $e_{a_i} = \{v_{a_{i-1}}, v_{a_i}\}$ for $i = 1, \cdots, r_a - 1$ and the closing edge $e_{a_{r_a}} = \{v_{a_{r_a-1}}, v_{a_0}\}$. Similarly, a \emph{path-graphification} of $H$ is given by a graph $\mathcal{G}(\mathcal{V}, \mathcal{E}, \mathcal{C})$ such that every row of $H$ can be represented as a path in $\mathcal{G}$:
    \begin{align}\label{expr: path-graphification-parity-check}
        (v_{a_0}, e_{a_1}, v_{a_1}, e_{a_2}, v_{a_2}, \ldots, e_{a_{r_a-1}}, v_{a_{r_a-1}})
    \end{align}
    that is, only the first $r_a - 1$ edges are attached and the closing edge is omitted. A row of weight $1$ contributes no edge, and a row of weight $2$ contributes a doubled edge under cycle-graphification, which we allow as a multigraph (Definition~\ref{def: formal-graph-definition}). Both ways result in a graph chain $\mathcal{G}: \mathbb{F}_2[\mathcal{V}] \xrightarrow{\partial_1} \mathbb{F}_2[\mathcal{E}] \xrightarrow{\partial_0} \mathbb{F}_2[\mathcal{C}]$ such that $\mathbb{F}_2[\mathcal{V}] \cong \mathbb{F}^{n_H}_2$, with $\mathbb{F}_2[\mathcal{E}] \cong \mathbb{F}^{\sum_{a:\, r_a \geq 2} r_a}_2$ for cycle-graphification (rows of weight $1$ contribute no edge) and $\mathbb{F}_2[\mathcal{E}] \cong \mathbb{F}^{\sum^{m_H}_{a=1} (r_a-1)}_2$ for path-graphification.
\end{definition}

The graphification of a parity-check matrix is useful in constructing sparse connectivity maps. 

\begin{lemma}\label{lemma: graphify-parity-check}
   Let $H \in \mathbb{F}^{m_H \times n_H}_2$ and let $\mathcal{G}(\mathcal{V}, \mathcal{E}, \mathrm{v})$
   be a graphification of $H$. Then for every restriction map $M^{(L)}_1$ such that every row of $H M^{(L)}_1$ has even Hamming weight, there exists $M^{(L)}_0$ such that
   \begin{align}
    H M^{(L)}_1 =  M^{(L)}_0 \partial_1.
   \end{align}
    Moreover, $\omega_{\mathrm{row}}(M^{(L)}_0)\leq \lfloor \omega_{\mathrm{row}}(H)/2 \rfloor$ for cycle-graphification, and $\omega_{\mathrm{row}}(M^{(L)}_0) \leq \omega_{\mathrm{row}}(H)-1$ for path-graphification, and $\omega_{\mathrm{col}}(M^{(L)}_0) \leq 1$ in both cases. Furthermore, the graph $\mathcal{G}$ has maximum vertex degree at most $2 \omega_{\mathrm{col}}(H)$.
   \begin{proof}
        To see this, fix a row index $a$ and write $r := |\supp H_a|$. The vertices and edges attached to this row are those of expression~\eqref{expr: cycle-graphification-parity-check}, namely the cycle $(v_{a_0}, e_{a_1}, v_{a_1}, \cdots, e_{a_r})$ on the $r$ vertices carrying the support of $H_a$, with the closing edge $e_{a_r} = \{v_{a_{r-1}}, v_{a_0}\}$ (for path-graphification, expression~\eqref{expr: path-graphification-parity-check}, the closing edge is omitted). Any $r-1$ of the vectors
        \begin{align}
             (\partial_1)^T y_{e_{a_1}}, \cdots, (\partial_1)^T y_{e_{a_{r-1}}}
        \end{align}
         are linearly independent and form a basis of the $(r-1)$-dimensional space of even Hamming-weight vectors supported on the cycle's vertices, as
         \begin{align}
            \sum_{j \geq 0}\binom{r}{2j} = 2^{r-1} = \left| \mathrm{span}\{  y_{e_{a_1}}, \cdots, y_{e_{a_{r-1}}}\} \right|. 
         \end{align}
        This implies that for every restriction $M^L_1$ for which $H_a M^L_1$ has even Hamming weight, there exists a projection $M^L_0$ such that, writing $y^{(L)}_{a} :=  (M^L_0)_{a}$ for its $a$-th row,
        \begin{align}
            (\partial_1)^T (y^{(L)}_{a})^T = (H_a M^{(L)}_1)^T,
        \end{align}

        by the dimension count above; that is, $H_a M^{(L)}_1 = y^{(L)}_{a} \partial_1$ row by row, which assembled over $a$ is the claimed identity.
        Note that for path-graphification~\eqref{expr: path-graphification-parity-check}, the row weight of $y^{(L)}_a$ is at most $\omega_{\mathrm{row}}(H)-1$, and for cycle-graphification~\eqref{expr: cycle-graphification-parity-check}, the row weight of $y^{(L)}_a$ is at most $\lfloor \omega_{\mathrm{row}}(H)/2 \rfloor$.
        For the column weight, the edges $e_{a_1}, \cdots, e_{a_r}$ are created by row $a$ alone, so distinct rows of $M^{(L)}_0$ have disjoint supports and $\omega_{\mathrm{col}}(M^{(L)}_0) \leq 1$. Furthermore, each vertex can only be touched by at most $2 \omega_{\mathrm{col}}(H)$ many edges (one cycle per row containing it, two edges each), so that the resultant graph $\mathcal{G}$ has vertex degree $\Delta \leq 2\omega_{\mathrm{col}}(H)$.
   \end{proof}
\end{lemma}

Lemma~\ref{lemma: graphify-parity-check} shows how the elementary connectivity map can be constructed from a given parity-check matrix $H$ of a classical code. In this section, we apply the formalism in building the various surgery gadgets discussed in Section~\ref{sec: main-graph-surgery}. Before we proceed, we introduce a final ingredient, a property of graphs. 

\begin{definition}[Graph automorphism]\label{def: graph-isomorphism}
    Let $\mathcal{G}(\mathcal{V}, \mathcal{E})$ be a graph. A graph \emph{automorphism} of $\mathcal{G}$ is a permutation $\sigma$ of $\mathcal{V}$ such that $\{u, v\} = \mathrm{v}(e) \in \mathcal{E}$ if and only if $\{\sigma(u), \sigma(v)\} \in \mathcal{E}$. It induces a permutation of edges, $\sigma(\{u, v\}) = \{\sigma(u), \sigma(v)\}$; for multigraphs, we fix in addition a bijection between the parallel edge classes $\mathrm{v}^{-1}(\{u,v\})$ and $\mathrm{v}^{-1}(\{\sigma(u),\sigma(v)\})$, matching multiplicities.
\end{definition}

An important observation is that the permutation on vertices induces an intertwining property on the edge-vertex incidence matrix $\partial_1$ of the graph.

\begin{lemma}[Graph isomorphism]\label{lemma: graph-isomorphism}
    Let $\mathcal{G}: \mathbb{F}_2[\mathcal{V}] \xrightarrow{\partial_1} \mathbb{F}_2[\mathcal{E}] \xrightarrow{\partial_0} \mathbb{F}_2[\mathcal{C}]$ be a connected graph chain, where $\partial_1$ and $\partial_0$ denote, respectively, the edge-vertex and cycle-edge incidence matrices, with $\IM \partial_0^T = \ker \partial_1^T$. Let $\sigma$ be a permutation of $\mathcal{V}$. Define matrices $P^{\mathcal{V}}_{\sigma}$ and $P^{\mathcal{E}}_{\sigma}$ by
    \begin{align}
        (P^{\mathcal{V}}_{\sigma} x)(v) &= x(\sigma^{-1}(v)), \\
        (P^{\mathcal{E}}_{\sigma} y)(e) &= \sum_{f \in P_{\sigma^{-1}(e)}} y(f),
    \end{align}
    where, for each edge $e$, $P_{\sigma^{-1}(e)} \subseteq \mathcal{E}$ is a path joining the two preimage vertices $\sigma^{-1}(\mathrm{v}_0(e))$ and $\sigma^{-1}(\mathrm{v}_1(e))$, which exists since $\mathcal{G}$ is connected. Equivalently, on the indicator vectors $\mathbf{1}_v \in \mathbb{F}_2[\mathcal{V}]$, we have $P^{\mathcal{V}}_{\sigma} \mathbf{1}_v = \mathbf{1}_{\sigma(v)}$, and the row of $P^{\mathcal{E}}_{\sigma}$ at $e$ is the indicator of $P_{\sigma^{-1}(e)}$. Then we have the following properties.
    \begin{itemize}
        \item (Intertwining property). For every choice of the paths $\{P_{\sigma^{-1}(e)}\}_{e \in \mathcal{E}}$, we have
        \begin{align}\label{eq: graph-isomorphism-intertwining}
            P^{\mathcal{E}}_{\sigma} \partial_1 = \partial_1 P^{\mathcal{V}}_{\sigma}.
        \end{align}
        \item (Graph automorphism). If $\sigma$ is a graph automorphism, then taking $P_{\sigma^{-1}(e)} = \{\sigma^{-1}(e)\}$ makes $P^{\mathcal{E}}_{\sigma}$ the permutation matrix of the induced edge permutation, $P^{\mathcal{E}}_{\sigma} \mathbf{1}_e = \mathbf{1}_{\sigma(e)}$.
        \item (Uniqueness). Each row of $P^{\mathcal{E}}_{\sigma}$ is unique up to the addition of elements of $\IM \partial_0^T = \ker \partial_1^T$: deforming any row by an $\mathbb{F}_2$-linear combination of cycles preserves Eq.~\eqref{eq: graph-isomorphism-intertwining}, and any two solutions differ rowwise by such elements.
        \item (Row weight). Choosing shortest paths,
        \begin{align}
          \begin{aligned}
            |P_{\sigma^{-1}(e)}| \leq 1
            &+ \mathrm{dist}_{\mathcal{G}}\bigl(\mathrm{v}_0(e), \sigma^{-1}(\mathrm{v}_0(e))\bigr) \\
            &+ \mathrm{dist}_{\mathcal{G}}\bigl(\mathrm{v}_1(e), \sigma^{-1}(\mathrm{v}_1(e))\bigr).
          \end{aligned}
        \end{align}
    \end{itemize}
    \begin{proof}
        Both sides of Eq.~\eqref{eq: graph-isomorphism-intertwining} are linear maps $\mathbb{F}_2[\mathcal{V}] \rightarrow \mathbb{F}_2[\mathcal{E}]$, so it suffices to evaluate them on $x \in \mathbb{F}_2[\mathcal{V}]$ at each edge $e$. For the right-hand side, since $(\partial_1 x)(e) = x(\mathrm{v}_0(e)) + x(\mathrm{v}_1(e))$ and by the definition of $P^{\mathcal{V}}_{\sigma}$,
        \begin{align}
            (\partial_1 P^{\mathcal{V}}_{\sigma} x)(e)
            = x(\sigma^{-1}(\mathrm{v}_0(e))) + x(\sigma^{-1}(\mathrm{v}_1(e))).
        \end{align}
        For the left-hand side, write the path as
        \begin{align}
            P_{\sigma^{-1}(e)}: \quad u_0 \xrightarrow{\;f_1\;} u_1 \xrightarrow{\;f_2\;} \cdots \xrightarrow{\;f_s\;} u_s,
        \end{align}
        with $u_0 = \sigma^{-1}(\mathrm{v}_0(e))$ and $u_s = \sigma^{-1}(\mathrm{v}_1(e))$. Then, by the definition of $P^{\mathcal{E}}_{\sigma}$,
        \begin{align}
            (P^{\mathcal{E}}_{\sigma}\, \partial_1 x)(e)
            = \sum_{t=1}^{s} \bigl[x(u_{t-1}) + x(u_t)\bigr]
            = x(u_0) + x(u_s),
        \end{align}
        where the last equality holds because each interior vertex $u_t$, for $1 \leq t \leq s-1$, appears in exactly two consecutive terms and cancels over $\mathbb{F}_2$; the pairwise-distinct-vertices condition in Definition~\ref{def: path-cycle} ensures that these are the only repetitions. The two sides agree for all $x$ and $e$, proving Eq.~\eqref{eq: graph-isomorphism-intertwining}. If $\sigma$ is a graph automorphism, then $\{\sigma^{-1}(e)\}$ is such a path and $e \mapsto \sigma^{-1}(e)$ is a bijection on $\mathcal{E}$, so $P^{\mathcal{E}}_{\sigma}$ is a permutation matrix. If a matrix $M$ also satisfies $M \partial_1 = \partial_1 P^{\mathcal{V}}_{\sigma}$, then $(P^{\mathcal{E}}_{\sigma} + M) \partial_1 = 0$, so each row of $P^{\mathcal{E}}_{\sigma} + M$ lies in $\ker \partial_1^T = \IM \partial_0^T$, which is the uniqueness claim. Finally, choose $P_{\sigma^{-1}(e)}$ as a shortest path, so that $|P_{\sigma^{-1}(e)}| = \mathrm{dist}_{\mathcal{G}}\bigl(\sigma^{-1}(\mathrm{v}_0(e)), \sigma^{-1}(\mathrm{v}_1(e))\bigr)$, the shortest path (as in number of edges) connecting these two endpoints. By the repeated applications of the triangle inequality: 
        \begin{align}
           \begin{aligned}
             |P_{\sigma^{-1}(e)}| &= \mathrm{dist}_{\mathcal{G}}\bigl(\sigma^{-1}(\mathrm{v}_0(e)), \sigma^{-1}(\mathrm{v}_1(e))\bigr) \\
             &\leq \mathrm{dist}_{\mathcal{G}}\bigl(\sigma^{-1}(\mathrm{v}_0(e)), \mathrm{v}_0(e)\bigr) \\
             &+ \mathrm{dist}_{\mathcal{G}}\bigl(\mathrm{v}_1(e), \sigma^{-1}(\mathrm{v}_1(e))\bigr) \\
             &+ \mathrm{dist}_{\mathcal{G}}\bigl(\mathrm{v}_0(e), \mathrm{v}_1(e)\bigr).
           \end{aligned}
        \end{align}
        Note that $\mathrm{dist}_{\mathcal{G}}\bigl(\mathrm{v}_0(e), \mathrm{v}_1(e)\bigr)=1$, which proves the last assertion. 
    \end{proof}
\end{lemma}

We say that, for undirected graphs, an edge is invariant under $\sigma$ if $e = \sigma(e)$; that is, $\sigma$ only permutes the vertices of the edge. In certain cases, we wish to have a concept called \emph{one-way permutation}: we wish to identify $v_1$ by $v_2$ but not vice versa, so that we may \emph{prune} the vertex $v_1$ after such an identification. 

\begin{definition}[One-way permutation along a path]\label{def: one-way-permutation-graph}
    Let $\mathcal{G}(\mathcal{V}, \mathcal{E})$ be a graph. An \emph{index path} is a sequence of pairwise distinct vertices $v_1, v_2, \cdots, v_k \in \mathcal{V}$, $k\geq 2$. An index cycle is an index path except at the end we close with $v_k \rightarrow v_1$. The associated \emph{one-way permutation} along the index path is the vertex map: 
    \begin{align}
        \pi(v_i) = v_{i+1},
    \end{align}
    and $\pi(v_k) = v_k$ for an index path and $\pi(v_k) = v_1$ for an index cycle, where, in this case, we refer $\pi$ as simply a permutation. 
\end{definition}

Whenever the index path becomes a cycle, Definition~\ref{def: one-way-permutation-graph} reduces to a graph isomorphism. This gives a more general and complete characterization of the framework that enables us to construct various gadgets measuring arbitrary logical operators. 

\begin{lemma}[One-way permutation intertwining property]\label{lemma: one-way-intertwining}
      Let $\mathcal{G}: \mathbb{F}_2[\mathcal{V}] \xrightarrow{\partial_1} \mathbb{F}_2[\mathcal{E}] \xrightarrow{\partial_0} \mathbb{F}_2[\mathcal{C}]$ be a connected graph chain, where $\partial_1$ and $\partial_0$ denote, respectively, the edge-vertex and cycle-edge incidence matrices, with $\IM \partial_0^T = \ker \partial_1^T$. Let $\pi$ be a one-way permutation along an index path. Define matrices $Q^{\mathcal{V}}_{\pi}$ and $Q^{\mathcal{E}}_{\pi}$ by
      \begin{align}
        Q^{\mathcal{V}}_{\pi}\, \mathbf{1}_v &= \mathbf{1}_{\pi(v)}, \\
        Q^{\mathcal{E}}_{\pi}\, \mathbf{1}_e &= \mathbf{1}_{P_{\pi(e)}},
      \end{align}
      where $\mathbf{1}_v \in \mathbb{F}_2[\mathcal{V}]$ and $\mathbf{1}_e \in \mathbb{F}_2[\mathcal{E}]$ denote the indicator vectors, $\mathbf{1}_{P_{\pi(e)}} = \sum_{f \in P_{\pi(e)}} \mathbf{1}_f$ is the indicator of the path, and, for each edge $e$, $P_{\pi(e)} \subseteq \mathcal{E}$ is a path joining the two image vertices $\pi(\mathrm{v}_0(e))$ and $\pi(\mathrm{v}_1(e))$, taken empty ($\mathbf{1}_{P_{\pi(e)}} = 0$) if $\pi(\mathrm{v}_0(e)) = \pi(\mathrm{v}_1(e))$. Then we have the following properties.
      \begin{itemize}
        \item (Intertwining property). For every choice of the paths $\{P_{\pi(e)}\}_{e \in \mathcal{E}}$, we have 
        \begin{align}\label{eq: one-way-intertwining}
            Q^{\mathcal{V}}_{\pi} \partial^T_1 = \partial^T_1 Q^{\mathcal{E}}_{\pi}
        \end{align}
        or, stated equivalently, 
        \begin{align}
              \partial_1 J^{\mathcal{V}}_{\pi}= J^{\mathcal{E}}_{\pi}\partial_1,
        \end{align}
        for $(Q^{\mathcal{V}}_{\pi})^T =J^{\mathcal{V}}_{\pi} $ and $(Q^{\mathcal{E}}_{\pi})^T = J^{\mathcal{E}}_{\pi}$, that is, 
        \begin{align}
            (J^{\mathcal{V}}_{\pi} x)(v) &= x(\pi(v)) \\
            (J^{\mathcal{E}}_{\pi} y)(e) &= \sum_{f \in P_{\pi(e)}} y(f).
        \end{align}

        \item (Invariance). $J^{\mathcal{V}}_{\pi} \mathbf{1}_{\mathcal{V}} = \mathbf{1}_{\mathcal{V}}$, but it is not true for $Q^{\mathcal{V}}_{\pi}$ unless $\pi$ is a proper permutation: along an index path,
        \begin{align}
            Q^{\mathcal{V}}_{\pi} \mathbf{1}_{\mathcal{V}} = \mathbf{1}_{\mathcal{V}} + \mathbf{1}_{v_1} + \mathbf{1}_{v_k}.
        \end{align}
        Similarly, $J^{\mathcal{E}}_{\pi}$ preserves $\IM \partial_1$: $J^{\mathcal{E}}_{\pi} \partial_1 x = \partial_1 J^{\mathcal{V}}_{\pi} x$ for every $x \in \mathbb{F}_2[\mathcal{V}]$.

        \item (Reduction to graph isomorphism for index cycles). For any index cycle, the one-way permutation $\pi$ is a permutation, $Q^{\mathcal{V}}_{\pi} = P^{\mathcal{V}}_{\pi}$, and $J^{\mathcal{V}}_{\pi}$ and $J^{\mathcal{E}}_{\pi}$ are the corresponding matrices in Lemma~\ref{lemma: graph-isomorphism} for the inverse permutation $\sigma = \pi^{-1}$; in particular Eq.~\eqref{eq: one-way-intertwining} reduces to Eq.~\eqref{eq: graph-isomorphism-intertwining}.
        \item (Uniqueness). Each column of $Q^{\mathcal{E}}_{\pi}$ is unique up to the addition of elements of $\IM \partial_0^T = \ker \partial_1^T$: deforming any column by an $\mathbb{F}_2$-linear combination of cycles preserves Eq.~\eqref{eq: one-way-intertwining}, and any two solutions differ columnwise by such elements.
        \item (Column weight). Choosing shortest paths,
        \begin{align}
          \begin{aligned}
            |P_{\pi(e)}| \leq 1
            &+ \mathrm{dist}_{\mathcal{G}}\bigl(\mathrm{v}_0(e), \pi(\mathrm{v}_0(e))\bigr) \\
            &+ \mathrm{dist}_{\mathcal{G}}\bigl(\mathrm{v}_1(e), \pi(\mathrm{v}_1(e))\bigr),
          \end{aligned}
        \end{align}
        with $|P_{\pi(e)}| = 0$ exactly for the collapsed edges, $\pi(\mathrm{v}_0(e)) = \pi(\mathrm{v}_1(e))$.
    \end{itemize}
    \begin{proof}
        Both sides of Eq.~\eqref{eq: one-way-intertwining} are linear maps $\mathbb{F}_2[\mathcal{E}] \rightarrow \mathbb{F}_2[\mathcal{V}]$, so it suffices to evaluate them on the indicator vector $\mathbf{1}_e$ of each edge $e$. For the left-hand side, since $\partial^T_1 \mathbf{1}_e = \mathbf{1}_{\mathrm{v}_0(e)} + \mathbf{1}_{\mathrm{v}_1(e)}$ and by the definition of $Q^{\mathcal{V}}_{\pi}$,
        \begin{align}
            Q^{\mathcal{V}}_{\pi}\, \partial^T_1 \mathbf{1}_e
            = Q^{\mathcal{V}}_{\pi}\bigl(\mathbf{1}_{\mathrm{v}_0(e)} + \mathbf{1}_{\mathrm{v}_1(e)}\bigr)
            = \mathbf{1}_{\pi(\mathrm{v}_0(e))} + \mathbf{1}_{\pi(\mathrm{v}_1(e))}.
        \end{align}
        For the right-hand side, write the path as
        \begin{align}
            P_{\pi(e)}: \quad u_0 \xrightarrow{\;f_1\;} u_1 \xrightarrow{\;f_2\;} \cdots \xrightarrow{\;f_s\;} u_s,
        \end{align}
        with $u_0 = \pi(\mathrm{v}_0(e))$ and $u_s = \pi(\mathrm{v}_1(e))$. Then, by the definition of $Q^{\mathcal{E}}_{\pi}$,
        \begin{align}
            \partial^T_1\, Q^{\mathcal{E}}_{\pi} \mathbf{1}_e
            = \partial^T_1 \mathbf{1}_{P_{\pi(e)}}
            = \sum_{t=1}^{s} \bigl(\mathbf{1}_{u_{t-1}} + \mathbf{1}_{u_t}\bigr)
            = \mathbf{1}_{u_0} + \mathbf{1}_{u_s},
        \end{align}
        where the last equality holds because each interior vertex $u_t$, for $1 \leq t \leq s-1$, appears in exactly two consecutive terms and cancels over $\mathbb{F}_2$; the pairwise-distinct-vertices condition in Definition~\ref{def: path-cycle} ensures that these are the only repetitions. For a collapsed edge, $\pi(\mathrm{v}_0(e)) = \pi(\mathrm{v}_1(e))$, the path is empty and both sides vanish. The two sides agree on every $\mathbf{1}_e$, proving Eq.~\eqref{eq: one-way-intertwining}. For the $J$-form, evaluate both sides on $x \in \mathbb{F}_2[\mathcal{V}]$ at each edge $e$: by the definitions of $J^{\mathcal{V}}_{\pi}$ and $J^{\mathcal{E}}_{\pi}$,
        \begin{align}
            (\partial_1 J^{\mathcal{V}}_{\pi} x)(e) &= x(\pi(\mathrm{v}_0(e))) + x(\pi(\mathrm{v}_1(e))), \\
            (J^{\mathcal{E}}_{\pi}\, \partial_1 x)(e) &= \sum_{t=1}^{s} \bigl[x(u_{t-1}) + x(u_t)\bigr] = x(u_0) + x(u_s),
        \end{align}
        by the same telescoping cancellation; the two sides agree, and both vanish on a collapsed edge. This identity is also the preservation of $\IM \partial_1$ by $J^{\mathcal{E}}_{\pi}$. For the invariance, each row of $J^{\mathcal{V}}_{\pi}$ is the single indicator $\mathbf{1}^T_{\pi(v)}$, so $(J^{\mathcal{V}}_{\pi} \mathbf{1}_{\mathcal{V}})(v) = \mathbf{1}_{\mathcal{V}}(\pi(v)) = 1$ for every $v$. For $Q^{\mathcal{V}}_{\pi}$, $(Q^{\mathcal{V}}_{\pi} \mathbf{1}_{\mathcal{V}})(v) = |\pi^{-1}(v)| \bmod 2$; along an index path $\pi^{-1}(v_1) = \emptyset$ and $\pi^{-1}(v_k) = \{v_{k-1}, v_k\}$, while every other preimage is a singleton, giving $Q^{\mathcal{V}}_{\pi}\mathbf{1}_{\mathcal{V}} = \mathbf{1}_{\mathcal{V}} + \mathbf{1}_{v_1} + \mathbf{1}_{v_k}$, which is not invariant. However, for a proper permutation, every preimage is a singleton and $Q^{\mathcal{V}}_{\pi}\mathbf{1}_{\mathcal{V}} = \mathbf{1}_{\mathcal{V}}$. If a matrix $M$ also satisfies $Q^{\mathcal{V}}_{\pi} \partial^T_1 = \partial^T_1 M$, then $\partial^T_1 (Q^{\mathcal{E}}_{\pi} + M) = 0$, so each column of $Q^{\mathcal{E}}_{\pi} + M$ lies in $\ker \partial^T_1 = \IM \partial^T_0$, which is the uniqueness claim. Finally, choose $P_{\pi(e)}$ as a shortest path, so that $|P_{\pi(e)}| = \mathrm{dist}_{\mathcal{G}}\bigl(\pi(\mathrm{v}_0(e)), \pi(\mathrm{v}_1(e))\bigr)$, the shortest path (as in number of edges) connecting the two image vertices. By the repeated application of the triangle inequality:
        \begin{align}
           \begin{aligned}
             |P_{\pi(e)}| &= \mathrm{dist}_{\mathcal{G}}\bigl(\pi(\mathrm{v}_0(e)), \pi(\mathrm{v}_1(e))\bigr) \\
             &\leq \mathrm{dist}_{\mathcal{G}}\bigl(\pi(\mathrm{v}_0(e)), \mathrm{v}_0(e)\bigr) + \mathrm{dist}_{\mathcal{G}}\bigl(\mathrm{v}_1(e), \pi(\mathrm{v}_1(e))\bigr) \\
             &+ \mathrm{dist}_{\mathcal{G}}\bigl(\mathrm{v}_0(e), \mathrm{v}_1(e)\bigr).
           \end{aligned}
        \end{align}
        Note that $\mathrm{dist}_{\mathcal{G}}\bigl(\mathrm{v}_0(e), \mathrm{v}_1(e)\bigr) = 1$ and, for a collapsed edge, $|P_{\pi(e)}| = 0$, which proves the last assertion.
    \end{proof}
\end{lemma}

\begin{example}[Permutation and one-way permutation on the cycle graph $C_4$]\label{example: C4-one-way-permutation}
\begin{figure}[H]
    \centering
    \begin{tikzpicture}[
        every node/.style={font=\scriptsize},
        vtx/.style={circle, fill=black, inner sep=1.1pt},
        pimap/.style={->, densely dashed, violet, thick},
        pathhl/.style={line width=3.2pt, orange, opacity=0.35, line cap=round},
        collapsed/.style={densely dotted, gray, thick},
        edg/.style={thick}
    ]
    \foreach \dx in {0, 3.0, 6.0}{
        \begin{scope}[xshift=\dx cm]
            \coordinate (v1) at (0,1.5);
            \coordinate (v2) at (1.5,1.5);
            \coordinate (v3) at (1.5,0);
            \coordinate (v4) at (0,0);
            \node[above left=-1pt] at (v1) {$v_1$};
            \node[above right=-1pt] at (v2) {$v_2$};
            \node[below right=-1pt] at (v3) {$v_3$};
            \node[below left=-1pt] at (v4) {$v_4$};
        \end{scope}
    }
    \begin{scope}
        \draw[pathhl] (1.5,0) -- (0,0) -- (0,1.5);
        \draw[edg] (0,1.5) -- node[above] {$e_1$} (1.5,1.5);
        \draw[edg] (1.5,1.5) -- node[right] {$e_2$} (1.5,0);
        \draw[edg] (1.5,0) -- node[below] {$e_3$} (0,0);
        \draw[edg] (0,0) -- node[left] {$e_4$} (0,1.5);
        \draw[pimap] (0,1.5) to[bend right=28] (1.5,1.5);
        \draw[pimap] (1.5,1.5) to[bend right=28] (1.5,0);
        \draw[pimap] (1.5,0) -- (0.12,1.38);
        \foreach \p in {(0,1.5),(1.5,1.5),(1.5,0),(0,0)} \node[vtx] at \p {};
        \node at (0.75,-0.62) {(a) $1{\to}2{\to}3{\to}1$};
    \end{scope}
    \begin{scope}[xshift=3.0cm]
        \draw[pathhl] (0,1.5) -- (1.5,1.5) (0,0) -- (0,1.5);
        \draw[collapsed] (0,1.5) -- node[above] {$e_1$} (1.5,1.5);
        \draw[edg] (1.5,1.5) -- node[right] {$e_2$} (1.5,0);
        \draw[edg] (1.5,0) -- node[below] {$e_3$} (0,0);
        \draw[edg] (0,0) -- node[left] {$e_4$} (0,1.5);
        \draw[pimap] (0,1.5) to[bend right=28] (1.5,1.5);
        \foreach \p in {(0,1.5),(1.5,1.5),(1.5,0),(0,0)} \node[vtx] at \p {};
        \node at (0.75,-0.62) {(b) $1{\to}2$};
    \end{scope}
    \begin{scope}[xshift=6.0cm]
        \draw[pathhl] (0,1.5) -- (1.5,1.5) (0,0) -- (0,1.5);
        \draw[edg] (0,1.5) -- node[above] {$e_1$} (1.5,1.5);
        \draw[collapsed] (1.5,1.5) -- node[right] {$e_2$} (1.5,0);
        \draw[edg] (1.5,0) -- node[below] {$e_3$} (0,0);
        \draw[edg] (0,0) -- node[left] {$e_4$} (0,1.5);
        \draw[pimap] (0,1.5) to[bend right=28] (1.5,1.5);
        \draw[pimap] (1.5,1.5) to[bend right=28] (1.5,0);
        \foreach \p in {(0,1.5),(1.5,1.5),(1.5,0),(0,0)} \node[vtx] at \p {};
        \node at (0.75,-0.62) {(c) $1{\to}2{\to}3$};
    \end{scope}
    \end{tikzpicture}
    \caption{\textbf{Permutation and one-way permutation on the cycle graph $C_4$} (Example~\ref{example: C4-one-way-permutation}). Dashed violet arrows depict the vertex map $\pi$; a dotted gray edge is collapsed, $\pi(\mathrm{v}_0(e)) = \pi(\mathrm{v}_1(e))$, and receives the zero row of $J^{\mathcal{E}}_{\pi}$; the orange overlay marks the weight-$2$ path row of $J^{\mathcal{E}}_{\pi}$. \textbf{(a)} The index cycle $1 \to 2 \to 3 \to 1$ is a proper permutation but not a graph automorphism, so the row of $e_2$ is the path $\{e_3, e_4\}$ (highlighted). \textbf{(b)} The index path $1 \to 2$ collapses $e_1$ and prunes $v_1$; the row of $e_4$ is the path $\{e_1, e_4\}$ (highlighted). \textbf{(c)} The index path $1 \to 2 \to 3$ moves the collapsed edge to $e_2$; the row of $e_4$ is again $\{e_1, e_4\}$.}
    \label{fig: C4-one-way-example}
\end{figure}
    We illustrate the (one-way) permutations and their intertwining properties (Lemmas~\ref{lemma: graph-isomorphism} and~\ref{lemma: one-way-intertwining}) with an explicit cyclic graph of order $4$ (Fig.~\ref{fig: C4-one-way-example}): $\mathcal{V} = \{v_1, \ldots, v_4\}$, $\mathcal{E} = \{e_i: \mathrm{v}(e_i) = \{v_i, v_{i+1}\}\}^{4}_{i=1}$ with $v_5 \equiv v_1$, and the single cycle $c = e_1 + e_2 + e_3 + e_4$, so that $\ker \partial^T_1 = \mathrm{span}\{\mathbf{1}_{\mathcal{E}}\}$. In the ordered bases $(v_1, \ldots, v_4)$ and $(e_1, \ldots, e_4)$,
    \begin{align}
        \partial_1 = \begin{pmatrix} 1&1&0&0 \\ 0&1&1&0 \\ 0&0&1&1 \\ 1&0&0&1 \end{pmatrix}.
    \end{align}
    \begin{enumerate}
        \item (Index cycle $1 \to 2 \to 3 \to 1$). Here $\pi = (v_1\, v_2\, v_3)$, fixing $v_4$, is a proper permutation, so $J^{\mathcal{V}}_{\pi} = P^{\mathcal{V}}_{\sigma}$ and $J^{\mathcal{E}}_{\pi} = P^{\mathcal{E}}_{\sigma}$ for the graph isomorphism $\sigma = \pi^{-1} = (v_1\, v_3\, v_2)$:
        \begin{align}
            P^{\mathcal{V}}_{\sigma} = \begin{pmatrix} 0&1&0&0 \\ 0&0&1&0 \\ 1&0&0&0 \\ 0&0&0&1 \end{pmatrix}, \quad
            P^{\mathcal{E}}_{\sigma} = \begin{pmatrix} 0&1&0&0 \\ 0&0&1&1 \\ 0&0&0&1 \\ 1&0&0&1 \end{pmatrix},
        \end{align}
        with path rows $e_1 \mapsto \{e_2\}$, $e_2 \mapsto \{e_3, e_4\}$, $e_3 \mapsto \{e_4\}$, and $e_4 \mapsto \{e_4, e_1\}$. Since $\{v_1, v_3\} \notin \mathcal{E}$, $\sigma$ is not a graph automorphism of $C_4$, and $P^{\mathcal{E}}_{\sigma}$ is accordingly not a permutation matrix: the row of $e_2$ is a genuine weight-$2$ path:
        \begin{align}
           \begin{aligned}
             (P^{\mathcal{E}}_{\sigma} \partial_1 x)(e_2) &= (\partial_1 x)(e_3) + (\partial_1 x)(e_4) = x(v_3) + x(v_1); \\
             (\partial_1 P^{\mathcal{V}}_{\sigma}x)(e_2) &= (P^{\mathcal{V}}_{\sigma}x)(v_2) + (P^{\mathcal{V}}_{\sigma}x)(v_3) = x(v_3) + x(v_1),
           \end{aligned}
        \end{align}
        with the interior vertex $v_4$ cancelling over $\mathbb{F}_2$. Replacing it by the complementary path $\{e_1, e_2\}$ changes the row by the cycle $\mathbf{1}_{\mathcal{E}} \in \ker \partial^T_1$, illustrating the uniqueness property.
        \item (Index path $1 \to 2$). The one-way permutation $\pi(v_1) = v_2$, fixing all other vertices, gives
        \begin{align}
            J^{\mathcal{V}}_{\pi} = \begin{pmatrix} 0&1&0&0 \\ 0&1&0&0 \\ 0&0&1&0 \\ 0&0&0&1 \end{pmatrix}, \quad
            J^{\mathcal{E}}_{\pi} = \begin{pmatrix} 0&0&0&0 \\ 0&1&0&0 \\ 0&0&1&0 \\ 1&0&0&1 \end{pmatrix}.
        \end{align}
        The column of $v_1$ vanishes, so $J^{\mathcal{V}}_{\pi}$ is singular --- the vertex $v_1$ is pruned, which is what makes the identification one-way. The edge $e_1$ is collapsed ($\pi(v_1) = \pi(v_2) = v_2$) and receives the zero row, 
        \begin{align}
           \begin{aligned}
             &(J^{\mathcal{E}}_{\pi} \partial_1 x)(e_1) = 0; \\
             (\partial_1 J^{\mathcal{V}}_{\pi}x)(e_1) &= (J^{\mathcal{V}}_{\pi}x)(v_1) + (J^{\mathcal{V}}_{\pi}x)(v_2) = 0,
           \end{aligned}
        \end{align}
        while the row of $e_4$ is the path $\{e_1, e_4\}$ joining the images $v_4$ and $v_2$:
        \begin{align}
            \begin{aligned}
                  (J^{\mathcal{E}}_{\pi} \partial_1 x)(e_4) &= (\partial_1 x)(e_1) + (\partial_1 x)(e_4) = x(v_2) + x(v_4); \\
             (\partial_1 J^{\mathcal{V}}_{\pi}x)(e_4) &= (J^{\mathcal{V}}_{\pi}x)(v_1) + (J^{\mathcal{V}}_{\pi}x)(v_4) = x(v_2) + x(v_4).
            \end{aligned}
        \end{align}
        \item (Index path $1 \to 2 \to 3$). The one-way permutation $\pi(v_1) = v_2$, $\pi(v_2) = v_3$, fixing $v_3$ and $v_4$, gives
        \begin{align}
            J^{\mathcal{V}}_{\pi} = \begin{pmatrix} 0&1&0&0 \\ 0&0&1&0 \\ 0&0&1&0 \\ 0&0&0&1 \end{pmatrix}, \quad
            J^{\mathcal{E}}_{\pi} = \begin{pmatrix} 0&1&0&0 \\ 0&0&0&0 \\ 0&0&1&0 \\ 1&0&0&1 \end{pmatrix}.
        \end{align}
        The collapsed edge moves with the head of the index path: now $e_2$ receives the zero row ($\pi(v_2) = \pi(v_3) = v_3$),
        \begin{align}
           \begin{aligned}
             &(J^{\mathcal{E}}_{\pi} \partial_1 x)(e_2) = 0; \\
             (\partial_1 J^{\mathcal{V}}_{\pi}x)(e_2) &= (J^{\mathcal{V}}_{\pi}x)(v_2) + (J^{\mathcal{V}}_{\pi}x)(v_3) = 0,
           \end{aligned}
        \end{align}
        the row of $e_1$ is $\{e_2\}$. The row of $e_4$ is again $\{e_1, e_4\}$, attaining the column-weight bound $1 + \mathrm{dist}_{\mathcal{G}}(v_1, v_2) = 2$; its intertwining check is identical to that of the previous item.
    \end{enumerate}
    On the remaining edges both sides act identically, and in all three cases the full intertwining $\partial_1 J^{\mathcal{V}}_{\pi} = J^{\mathcal{E}}_{\pi} \partial_1$ is confirmed by direct matrix multiplication.
\end{example}

We will revisit an equivalent, slightly more general, formulation of one-way permutation in Section~\ref{sec: parallel-surgery-arbitrary-addressable}, which applies to parallel, hypergraph surgery techniques.

\subsubsection*{Generalized bridging techniques for graph-surgery gadgets}\label{sec: generalized-bridging-techniques}

Let $\mathfrak{X}$ be a set of surgery gadgets. A standard technique is to apply bridging qubits and checks between two surgery gadgets, each measuring a logical operator, to measure the product of the two logical operators. Here, we treat bridging in a more general setting, by unifying the bridges between two gadgets measuring the same-type ($XX$ or $ZZ$) logical operators and mixed types, as well as within a single surgery gadget for measuring mixed-type (including $Y$-type) logical operators. 

\begin{definition}[Generalized bridge on a connected graph]\label{def: generalized-bridging-technique-graph}
     Let  $\mathcal{G}: \mathbb{F}_2[\mathcal{V}] \xrightarrow{\partial_1} \mathbb{F}_2[\mathcal{E}] \xrightarrow{\partial_0} \mathbb{F}_2[\mathcal{C}]$ be a graph chain and $\gamma$ be a size-$r$ subset of vertices from $\mathcal{V}$. Then a generalized bridge to $\mathcal{G}$ is given as follows. 
     \begin{itemize}
        \item For $\gamma$, label the vertices by $\gamma = \{v_{1}, \cdots, v_{r} \}$ and $r$-many bridge edges $\mathcal{E}_b = \{e_1, \cdots, e_r\}$ with $\mathrm{v}(e_i) = \{v_i, v_{i+1}\}$ for $i=1, \cdots, r-1$ and $\mathrm{v}(e_r) = \{v_1, v_r \}$, closing $\gamma$ into an $r$-cycle; for $r=2$ this is a doubled edge, which we allow as a multigraph (Definition~\ref{def: formal-graph-definition}).
        \item Add $|\mathcal{C}_b|$-many bridge checks to $\mathcal{G}$, resulting in $\mathcal{C}' = \mathcal{C} \sqcup \mathcal{C}_b$ with $|\mathcal{C}_b| \geq r$,
     \end{itemize}
     such that the updated graph chain $\mathcal{G}': \mathbb{F}_2[\mathcal{V}] \xrightarrow{\partial'_1} \mathbb{F}_2[\mathcal{E}'] \xrightarrow{\partial'_0} \mathbb{F}_2[\mathcal{C}']$ is exact and connected, where $\partial'_1$ is the updated edge-vertex incidence matrix including the bridge edges.
\end{definition}

We also give another definition on disconnected graphs for concreteness. 

\begin{definition}[Generalized bridge on disconnected (hyper)graphs]\label{def: generalized-bridging-technique-graph-disconnected}
  Let  $\mathcal{G}: \mathbb{F}_2[\mathcal{V}] \xrightarrow{\partial_1} \mathbb{F}_2[\mathcal{E}] \xrightarrow{\partial_0} \mathbb{F}_2[\mathcal{C}]$ be a (hyper)graph chain. Let $\gamma_1$ and $\gamma_2$ be two equal-size subsets of vertices from $\mathcal{V}$ so that we, without loss of generality, assume that $\gamma_1 = \{ v_1, \cdots, v_s \}$ and $\gamma_2 = \{v_{s+1}, \cdots, v_{2s} \}$. Then a generalized bridge to $\mathcal{G}$ is given as follows. 
  \begin{itemize}
    \item Add $s$-many bridge edges $\mathcal{E}_b = \{ e_1, \cdots, e_s\}$ to $\mathcal{G}$, where $\mathrm{v}(e_j) = \{v_j, v_{s+j}\}$ for $j=1, \cdots, s$, resulting in $\mathcal{E}' = \mathcal{E} \sqcup \mathcal{E}_b$.
    \item Add $|\mathcal{C}_b|$-many bridge checks to $\mathcal{G}$, resulting in $\mathcal{C}' = \mathcal{C} \sqcup \mathcal{C}_b$.
  \end{itemize}
  such that the updated graph chain $\mathcal{G}': \mathbb{F}_2[\mathcal{V}] \xrightarrow{\partial'_1} \mathbb{F}_2[\mathcal{E}'] \xrightarrow{\partial'_0} \mathbb{F}_2[\mathcal{C}']$ is exact and connected, where $\partial'_1$ is the updated edge-vertex incidence matrix including the bridge edges.
\end{definition}

\begin{remark}
    Evidently, this definition generalizes the notion of bridging to disconnected (hyper)graphs, allowing us to connect two separate components via bridge edges, which is useful for hypergraph surgery in Section~\ref{sec: parallel-surgery-arbitrary-addressable}.
\end{remark}

Applying the above definition to a disjoint union of two graph-surgery gadgets results in a bridged gadget measuring the product of the two logical operators. For a single, connected graph, bridging can never decrease the soundness of the gadget, since the kernel of the edge-vertex incidence matrix is spanned by the all-one vector.
The following lemma shows that the bridged union of two disjoint gadgets is still sound, whenever its constituent gadgets are sound. 

\begin{lemma}\label{lemma: bridging-disjoint-graph-sound}
    Let two disjoint graph chains $\mathcal{G}: \mathbb{F}_2[\mathcal{V}] \xrightarrow{\partial_1} \mathbb{F}_2[\mathcal{E}] \xrightarrow{\partial_0} \mathbb{F}_2[\mathcal{C}]$ and $\mathcal{G}^{\prime}: \mathbb{F}_2[\mathcal{V}^{\prime}] \xrightarrow{\partial^{\prime}_1} \mathbb{F}_2[\mathcal{E}^{\prime}] \xrightarrow{\partial^{\prime}_0} \mathbb{F}_2[\mathcal{C}^{\prime}]$ be at least $(d, \rho_d)$-sound with $\rho_d \geq 1$. Let $\gamma$ and $\gamma^{\prime}$ be, respectively, two equal-size subsets of vertices from $\mathcal{V}$ and $\mathcal{V}^{\prime}$, of size $|\gamma| = |\gamma^{\prime}| = |\mathcal{E}_b| \geq d$. Then a generalized bridge to $\mathcal{G} \sqcup \mathcal{G}^{\prime}$, as in Definition~\ref{def: generalized-bridging-technique-graph-disconnected}, results in an edge-vertex incidence matrix that is also at least $(d, \rho_d)$-sound. The hypothesis $\rho_d \geq 1$ is automatic for the gadgets to which we apply this lemma, since item~(ii) of Definition~\ref{def: formal-surgery-gadgets-desiderata} requires $\rho_d \geq \max_i \max(\omega_{\mathrm{col}}(\Phi^{(i,X)}_1), \omega_{\mathrm{col}}(\Phi^{(i,Z)}_1)) \geq 1$.
   \begin{proof}
    Write $\mathcal{G}^{(i)} := \mathcal{G}$ and $\mathcal{G}^{(i')} := \mathcal{G}^{\prime}$ for the two bridged chains, and $\mathcal{V}^{(i)}, \mathcal{V}^{(i')}$ for their vertex sets. Without loss of generality, we assume that the bridge edges are added according to
        \begin{align}
            \partial_1 := \left( \begin{array}{cc}
                \partial^{(i)}_1 & 0 \\ 
                0 & \partial^{(i')}_1 \\
                B^{(i)} & B^{(i')}
            \end{array} \right),
        \end{align}
        where $B^{(s)} \in \mathbb{F}_2^{|\mathcal{E}_b| \times |\mathcal{V}^{(s)}|}$, $s \in \{i, i'\}$, is the selector of the $|\mathcal{E}_b|$ bridged vertices on side $s$ (bridge edge $t$ reads the entry $u_t + v_t$), for $|\mathcal{E}_b| \geq d$. Recall that $\partial^{(i)}_1$ and $\partial^{(i')}_1$ are edge-vertex incidence matrices of two connected graphs, their kernels spanned by the all-one vectors. For any vector $\begin{pmatrix} u \\ v \end{pmatrix}$, if $u \notin \ker \partial^{(i)}_1$ and $v \notin \ker \partial^{(i')}_1$, then we have that $|\partial_1 \begin{pmatrix} u \\ v \end{pmatrix}| \geq |\partial^{(i)}_1 u| + |\partial^{(i')}_1 v| \geq \rho_d \, \mathrm{dist}(u, \ker \partial^{(i)}_1) + \rho_d \, \mathrm{dist}(v, \ker \partial^{(i')}_1) \geq \rho_d \, \mathrm{dist}(\begin{pmatrix} u \\ v \end{pmatrix}, \ker \partial_1)$ in the \emph{aligned} case, where the two component distances are attained on matching kernel branches; the \emph{anti-aligned} case, where, without loss of generality, $|u| \leq |u + 1_{\mathcal{V}^{(i)}}|$ while $|v + 1_{\mathcal{V}^{(i')}}| < |v|$. The bridge rows of the above give $|B^{(i)} u + B^{(i')} v| \geq |\mathcal{E}_b| - |u| - |v + 1_{\mathcal{V}^{(i')}}|$, whence $|\partial_1 \begin{pmatrix} u \\ v \end{pmatrix}| \geq \rho_d (|u| + |v + 1_{\mathcal{V}^{(i')}}|) + |\mathcal{E}_b| - (|u| + |v + 1_{\mathcal{V}^{(i')}}|) \geq d$. Hence, it remains to consider the case whenever $u \in \ker \partial^{(i)}_1$ or $v \in \ker \partial^{(i')}_1$, and $\begin{pmatrix} u \\ v \end{pmatrix} \notin \ker \partial_1$. Without loss of generality, we assume that $u \in \ker \partial^{(i)}_1$ and $v \notin \ker \partial^{(i')}_1$, so that $u \in \{0, 1_{\mathcal{V}^{(i)}}\}$; the case $u=0$ is the anti-aligned estimate above applied to $\begin{pmatrix} 0 \\ v \end{pmatrix}$, so we may take $u = 1_{\mathcal{V}^{(i)}}$, the all-one vector on the vertices of $\mathcal{G}^{(i)}$ (if also $v \in \ker \partial^{(i')}_1$, the third row alone gives $|\partial_1 \begin{pmatrix} u \\ v \end{pmatrix}| = |\mathcal{E}_b| \geq d$). Suppose that 
        \begin{align}
            |\partial_1 \begin{pmatrix} 1_{\mathcal{V}^{(i)}} \\ v \end{pmatrix}| = |\begin{pmatrix} 0 \\ \partial^{(i')}_1 v \\ B^{(i)}1_{\mathcal{V}^{(i)}} + B^{(i')}v \end{pmatrix}| <d.
        \end{align}
        Note that we have that, on the third row, 
        \begin{align}
            \begin{pmatrix}
                B^{(i)} & B^{(i')}
            \end{pmatrix}  \begin{pmatrix} 1_{\mathcal{V}^{(i)}} \\ v \end{pmatrix} = \begin{pmatrix}
                B^{(i)} & B^{(i')}
            \end{pmatrix} \begin{pmatrix} 0 \\ v + 1_{\mathcal{V}^{(i')}} \end{pmatrix}
        \end{align}
        as $\begin{pmatrix}
            1_{\mathcal{V}^{(i)}} \\ 1_{\mathcal{V}^{(i')}}
        \end{pmatrix}$ is the nonzero kernel vector of $\partial_1$. In fact the same shift applies to all of $\partial_1$: $\partial_1 \begin{pmatrix} 1_{\mathcal{V}^{(i)}} \\ v \end{pmatrix} = \partial_1 \begin{pmatrix} 0 \\ v + 1_{\mathcal{V}^{(i')}} \end{pmatrix}$. Write $z := v + 1_{\mathcal{V}^{(i')}}$. If $|z| \leq |v|$, then, combined with the (small-set) soundness of $\partial^{(i')}_1$,
        \begin{align}
              |\partial_1 \begin{pmatrix} 1_{\mathcal{V}^{(i)}} \\ v \end{pmatrix}| \geq |\partial^{(i')}_1 z| \geq \rho_d |z| = \rho_d\mathrm{dist}(\begin{pmatrix} 1_{\mathcal{V}^{(i)}} \\ v \end{pmatrix}, \ker \partial_1),
        \end{align}
        since $\mathrm{dist}(z, \ker \partial^{(i')}_1) = |z|$ and $\mathrm{dist}(\begin{pmatrix} 1_{\mathcal{V}^{(i)}} \\ v \end{pmatrix}, \ker \partial_1) = |z|$. Otherwise $|v| < |z|$: the third row gives $|B^{(i)} 1_{\mathcal{V}^{(i)}} + B^{(i')} v| \geq |\mathcal{E}_b| - |v|$, so that $|\partial_1 \begin{pmatrix} 1_{\mathcal{V}^{(i)}} \\ v \end{pmatrix}| \geq \rho_d |v| + |\mathcal{E}_b| - |v| \geq |\mathcal{E}_b| \geq d$ for $\rho_d \geq 1$, contradicting the supposition.
   \end{proof}
\end{lemma}

On the other hand, the bridged (cycle) checks are given by 
\begin{align}\label{eq: bridged-disjoint-graphs-cycles}
    \partial_0 = \left( \begin{array}{ccc}
        \partial^{(i)}_0 & 0 & 0 \\
        0 & \partial^{(i')}_0 & 0 \\
        \Lambda^{(i)} & \Lambda^{(i')} & D
    \end{array} \right), 
\end{align}
where $D$ is given by the line graph/parity check of the repetition code on $|\mathcal{E}_b|$-many bridge edges. In particular, we assume that these bridged checks are supported on four-cycles, so that $\Lambda^{(i)}$ and $\Lambda^{(i')}$ are diagonal matrices.  Technically, we cannot directly assume that there always exist four-cycles for all bridge checks. We assume this for practical convenience and refer to a more rigorous treatment in Ref.~\cite{he2025extractorsqldpcarchitecturesefficient}. Next, we show how the bridging connects with graph isomorphisms which allows us to measure a product of \emph{mixed} logical operators. 

\begin{lemma}\label{lemma: bridge-measure-Y}
    Let $\mathcal{G}: \mathbb{F}_2[\mathcal{V}] \xrightarrow{\partial_1} \mathbb{F}_2[\mathcal{E}] \xrightarrow{\partial_0} \mathbb{F}_2[\mathcal{C}]$ be a graph chain. Let the permutation group on vertices be defined by $S^{\oplus l}_r$; that is, permutations on the disjoint union of $l$-many copies of $r$-many vertices, generated by cyclic-adjacent Coxeter generators (adjacent transpositions) $\tau^{(m)}_i = (i \; i+1)$ for $i=1, \cdots, r-1$, with $\tau^{(m)}_r = (r \; 1)$ for any $m$-th copy. Let the bridged edges be given by the set 
    \begin{align}
        \mathcal{E}_b = \bigsqcup_{i=1}^{r} \left\{ \mathrm{v}(e^b_{m,i}) = \{v_{m, i}, v_{m, (i \bmod r) + 1}\}, m=1, \cdots l   \right\}.
    \end{align}
    Denote the updated graph chain by $\mathcal{G}': \mathbb{F}_2[\mathcal{V}] \xrightarrow{\partial'_1} \mathbb{F}_2[\mathcal{E}']$, where
    \begin{align}
        \mathcal{E}' = \mathcal{E} \cup \mathcal{E}_b.
    \end{align}
     Let any $\sigma \in S^{\oplus l}_r$ and $P^{\mathcal{V}}_{\sigma}$ be the induced matrices on  $\mathbb{F}_2[\mathcal{V}]$. There exists a choice $P^{\mathcal{E}'}_{\sigma}$ such that 
     \begin{align}\label{eq: permutation-ring-bridge-intertwining}
        P^{\mathcal{E}'}_{\sigma} \partial'_1 = \partial'_1 P^{\mathcal{V}}_{\sigma},
     \end{align} 
     where the row weight of $P^{\mathcal{E}'}_{\sigma}$ is at most $r+1$ and the updated graph vertex degree is at most increased by $2$.
     \begin{proof}
        Recall from Lemma~\ref{lemma: graph-isomorphism} that any choice of paths $P_{\sigma^{-1}(e)} \subseteq \mathcal{E}'$ joining $\sigma^{-1}(\mathrm{v}_0(e))$ and $\sigma^{-1}(\mathrm{v}_1(e))$, one for each $e \in \mathcal{E}'$, yields a matrix $P^{\mathcal{E}'}_{\sigma}$ satisfying Eq.~\eqref{eq: permutation-ring-bridge-intertwining}: the telescoping argument there uses only the existence of the chosen paths, so the connectivity of $\mathcal{G}$ is not required. By construction, each $v_{m,i}$ and its preimage $\sigma^{-1}(v_{m,i})$ lie on the common bridge $r$-cycle $(e^b_{m,1}, \cdots, e^b_{m,r})$, joined by a ring arc of length at most $\lfloor r/2 \rfloor$, while every other vertex is invariant under $\sigma$. Choose the paths as follows:
        \[\begin{tikzcd}
            {\mathrm{v}_0(e)} & {\mathrm{v}_1(e)} \\
            {\sigma^{-1}(\mathrm{v}_0(e))} & {\sigma^{-1}(\mathrm{v}_1(e))}
            \arrow["e", from=1-1, to=1-2]
            \arrow[two heads, from=1-1, to=2-1]
            \arrow[two heads, from=1-2, to=2-2]
            \arrow[from=2-1, to=2-2].
        \end{tikzcd}\]
        The row path (the bottom edge) is realized as the left vertical, followed by $e$, followed by the right vertical, where the vertical double arrows are the ring arcs from each preimage to its endpoint, empty at fixed endpoints (in particular, $P_{\sigma^{-1}(e)} = \{e\}$ if $\sigma$ fixes both). The three segments meet only at their joints, so the concatenation is a path, and
        \begin{align}
            |P_{\sigma^{-1}(e)}| \leq \lfloor r/2 \rfloor + 1 + \lfloor r/2 \rfloor \leq r + 1.
        \end{align}
        Finally, each $v_{m,i}$ is incident to exactly the two new bridge edges $e^b_{m,i-1}$ and $e^b_{m,i}$, indices modulo $r$, and no other degree changes: the vertex degree increases by at most $2$. 
     \end{proof}
\end{lemma}

\begin{remark}
It is clear that the above can be directly applied to the case of one-way permutation, where the row weight of $J^{\mathcal{E}}_{\pi}$, for any such $\pi$ acting along some index path with the same set of vertices from the action of $S^{\oplus l}_r$, is at most $1 + \lfloor r/2 \rfloor$: a one-way permutation along an index path fixes every vertex outside it, so at most one endpoint of each non-bridge edge is moved of length at most $\lfloor r/2 \rfloor$.
\end{remark}

There is also a way to ensure that $P^{\mathcal{E}}_{\sigma}$ (resp. $J^{\mathcal{E}}_{\sigma}$) has row- and column-weight at most $1$. For a graph permutation, this can be achieved by adding bridged edges so that the desired operation $\sigma$ becomes a graph automorphism. We call this process \emph{graph symmetrization}. 

\begin{lemma}\label{lemma: graph-symmetrization-single-permutation}
    Let a graph be given by $\mathcal{G}(\mathcal{V}, \mathcal{E})$.  Under the above setting with $l=1$ for simplicity, let $\sigma = (i \; j) \in S_r$ be a transposition of $v_i$ and $v_j$. Then denote the bridged edge set by 
    \begin{align}
        \mathcal{E}^{\sigma}_b := \left\{ e^{\sigma} :\, \mathrm{v}(e^{\sigma}) = \sigma(\mathrm{v}(e)),\; e \in \mathcal{E}_{ij} \right\},
    \end{align}
    where $\mathcal{E}_{ij} := \{ e \in \mathcal{E} :\, |\mathrm{v}(e) \cap \{v_i, v_j\}| = 1 \}$ collects the edges incident to exactly one of $v_i, v_j$, and $\sigma$ acts on endpoint pairs by $\sigma(\{u, w\}) = \{\sigma(u), \sigma(w)\}$; that is, one bridge edge for each such $e$, carrying its moved endpoint. Then $\sigma$ is a graph automorphism on the bridged graph $\mathcal{G}(\mathcal{V}, \mathcal{E} \sqcup \mathcal{E}^{\sigma}_b)$.
    \begin{proof}
        By inspection, $\sigma$ only acts on $v_i, v_j$ and the graph is undirected. Hence, the only edges that are not invariant under $\sigma$ are those incident to either $v_i$ or $v_j$; an edge $e$ with $\mathrm{v}(e) = \{v_i, v_j\}$ is already invariant and needs no bridge, which is why it is excluded above --- bridging it would create a self-loop. By the construction of $\mathcal{E}^{\sigma}_b$, for every edge $e$ with $\mathrm{v}(e) = \{v_i, v^{\prime\prime}\}$ the bridged graph contains a corresponding edge with endpoints $\{v_j, v^{\prime\prime}\}$, and vice versa, one for each such $e$, so that parallel edges are matched with the correct multiplicity. Therefore, $\sigma$ preserves the adjacency relations in the bridged graph, making it a graph automorphism.
    \end{proof}
\end{lemma}

Hence, in this way, we can ensure that $P^{\mathcal{E}}_{\sigma}$ is a proper permutation matrix whose row- and column-weights are at most $1$. The above corresponds precisely to symmetrization of edges $\mathcal{E} \rightarrow \mathcal{E} \cup \sigma(\mathcal{E})$, where $\sigma(\mathcal{E}) = \{ \sigma(e), e \in \mathcal{E}\}$. This comes at the expense of doubling the graph vertex degree. For a constant-degree graph, this is asymptotically negligible, though in practice, this might pose problems. 

\begin{remark}
    A more problematic aspect of this construction is that this would only ensure a symmetrization of a single permutation $\sigma$. In practice, such as in extractor architecture, we wish to symmetrize for all possible permutations in a symmetric group $S^{\oplus l}_r$, as we wish to \emph{fix} the ancilla graph for different Pauli-product measurements. Let $\sigma$ be a general permutation and let $T = \mathrm{ord}(\sigma)$ be the least integer with $\sigma^T=1$. The symmetrization process takes 
    \begin{align}\label{expr: symmetrizing-edges}
        \mathcal{E} \rightarrow \mathcal{E} \cup \sigma(\mathcal{E}) \cup \cdots \cup \sigma^{T-1}(\mathcal{E}),
    \end{align}
    where the resulting graph vertex degree might increase by $T \Delta$. 
\end{remark}

The systematic treatment of bridges above allows us to state formally the surgery gadget desiderata (informally anticipated by Definition~\ref{def: main-surgery-gadgets-desiderata}), which we will use for the following sections. 

\begin{definition}[Formal surgery gadget(s) desiderata]\label{def: formal-surgery-gadgets-desiderata}
    Let $\mathfrak{X} = \{ \mathcal{S}^{(i)}[\mathcal{G}^{(i)}; (\Phi^{(i)}_1, \Phi^{(i)}_0)] \}_{i}$ be a set of graph-surgery gadget(s) to the data (CSS) code $\mathcal{Q}$.  We say that the graph-surgery gadget(s) satisfies graph-surgery gadget(s) desiderata with added degree $\delta$ and soundness $\rho_d$ if:
    \begin{enumerate}[label=(\roman*)]
        \item \label{graph-desiderata-measurement} Their connectivity map satisfies Eq.~\eqref{eq: main-commutation_requirement}, are elementary according to Definition~\ref{def: elementary-connectivity-surgery}, and, for each $i$, $\Phi^{(i)}_1(\ker \partial^{(i)}_1) = \{0, L^{(i)}\}$, in symplectic notation, with $L^{(i)}$ a representative of some nontrivial logical operator $\bar{L}^{(i)}$ of $\mathcal{Q}$.

        \item \label{graph-desiderata-soundness} The graph edge-vertex incidence matrix of each surgery gadget in $\mathfrak{X}$ is $(d, \rho_d)$-sound, with $\rho_d \geq \max_i \max(\omega_{\mathrm{col}}(\Phi^{(i,X)}_1), \omega_{\mathrm{col}}(\Phi^{(i,Z)}_1))$, over the defined components. 
       
        \item \label{graph-desiderata-added-degree} There exists some constant $\delta$ such that any graph-surgery gadget from $\mathfrak{X}$ has a merged code (with data code $\mathcal{Q}$) with added degree at most $\delta$. 
       
        \item \label{graph-desiderata-bridging} The generalized bridge to $\mathcal{S}^{(i)}[\mathcal{G}^{(i)}; (\Phi^{(i)}_1, \Phi^{(i)}_0)]$ according to Definition~\ref{def: generalized-bridging-technique-graph}, or any subset of $\mathfrak{X}$ by grouping their respective disjoint graphs according to Definition~\ref{def: generalized-bridging-technique-graph-disconnected}, results in an updated graph chain that is connected, exact, whose edge-vertex incidence matrix is at least $(d, \rho_d)$-sound, and whose maximum cycle-basis congestion and length ($\omega_{\mathrm{col}}(\partial_0)$ and $\omega_{\mathrm{row}}(\partial_0)$, respectively) remain unchanged after bridging.
    \end{enumerate}
\end{definition}

\begin{remark}
    Technically, we cannot directly assume property~\ref{graph-desiderata-bridging}. We assume this for practical convenience and refer to a more rigorous treatment in Ref.~\cite{he2025extractorsqldpcarchitecturesefficient}.
\end{remark}

In Section~\ref{sec: FT-surgery-gadgets}, we will put the above formalism into action, whereupon we construct seed surgery gadgets $\mathfrak{X}_{\mathrm{seed}}$ and a canonical extractor for measuring an arbitrary logical operator.

\subsection{Fault-tolerant graph-surgery gadgets}\label{sec: FT-surgery-gadgets}

Here we provide a systematic application of the above formalisms in constructing bridged surgery gadgets for measuring arbitrary logical operators.

\subsubsection*{Bridged seed surgery gadgets}
Let $\mathfrak{X}_{\mathrm{seed}}$ be a set of seed surgery gadgets according to Definition~\ref{def: main-seed-logical-operators-gadget}. In this section, we show that the bridged seed surgery gadgets can measure an arbitrary logical operator, provided that $\mathfrak{X}_{\mathrm{seed}}$ satisfies the surgery gadgets desiderata Definition~\ref{def: formal-surgery-gadgets-desiderata} with $(d, \rho_d)$-soundness and added degree $\delta$, where $d$ is the distance of the underlying data code $\mathcal{Q}$. The rough picture is presented as follows: 

\begin{itemize}
    \item ($XX$- and $ZZ$-type measurements). This is the standard bridging technique discussed in Lemma~\ref{lemma: bridging-disjoint-graph-sound}. Since we assume the surgery gadgets desiderata Definition~\ref{def: formal-surgery-gadgets-desiderata}, this is guaranteed to be $(d, \rho_d)$-sound and has added degree $\delta$.

    \item ($Y$-type measurements). As we will see later, we will use the generalized bridge for disconnected graphs Definition~\ref{def: generalized-bridging-technique-graph-disconnected} and Lemma~\ref{lemma: graph-isomorphism} (or Lemma~\ref{lemma: one-way-intertwining}) for such Pauli-product measurements. 
    
    \item (Mixed $XZ$-type measurements). For the canonical basis, any two opposite-type basis operators with distinct indices have disjoint supports; see Ref.~\cite{gottesman1997stabilizercodesquantumerror}. For other bases, the overlaps have a physical operator of $Y$-type, which can be measured using the above method.
\end{itemize}

Furthermore, we will need a practical assumption on the graphs constructed for the seed surgery gadgets.

\begin{enumerate}
    \item \label{seed-logical-assumption-graphify}(Cycle-graphification). Let $\bar{L}^{(s)}_Z$ be the $Z$-type (resp. $X$-type) logical basis operator for the $s$-th seed surgery gadget $\mathcal{S}^{(s)}[\mathcal{G}^{(s)}; (\Phi^{(s)}_1, \Phi^{(s)}_0)]$. Then the graph chain
    \begin{align}
    \mathcal{G}^{(s)}: \mathbb{F}_2[\mathcal{V}^{(s)}] \xrightarrow{\partial^{(s)}_1} \mathbb{F}_2[\mathcal{E}^{(s)}] \xrightarrow{\partial^{(s)}_0} \mathbb{F}_2[\mathcal{C}^{(s)}]
    \end{align}
    with $\partial^{(s)}_1$ and the connectivity map $\Phi^{(s)}_1$ and $\Phi^{(s)}_0$ are given by the cycle-graphification~\eqref{expr: cycle-graphification-parity-check} of $H_X$ (resp.\ $H_Z$) according to Lemma~\ref{lemma: graphify-parity-check}, with 
    \begin{align}
      \begin{gathered}
        H_X M^{(L_Z)}_1 = M^{(L_Z)}_0 \partial^{(s)}_1 \\
        (\text{resp.\ } H_Z M^{(L_X)}_1 = M^{(L_X)}_0 \partial^{(s)}_1),
      \end{gathered}
    \end{align}
    where $\Phi^{(Z_s)}_1 = M^{(L_Z)}_1$ and $\Phi^{(Z_s)}_0 = M^{(L_Z)}_0$ (resp.\ $\Phi^{(X_s)}_1 = M^{(L_X)}_1$ and $\Phi^{(X_s)}_0 = M^{(L_X)}_0$).

    \item (Degree of connectivity map). The overall connectivity maps satisfy $\omega_{\mathrm{row}}(\Phi^{(Z)}_1) = \omega_{\mathrm{col}}(\Phi^{(Z)}_1) = 1$ (resp. $\omega_{\mathrm{row}}(\Phi^{(X)}_1) = \omega_{\mathrm{col}}(\Phi^{(X)}_1) = 1$), ensured by the construction below.

    \item \label{seed-logical-assumption-degree-1-transforming}(Transforming maps). The transforming maps $(T^{(i, s)}_1, T^{(i, s)}_0)$ for transforming to a $Z$-type basis operator $\bar{Z}^{(i)}$ (resp. $X$-type basis operator $\bar{X}^{(i)}$) for all $i$ are degree-$1$, and $T^{(i, s)}_1 \Phi^{(Z_s)}_1 = M^{(Z_i)}_1$ and $T^{(i, s)}_0 \Phi^{(Z_s)}_0 = M^{(Z_i)}_0$ (resp. $T^{(i, s)}_1 \Phi^{(X_s)}_1 = M^{(X_i)}_1$ and $T^{(i, s)}_0 \Phi^{(X_s)}_0 = M^{(X_i)}_0$).
\end{enumerate}

\begin{theorem}[Restatement of Proposition~\ref{prop: bridge_get_full_Clifford-informal}]\label{thm: bridge_get_full_Clifford}
     Let $\mathcal{Q}$ be a (CSS) code with distance $d$, with the set of seed logical operators in Eq.~\eqref{eq: main-seed-logical-LPAA-canonical}, and let $\mathfrak{X}_{\mathrm{seed}}$ be a set of seed surgery gadgets which satisfy the surgery desiderata Definition~\ref{def: formal-surgery-gadgets-desiderata}, with (small-set) soundness $\rho_d$ and added degree $\delta$, which satisfy the assumptions \ref{seed-logical-assumption-graphify} and \ref{seed-logical-assumption-degree-1-transforming}. Let $|\mathcal{V}|$, $|\mathcal{E}|$, and $|\mathcal{C}|$ denote the maximum numbers of vertices, edges, and cycles, respectively, among all seed surgery gadgets, and let $|\mathcal{E}_b|$ and $|\mathcal{C}_b|$ denote the maximum numbers of edges and cycles, respectively, over all bridges between pairs of seed gadgets. Let $J^{(Z)}$ and $J^{(X)}$ be arbitrary sets of indices of the $Z$-type (resp.\ $X$-type) logical basis, and denote the corresponding logical operator 
   \begin{align}\label{eq: app-generic-logical-operator}
       \bar{L} = \prod_{i \in J^{(Z)}} \bar{Z}_{i} \prod_{j \in J^{(X)}} \bar{X}_{j},
   \end{align}
    where $\mathrm{b}(\bar{L}) = |J^{(Z)}| + |J^{(X)}|$. Then $\bar{L}$ can be measured using at most $\mathrm{b}(\bar{L})$ bridged seed surgery gadgets, with size at most $\mathrm{b}(\bar{L}) \bigl( |\mathcal{V}| + |\mathcal{E}| + |\mathcal{C}| \bigr) + \bigl( \mathrm{b}(\bar{L}) - 1 \bigr)(|\mathcal{E}_b| + |\mathcal{C}_b|)$. Hence, the resulting surgery gadget satisfies the surgery gadgets desiderata Definition~\ref{def: formal-surgery-gadgets-desiderata} with $(d, \rho_d)$-soundness and added degree at most
    \begin{align}
          \begin{aligned}
            \delta + \min\Bigl( & \max\bigl( \min(s, \omega_{\mathrm{row}}(H_X)), \omega_{\mathrm{col}}(H_X) \bigr), \\
            & \max\bigl( \min(s, \omega_{\mathrm{row}}(H_Z)), \omega_{\mathrm{col}}(H_Z) \bigr) \Bigr),
          \end{aligned}
    \end{align}
    where $s$ is the number of physical qubits for which the $Z$-type and $X$-type logical operators have nontrivial overlap, and $\omega_{\mathrm{row}}(\cdot)$ and $\omega_{\mathrm{col}}(\cdot)$ are the maximum row and column weights of the parity-check matrices of the underlying CSS code $\mathcal{Q}$; the two branches correspond to rewiring the $Z$-type or the $X$-type connectivity maps, respectively, and we take whichever side yields the smaller bound. When $H_X$ and $H_Z$ have the same maximum row and column weights, as for the symmetric canonical LP codes considered in this work, the bound simplifies to
    \begin{align}
        \delta + \max\bigl( \min(s, \omega_{\mathrm{row}}(H_X)), \omega_{\mathrm{col}}(H_X) \bigr).
    \end{align}
    \begin{proof}
        \textbf{Step 1 (Bridging same-types)}. For a generic logical operator given in Eq.~\eqref{eq: app-generic-logical-operator}, we can use $|J^{(Z)}|$-many surgery gadgets constructed and transformed from seed surgery gadgets to measure $\prod_{i \in J^{(Z)}} \bar{Z}_{i}$, and $|J^{(X)}|$-many surgery gadgets to measure $\prod_{j \in J^{(X)}} \bar{X}_{j}$. This involves the surgery gadgets, respectively, $\mathcal{S}[\mathcal{G}^{(Z)}; (\Phi^{(Z)}_1, \Phi^{(Z)}_0)]$ and $\mathcal{S}[\mathcal{G}^{(X)}; (\Phi^{(X)}_1, \Phi^{(X)}_0)]$ such that we can construct 
        \begin{align}
          \begin{aligned}
            \Phi^{(Z)}_1 &= \veebar_{i\in J^{(Z)}} \Phi^{(Z_i)}_1, \\
            \Phi^{(X)}_1 &= \veebar_{j\in J^{(X)}} \Phi^{(X_j)}_1,
          \end{aligned}
        \end{align}
    where the wedge product refers to the symmetric difference: if $\Phi^{(Z_i)}(v_i) = \Phi^{(Z_j)}(v_j)$, then $ \Phi^{(Z)}_1(v_i) = \Phi^{(Z)}_1(v_j)=0$. Note that this ensures that $\omega_{\mathrm{row}}(\Phi^{(Z)}_1)=1$ and $\omega_{\mathrm{row}}(\Phi^{(X)}_1)=1$. Furthermore, since each constituent $\Phi^{(Z_i)}_1$ (resp. $\Phi^{(X_j)}_1$) has domains that are vertices of disjoint graphs, the column weights are $1$, i.e., $\omega_{\mathrm{col}}(\Phi^{(Z)}_1)=1$ and $\omega_{\mathrm{col}}(\Phi^{(X)}_1)=1$, since by assumption~\ref{seed-logical-assumption-graphify}, the constituent connectivity maps $\Phi^{(Z_i)}_1$ and $\Phi^{(X_j)}_1$ are transversal (restriction maps). In other words, $\Phi^{(Z)}_1$ (resp. $\Phi^{(X)}_1$) is the restriction map for measuring $\prod_{i \in J^{(Z)}} \bar{Z}_{i}$ (resp. $\prod_{j \in J^{(X)}} \bar{X}_{j}$). In this case, by \ref{seed-logical-assumption-degree-1-transforming}, the $\Phi^{(Z)}_0$ (resp. $\Phi^{(X)}_0$) can be constructed via Lemma~\ref{lemma: graphify-parity-check}, where as a caveat, we assume that the added bridge edges according to Definition~\ref{def: generalized-bridging-technique-graph-disconnected} ensure that checks (row vectors of parity-check matrices) are graphified as cycles according to expression~\eqref{expr: cycle-graphification-parity-check}, which ensures that the row- and column-weight of $\Phi^{(Z)}_0$ (resp. $\Phi^{(X)}_0$) are limited by these of the underlying parity-check matrix $H_X$ (resp. $H_Z$); in particular, the row weights are at most $\lfloor \omega_{\mathrm{row}}(H_X)/2 \rfloor$ (resp. $\lfloor \omega_{\mathrm{row}}(H_Z)/2 \rfloor$). The added bridge edges also affect the graph itself: by property~\ref{graph-desiderata-bridging} of Definition~\ref{def: formal-surgery-gadgets-desiderata}, the cycle row and column weights are unchanged by the bridging, and the graph vertex degree increases by at most $1$, since each vertex carries at most one bridge edge. As another caveat, we choose to ignore this degree increment in the accounting below, here and again in Step~2. Hence, we say that the separate $Z$-type and $X$-type surgery gadgets, $\mathcal{S}[\mathcal{G}^{(Z)}; (\Phi^{(Z)}_1, \Phi^{(Z)}_0)]$ and $\mathcal{S}[\mathcal{G}^{(X)}; (\Phi^{(X)}_1, \Phi^{(X)}_0)]$ are $(d, \rho_d)$-sound and have added degree $\delta$.

        \textbf{Step 2 (Bridging opposite-types)}. Here, we utilize the generalized bridging techniques. By construction, $\Phi^{(X)}_1, \Phi^{(Z)}_1$ are restriction maps to physical representatives of $\prod_{i \in J^{(Z)}} \bar{Z}_i$ and $\prod_{j \in J^{(X)}} \bar{X}_j$, respectively. Suppose that $Z$-type and $X$-type logical operators have nontrivial overlap, i.e., $\supp \prod_{i \in J^{(Z)}} \bar{Z}_i \cap \supp \prod_{j \in J^{(X)}} \bar{X}_j \neq \emptyset$. Since $\Phi^{(X)}_1$ and $\Phi^{(Z)}_1$ are elementary (restriction maps), there exist equal-size subsets of vertices $\gamma_X \subseteq \mathcal{V}_X$ and $\gamma_Z \subseteq \mathcal{V}_Z$, enumerating \emph{all} of the overlap, such that, for $m \in [s]$ with $s = |\gamma_X| = |\gamma_Z|$, 
        \begin{align}
            \Phi^{(X)}_1(v_{m, 1}) = \Phi^{(Z)}_1(v_{m, 0})
        \end{align}
        where, consistently with the labeling of Theorem~\ref{thm: bridged-hypergraph-parallel-surgery}, we write $\gamma_Z = \{ v_{m, 0};\ m \in [s] \}$ and $\gamma_X = \{ v_{m, 1};\ m \in [s] \}$. Now, we define the joint graph $\mathcal{G}(\mathcal{V}, \mathcal{E}, \mathcal{C})$: 
        \begin{align}
            \mathcal{V} = \mathcal{V}_X \sqcup \mathcal{V}_Z, \quad \mathcal{E} = \mathcal{E}_X \sqcup \mathcal{E}_Z \sqcup \mathcal{E}_b , \quad \mathcal{C} = \mathcal{C}_X \sqcup \mathcal{C}_Z \sqcup \mathcal{C}_b
        \end{align}
        according to Definition~\ref{def: generalized-bridging-technique-graph-disconnected} so that 
        \begin{align}\label{eq: seed-bridged-edges}
          \begin{aligned}
            \mathcal{E}_b = \bigl\{e^b_m:\  & \mathrm{v}(e^b_m) = \{ \mathrm{v}_0(e^b_m), \mathrm{v}_1(e^b_m) \} \\
            & = \{ v_{m, 0}, v_{m, 1} \} \bigr\}_{m \in [s]},
          \end{aligned}
        \end{align}
        with bridge checks as in Eq.~\eqref{eq: bridged-disjoint-graphs-cycles}, whose cycle lengths and congestion are controlled by property~\ref{graph-desiderata-bridging} of Definition~\ref{def: formal-surgery-gadgets-desiderata}, and it is safe to assume that $|\mathcal{E}_b | \geq d$ (as we can always add more bridge edges where the permutation does not act). Hence, the joint chain $\mathcal{G}: \mathbb{F}_2[\mathcal{V}] \xrightarrow{\partial_1} \mathbb{F}_2[\mathcal{E}] \xrightarrow{\partial_0} \mathbb{F}_2[\mathcal{C}]$ is connected, exact, and $(d, \rho_d)$-sound by Lemma~\ref{lemma: bridging-disjoint-graph-sound}. Note that this is the case for $S^{\oplus s}_2$ and we use graph isomorphism for simplicity and keep track of the necessary adaptations for one-way permutation. Define the transpositions $\sigma = \prod_{m \in [s]} (v_{m, 0} \; v_{m, 1})$. Let the connectivity map be given by 
        \begin{align}
           \begin{aligned}
             \Phi_1 &= \begin{pmatrix} \Phi^{(X)}_1 P^{\mathcal{V}}_{\sigma} \\ \Phi^{(Z)}_1 \end{pmatrix}, \\
             \Phi_0 &= \begin{pmatrix} \Phi^{(X)}_0 P^{\mathcal{E}}_{\sigma} \\ \Phi^{(Z)}_0 \end{pmatrix}.
           \end{aligned}
        \end{align}Next, we show that these connectivity maps satisfy the commutation relations Eq.~\eqref{eq: main-commutation_requirement}. By construction and the graph isomorphism property Eq.~\eqref{eq: graph-isomorphism-intertwining}, we have that 
        \begin{align}
            H_Z \Phi^{(X)}_1 P^{\mathcal{V}}_\sigma = \Phi^{(X)}_0 \partial_1 P^{\mathcal{V}}_{\sigma} = \bigl(\Phi^{(X)}_0 P^{\mathcal{E}}_\sigma\bigr) \partial_1.
        \end{align}
        Furthermore, by construction, we have that 
        \begin{align}
            \bigl(\Phi^{(X)}_1 P^{\mathcal{V}}_{\sigma}\bigr)^{T} \Phi^{(Z)}_1 
        \end{align}
        is diagonal --- its $(v_{m, 0}, v_{m, 0})$ entry is $\langle \Phi^{(X)}_1(v_{m, 1}), \Phi^{(Z)}_1(v_{m, 0}) \rangle = 1$, and every other entry vanishes since $\gamma$ enumerates all of the overlap --- hence symmetric. Hence, this defines a valid surgery gadget according to Definition~\ref{def: main-surgery-gadgets-desiderata} or Definition~\ref{def: surgery-cone-code}. Since $\mathcal{G}$ is connected, $\ker \partial_1 = \{0, \mathbf{1}_{\mathcal{V}}\}$, and $P^{\mathcal{V}}_{\sigma} \mathbf{1}_{\mathcal{V}} = \mathbf{1}_{\mathcal{V}}$, so that $\Phi_1(\mathbf{1}_{\mathcal{V}}) = l_X + l_Z$ for representatives $l_X, l_Z$ of the two products: by Remark~\ref{remark: mixed-type-surgery}, the gadget measures $\bar{L}_X \bar{L}_Z$ up to a global phase. For a one-way permutation $\pi$ in place of $\sigma$, the same holds with the matrices of Lemma~\ref{lemma: one-way-intertwining}; in particular $J^{\mathcal{V}}_{\pi}\mathbf{1}_{\mathcal{V}} = \mathbf{1}_{\mathcal{V}}$ by its invariance property.

        \textbf{Step 3 (Sparsity).} We now ensure the bridged gadgets are still LDPC; in other words, we need to bound the row- and column-weights of $\Phi^{(X)}_0 P^{\mathcal{E}}_{\sigma}$.

        By Lemma~\ref{lemma: graph-isomorphism} applied on $\mathcal{G}$, a valid choice of $P^{\mathcal{E}}_{\sigma}$ is, rowwise,
        \begin{align}\label{eq: matching-bridge-row-law}
            \begin{aligned}
                &\bigl(P^{\mathcal{E}}_{\sigma}y\bigr)(e) = \sum_{e' \in \mathcal{E}} y(e')\mathrm{1}_{e' \in P_{\sigma^{-1}(e)}}\\
                &= y(e) + \sum_{v \in \{\mathrm{v}_0(e),\, \mathrm{v}_1(e)\} \cap \gamma} y(e_b(v)), \qquad \gamma := \gamma_X \cup \gamma_Z,
            \end{aligned}
        \end{align}
        where $e_b(v)$ denotes the matching bridge edge at $v$ ($e_b(v) = e^b_m$ for $v \in \mathrm{v}(e^b_m)$): the path joining the $\sigma^{-1}$-images of the endpoints. By the construction of bridged edges Eq.~\eqref{eq: seed-bridged-edges}, there are at most three (including $y(e)$) summands in Eq.~\eqref{eq: matching-bridge-row-law}, and this states that $P^{\mathcal{E}}_{\sigma}$ has row weight at most $3$. Hence, the row weights of $\Phi^{(X)}_0 P^{\mathcal{E}}_{\sigma}$ are at most $\lfloor 3\omega_{\mathrm{row}}(H_Z)/2 \rfloor$, since those of $\Phi^{(X)}_0$ are at most $\lfloor \omega_{\mathrm{row}}(H_Z)/2 \rfloor$. Furthermore, a key caveat is that, for a large family of codes such as LP canonical codes, not all edges in rows of $\Phi^{(X)}_0$ are acted on nontrivially by $P^{\mathcal{E}}_{\sigma}$. If each check (row index of $H_Z$) contains only a small constant number of edges acted on by $P^{\mathcal{E}}_{\sigma}$, then the row weights of $\Phi^{(X)}_0 P^{\mathcal{E}}_{\sigma}$ are bounded by a small constant. Furthermore, for measuring low-logical-weight logical operators such as a single $Y$, the row weight is increased at most by $1$. The bridged edges join the two disconnected copies, so only the $s$ vertices in $\gamma_X$ act nontrivially on the rows of $\Phi^{(X)}_0$. Each vertex contributes its bridged edge at most once modulo $2$, so the increase is the number of vertices touched by the row. Hence, the row weight of $\Phi^{(X)}_0 P^{\mathcal{E}}_{\sigma}$ increases by at most: 
        \begin{align}
            \min(s, \omega_{\mathrm{row}}(H_Z)). 
        \end{align}
        It remains to bound the column weight of $\Phi^{(X)}_0 P^{\mathcal{E}}_{\sigma}$. The graphified edges are private to their own check, so the columns of $\Phi^{(X)}_0$ carry weight $1$ and congestion can only arise on the bridge edges. By Eq.~\eqref{eq: matching-bridge-row-law}, a row of $H_Z$ acquires a bridge edge $e_b$ only if $\mathrm{v}_1(e_b) \cap \gamma_X \neq \emptyset$; there are at most $\omega_{\mathrm{col}}(H_Z)$ such rows. Note that the construction is symmetric: rewiring the $Z$-type connectivity maps instead yields the same bounds with $H_X$ in place of $H_Z$, and we take whichever side gives the smaller bound.
        In summary, the resulting bridged surgery gadget satisfies the surgery gadgets desiderata Definition~\ref{def: formal-surgery-gadgets-desiderata} with $(d, \rho_d)$-soundness and added degree at most
\begin{align}
          \begin{aligned}
            \delta + \min\Bigl( & \max\bigl( \min(s, \omega_{\mathrm{row}}(H_X)), \omega_{\mathrm{col}}(H_X) \bigr), \\
            & \max\bigl( \min(s, \omega_{\mathrm{row}}(H_Z)), \omega_{\mathrm{col}}(H_Z) \bigr) \Bigr).
          \end{aligned}
        \end{align}
    \end{proof}
\end{theorem}

\begin{remark}
For $\LP^{3 \times 5}_{33}$ with the canonical logical basis in Ref.~\cite{gottesman1997stabilizercodesquantumerror} and seed gadgets constructed as in Table~\ref{tab: seed-surgery-gadget-resources}, any two-body logical operator can be measured with at most $4$ bridged seed surgery gadgets --- together with the single-body measurements, these generate the full Clifford group --- and the merged degree increases by at most $3$, independently of the lift $l$: each conjugate pair of canonical basis operators overlaps at precisely one physical qubit and is disjoint otherwise, contributing one unit through the rewiring.
\end{remark}

We refer to Table~\ref{tab: tenY-bridged-seed-gadgets} for a concrete numerical example of this theorem: a ten-body $Y$-type Pauli-product measurement on $\LP^{3 \times 5}_{33}$, built from the two seed surgery gadgets of Table~\ref{tab: seed-surgery-gadget-resources}.

\begin{table}[!ht]
\centering
\footnotesize
\setlength{\tabcolsep}{4pt}
\begin{tabular}{@{}l ccc@{}}
\toprule
 & \multicolumn{3}{c}{Merged degree} \\
\cmidrule(l){2-4}
Measured operator
 & \shortstack{Stabilizer\\weight}
 & \shortstack{Qubit\\$X$-degree}
 & \shortstack{Qubit\\$Z$-degree} \\
\midrule
\quad $\bar{L}_Z$
 & $10$\,{\scriptsize$(+2)$} & $7$\,{\scriptsize$(+2)$} & $5$\,{\scriptsize$(+0)$} \\
\addlinespace
\quad $\bar{L}_X$
 & $10$\,{\scriptsize$(+2)$} & $8$\,{\scriptsize$(+3)$} & $5$\,{\scriptsize$(+0)$} \\
\addlinespace
\quad $\bar{L}_Y = \bar{L}_X \bar{L}_Z$
 & $11$\,{\scriptsize$(+3)$} & $8$\,{\scriptsize$(+3)$} & $5$\,{\scriptsize$(+0)$} \\
\bottomrule
\end{tabular}
\caption{\textbf{Bridged seed surgery gadgets: a ten-body $Y$-type measurement on $\LP^{3 \times 5}_{33}$.} A concrete numerical example of the measurement of $\bar{L}_{Y} = \prod_{m=0}^{4} \bar{Y}_{0,1,m}  \prod_{m'=28}^{32} \bar{Y}_{1,0,m'}$ up to a phase, per Theorem~\ref{thm: bridge_get_full_Clifford}, built from the two seed surgery gadgets of Table~\ref{tab: seed-surgery-gadget-resources} through degree-$1$ transforming maps. In Step~1, the ten transported gadgets of each type are chained by generalized bridges, achieving $\omega_{\mathrm{row}}(\Phi_0) = 2$, giving the increase of stabilizer weight $2$ (refer to the merged parity-check matrix given in Eq.~\eqref{eq: main-merged-parity-check}). In Step~2, the two chains are bridged by a one-way permutation $\pi$ according to Definition~\ref{def: one-way-permutation-graph} on the $s = 10$ conjugate-pair qubits, one bridge edge per pair, which rewires the $X$-type connectivity maps and increases their row weights by at most $1$, with at least $d-10$ additional bridged edges untouched by $\pi$ ($13$ in this instance, so that the two chains are joined by at least $d = 20$ bridge edges), so the (small-set) soundness $\rho_d \geq 1$ of the certified seeds is inherited by Lemma~\ref{lemma: bridging-disjoint-graph-sound}. Parenthetical entries indicate the increase over the data code.}
\label{tab: tenY-bridged-seed-gadgets}
\end{table}

\subsubsection*{Linear and canonical extractors}

Next, we study the construction of extractors. 

\begin{definition}[Linear extractor]\label{def: linear-extractor}
    Let $\mathcal{Q}$ be a data (CSS) code, and let $\mathcal{G}: \mathbb{F}_2[\mathcal{V}] \xrightarrow{\partial_1} \mathbb{F}_2[\mathcal{E}] \xrightarrow{\partial_0} \mathbb{F}_2[\mathcal{C}]$ be a graph chain. Let $\{\Phi^{(\mu)}_1, \Phi^{(\mu)}_0\}_{\mu}$ be a set of connectivity maps. We say that $\mathcal{X}[\mathcal{G}; \{\Phi^{(\mu)}_1, \Phi^{(\mu)}_0\}_{\mu}]$ is a linear extractor if the following conditions hold. 
    \begin{enumerate}
        \item For each index $\mu$, the connectivity map satisfies the commutation relation Eq.~\eqref{eq: main-commutation_requirement}, and the surgery gadget $\mathcal{S}[\mathcal{G}; (\Phi^{(\mu)}_1, \Phi^{(\mu)}_0)]$ is a graph-surgery gadget according to Definition~\ref{def: surgery-cone-code}. 
        \item  Any $\mathbb{F}_2$-linear combination of the pairs $(\Phi^{(\mu)}_1, \Phi^{(\mu)}_0)$, taken with the same coefficients on both families, defines a graph-surgery gadget according to Definition~\ref{def: surgery-cone-code}.
    \end{enumerate}
\end{definition}

In the next section, we show how the extractor is used: the notion encompasses seed gadgets and LP canonical extractors. We discuss two types of linear extractors: $(i)$ $Z$-type (resp. $X$-type) linear extractor, which measures any $Z$-type (resp. $X$-type) logical operators, and $(ii)$ Mixed-type linear extractor, which measures some selective patterns of mixed-type logical operators.

\begin{lemma}[$Z$-type linear extractor]\label{lemma: Z-type-linear-extractor}
    Let $\mathcal{X}[\mathcal{G}; \{\Phi^{(Z_\mu)}_1, \Phi^{(Z_\mu)}_0\}_{\mu}]$ be a linear extractor according to Definition~\ref{def: linear-extractor} on a \emph{connected} graph chain $\mathcal{G}$, indexed by a $Z$-type logical basis $\mathfrak{S}^{(Z)}$ such that the $\mu$-th gadget measures $\bar{Z}_{\mu}$, i.e., $\Phi^{(Z_\mu)}_1(\mathbf{1}_{\mathcal{V}}) = z_{\mu}$ for a binary representative $z_{\mu}$ of $\bar{Z}_{\mu}$. Then any $Z$-type logical operator $\bar{L}_Z$ can be measured using a linear combination of the connectivity maps $\{\Phi^{(Z_\mu)}_1, \Phi^{(Z_\mu)}_0\}_{\mu}$, which defines a valid graph-surgery gadget according to Definition~\ref{def: surgery-cone-code}. 
    \begin{proof}
        Let a generic $Z$-type logical operator $\bar{L}_Z$ be given by
        \begin{align}
            \bar{L}_Z = \prod_{\mu \in J} \bar{Z}_{\mu}
        \end{align}
        as a product of $Z$-type logical basis operators $\bar{Z}_{\mu}$, where $J$ is a subset of indices of the logical basis. Working with the binary supports, with $\supp \bar{Z}_\mu = z_{\mu}$ and $\supp \bar{L}_Z = l_Z$, then $\ker \partial_1 = \{0, \mathbf{1}_{\mathcal{V}}\}$ for a connected graph chain and $\Phi^{(Z_\mu)}_1(\mathbf{1}_{\mathcal{V}}) = z_{\mu}$, so that
        \begin{align}
            l_Z = \sum_{\mu \in J} z_{\mu} = \sum_{\mu \in J} \Phi^{(Z_\mu)}_1(\mathbf{1}_{\mathcal{V}}) = \Bigl( \sum_{\mu \in J} \Phi^{(Z_\mu)}_1 \Bigr)(\mathbf{1}_{\mathcal{V}}).
        \end{align}
        Note that this defines a graph-surgery gadget according to Definition~\ref{def: surgery-cone-code} since the linear combination of connectivity maps satisfies the commutation relations Eq.~\eqref{eq: main-commutation_requirement}.
    \end{proof}
\end{lemma}

\begin{remark}
    Note that, stated in generality, a linear combination of connectivity maps $\{\Phi^{(\mu)}_1, \Phi^{(\mu)}_0\}_{\mu}$ may increase the weights of connectivity maps, especially for the column weights of $\Phi_1$ and $\Phi_0$. 
\end{remark}

A simple construction of a linear extractor which ensures the low-degree and distance-preserving properties utilizes the graphification of a parity-check matrix:

\begin{lemma}\label{lemma: graphify-parity-check-sum}
    Let $H \in \mathbb{F}^{m_H \times n_H}_2$ and the graph chain $\mathcal{G}: \mathbb{F}_2[\mathcal{V}] \xrightarrow{\partial_1} \mathbb{F}_2[\mathcal{E}]$ be constructed from $H$ according to Definition~\ref{def: graphify-parity-check}. Here the maps $M^{(\mu)}_1, M^{(\mu)}_0$ play the role of the connectivity maps $\Phi^{(Z_\mu)}_1, \Phi^{(Z_\mu)}_0$ of the graphified construction. Let $\{M^{(\mu)}_1, M^{(\mu)}_0\}$ be a tuple of maps, with $M^{(\mu)}_1$ a restriction map, such that $H M^{(\mu)}_1 = M^{(\mu)}_0 \partial_1$ as defined in Lemma~\ref{lemma: graphify-parity-check}. Then, for any index set $J$, 
    \begin{align}\label{eq: graphify-sum-chain}
       H \sum_{\mu \in J} M^{(\mu)}_1  = \sum_{\mu \in J} M^{(\mu)}_0 \partial_1.
    \end{align}
    Denote $M^{(J)}_1 := \sum_{\mu \in J} M^{(\mu)}_1$ and $M^{(J)}_0 := \sum_{\mu \in J} M^{(\mu)}_0$. Then there exists a choice of $M^{(J)}_0$ satisfying Eq.~\eqref{eq: graphify-sum-chain} with
    \begin{align}
       \begin{aligned}
         &\omega_{\mathrm{row}}(M^{(J)}_0) \leq  \lfloor \omega_{\mathrm{row}}(H)/2 \rfloor \quad (\text{resp.\ } \leq \omega_{\mathrm{row}}(H)-1), \\
         &\omega_{\mathrm{col}}(M^{(J)}_0) \leq 1,
       \end{aligned}
    \end{align}
    where the parenthesis denotes the path-graphification~\eqref{expr: path-graphification-parity-check}, for which the choice is unique.
    \begin{proof}
        Eq.~\eqref{eq: graphify-sum-chain} is the linearity of matrix multiplication. Fix a row index $a$ of $H$, write $r := |\supp H_a|$, and let its graphification cycle~\eqref{expr: cycle-graphification-parity-check} be
        \begin{align}
            (v_{a_0}, e_{a_1}, v_{a_1}, \ldots, v_{a_{r-1}}, e_{a_r}), \qquad e_{a_t} = \{ v_{a_{t-1}}, v_{a_t} \}, 
        \end{align}
        where we set $v_{a_r} := v_{a_0}$ so that $v_{a_t}$ is incident exactly to $e_{a_t}$ and $e_{a_{t+1}}$, indices modulo $r$. Every $M^{(\mu)}_0$ has its row $a$ supported on $\mathcal{E}_a = \{e_{a_1}, \ldots, e_{a_r}\}$; hence so does $M^{(J)}_0$; write $u_t := (M^{(J)}_0)_{a, e_{a_t}}$ and $b_t := (H M^{(J)}_1)_{a, v_{a_t}}$. Then, for every $x \in \mathbb{F}_2[\mathcal{V}]$,
        \begin{align}
            \begin{aligned}
                (M^{(J)}_0 \partial_1 x)(a) &= \sum_{e \in \mathcal{E}} (M^{(J)}_0)_{ae} \left( x(\mathrm{v}_0(e)) + x(\mathrm{v}_1(e)) \right) \\
                &= \sum^{r}_{t=1} u_t \left( x(v_{a_{t-1}}) + x(v_{a_t}) \right) \\
                &= \sum^{r-1}_{t=0} (u_t + u_{t+1})\, x(v_{a_t}), \qquad u_0 := u_r,
            \end{aligned}
        \end{align}
        by regrouping the sum by vertices, each edge contributing its value to both of its endpoints. As the row $a$ of $H M^{(J)}_1$ is supported on $\supp H_a$ and Eq.~\eqref{eq: graphify-sum-chain} holds for every $x$, comparing the coefficients of $x(v_{a_t})$ gives the cyclic system
        \begin{align}
            u_t + u_{t+1} = b_t, \qquad t = 0, \ldots, r-1. 
        \end{align}
        Then we have that $u_{t+1} = u_{t} + b_t$, which we telescope to $u_{r} = u_0 + \sum_{s < r} b_s$. Note that $u_r = u_0$; it is necessary that $\sum_s b_s = 0$ which holds as every row of $H M^{(J)}_1$ has even Hamming weight. Hence, the solutions are $u_t = \epsilon + \sum_{s < t} b_s$ with $\epsilon \in \mathbb{F}_2$ a free variable. In vector form, this translates to solutions $u, u + \mathbf{1}_{\mathcal{E}_a}$, where $\mathbf{1}_{\mathcal{E}_a}$ is the all-ones vector on $\mathcal{E}_a$. Since these two solution vectors' supports are disjoint and complementary to each other,
        \begin{align}
            |u| + |u + \mathbf{1}_{\mathcal{E}_a}| = r \leq \omega_{\mathrm{row}}(H),
        \end{align}
        so the lighter one has weight at most $\lfloor \omega_{\mathrm{row}}(H)/2 \rfloor$. Note that $\mathbf{1}_{\mathcal{E}_a} \partial_1 = 0$, so this choice preserves Eq.~\eqref{eq: graphify-sum-chain}. For the path-graphification~\eqref{expr: path-graphification-parity-check}, the same regrouping leaves the two end terms $u_1 x(v_{a_0})$ and $u_{r-1} x(v_{a_{r-1}})$ unpaired, so the system is triangular with the unique solution $u_t = \sum_{s<t} b_s$, of weight at most $\omega_{\mathrm{row}}(H) - 1$. Finally, each edge lies on exactly one graphification cycle (resp. path), which proves the column bound.
    \end{proof}
\end{lemma}

\begin{lemma}[Mixed-type linear extractor]\label{lemma: mixed-type-linear-extractor}
    Let $\mathcal{X}[\mathcal{G}; \{\Phi^{(Z_\mu)}_1, \Phi^{(Z_\mu)}_0\}_{\mu}]$ be a linear extractor according to Definition~\ref{def: linear-extractor}, indexed by a $Z$-type logical basis $\mathfrak{S}^{(Z)}$. Suppose further that $\mathcal{Q}$ satisfies the canonical $ZX$-duality in Definition~\ref{def: canonical-ZX-duality}, and let, for every $\mu$, 
    \begin{align}
        \Phi^{(\mu)}_1 := (\Psi_1 \Phi^{(Z_\mu)}_1, \Phi^{(Z_\mu)}_1), \quad \Phi^{(\mu)}_0 := (\Psi_0 \Phi^{(Z_\mu)}_0, \Phi^{(Z_\mu)}_0). 
    \end{align}
    Then $\mathcal{X}[\mathcal{G}; \{\Phi^{(\mu)}_1, \Phi^{(\mu)}_0\}_{\mu}]$ defines a mixed-type linear extractor.
    \begin{proof}
        Following from the preceding proof, it remains to show that the linear combination of connectivity maps $\{\Phi^{(\mu)}_1, \Phi^{(\mu)}_0\}_{\mu}$ satisfies the commutation relations Eq.~\eqref{eq: main-commutation_requirement}. In particular, for any $\mu$ and $\nu$, \begin{align}
            \begin{aligned}
                &(\Phi^{(X_\mu)}_1 + \Phi^{(X_\nu)}_1)^T (\Phi^{(Z_\mu)}_1 + \Phi^{(Z_\nu)}_1) \\
                &= \Phi^{(X_\mu)T}_1 \Phi^{(Z_\mu)}_1 + \Phi^{(X_\mu)T}_1 \Phi^{(Z_\nu)}_1 \\
                &\quad + \Phi^{(X_\nu)T}_1 \Phi^{(Z_\mu)}_1 + \Phi^{(X_\nu)T}_1 \Phi^{(Z_\nu)}_1. 
            \end{aligned}
        \end{align}
        For the cross-term, 
        \begin{align}
            \begin{aligned}
                &\Phi^{(X_\mu)T}_1 \Phi^{(Z_\nu)}_1 + \Phi^{(X_\nu)T}_1 \Phi^{(Z_\mu)}_1 \\
                &= \Phi^{(Z_\mu)T}_1 \Psi_1^T \Phi^{(Z_\nu)}_1 + \Phi^{(Z_\nu)T}_1 \Psi_1^T \Phi^{(Z_\mu)}_1,
            \end{aligned}
        \end{align}
        which is symmetric since $\Psi_1$ is a permutation representation of an involution (Definition~\ref{def: canonical-ZX-duality}), hence $\Psi_1 = \Psi^T_1$.
    \end{proof}
\end{lemma}

Finally, we comment that, once one has constructed two linear extractors measuring $Z$-type logical operators and $X$-type logical operators, one can combine them to form a single linear extractor that measures both types of logical operators simultaneously, which utilizes Steps~2 and~3 of Theorem~\ref{thm: bridge_get_full_Clifford}. 

\begin{definition}[Canonical extractor desiderata]\label{def: canonical-extractor-desiderata}
    Let $\mathcal{Q}$ be a quantum (CSS) code with distance $d$.
    Let $\mathcal{X}[\mathcal{G}; \{\Phi^{(\mu)}_1, \Phi^{(\mu)}_0 \}_{\mu}]$ be a linear extractor. We say that it satisfies the canonical extractor desiderata with soundness $\rho_d$ and added degree $\delta$ if, for every logical operator $\bar{L}$ of the addressed type, there exists an ($\mathbb{F}_2$) linear combination $(\Phi^{(L)}_1, \Phi^{(L)}_0)$ of the connectivity maps $\{(\Phi^{(\mu)}_1, \Phi^{(\mu)}_0)\}_{\mu}$, taken with the same coefficients on both families, such that $\mathcal{S}[\mathcal{G}; (\Phi^{(L)}_1, \Phi^{(L)}_0)]$ is a valid graph-surgery gadget measuring $\bar{L}$ and satisfies the graph-surgery gadget(s) desiderata Definition~\ref{def: formal-surgery-gadgets-desiderata} with soundness $\rho_d$ and added degree $\delta$. 
\end{definition}

A simple corollary realizing a linear extractor satisfying the canonical extractor desiderata is as follows:

\begin{corollary}
    Let $\mathcal{X}[\mathcal{G}_Z; \{(\Phi^{(Z_\mu)}_1, \Phi^{(Z_\mu)}_0)\}_{\mu}]$ and $\mathcal{X}[\mathcal{G}_X; \{(\Phi^{(X_\mu)}_1, \Phi^{(X_\mu)}_0)\}_{\mu}]$ be two linear extractors measuring $Z$-type and $X$-type logical (basis) operators, respectively, whose connectivity maps are given by the cycle-graphification of assumption~\ref{seed-logical-assumption-graphify} (Definition~\ref{def: graphify-parity-check}, Lemma~\ref{lemma: graphify-parity-check-sum}) and are elementary according to Definition~\ref{def: elementary-connectivity-surgery}. Suppose that $\mathcal{X}[\mathcal{G}_Z; \{(\Phi^{(Z_\mu)}_1, \Phi^{(Z_\mu)}_0)\}_{\mu}]$ (resp. $\mathcal{X}[\mathcal{G}_X; \{(\Phi^{(X_\mu)}_1, \Phi^{(X_\mu)}_0)\}_{\mu}]$) satisfies the canonical extractor desiderata Definition~\ref{def: canonical-extractor-desiderata} with soundness $\rho_d$ and added degree $\delta$ for measuring arbitrary $Z$-type (resp.\ $X$-type) logical operators. Then the extractor obtained by bridging $\mathcal{G}_Z \sqcup \mathcal{G}_X$ at the overlap pairs according to Steps~2--3 of Theorem~\ref{thm: bridge_get_full_Clifford}, with $|\mathcal{E}_b| \geq d$, satisfies the canonical extractor desiderata for arbitrary logical operators, with soundness $\rho_d$ and added degree: 
    \begin{align}
         \begin{aligned}
            \delta + \min\Bigl( & \max\bigl(  \omega_{\mathrm{row}}(H_X), \omega_{\mathrm{col}}(H_X) \bigr), \\
            & \max\bigl(  \omega_{\mathrm{row}}(H_Z), \omega_{\mathrm{col}}(H_Z) \bigr) \Bigr).
          \end{aligned}
    \end{align}
    \begin{proof}
        The proof follows from the simple adaptation from Theorem~\ref{thm: bridge_get_full_Clifford}. In particular, we exchange Step~1 of Theorem~\ref{thm: bridge_get_full_Clifford} for measuring an arbitrary same-type logical operator with a linear combination of connectivity maps for the respective basis operators, which ensures sparsity from Lemma~\ref{lemma: graphify-parity-check-sum}, while the rest of the steps follow directly, preserving the soundness and added degree as claimed.
    \end{proof}
\end{corollary}

In the case of product codes with $ZX$-duality, we show in Section~\ref{sec: LP-canonical-extractor} that it suffices to construct a single $Z$-type linear extractor which can be upgraded to satisfy the canonical extractor desiderata with the generalized bridging techniques studied in Section~\ref{sec: generalized-bridging-techniques}.

\subsubsection*{Parallel measurement using a disconnected set of graphs}
Next we discuss the parallel measurement using a disconnected set of graphs.

\begin{lemma}[Restatement of Proposition~\ref{prop: main-parallel-low-rate}]\label{lemma: distance-preserving-graph-parallel}
    Let $\mathcal{Q}$ be a quantum code equipped with some set of graph-surgery gadgets $\{ \mathcal{S}^{(i)}[\mathcal{G}^{(i)}; (\Phi^{(i)}_1, \Phi^{(i)}_0)] \}$ according to Definition~\ref{def: surgery-cone-code} and/or Definition~\ref{def: main-graph-surgery-gadget}, which satisfies the surgery gadgets desiderata Definition~\ref{def: formal-surgery-gadgets-desiderata}, with soundness $\rho_d$ and added degree $\delta$. Denote $\bar{L}^{(i)}$ to be the logical operator measured by $\mathcal{S}^{(i)}[\mathcal{G}^{(i)}; (\Phi^{(i)}_1, \Phi^{(i)}_0)]$, and $\mathcal{L} := \{ \bar{L}^{(i)}\}$ to be the set of (commuting) logical operators measured by the surgery gadgets. Let $\mathcal{S}[\mathcal{G}; (\Phi_1, \Phi_0)]$ be defined as a surgery gadget with a disjoint union of the individual graphs, i.e.,
   \begin{equation}\label{eq: app-union-of-graphs}
  \begin{aligned}
      \mc{G} & = \sqcup_i \mc{G}^{(i)}; \\
      \Phi_1 & = \mathrm{hstack}\{\Phi_1^{(i)}\};\\
      \Phi_0 &= \mathrm{hstack}\{\Phi_0^{(i)}\}. 
  \end{aligned} 
  \end{equation}
  Let $s = \max(\omega_{\mathrm{row}}(\Phi_0), \omega_{\mathrm{row}}(\Phi_1))$. Then $\mathcal{S}[\mathcal{G}; (\Phi_1, \Phi_0)]$ is a graph-surgery gadget according to Definition~\ref{def: surgery-cone-code}, which measures the logical operators in $\mathcal{L}$ and is $(d, \rho_d)$-sound, with added degree at most $\delta+s$. The bound improves to $\delta$ if the row supports of both $\Phi^{(i)}_1$ and $\Phi^{(i)}_0$ are pairwise disjoint across gadgets:

  \begin{itemize}
    \item (Same-type gadgets). Without loss of generality, assume that each surgery gadget measures a $Z$-type logical operator $\bar{L}^{(i)}$. 
    \item (Symmetric mixed-type gadgets). Suppose in addition the $ZX$-duality Definition~\ref{def: canonical-ZX-duality} for $\mathcal{Q}$. Let 
    \begin{align}
        \Phi^{(i)}_1 = (\Psi_1 \Phi^{(Z_i)}_1, \Phi^{(Z_i)}_1), \quad \Phi^{(i)}_0 = (\Psi_0 \Phi^{(Z_i)}_0, \Phi^{(Z_i)}_0).
    \end{align}
    \item (Non-overlapping gadgets). Suppose that the supports of the logical operators $\bar{L}^{(i)}$ are pairwise disjoint.
  \end{itemize}
  \begin{proof}
    It suffices to prove the symmetric mixed case; the other cases follow directly. For the soundness, $|\partial_1 u| = \sum_i |\partial^{(i)}_1 u_i|$ and $\mathrm{dist}(u, \ker \partial_1) \leq \sum_i \mathrm{dist}(u_i, \ker \partial^{(i)}_1)$ over the disjoint union, so the componentwise soundness bounds add. To see that, we only need to show that the added vertex checks from $\Phi_1$ commute and by commutation relations Eq.~\eqref{eq: main-commutation_requirement}: 
    \begin{align}\label{expr: hstacked-Phi-X-Z-commutation}
        \begin{aligned}
           & \begin{pmatrix}
             \Phi^{(X_\mu) T}_1 \\
              \Phi^{(X_\nu)T}_1
        \end{pmatrix}  \begin{pmatrix}
                \Phi^{(Z_\mu)}_1 & \Phi^{(Z_\nu)}_1
        \end{pmatrix} \\
        &= \begin{pmatrix}
            \Phi^{(X_\mu) T}_1 \Phi^{(Z_\mu)}_1 & \Phi^{(X_\mu) T}_1 \Phi^{(Z_\nu)}_1 \\
            \Phi^{(X_\nu) T}_1 \Phi^{(Z_\mu)}_1 & \Phi^{(X_\nu) T}_1 \Phi^{(Z_\nu)}_1
        \end{pmatrix}
        \end{aligned}
    \end{align}
    Hence, by the $ZX$-duality Definition~\ref{def: canonical-ZX-duality}
    \begin{align}
        \begin{aligned}
            \Phi^{(X_\mu) T}_1 \Phi^{(Z_\nu)}_1 = \Phi^{(Z_\mu)T}_1 \Psi^T_1 \Phi^{(Z_\nu)}_1 = \Phi^{(Z_\mu) T}_1 \Psi_1 \Phi^{(Z_\nu)}_1.
        \end{aligned}
    \end{align}
    The two cross blocks are transposes of one another, and each diagonal block $\Phi^{(Z_\mu)T}_1 \Psi_1 \Phi^{(Z_\mu)}_1$ is symmetric since $\Psi_1$ is; hence expression~\eqref{expr: hstacked-Phi-X-Z-commutation} is symmetric. We can generalize the above for horizontally stacked multiple surgery gadgets.
  \end{proof}
\end{lemma}

\subsection{Parallel surgery for arbitrary addressability}\label{sec: parallel-surgery-arbitrary-addressable}
We now introduce a more general construction of parallel surgery gadgets, which allows for arbitrary addressability of logical operators, not limited only to same-type Pauli-product measurements. Utilizing puncture and augmentation techniques, we can construct a parallel surgery gadget that measures any arbitrary set of logical operators of the same type, as demonstrated in hypergraph-product codes as well as canonical lifted-product (LP) codes studied in this work, particularly in Section~\ref{sec: main-LP-parallel-logic} (intra-column surgery as in Definition~\ref{def: main-LP-intracolumn-hyerpgraph-surgery}). As alluded to earlier, the key to $Y$-type measurements is to introduce the graph isomorphism $(P^{\mathcal{V}}_{\sigma}, P^{\mathcal{E}}_{\sigma})$ with the intertwining property Eq.~\eqref{eq: graph-isomorphism-intertwining}. Here we generalize this notion to hypergraphs, where we perform parallel measurements for arbitrary addressability of logical operators. We mean, for the sake of consistency, by a \emph{hypergraph chain}, denoted $\tilde{\mathcal{G}}: \mathbb{F}_2[\mathcal{V}] \xrightarrow{\partial_1} \mathbb{F}_2[\mathcal{E}] \xrightarrow{\partial_0} \mathbb{F}_2[\mathcal{C}]$, a generic, not necessarily exact, $2$-chain, with the elements of $\mathcal{E}$ the hyperedges. 

\begin{figure*}[!t]
    \centering
    \includegraphics[width=\textwidth]{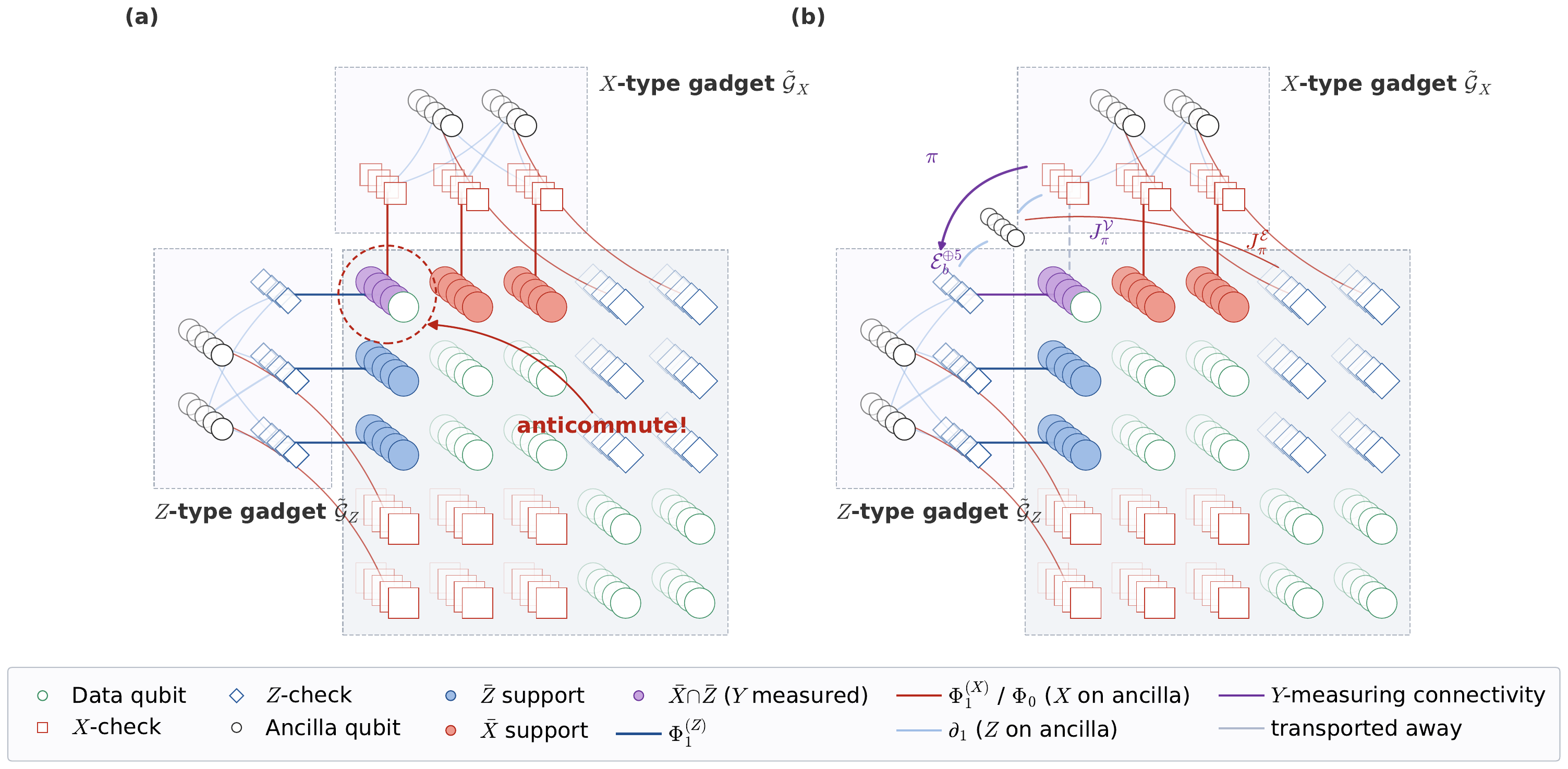}
    \caption{\textbf{Parallel $Y$-type measurements by bridged hypergraph surgery} (Theorem~\ref{thm: bridged-hypergraph-parallel-surgery}), drawn on the toy canonical LP code $\LP^{2\times3}_{5} = [[65, 5]]$, measuring the logically disjoint set $\mathcal{L} = \{\bar{Y}_{0,0,m}\}_{m=0}^{3}$ --- four of the five logical qubits of the single information fibre. \textbf{(a)} The vanilla $Z$- and $X$-type hypergraph-surgery gadgets $\tilde{\mathcal{G}}_Z$ and $\tilde{\mathcal{G}}_X$, satisfying Assumptions~\ref{parallel-canonical-basis} and~\ref{parallel-transversal-phi-1}. The vertex checks of both gadgets act $Z$-type on their incident ancilla edges (blue, $\partial_1$), while the extended data checks act $X$-type on the ancilla qubits (red, $\Phi_0$). Hence each vertex-check (corresponding to a logical qubit of the information fibre) pair $v_{m,0} \in \gamma_Z$, $v_{m,1} \in \gamma_X$ overlaps on exactly one data qubit and anticommutes: the plain disjoint union fails the surgery commutation relations Eqs.~\eqref{eq: main-commutation_requirement}. 
    \textbf{(b)} Bridging as generalized augmentation (Definition~\ref{def: bridging-hypergraph-surgery}): the one-way permutation $\pi\colon v_{m,1} \mapsto v_{m,0}$ (purple arrow $=$ the $Y$-measuring connectivity) Eq.~\eqref{eq: bridged-edges-hypergraph}. The $X$-component of the connectivity map is deformed: $(\Phi^{(X)}_1, \Phi^{(X)}_0) \mapsto (\Phi^{(X)}_1 J^{\mathcal{V}}_{\pi}, \Phi^{(X)}_0J^{\mathcal{E}}_{\pi})$. As a result, the deformed $H_Z$ data-check picks up one bridge copy along its path $P_{\pi(e)}$ giving $\omega_{\mathrm{row}}(\Phi^{(X)}_0 J^{\mathcal{E}}_{\pi}) = 2$. }
    \label{fig: hirate-Y-toy}
\end{figure*}

\begin{lemma}[Existence criterion]\label{lemma: hypergraph-connectivity-existence}
     Let $\tilde{\mathcal{G}}(\mathcal{V}, \mathcal{E})$ be a hypergraph with hyperedges $\mathcal{E}$. Let $\partial_1: \mathbb{F}_2[\mathcal{V}] \to \mathbb{F}_2[\mathcal{E}]$ be the incidence map and let $H \in \mathbb{F}^{m_H \times n_H}_2$ be a matrix. Then the following are equivalent:
        \begin{enumerate}
            \item There exists a tuple $(\Phi_1, \Phi_0)$ such that 
            \begin{align}
                H \Phi_1 = \Phi_0 \partial_1.
            \end{align}
            \item $\Phi_1(\ker \partial_1) \subseteq \ker H$.
        \end{enumerate}
         \begin{proof}
         From $2.$ to $1.$: $\Phi_1(\ker \partial_1) \subseteq \ker H$ gives $\ker \partial_1 \subseteq \ker (H \Phi_1)$. Over a field this is a row-space statement: every row functional of $H\Phi_1$ vanishes on $\ker \partial_1 = (\rs \partial_1)^{\perp}$, hence lies in $\rs \partial_1$. Expressing each row of $H\Phi_1$ as an $\mathbb{F}_2$-linear combination of the rows of $\partial_1$ assembles a matrix $\Phi_0$ with $H\Phi_1 = \Phi_0 \partial_1$. From $1.$ to $2.$: for $g \in \ker \partial_1$, $H \Phi_1 g = \Phi_0 \partial_1 g = 0$.
    \end{proof}
\end{lemma}

\begin{remark}\label{rem: hypergraph-kernel-invariance}
    Let $\tilde{\mathcal{G}}(\mathcal{V}, \mathcal{E})$ be a hypergraph with hyperedges $\mathcal{E}$. Let $\partial_1: \mathbb{F}_2[\mathcal{V}] \to \mathbb{F}_2[\mathcal{E}]$ be the incidence map. Let $\sigma$ be a permutation of the vertices $\mathcal{V}$, and let $P^{\mathcal{V}}_{\sigma}$ be the corresponding permutation matrix. We have that $P^{\mathcal{V}}_{\sigma}(\ker \partial_1) = \ker \partial_1 $ if and only if there exists $P^{\mathcal{E}}_{\sigma}$ such that
 \begin{align}\label{eq: hypergraph-intertwining}
            \partial_1 P^{\mathcal{V}}_{\sigma} = P^{\mathcal{E}}_{\sigma} \partial_1.
\end{align}
\end{remark}

\begin{remark}
    In graph isomorphism or one-way permutation, $P^{\mathcal{V}}_{\sigma}(\ker \partial_1) = \ker \partial_1$ or $J^{\mathcal{V}}_{\sigma}(\ker \partial_1) = \ker \partial_1$ is automatic (for one-way permutations, on connected graph chains, by the invariance property of Lemma~\ref{lemma: one-way-intertwining}). In general, for hypergraphs, it is necessary to require such a condition, as we need to ensure the correct logical action. Interestingly, this condition is satisfied whenever the intertwining property Eq.~\eqref{eq: hypergraph-intertwining} holds, which promises a systematic approach for $Y$-type measurements. 
\end{remark}

Here, we aim to measure an arbitrary set of logical operators $\mathcal{L} = \{ \bar{L}^{(i)}\}$ which are logically disjoint under the following assumptions. 

\begin{enumerate}
    \item \label{parallel-canonical-basis} We assume that we are working with a canonical logical basis set (see Ref.~\cite{gottesman1997stabilizercodesquantumerror}{, Section~4, standard form}; for canonical LP codes this is the canonical basis of Section~\ref{sub: main-canonical-LP-codes}). In particular,
    \begin{enumerate}
        \item $\bar{Z}_{i, j, m}$ and $\bar{X}_{i', j', m'}$ have physical overlap if and only if $i =i', j=j', m=m'$. Hence, they may only have overlap if they are conjugate partners, in which case they overlap at precisely one physical qubit.
        \item Each $\bar{Z}_{i, j, m}$ (resp. $\bar{X}_{i, j, m}$) has its information bit(s), on which no other basis element of the same type is supported.
    \end{enumerate} 
    \item \label{parallel-transversal-phi-1} We assume that $\Phi^{(X)}_1, \Phi^{(Z)}_1$ and $\Phi^{(X)}_0, \Phi^{(Z)}_0$ are transversally connected: their row and column weights are at most $1$. 
\end{enumerate}

{In this setting, bridging along with the (one-way) permutation, for measuring such a logically disjoint set of logical operators, is reduced to simply adding augmentation checks for measuring joint logical operators. However, a key distinction emerges when measuring different types of logical operators. More precisely, we can straightforwardly bridge/augment checks for a measured set containing only logically disjoint operators of the same type, e.g., $XX$ or $ZZ$, or of mixed type, e.g., $XZ$, without changing the connectivity maps. This is due to Assumption~\ref{parallel-canonical-basis}. Hence, the nontrivial part arises exclusively when measuring a set of $Y$-type logical operators, where the commutation relation in Eq.~\eqref{eq: main-commutation_requirement}(i) generally fails, as the $X$-component and $Z$-component logical operators overlap even in logically disjoint measurement sets.} See Fig.~\ref{fig: hirate-Y-toy} for an illustration on the toy code $\LP^{2\times3}_{5}$.

We now systematically explain this, starting with an analogous explanation to the graph-surgery case.

\begin{definition}[Hypergraph incidence relation and paths]\label{def: hypergraph-incidence-relation}
      Let $\tilde{\mathcal{G}}(\mathcal{V}, \mathcal{E})$ be a hypergraph with hyperedges $\mathcal{E}$ and incidence map $\partial_1: \mathbb{F}_2[\mathcal{V}] \to \mathbb{F}_2[\mathcal{E}]$, so that $(\partial_1)_{e, v} = 1$ if and only if $v \in \mathrm{v}(e)$, where $\mathrm{v}(e) \subseteq \mathcal{V}$ denotes the set of vertices incident to $e$ by containment. Let $\pi$ be a one-way permutation on $\mathcal{V}$ (Definition~\ref{def: one-way-permutation-graph}). For a hyperedge $e \in \mathcal{E}$, a \emph{$\pi$-path for $e$} is any subset of hyperedges $P_{\pi}(e) \subseteq \mathcal{E}$ that covers, modulo $2$, the $\pi$-images of the vertices of $e$:
      \begin{align}\label{eq: hypergraph-path}
        \sum_{f \in P_{\pi}(e)} \; \sum_{u \in \mathrm{v}(f)} \mathbf{1}_{u} \;=\; \sum_{v \in \mathrm{v}(e)} \mathbf{1}_{\pi(v)},
      \end{align}
      that is, $\partial^T_1 \mathbf{1}_{P_{\pi}(e)} = \sum_{v \in \mathrm{v}(e)} \mathbf{1}_{\pi(v)}$. We set $P_{\pi}(e) = \emptyset$ when the right-hand side vanishes, and call such $e$ \emph{collapsed}.
\end{definition}

Note that no ordering or adjacency relation among the elements of $P_{\pi}(e)$ is necessary: only the boundary condition Eq.~\eqref{eq: hypergraph-path} enters the intertwining argument below, and interior cancellations are automatic over $\mathbb{F}_2$. The graph paths $P_{\sigma^{-1}(e)}$ of Lemma~\ref{lemma: graph-isomorphism} and $P_{\pi(e)}$ of Lemma~\ref{lemma: one-way-intertwining} are precisely the connected realizations of this definition on $2$-edges, so the definition subsumes both; the one-way form is the case used in what follows, which suffices for quite general constructions of hypergraph-surgery gadgets.

\begin{lemma}[Hypergraph one-way intertwining property]\label{lemma: hypergraph-intertwining-one-way}
    Let $\tilde{\mathcal{G}}: \mathbb{F}_2[\mathcal{V}] \xrightarrow{\partial_1} \mathbb{F}_2[\mathcal{E}] \xrightarrow{\partial_0} \mathbb{F}_2[\mathcal{C}]$ be a hypergraph chain and let $\pi$ be a one-way permutation on $\mathcal{V}$. Define
    \begin{align}
        (J^{\mathcal{V}}_{\pi} x)(v) = x(\pi(v)),
        \qquad
        (J^{\mathcal{E}}_{\pi} y)(e) = \sum_{f \in P_{\pi}(e)} y(f).
    \end{align}
    Then a $\pi$-path $P_{\pi}(e)$ exists for every $e \in \mathcal{E}$ if and only if
    \begin{align}\label{expr: hypergraph-kernel-condition}
        J^{\mathcal{V}}_{\pi}(\ker \partial_1) \subseteq \ker \partial_1,
    \end{align}
    and in that case, for any choice of the $\pi$-paths $\{P_{\pi}(e)\}_{e \in \mathcal{E}}$,
    \begin{align}\label{eq: hypergraph-one-way-intertwining}
        \partial_1 J^{\mathcal{V}}_{\pi} = J^{\mathcal{E}}_{\pi} \partial_1.
    \end{align}
    \begin{proof}
        A $\pi$-path for $e$ exists if and only if $\sum_{v \in \mathrm{v}(e)} \mathbf{1}_{\pi(v)} \in \IM \partial^T_1 = (\ker \partial_1)^{\perp}$. For any $g \in \ker \partial_1$,
        \begin{align}
            \langle \sum_{v \in \mathrm{v}(e)} \mathbf{1}_{\pi(v)},\, g \rangle
            = \sum_{v \in \mathrm{v}(e)} g(\pi(v))
            = (\partial_1 J^{\mathcal{V}}_{\pi} g)(e),
        \end{align}
        so $\pi$-paths exist for every $e$ if and only if $\partial_1 J^{\mathcal{V}}_{\pi} g = 0$ for all $g \in \ker \partial_1$, which is Eq.~\eqref{expr: hypergraph-kernel-condition}. Given the $\pi$-paths, row $e$ of $J^{\mathcal{E}}_{\pi}$ is the indicator $\mathbf{1}^T_{P_{\pi}(e)}$, so for every $x \in \mathbb{F}_2[\mathcal{V}]$,
        \begin{align}
           \begin{aligned}
             (J^{\mathcal{E}}_{\pi} \partial_1 x)(e)
           &= \langle \mathbf{1}_{P_{\pi}(e)}, \partial_1 x \rangle
            = \langle \partial^T_1 \mathbf{1}_{P_{\pi}(e)}, x \rangle
            \\
            &= \sum_{v \in \mathrm{v}(e)} x(\pi(v))
            = (\partial_1 J^{\mathcal{V}}_{\pi} x)(e),
           \end{aligned}
        \end{align}
        by Eq.~\eqref{eq: hypergraph-path}, proving Eq.~\eqref{eq: hypergraph-one-way-intertwining}; both sides vanish on collapsed edges.
    \end{proof}
\end{lemma}

We now give a standard construction of hypergraph surgery for an arbitrary CSS code with a canonical logical basis. The hypergraphs are taken from the corresponding parity-check matrices, which we assume have full row rank. 

\begin{remark}
    For any CSS code with a canonical logical basis set as defined in Ref.~\cite{gottesman1997stabilizercodesquantumerror}{ (the standard form of its Section~4, which every CSS code admits)}, we can always construct hypergraph-surgery gadgets $\mathcal{S}[\tilde{\mathcal{G}}_Z; (\Phi^{(Z)}_1, \Phi^{(Z)}_0)]$ and $\mathcal{S}[\tilde{\mathcal{G}}_X; (\Phi^{(X)}_1, \Phi^{(X)}_0)]$ with sparse hypergraphs $\tilde{\mathcal{G}}_X$ and $\tilde{\mathcal{G}}_Z$ satisfying Assumptions~\ref{parallel-canonical-basis} and~\ref{parallel-transversal-phi-1}, for measuring, respectively, arbitrary logically disjoint $Z$-type and $X$-type logical operators.
\end{remark}

Here, we present a simple scheme upgrading hypergraph-surgery gadgets from above for measuring an arbitrary logically disjoint set of logical operators. Let $\mathcal{L}$ be the set of logically disjoint logical operators; then we denote 
\begin{align}\label{eqs: X-Z-component-logicals}
    \begin{aligned}
        \mathcal{L}_X &= \text{the set of $X$-components of the elements of $\mathcal{L}$}, \\
        \mathcal{L}_Z &= \text{the set of $Z$-components of the elements of $\mathcal{L}$},
    \end{aligned}
\end{align}
where each element $\bar{L} \in \mathcal{L}$ factors uniquely as $\bar{L} = \bar{L}_X \bar{L}_Z$ in the canonical basis (trivial components are discarded), so that $\mathcal{L}_X$ and $\mathcal{L}_Z$ are, respectively, two logically disjoint sets of $X$-type and $Z$-type logical operators. Let two, respectively, $Z$-type and $X$-type hypergraph-surgery gadgets $\mathcal{S}[\tilde{\mathcal{G}}_Z; (\Phi^{(Z)}_1, \Phi^{(Z)}_0)]$ and $\mathcal{S}[\tilde{\mathcal{G}}_X; (\Phi^{(X)}_1, \Phi^{(X)}_0)]$ be given with two disjoint hypergraphs $\tilde{\mathcal{G}}_Z$ and $\tilde{\mathcal{G}}_X$, measuring $\mathcal{L}_Z$ and $\mathcal{L}_X$ in Eqs.~\eqref{eqs: X-Z-component-logicals}. Let $\gamma_Z$ and $\gamma_X$ be the respective sets of vertices from $\tilde{\mathcal{G}}_Z$ and $\tilde{\mathcal{G}}_X$ {that measure the $Z$-component and the $X$-component of the same mixed element of $\mathcal{L}$}. {By Assumption~\ref{parallel-canonical-basis}, their supports under $\Phi^{(Z)}_1$ and $\Phi^{(X)}_1$ overlap exactly when that element contains a $\bar{Y}$ factor. Each mixed element contributes one vertex to each set, so} $|\gamma_X| = |\gamma_Z| = l$. Furthermore, denote
\begin{align}
    \gamma_Z := \{ v_{m, 0}; m \in [l] \}; \quad  \gamma_X := \{ v_{m, 1}; m \in [l] \},
\end{align}
so that {$v_{m, 0}$ and $v_{m, 1}$ measure the two components of the same element. If that element contains a $\bar{Y}$ factor,} $\Phi^{(X)}_1(v_{m, 1})$ overlaps precisely only with $\Phi^{(Z)}_1(v_{m, 0})${. Otherwise, $\Phi^{(X)}_1(v_{m, 1})$ overlaps with no $\Phi^{(Z)}_1(v)$}. Let $\pi$ be the one-way permutation acting disjointly on the vertex pairs {whose supports overlap} and fixing every other vertex{. On a pair with disjoint supports, the bridge edges of Definition~\ref{def: bridging-hypergraph-surgery} below are plain augmentation checks, and $\pi$ fixes both vertices. Explicitly,}
\begin{align}\label{expr: hypergraph-one-way-permutation}
    \pi: v_{m, 1} \mapsto v_{m,0}; \quad  m \in [l] {\text{ with } \Phi^{(X)}_1(v_{m, 1}) \cap \Phi^{(Z)}_1(v_{m, 0}) \neq \emptyset}.
\end{align}
Now we describe the construction of bridges in the context of hypergraph surgery. 

\begin{definition}[Bridging as generalized augmentation in hypergraph surgery]\label{def: bridging-hypergraph-surgery}
    The bridged hypergraph $\tilde{\mathcal{G}}$ is given by disjoint union of $\tilde{\mathcal{G}}_Z$ and $\tilde{\mathcal{G}}_X$:
    \begin{align}
        \tilde{\mathcal{G}} := \tilde{\mathcal{G}}_Z \sqcup \tilde{\mathcal{G}}_X,
    \end{align}
    where the disjoint union is understood as the disjoint unions of vertices, hyperedges, and cycles, respectively. Additionally, add bridge edges connecting vertices in $\gamma_Z$ and $\gamma_X$ with the one-way permutation $\pi$~\eqref{expr: hypergraph-one-way-permutation} according to Definition~\ref{def: generalized-bridging-technique-graph-disconnected}, with the bridged edge set of multiplicity $t$
    \begin{multline}\label{eq: bridged-edges-hypergraph}
        \mathcal{E}^{\oplus t}_b := \bigl\{ e^b_{m, i}:
        \mathrm{v}(e^b_{m, i}) = \bigl\{ \mathrm{v}_0(e^b_{m, i}),
        \mathrm{v}_1(e^b_{m, i}) \bigr\} \\
        = \{v_{m, 0}, v_{m, 1}\},\quad
        m \in [l],\quad i \in [t] \bigr\},
    \end{multline}
    that is, each bridge edge is included as $t$ parallel multiedges, together with the corresponding set of bridged cycles $\mathcal{C}_b$ recording the multiedge dependencies. The multiplicity is deliberately taken to be the soundness parameter $t$ of the constituent gadgets: $t$ parallel copies are exactly the threshold at which the added weight-$2$ rows preserve $(t, \rho_t)$-soundness, since any direction violating a bridge pair acquires syndrome weight $t$ and thereby leaves the small-set window $|\partial_1 u| < t$.
\end{definition}

The addition of bridges allows us to conclude that arbitrary sets of logically disjoint logical operators can be measured using hypergraph-surgery gadgets explicitly with ensured fault tolerance and the LDPC property: 

\begin{theorem}\label{thm: bridged-hypergraph-parallel-surgery}
Let two, respectively, $Z$-type and $X$-type hypergraph-surgery gadgets $\mathcal{S}[\tilde{\mathcal{G}}_Z; (\Phi^{(Z)}_1, \Phi^{(Z)}_0)]$ and $\mathcal{S}[\tilde{\mathcal{G}}_X; (\Phi^{(X)}_1, \Phi^{(X)}_0)]$ be given with two disjoint hypergraph chains $\tilde{\mathcal{G}}_Z$ and $\tilde{\mathcal{G}}_X$, measuring $\mathcal{L}_Z$ and $\mathcal{L}_X$ in Eqs.~\eqref{eqs: X-Z-component-logicals}, satisfying Assumptions~\ref{parallel-canonical-basis} and~\ref{parallel-transversal-phi-1}. Then there exist $J^{\mathcal{V}}_{\pi}$ and $J^{\mathcal{E}}_{\pi}$ (Lemma~\ref{lemma: hypergraph-intertwining-one-way}) such that the following hold:
\begin{enumerate}
    \item With respect to the splits $\mathcal{V} = \mathcal{V}_X \sqcup \mathcal{V}_Z$ and $\mathcal{E} = (\mathcal{E}_X \sqcup \mathcal{E}^{\oplus t}_b) \sqcup \mathcal{E}_Z$, the connectivity map $(\Phi_1, \Phi_0)$ is given by
    \begin{align}
        \Phi_1 = \left( \begin{array}{c}
        \Phi^{(X)}_1J^{\mathcal{V}}_{\pi}  \\
         \Phi^{(Z)}_1
        \end{array} \right), \quad
        \Phi_0 = \left( \begin{array}{c}
        \Phi^{(X)}_0J^{\mathcal{E}}_{\pi} \\
        \Phi^{(Z)}_0
        \end{array} \right),
    \end{align}
    and let the hypergraph with bridging be given by Definition~\ref{def: bridging-hypergraph-surgery}. Then $\mathcal{S}[\tilde{\mathcal{G}}; (\Phi_1, \Phi_0)]$ defines a hypergraph-surgery gadget satisfying commutation relations~Eqs.~\eqref{eq: main-commutation_requirement} and measures $\mathcal{L}${, i.e. $\Phi_1(\ker \partial_1)$ is generated by $\mathcal{L}$}.
    \item Let $\partial_1: \mathbb{F}_2[\mathcal{V}] \rightarrow \mathbb{F}_2[\mathcal{E}]$ be the hyperedge-vertex incidence matrix of the hypergraph $\tilde{\mathcal{G}}$.
    Then the column weight $\omega_{\mathrm{col}}(\partial_1)$ increases by at most $t$ from the maximum column weight of the incidence matrices of hypergraphs $\tilde{\mathcal{G}}_Z$ and $\tilde{\mathcal{G}}_X$, denoted by, respectively, $\partial^{(Z)}_1$ and $\partial^{(X)}_1$.
    \item $\omega_{\mathrm{row}}(\Phi^{(X)}_0J^{\mathcal{E}}_{\pi})$ increases from $\omega_{\mathrm{row}}(\Phi^{(X)}_0)$ by at most $\omega_{\mathrm{row}}(\partial^{(X)}_1)$ and $\omega_{\mathrm{col}}(\Phi^{(X)}_0J^{\mathcal{E}}_{\pi})$ increases from $\omega_{\mathrm{col}}(\Phi^{(X)}_0)$ by at most $\lceil \omega_{\mathrm{col}}(\partial^{(X)}_1) /t \rceil$.
    \item 
    
    Let $\partial^{'(Z)}_1$ and $\partial^{'(X)}_1$ be the resulting matrices where the information columns (vertices) are punctured, according to Lemma~\ref{lemma: soundness-puncture}. With the multiplicity-$t$ bridged edges Eq.~\eqref{eq: bridged-edges-hypergraph}, $\partial_1$ is $(t, \rho_t)$-sound whenever $\partial^{'(Z)}_1$ and $\partial^{'(X)}_1$ are $(t, \rho_t)$-sound, provided that these bridges are weight-$2$ augmenting checks on the respective information bits (vertices) of $\partial^{(Z)}_1$ ($\tilde{\mathcal{G}}_Z$) and $\partial^{(X)}_1$ ($\tilde{\mathcal{G}}_X$).
\end{enumerate}
\begin{proof}
    The first assertion is given by the explicit construction of bridging and hypergraph $\tilde{\mathcal{G}}$ in Definition~\ref{def: bridging-hypergraph-surgery}, as well as Assumptions~\ref{parallel-canonical-basis} and~\ref{parallel-transversal-phi-1}. 
    In particular, 
    \begin{align}
        \Phi^{(Z)T}_1 \Phi^{(X)}_1 J^{\mathcal{V}}_{\pi}
    \end{align}
    is diagonal. The second assertion follows straightforwardly from Eq.~\eqref{eq: bridged-edges-hypergraph}. For the third assertion, by row-transversality of $\Phi^{(X)}_0$ (Assumption~\ref{parallel-transversal-phi-1}), row $a$ of $\Phi^{(X)}_0$ is supported on a single hyperedge $e_a$, so
    \begin{align}
      \begin{aligned}
         (\Phi^{(X)}_0 J^{\mathcal{E}}_{\pi})_{a e'}
         &= \sum_{e \in \mathcal{E}} (\Phi^{(X)}_0)_{ae} (J^{\mathcal{E}}_{\pi})_{ee'} \\
         &= (J^{\mathcal{E}}_{\pi})_{e_a e'}
         = \mathbf{1}_{e' \in P_{\pi}(e_a)},
      \end{aligned}
    \end{align}
    i.e., row $a$ of $\Phi^{(X)}_0 J^{\mathcal{E}}_{\pi}$ is the indicator of the $\pi$-path $P_{\pi}(e_a)$. For the pairing $\pi$ of Eq.~\eqref{expr: hypergraph-one-way-permutation}, one such $\pi$-path is
    \begin{align}
         \mathbf{1}_{P_{\pi}(e_a)} = \mathbf{1}_{e_a} + \sum_{m:\ v_{m, 1} \in \mathrm{v}(e_a) \cap \gamma_X} \mathbf{1}_{e^b_{m, i_m}},
    \end{align}
    one bridge edge per moved vertex, with a copy $i_m \in [t]$ chosen freely for each; its weight is at most $1 + |\mathrm{v}(e_a)| \leq 1 + \omega_{\mathrm{row}}(\partial^{(X)}_1)$, which is the row-weight bound. For the column weights, two different rows $a, a'$ of $\Phi^{(X)}_0J^{\mathcal{E}}_{\pi}$ may only be incident to the same hyperedge if it is a bridge edge whose vertex lies in $\mathrm{v}(e_a) \cap \mathrm{v}(e_{a'})$: since each vertex $v_{m, 1}$ lies in at most $\omega_{\mathrm{col}}(\partial^{(X)}_1)$ hyperedges, at most $\omega_{\mathrm{col}}(\partial^{(X)}_1)$ rows route through the bridge edges at $m$, and distributing them over the $t$ parallel copies (choosing different representatives $i_m$ from the multiplicity) leaves at most $\lceil \omega_{\mathrm{col}}(\partial^{(X)}_1)/t \rceil$ per copy. Finally, for $|\partial_1 u| < t$, every added (weight-$2$) check is satisfied, since each is repeated $t$ times. Assume that all other information bits (vertices) are punctured. If $|\partial_1 u| > 0$, then $u \notin \ker \partial_1$; furthermore, $u= \begin{pmatrix}
        u_X \\ u_Z
    \end{pmatrix}$ with $u_X \notin \ker \partial^{(X)}_1$ or $u_Z \notin \ker \partial^{(Z)}_1$. This ensures that 
    \begin{align}
        |\partial_1 u| \geq \rho_t \mathrm{dist}(u, \ker \partial_1),
    \end{align}
    since $\partial^{'(Z)}_1$ and $\partial^{'(X)}_1$ are $(t, \rho_t)$-sound and 
    \begin{align}
        \ker \left( \begin{array}{cc}
            \partial^{'(X)}_1 & 0 \\
            0 & \partial^{'(Z)}_1
        \end{array} \right) \subset \ker \partial_1.
    \end{align}
    By the ordering of soundness under puncture (Lemma~\ref{lemma: soundness-puncture}), we conclude the result. 
\end{proof}
\end{theorem}

We illustrate the construction concretely in Fig.~\ref{fig: hirate-Y-toy} on the toy canonical LP code $\LP^{2\times3}_{5} = [[65, 5]]$, whose $k = 5$ logical qubits form a single information fibre: the four single-$Y$ logical operators $\bar{Y}_{0,0,m}$, $m = 0, \dots, 3$, are measured in parallel by bridging the $Z$- and $X$-type gadgets at the $4$ vertex pairs.

Finally we comment on the increased degree from bridging. For $\mathcal{L}$ a logically disjoint set of Pauli operators to measure simultaneously, we wish to add bridges such that the vertex degree (of bridged $\partial_1$) increases at most $t$ for $(t, \rho_t)$-soundness regardless of logical weight of the operators in $\mathcal{L}$. This can be given with a smart design of bridges, which we illustrate with concrete examples as follows: 

\begin{example}[An efficient bridge design]
    Suppose we wish to measure a set of logically disjoint operators: 
    \begin{align}
        \mathcal{L} = \{ \bar{Y}_1 \bar{Y}_2, \bar{X}_3 \bar{Y}_4, \bar{Z}_5 \bar{Z}_6\}.
    \end{align}
    We now show how to design bridged edges with minimal increase in vertex degree. The bridged edges are specified by the incidence matrices below.

 Throughout the discussion in this example, we identify $\mathcal{V} = \mathcal{V}_Z \sqcup \mathcal{V}_X$ with $v_{i, 0} \in \mathcal{V}_Z$ and $v_{i, 1} \in \mathcal{V}_X$. 
    \begin{itemize}
        \item ($\bar Y_1 \bar Y_2$). In this case, the vertices $\{v_{1,0}, v_{1,1}, v_{2,0}, v_{2,1}\}$ associated with the information bits (vertices) of $\bar{Z}_1 \bar{Z}_2, \bar{X}_1, \bar{X}_2$ are involved in the bridging:
        \begin{align}
            \begin{aligned}
                &\partial_1\big|_{\bar Y_1 \bar Y_2} =
	\begin{array}{l|cccc}
	 & v_{1,0} & v_{1,1} & v_{2,0} & v_{2,1} \\
	\hline
	e_1\ (v_{1,0},v_{1,1}) & 1 & 1 & 0 & 0 \\
	e_2\ (v_{2,0},v_{2,1}) & 0 & 0 & 1 & 1 \\
	e_3\ (v_{1,0},v_{2,0}) & 1 & 0 & 1 & 0 \\
	e_4\ (v_{1,1},v_{2,1}) & 0 & 1 & 0 & 1 \\
	e_5\ (v_{1,0},v_{2,1})_1 & 1 & 0 & 0 & 1 \\
	e_6\ (v_{1,0},v_{2,1})_2 & 1 & 0 & 0 & 1 \\
	e_7\ (v_{1,1},v_{2,0})_1 & 0 & 1 & 1 & 0 \\
	e_8\ (v_{1,1},v_{2,0})_2 & 0 & 1 & 1 & 0 \\
	\hline
	\text{col sum} & 4 & 4 & 4 & 4
	\end{array}
	\\
	&\partial_0\big|_{\bar Y_1 \bar Y_2} =
	\begin{array}{l|cccccccc}
	 & e_1 & e_2 & e_3 & e_4 & e_5 & e_6 & e_7 & e_8 \\
	\hline
	c_{(v_{1,0},v_{2,1})} & 0 & 0 & 0 & 0 & 1 & 1 & 0 & 0 \\
	c_{(v_{1,1},v_{2,0})} & 0 & 0 & 0 & 0 & 0 & 0 & 1 & 1 \\
	c_{\triangle}(v_{1,0} v_{2,0} v_{2,1}) & 0 & 1 & 1 & 0 & 1 & 0 & 0 & 0 \\
	c_{\triangle}(v_{1,0} v_{1,1} v_{2,1}) & 1 & 0 & 0 & 1 & 1 & 0 & 0 & 0 \\
	c_{\triangle}(v_{1,0} v_{1,1} v_{2,0}) & 1 & 0 & 1 & 0 & 0 & 0 & 1 & 0 \\
	\end{array}.
            \end{aligned}
        \end{align}
        \item ($\bar X_3 \bar Y_4$). In this case, the vertices $\{ v_{3,1}, v_{4,0}, v_{4,1}\}$ associated with the information bits (vertices) of $\bar{X}_3, \bar{Z}_4, \bar{X}_4$ are involved in the bridging:
        \begin{align}
            \begin{aligned}
                &\partial_1\big|_{\bar X_3 \bar Y_4} =
	\begin{array}{l|ccc}
	 & v_{3,1} & v_{4,0} & v_{4,1} \\
	\hline
	e_1, e_2\ (v_{3,1},v_{4,0}) & 1 & 1 & 0 \\
	e_3, e_4\ (v_{3,1},v_{4,1}) & 1 & 0 & 1 \\
	e_5, e_6\ (v_{4,0},v_{4,1}) & 0 & 1 & 1 \\
	\hline
	\text{col sum} & 4 & 4 & 4
	\end{array}
	\\
	&\partial_0\big|_{\bar X_3 \bar Y_4} =
	\begin{array}{l|cccccc}
	 & e_1 & e_2 & e_3 & e_4 & e_5 & e_6 \\
	\hline
	c_{(v_{3,1},v_{4,0})} & 1 & 1 & 0 & 0 & 0 & 0 \\
	c_{(v_{3,1},v_{4,1})} & 0 & 0 & 1 & 1 & 0 & 0 \\
	c_{(v_{4,0},v_{4,1})} & 0 & 0 & 0 & 0 & 1 & 1 \\
	c_{\triangle}(v_{3,1} v_{4,0} v_{4,1}) & 1 & 0 & 1 & 0 & 1 & 0 \\
	\end{array}.
            \end{aligned}
        \end{align}
        \item ($\bar Z_5 \bar Z_6$). In this case, the vertices $\{ v_{5,0}, v_{6,0} \}$ associated with the information bits (vertices) of $\bar{Z}_5, \bar{Z}_6$ are involved in the bridging:
        \begin{align}
            \begin{aligned}
               &\partial_1\big|_{\bar Z_5 \bar Z_6} =
	\begin{array}{l|cc}
	 & v_{5,0} & v_{6,0} \\
	\hline
	e_1, \ldots, e_4\ (v_{5,0},v_{6,0}) & 1 & 1 \\
	\hline
	\text{col sum} & 4 & 4
	\end{array}
	\\
	&\partial_0\big|_{\bar Z_5 \bar Z_6} =
	\begin{array}{l|cccc}
	 & e_1 & e_2 & e_3 & e_4 \\
	\hline
	c_{12} & 1 & 1 & 0 & 0 \\
	c_{23} & 0 & 1 & 1 & 0 \\
	c_{34} & 0 & 0 & 1 & 1 \\
	\end{array}.
            \end{aligned}
        \end{align}
    Combined, let $\partial^{'(Z)}_1$ and $\partial^{'(X)}_1$ denote, respectively, the non-information restrictions of $\partial^{(Z)}_1$ and $\partial^{(X)}_1$, as in Theorem~\ref{thm: bridged-hypergraph-parallel-surgery}(4). The bridged hypergraph chain is given by
    \begin{align}
        \begin{aligned}
        &\partial_1 \;=\;
	\left(\begin{array}{cc|ccc}
	\partial'^{(Z)}_1 & 0 & \partial^{(Z)}_1\big|_{V_{\bar Y_1\bar Y_2}} & \partial^{(Z)}_1\big|_{V_{\bar X_3\bar Y_4}} & \partial^{(Z)}_1\big|_{V_{\bar Z_5\bar Z_6}} \\
	0 & \partial'^{(X)}_1 & \partial^{(X)}_1\big|_{V_{\bar Y_1\bar Y_2}} & \partial^{(X)}_1\big|_{V_{\bar X_3\bar Y_4}} & 0 \\
	\hline
	0 & 0 & \partial_1\big|_{\bar Y_1\bar Y_2} & 0 & 0 \\
	0 & 0 & 0 & \partial_1\big|_{\bar X_3\bar Y_4} & 0 \\
	0 & 0 & 0 & 0 & \partial_1\big|_{\bar Z_5\bar Z_6}
	\end{array}\right),
	\\
	&\partial_0 \;=\;
	\left(\begin{array}{cc|ccc}
	0 & 0 & \partial_0\big|_{\bar Y_1\bar Y_2} & 0 & 0 \\
	0 & 0 & 0 & \partial_0\big|_{\bar X_3\bar Y_4} & 0 \\
	0 & 0 & 0 & 0 & \partial_0\big|_{\bar Z_5\bar Z_6}
	\end{array}\right),
        \end{aligned}
    \end{align}
    where the skeleton rows contain the punctured $\partial'^{(Z)}_1 \oplus \partial'^{(X)}_1$ as a block diagonal matrix on the non-information columns. The blocks $\partial^{(Z)}_1\big|_{V_C}$ and $\partial^{(X)}_1\big|_{V_C}$ contain the columns of the unpunctured skeletons deleted by puncturing, namely those on the information vertices $V_C$. Since the $\bar Z_5\bar Z_6$ component is purely $Z$-type, $\partial^{(X)}_1\big|_{V_{\bar Z_5\bar Z_6}} = 0$.
    \end{itemize}
\end{example}

We can straightforwardly generalize this technique for arbitrarily logically disjoint logical operators, where the bridged $\partial_1$ is at least $(t, \rho_t)$-sound if $\partial^{'(Z)}_1$ and $\partial^{'(X)}_1$ are each at least $(t, \rho_t)$-sound. Futhermore, the bridged $\partial_1$ increases vertex degree by at most $t$, with $\omega_{\mathrm{col}}(\partial_0) \leq 3$ and $\omega_{\mathrm{row}}(\partial_0) \leq 3$. Table~\ref{tab: high-rate-arbitrary-resources} benchmarks the full construction --- bridging combined with the thickening of the following subsection --- on four canonical LP codes, at the worst case over all logically disjoint sets of logical operators{ of any number of bodies (any-body); this worst case coincides with that over all 2-PPMs, while the two-body Clifford menu is reported in Table~\ref{tab: high-rate-two-body-addressable}}.

Often times, Theorem~\ref{thm: bridged-hypergraph-parallel-surgery}(iii) gives a crude upper bound on the merged degree, and, depending on the specific structure of $\mathcal{L}$ as well as more optimized method, we can reduce the merged degree even further. One such method is to symmetrize the underlying hypergraph with added (hyper)edges, introduced in~\eqref{expr: symmetrizing-edges}, since $\pi$ is (at most) an involution on the set of (hyper)edges. Denote $\pi := \oplus_{m\in [l]} s_{m}\pi_m$ where $s_m \in \{ 0, 1\}$ and is used to select which involutions are applied, depending on the particular measurement pattern. In this case, a slightly better design--which we refer to as \emph{merging} along one-way permutation--can emerge: Let a hyperedge $e \in \mathcal{E}_X$
\begin{enumerate}
    \item \label{hypergraph-symmetrization} If $\mathrm{v}(\pi(e)) \neq \mathrm{v}(e)$, then we add $\pi(e)$ to $\mathcal{E} := \mathcal{E}_Z \cup \mathcal{E}_X$,
    \begin{align}
         \mathcal{E} \mapsto \mathcal{E} \cup \pi(\mathcal{E}),
    \end{align}
    so that $J^{\mathcal{E}}_{\pi}$ is an automorphism with $\omega_{\mathrm{row}}(J^{\mathcal{E}}_{\pi}) = \omega_{\mathrm{col}}(J^{\mathcal{E}}_{\pi})= 1$. 
    \item \label{hypergraph-pruning} For any vertices $v_{m, 1} \in \gamma_X$ acted nontrivially by $\pi$, i.e., $s_m=1$, we prune these and their associated hyperedges: $e \in \mathcal{E}_X$ such that $\mathrm{v}(e) {\ni} v_{m,1}$ for any $m \in [l]$. 
\end{enumerate}

This defines a valid, sparse, and {soundness}-preserving hypergraph gadget for measuring $Y$-type logically disjoint operators. More precisely, suppose $\mathcal{L}$ is a set of single-qubit logical $Y$-measurements, then the above procedure ensures that
\begin{align}\label{expr: merged-logical-Ys}
           \{  X(\Phi^{(X)}_1 J^{\mathcal{V}}_{\pi}(y))Z(\Phi^{(Z)}_1(y)): y \in \ker \partial_1 \}
\end{align}
is spanned by $\mathcal{L}$, where $\partial_1$ is the hyperedge-vertex incidence matrix of the symmetrized and pruned hypergraph. To see this, the merging along $\pi$ on $\partial^{(X)}_1$ is effectively the column permutation of vertices $v_{m,1} \in \gamma_X$ with $s_m=1$ to $v_{m, 0}$. Hence, let $y_X \in \mathbb{F}_2[\mathcal{V}_X]$ which lies in $\ker \partial^{(X)}_1$ which is supported on vertices (information bits) of at least one of such $v_{m, 1}$, then 
\begin{align}
    (J^{\mathcal{V}\dagger}_{\pi}(y_X))(v_{m, 0}) = y_X({v_{m,1}})
\end{align}
is supported on the corresponding $v_{m, 0} $ vertices lying in $\gamma_Z$, which are a subset of information bits to $\ker \partial^{(Z)}_1$. Hence, by Lemma~\ref{lemma: information-set}, there exists a unique extension to $y \in \mathbb{F}_2[\mathcal{V}_X \cup \mathcal{V}_Z]$ such that $\partial_1 y = 0$ and $y|_{\mathcal{V}_Z}$ lies in $\ker \partial^{(Z)}_1$. Hence, as long as $\Phi^{(X)}_1(\ker \partial^{(X)}_1)$ is spanned by $X$-components of $\mathcal{L}$ and $\Phi^{(Z)}_1(\ker \partial^{(Z)}_1)$ is spanned by $Z$-components of $\mathcal{L}$, the set in~\eqref{expr: merged-logical-Ys} is indeed spanned by $\mathcal{L}$. This explains the validity of single-qubit parallel measurement of logical $Y$-type operators, and multiple-qubit joints with mixed-types are generalized straightforwardly using the bridges/augmentation introduced earlier and in Section~\ref{app-sec: classical-error-correcting-codes}.

\begin{lemma}[Hypergraph symmetrization]\label{lemma: hypergraph-symmetrization}
    Under the same setting (and notations) as in Theorem~\ref{thm: bridged-hypergraph-parallel-surgery} and we, instead of directly adding bridged edges $\mathcal{E}_b$, merge along the one-way permutation $\pi$ prescribed by~\ref{hypergraph-symmetrization} and~\ref{hypergraph-pruning} above. Then the following hold: 
    \begin{enumerate}
        \item The rewired $\Phi^{(X)}_1J^{\mathcal{V}}_{\pi}$ and  $\Phi^{(X)}_0J^{\mathcal{E}}_{\pi}$ have the same column- and row-weights as {those} of $\Phi^{(X)}_1$ and $\Phi^{(X)}_0$, respectively. 
        \item Let $\partial_1: \mathbb{F}_2[\mathcal{V}] \rightarrow \mathbb{F}_2[\mathcal{E}]$ be the hyperedge-vertex incidence matrix of the hypergraph $\tilde{\mathcal{G}}$. Then $\omega_{\mathrm{col}}(\partial_1) \leq \omega_{\mathrm{col}}(\partial^{(X)}_1) + \omega_{\mathrm{col}}(\partial^{(Z)}_1)+ t$--if we add bridged edges/augmentations with repetition $t$--and $\omega_{\mathrm{row}}(\partial_1) = \max(\omega_{\mathrm{row}}(\partial^{(X)}_1), \omega_{\mathrm{row}}(\partial^{(Z)}_1))$.
        \item $\partial_1$ is (at least) $(t, \rho_t)$-sound in the sense of Theorem~\ref{thm: bridged-hypergraph-parallel-surgery}(iv).
    \end{enumerate}
    \begin{proof}
        Here it is easy and equivalent to treat $\pi$ as direct sums of transpositions, followed by pruning vertices in~\ref{hypergraph-pruning}. Hence, by~\ref{hypergraph-symmetrization}, this is a (hyper)graph automorphism and, hence, $J^{\mathcal{V}}_{\pi}$ and $J^{\mathcal{E}}_{\pi}$ are transversal, i.e., having row- and column-weights equal to $1$. The same holds naturally after pruning. For the second item, let $v_{m, 0} \in \gamma_Z$ be a vertex for which $s_m =1$. Then $v_{m, 0}$ is incident to at most $\omega_{\mathrm{col}}(\partial^{(Z)}_1)$ many hyperedges in $\mathcal{E}_Z$. By the symmetrization~\ref{hypergraph-symmetrization}, $v_{m, 0}$ is also incident to at most $\omega_{\mathrm{col}}(\partial^{(X)}_1)$ many hyperedges in $\mathcal{E}_X$; while the largest-weight hyperedges are bounded by $\max(\omega_{\mathrm{row}}(\partial^{(X)}_1), \omega_{\mathrm{row}}(\partial^{(Z)}_1))$. The column bound is further increased by $t$ if we add further bridged edges/augmentations for measuring joint operators. For the third part, puncturing all the information bits in $\partial_1$ results in a submatrix 
        \begin{align}
            \left(\begin{array}{cc}
                \partial^{'(X)}_1 & 0 \\
                0 & \partial^{'(Z)}_1
            \end{array}\right),
        \end{align}
        which is, by assumption, $(t, \rho_t)$-sound. Hence, the result follows from Lemma~\ref{lemma: soundness-puncture}, as in the proof of Theorem~\ref{thm: bridged-hypergraph-parallel-surgery}(4). 
    \end{proof}
\end{lemma}

\begin{table*}[!t]
{
\centering
\footnotesize
\setlength{\tabcolsep}{1.5pt}
\begin{tabular}{@{}l l cc ccc ccc@{}}
\toprule
 & & & & \multicolumn{3}{c}{Merged size} & \multicolumn{3}{c}{Merged degree} \\
\cmidrule(lr){5-7}\cmidrule(l){8-10}
 & Data code & \shortstack{Bridge\\multiplicity $t$} & \shortstack{Thickening\\length $m$} & Qubits & $X$-checks & $Z$-checks & \shortstack{Stabilizer\\weight} & \shortstack{Qubit\\$X$-degree} & \shortstack{Qubit\\$Z$-degree} \\
\midrule
 \multirow{4}{*}{Any-body} & $\LP^{2\times4}_{16} = [[320, 64, 10]]$ & $4$ & $3$ & $1472$\,{\scriptsize$(+1152)$} & $928$\,{\scriptsize$(+800)$} & $704$\,{\scriptsize$(+576)$} & $12$\,{\scriptsize$(+3)$} & $10$\,{\scriptsize$(+4)$} & $6$\,{\scriptsize$(+0)$} \\
  & $\LP^{2\times3}_{36} = [[468, 36, 20]]$ & $5$ & $4$ & $2412$\,{\scriptsize$(+1944)$} & $1692$\,{\scriptsize$(+1476)$} & $1080$\,{\scriptsize$(+864)$} & $15$\,{\scriptsize$(+5)$} & $10$\,{\scriptsize$(+5)$} & $6$\,{\scriptsize$(+1)$} \\
  & $\LP^{3\times5}_{20} = [[680, 80, 13]]$ & $5$ & $3$ & $3400$\,{\scriptsize$(+2720)$} & $2420$\,{\scriptsize$(+2120)$} & $1500$\,{\scriptsize$(+1200)$} & $14$\,{\scriptsize$(+4)$} & $9$\,{\scriptsize$(+3)$} & $6$\,{\scriptsize$(+0)$} \\
  & $\LP^{3\times5}_{28} = [[952, 112, 17]]$ & $5$ & $4$ & $6216$\,{\scriptsize$(+5264)$} & $4676$\,{\scriptsize$(+4256)$} & $2660$\,{\scriptsize$(+2240)$} & $14$\,{\scriptsize$(+5)$} & $9$\,{\scriptsize$(+4)$} & $5$\,{\scriptsize$(+0)$} \\
\bottomrule
\end{tabular}
\caption{\textbf{Arbitrary logically disjoint measurements by bridged hypergraph surgery.} Merged codes of the parallel logical Pauli-product measurements of any logically disjoint set $\mathcal{L}$ of products of canonical logical basis operators (any number of bodies, $\bar{Y}$ factors included) in one logical cycle, on the (symmetric) canonical LP codes $\LP^{2\times3}_{36} = [[468, 36, 20]]$, $\LP^{3\times5}_{28} = [[952, 112, 17]]$, $\LP^{2\times4}_{16} = [[320, 64, 10]]$, and $\LP^{3\times5}_{20} = [[680, 80, 13]]$; each cell is the worst case over all such sets. The worst case over all 2-PPMs (products on at most two logical qubits, $\bar{X}\bar{Y}$, $\bar{Z}\bar{Y}$, $\bar{Y}\bar{Y}$ included) coincides with this table in every entry; the two-body Clifford menu is reported in Table~\ref{tab: high-rate-two-body-addressable}. Single $\bar{Y}$ factors are measured by merged vertices (Lemma~\ref{lemma: hypergraph-symmetrization}) and the $\bar{Y}$ factors of products by rewired anchors, except for $\LP^{2\times4}_{16}$, where every $\bar{Y}$ factor is merged. The bridge multiplicity $t$ and the thickening length $m$ are chosen so that the thickened hypergraph-surgery gadget is distance-preserving by Theorem~\ref{thm: main-high-rate-surgery-distance-preserving}(ii) (Theorem~\ref{thm: formal-surgery-distance-preserving}): $m \rho_t \geq 1$ and $mt \geq d$, i.e.\ $m = \gamma_A$ (Eq.~\eqref{eq: gamma-A-canonical}) and $t = \lceil d/m \rceil$. The condition $mt \geq d$ cannot be weakened for joint measurements: a single factor of a measured pair lies outside $\ker \partial_1$ while its $\Phi_1$-image is a nontrivial logical operator, so the merged distance is capped at $mt$. Parentheses report the increase over the data code in the respective fields. For all four codes, the hypergraph chains are exact with full-row-rank punctured skeletons, so no ancilla-type logical operator arises ($d(\tilde{\mathcal{G}}) = \infty$) and Theorem~\ref{thm: formal-surgery-distance-preserving}(ii) applies.}
\label{tab: high-rate-arbitrary-resources}
}
\end{table*}

\subsection{Thickening of hypergraph surgery}\label{sec: surgery-thickening}

As observed from the proof, the argument (of distance-preserving) is symmetric in the $X$- and $Z$-directions. Hence, without loss of generality, we only prove the soundness for single-type Pauli operators, e.g., $Z$-type Pauli operators. In this regard, we will continue the abuse of notation by $\Phi_1 = \Phi^{(Z)}_1$, whenever the context is clear. A key question is how to boost soundness, especially in hypergraph surgery. In graph surgery, the Cheeger constant is the standard tool for ensuring distance preservation~\cite{he2025extractorsqldpcarchitecturesefficient, Ide_2025, cowtan2025fastfaulttolerantlogicalmeasurements}. A standard technique is through \emph{thickening} or tensor-product ancilla chains. 

\begin{definition}[Thickening]\label{def: length-m-thickening}
     Let $\mathcal{S}[\tilde{\mathcal{G}}; (\Phi_1, \Phi_0)]$ be a (hypergraph) surgery gadget associated with some quantum code $Q$. Let $\mathcal{D}_{\bullet}: \mathcal{D}_1 \xrightarrow{D} \mathcal{D}_0$ be the parity check to the length-$m$ repetition code with row vector
    \begin{align}\label{eq: parity-check-repetition}
        D_{i} = e^{(i)}_{\mathcal{D}_1} + e^{(i+1)}_{\mathcal{D}_1}
    \end{align}
    where $e^{(i)}_{\mathcal{D}_1}, e^{(i+1)}_{\mathcal{D}_1}$ are unit vectors and $i \in [m-1]$. Then the length-$m$ thickened surgery gadget $\mathcal{S}[\tilde{\mathcal{G}}_m; (\tilde{\Phi}_1, \tilde{\Phi}_0)]$ is given by the product graph $\tilde{\mathcal{G}}_m(\tilde{\mathcal{V}}, \tilde{\mathcal{E}}, \tilde{\mathcal{C}})$ ($\tilde{\mathcal{G}}_1 = \mathcal{G}$) with 
    \begin{align}
   \begin{aligned}
   \tilde{\mathcal{V}} &= \mathcal{V} \otimes \mathcal{D}_1 \\
    \tilde{\mathcal{E}} &= \mathcal{E} \otimes \mathcal{D}_1 \oplus \mathcal{V} \otimes \mathcal{D}_0 \\
     \tilde{\mathcal{C}} &= \mathcal{C} \otimes \mathcal{D}_1 \oplus \mathcal{E} \otimes \mathcal{D}_0.
   \end{aligned}
    \end{align}
    Correspondingly, the boundary maps for $\tilde{\mathcal{G}}_m$ are given by 
    \begin{align}\label{eq: thickened-ancilla-boundary-maps}
        \tilde{\partial}_1 = \left(\begin{array}{c}
             \partial_1 \otimes I_{\mathcal{D}_1}  \\
              I_{\mathcal{V}} \otimes D
        \end{array}\right); \quad \tilde{\partial}^T_0 = \left(\begin{array}{cc}
             \partial^T_0 \otimes I_{\mathcal{D}_1} & I_{\mathcal{E}} \otimes D^T \\
              0 & \partial^T_1 \otimes I_{\mathcal{D}_0}
        \end{array}\right).
    \end{align}

    The thickened connectivity maps $\tilde{\Phi} = (\tilde{\Phi}_1, \tilde{\Phi}_0)$ are given by 
    \begin{align}
        \begin{aligned}
            \tilde{\Phi}_1 &= \Phi_1 \otimes e^{(i)}_{\mathcal{D}_1} \\
            \tilde{\Phi}_0 &= \Phi_0 \otimes e^{(i)}_{\mathcal{D}_1}
        \end{aligned}
    \end{align}
    for an arbitrary choice of $i=1, \cdots, m$ and, with a slight notational abuse, $e^{(i)}_{\mathcal{D}_1}$ as a unit (row) vector supported on the $i$-th index. 
\end{definition}

\begin{lemma}[Soundness boost from thickening]\label{lemma: soundness-boost-thickening}
    Let $\partial_1$ be a map with $(t, \rho_t)$-soundness, and let $\Phi_1$ be elementary, so that the thickened map $\tilde{\Phi}_1 = \Phi_1 \otimes e^{(i)}_{\mathcal{D}_1}$ is elementary as well. Then $\tilde{\partial}_1$ has $(mt, \min(m\rho_t, 1))$-soundness relative to $\tilde{\Phi}_1$. 
    \begin{proof}
        Let $y = \sum^m_{i=1} y_{\mathcal{V}}^{(i)} \otimes e^{(i)}_{\mathcal{D}_1}$ with $\tilde{\partial}_1 y \neq 0$. Recall from Eq.~\eqref{eq: thickened-ancilla-boundary-maps} we have
        \begin{align}\label{eq: tilde-partial-1-y-decomposition}
            \tilde{\partial}_1 y = \begin{pmatrix}
                \partial_1 y_{\mathcal{V}}^{(1)} \otimes e^{(1)}_{\mathcal{D}_1} \\
                \vdots \\
                \partial_1 y_{\mathcal{V}}^{(m)} \otimes e^{(m)}_{\mathcal{D}_1} \\
                y_{\mathcal{V}}^{(1)} \otimes D e^{(1)}_{\mathcal{D}_1} \\
                \vdots \\
                y_{\mathcal{V}}^{(m)} \otimes D e^{(m)}_{\mathcal{D}_1}
            \end{pmatrix}.
        \end{align}
        We can decompose $y = y' + y''$ for $y' \in \mathbb{F}_2^{|\mathcal{V}|} \otimes (\mathbb{F}_2^m \setminus \ker D)$ and $y'' \in \mathbb{F}_2^{|\mathcal{V}|} \otimes \ker D$: 
        \begin{align}
            y' := \sum^m_{i=1} x_{\mathcal{V}}^{(i)} \otimes e^{(i)}_{\mathcal{D}_1}
        \end{align}
        where vectors $x_{\mathcal{V}}^{(i)} \in \mathbb{F}_2^{|\mathcal{V}|}$ satisfy 
        \begin{align}\label{eq: disjoint-supports-from-soundness-boost}
            \bigcap^m_{i=1} \mathrm{supp}(x_{\mathcal{V}}^{(i)}) = \emptyset,
        \end{align}
        and $e^{(i)}_{\mathcal{D}_1} \in \mathbb{F}_2^m $ are the unit vectors supported on the $i$-th entry. Their shared support is given by 
        \begin{align}
            y'' := g_{\mathcal{V}} \otimes 1_{\mathcal{D}_1}
        \end{align}
        where $1_{\mathcal{D}_1} \in \ker D$ is the all-one vector. For the latter half of the matrix in Eq.~\eqref{eq: tilde-partial-1-y-decomposition} and $y''$, this implies that 
        \begin{align}
            \tilde{\partial}_1 y'' = \partial_1 \otimes I_{\mathcal{D}_1} (g_{\mathcal{V}} \otimes 1_{\mathcal{D}_1}) = \partial_1 g_{\mathcal{V}} \otimes 1_{\mathcal{D}_1}
        \end{align}
        and 
        \begin{align}
           |\tilde{\partial}_1 y''| = m |\partial_1 g_{\mathcal{V}}| &\geq m\min(t, \rho_t \mathrm{dist}(g_{\mathcal{V}}, \ker \partial_1)). \label{expr: soundness-boost-upper-sector}
        \end{align}
        For the other case, with the latter half of the matrix in Eq.~\eqref{eq: tilde-partial-1-y-decomposition} and $y'$: 
        \begin{align}
           \begin{aligned}
           (I_{\mathcal{V}} \otimes D) y' &= \sum^m_{i=1} x_{\mathcal{V}}^{(i)} \otimes D e^{(i)}_{\mathcal{D}_1} \\
           &= \sum^{m-1}_{i=1} (x^{(i)}_{\mathcal{V}} + x^{(i+1)}_{\mathcal{V}}) \otimes e^{(i)}_{\mathcal{D}_0},
           \end{aligned}
        \end{align}
        where we used the form of $D$ given in Eq.~\eqref{eq: thickened-ancilla-boundary-maps}, the parity-check matrix of the length-$m$ repetition code. Combining the two cases, we write $y = y' + y''$ for $y' \in \mathbb{F}_2^{|\mathcal{V}|} \otimes (\mathbb{F}_2^m \setminus \ker D)$ and $y'' \in \mathbb{F}_2^{|\mathcal{V}|} \otimes \ker D$: 
        \begin{align}
           \begin{aligned}
             \tilde{\partial}_1 y &= (\partial_1 \otimes I_m)(y'' + y') + (I_{\mathcal{V}} \otimes D) y' \\
             &= \sum^m_{i=1}\partial_1 (g_{\mathcal{V}} +x^{(i)}_{\mathcal{V}})\otimes e^{(i)}_{\mathcal{D}_1} + \sum^{m-1}_{i=1} (x^{(i)}_{\mathcal{V}} + x^{(i+1)}_{\mathcal{V}}) \otimes e^{(i)}_{\mathcal{D}_0}.
           \end{aligned}
        \end{align}
        The two sums lie in the disjoint sectors $\mathbb{F}_2[\mathcal{E}] \otimes \mathcal{D}_1$ and $\mathbb{F}_2[\mathcal{V}] \otimes \mathcal{D}_0$, so their weights add:
        \begin{align}\label{eq: soundness-boost-weight-split}
         |\tilde{\partial}_1 y| = \sum^m_{i=1}|\partial_1 y^{(i)}_{\mathcal{V}}| + \sum^{m-1}_{i=1}|x^{(i)}_{\mathcal{V}} + x^{(i+1)}_{\mathcal{V}}|.
        \end{align}
        To conclude, let $i^*$ maximize $w_i := \mathrm{dist}(y^{(i)}_{\mathcal{V}}, \ker \partial_1)$ and $j$ minimize it. The difference terms of Eq.~\eqref{eq: soundness-boost-weight-split} telescope, $\sum^{m-1}_{k=1}|x^{(k)}_{\mathcal{V}} + x^{(k+1)}_{\mathcal{V}}| \geq |x^{(i^*)}_{\mathcal{V}} + x^{(j)}_{\mathcal{V}}| = |y^{(i^*)}_{\mathcal{V}} + y^{(j)}_{\mathcal{V}}| \geq w_{i^*} - w_j$, while each syndrome term satisfies $|\partial_1 y^{(i)}_{\mathcal{V}}| \geq \min(t, \rho_t\, w_i) \geq \min(t, \rho_t\, w_j)$---by the $(t, \rho_t)$-soundness of $\partial_1$ when $|\partial_1 y^{(i)}_{\mathcal{V}}| < t$, and trivially otherwise---so that $\sum^m_{i=1}|\partial_1 y^{(i)}_{\mathcal{V}}| \geq m \min(t, \rho_t w_j)$. If the minimum equals $t$, then $|\tilde{\partial}_1 y| \geq mt$; otherwise
        \begin{align}
            \begin{aligned}
                |\tilde{\partial}_1 y| &\geq (w_{i^*} - w_j) + m\rho_t\, w_j = w_{i^*} + (m\rho_t - 1)\, w_j\\
                &\geq \min(m\rho_t, 1)\, w_{i^*},
            \end{aligned}
        \end{align}
        since $w_j \leq w_{i^*}$. In both cases $|\tilde{\partial}_1 y| \geq \min(mt, \min(m\rho_t, 1) \max_i w_i)$, and the layer carried by $\tilde{\Phi}_1$ gives $\mathrm{dist}_{\tilde{\Phi}_1}(y, \ker \tilde{\partial}_1) \leq \max_i w_i$.
        In other words, $\tilde{\partial}_1$ is $(mt, \min(m\rho_t, 1))$-sound relative to $\tilde{\Phi}_1$, as desired. 
    \end{proof}

\end{lemma}

In the context of parallel surgery such as the use of puncture in performing addressable measurement, it is not sufficient to prove the small-set soundness for $\partial_1$ but for any of its punctured versions. More precisely, there exists a partial ordering: 
\begin{lemma}[Soundness related to puncture]\label{lemma: soundness-puncture}
    Let $H: \mathbb{F}_2^n \rightarrow \mathbb{F}_2^m$ be a linear map and $I$ be an information set with $|I|=k$. Let $\gamma$ be a subset of column indices of $H$ whose complement $\bar{\gamma} \subseteq I$ consists only of information bits, retaining at least one. Denote $H_\gamma$ to be its column restriction. Let $\rho_t(H)$ and $\rho_t(H_\gamma)$ be the optimal $(t, \rho_t)$-soundness constants of $H$ and $H_\gamma$, respectively (set to $+\infty$ when the small-set window is empty). Then we have that $\rho_t(H_\gamma) \leq \rho_t(H)$. Furthermore, there is a partial ordering:
    \begin{align}
        \rho_t(H) \geq \rho_t(H_{\gamma_1}) \geq \cdots \geq \rho_t(H_{\gamma_k}),
    \end{align}
    whenever $\gamma_1 \supset \cdots \supset \gamma_k$, each excluding only information bits of its predecessor. 
\end{lemma}

\begin{proof}
    Let $u$ attain the optimal ratio for $H$ at syndrome weight $|Hu| < t$. Adding the systematic kernel basis elements $g^{(i)}$ for the information bits $i \in \mathrm{supp}(u) \cap \bar{\gamma}$ clears $u$ on $\bar{\gamma}$ without changing $Hu$ or $\mathrm{dist}(u, \ker H)$, so we may assume $u$ is supported on $\gamma$; write $u_\gamma$ for its restriction, with $|H_\gamma u_\gamma| = |Hu| < t$. Every element of $\ker H_\gamma$ extends by zero to an element of $\ker H$, so $\mathrm{dist}(u_\gamma, \ker H_\gamma) \geq \mathrm{dist}(u, \ker H)$. Hence $\rho_t(H_\gamma) \leq |H_\gamma u_\gamma| / \mathrm{dist}(u_\gamma, \ker H_\gamma) \leq |Hu| / \mathrm{dist}(u, \ker H) = \rho_t(H)$. Iterating over the nested $\gamma_j$ gives the chain.
\end{proof}

Note that the soundness is only a sufficient condition and is not expected to be necessary. In particular, we can consider another criterion, called the expansion. 

\begin{definition}[Minimum syndrome weights of linear maps]\label{def: minimum-syndrome-weights}
    Let $H: \mathbb{F}_2^n \rightarrow \mathbb{F}_2^m$ be a linear map. We denote $\lambda(H)$ to be the minimum weight of a non-zero syndrome, i.e.,
    \begin{align}
        \lambda(H) := \min_{x \notin \ker H} |Hx|.
    \end{align}
\end{definition}

Note that if $H$ is full row rank, then $\lambda(H) = 1$, so to have any nontrivial minimum syndrome weight, $H$ must not be full row rank. 

\begin{lemma}\label{lemma: minimum-weight-syndrome-expansion}
Under any puncture with $\gamma$ defined in Lemma~\ref{lemma: soundness-puncture}, we have that $\lambda(H) = \lambda(H_{\gamma})$. Furthermore, for an elementary connectivity map, so that $\omega_{\mathrm{col}}(\Phi_1) \leq 1$, we have 
\begin{align}
    |\tilde{\partial}_1 y| \geq \min\bigl(\mathrm{dist}_{\tilde{\Phi}_1}(y, \ker \tilde{\partial}_1),\, m \lambda(\partial_1)\bigr). 
\end{align}
Equivalently, this implies that $\tilde{\partial}_1$ is at least $(m\lambda(\partial_1), 1)$-sound relative to $\tilde{\Phi}_1$.
\begin{proof}
    Since $\gamma$ only excludes bits that lie in some information set of $H$, let $\{g^{(i)}\}^k_{i=1}$ be the systematic basis of $\ker H$. Suppose that there exists some $x_\gamma$ such that $0 < |H_\gamma x_{\gamma}| < \lambda(H)$; we extend it as 
    \begin{align}
        x := (x_{\bar{\gamma}}, x_{\gamma})^T = (0, x_{\gamma})^T,
    \end{align}
    where $x_{\bar{\gamma}} =0$ is the trivial extension. This contradicts the minimality of $\lambda(H)$, hence $\lambda(H_\gamma) \geq \lambda(H)$. Conversely, let $x$ attain $\lambda(H)$; adding $g^{(i)}$ for every removed information bit $i \in \supp(x) \cap \bar{\gamma}$ makes $x$ vanish on $\bar{\gamma}$ without changing $Hx$, so $\lambda(H_\gamma) \leq \lambda(H)$. For the second part, write that 
    \begin{align}
        y = \sum^m_{i=1} y^{(i)}_{\mathcal{V}} \otimes e^{(i)}_{\mathcal{D}_1},
    \end{align}
    similarly to the proof of Lemma~\ref{lemma: soundness-boost-thickening}, and $y = y' + y''$ for $y' \in \mathbb{F}_2^{|\mathcal{V}|} \otimes (\mathbb{F}_2^m \setminus \ker D)$ and $y'' \in \mathbb{F}_2^{|\mathcal{V}|} \otimes \ker D$. Then we conclude that 
    \begin{align}
        |\tilde{\partial}_1 y'| \geq \max_{i \in [m]} |x^{(i)}_{\mathcal{V}}| \geq |\tilde{\Phi}_1(y')|
    \end{align}
    for elementary connectivity maps and 
    \begin{align}
        |\tilde{\partial}_1 y''| \geq m \lambda(\partial_1).
    \end{align}
    Drawing on a similar technique from Lemma~\ref{lemma: soundness-boost-thickening} and utilizing the thickened $\partial_1$ in Eq.~\eqref{eq: tilde-partial-1-y-decomposition}, we arrive at the conclusion. 
    
\end{proof}
\end{lemma}

This allows us to conclude the main statement in the soundness boost. 

\begin{theorem}[Restatement of Theorem~\ref{thm: main-high-rate-surgery-distance-preserving}]\label{thm: formal-surgery-distance-preserving}
    Let $\mathcal{S}[\tilde{\mathcal{G}}; (\Phi_1, \Phi_0)]$ be a (hypergraph) surgery gadget to a data code $\mathcal{Q}$ with distance $d$, where $\Phi_1$ is elementary (its column weight at most $1$), with (hyper)graph chain $\tilde{\mathcal{G}}: \mathbb{F}_2[\mathcal{V}] \xrightarrow{\partial_1} \mathbb{F}_2[\mathcal{E}] \xrightarrow{\partial_0} \mathbb{F}_2[\mathcal{C}]$ and let the $Z$-distance of $\tilde{\mathcal{G}}$ be $d(\tilde{\mathcal{G}})$, with $d(\tilde{\mathcal{G}}) := \infty$ when $\ker \partial_0 \setminus \IM \partial_1 = \emptyset$ (in particular, when there is no logical operator of type~\ref{surgery-logical-ancilla}). Let $\tilde{\mathcal{G}}_m$ be the length-$m$ thickened hypergraph, and let $\mc{S}[\tilde{\mc{G}}_m; (\tilde{\Phi}_1, \tilde{\Phi}_0)]$ be the corresponding surgery gadget, measuring $ \Phi_1(\ker \partial_1)$ in parallel. Then the following statements hold. 
    \begin{enumerate}[label=(\roman*)]
        \item The $m$-thickened surgery gadget $\mathcal{S}[\tilde{\mathcal{G}}_m; (\tilde{\Phi}_1, \tilde{\Phi}_0)]$ measures $ \Phi_1(\ker \partial_1)$ in parallel and is distance-preserving whenever $m \geq d$. 
        \item Suppose $m\, d(\tilde{\mathcal{G}}) \geq d$ and that $\mathcal{S}[\tilde{\mathcal{G}}; (\Phi_1, \Phi_0)]$ is $(t, \rho_t)$-sound for some $t$. Then the $m$-thickened surgery gadget $\mathcal{S}[\tilde{\mathcal{G}}_m; (\tilde{\Phi}_1, \tilde{\Phi}_0)]$ is distance-preserving whenever $mt \geq d$ and $m\rho_t \geq 1$.

        \item Suppose $m\, d(\tilde{\mathcal{G}}) \geq d$. 
        Then $\mc{S}[\tilde{\mc{G}}_m; (\tilde{\Phi}_1, \tilde{\Phi}_0)]$ is distance-preserving if 
        \begin{align}
             m \lambda(\partial_1) \geq d
        \end{align}
        where $ \lambda(\partial_1) := \min_{u \notin \ker \partial_1} |\partial_1 u|$.
       
    \end{enumerate}
    \begin{proof}
    The second assertion follows from $\Phi_1$ being an elementary connectivity map; hence, from Lemma~\ref{lemma: soundness-boost-thickening}, $\tilde{\partial}_1$ is $(mt, \min(m\rho_t, 1))$-sound relative to $\tilde{\Phi}_1$. Hence, the result follows from the common-representative argument in the proof of Lemma~\ref{lemma: distance-preserving-soundness}, whose $d(\mathcal{G})$-branch is discharged by $d(\tilde{\mathcal{G}}_m) \geq m\, d(\tilde{\mathcal{G}}) \geq d$: modulo $\IM \tilde{\partial}_1$, a type~\ref{surgery-logical-ancilla} operator of $\tilde{\mathcal{G}}_m$ restricts in each of the $m$ layers to one common class of $\ker \partial_0 / \IM \partial_1$, contributing weight $d(\tilde{\mathcal{G}})$ per layer when the class is nontrivial, and lying in $\IM \tilde{\partial}_1$ otherwise; the same bound serves the third assertion.

 For the third assertion, we only need to show that the added checks do not decrease the merged code distance. Let $l$ be the support vector of, say, a $Z$-type logical operator. Multiplying by the added vertex checks gives an operator of weight 
    \begin{align}
        |l + \tilde{\Phi}_1(y)| + |\tilde{\partial}_1 y|. 
    \end{align}
    By Lemma~\ref{lemma: minimum-weight-syndrome-expansion}, $|\tilde{\partial}_1 y| \geq \min(\mathrm{dist}_{\tilde{\Phi}_1}(y, \ker \tilde{\partial}_1), m \lambda(\partial_1))$. If the minimum is $\mathrm{dist}_{\tilde{\Phi}_1}(y, \ker \tilde{\partial}_1)$, let $h \in \ker \tilde{\partial}_1$ attain it, so that $|\tilde{\partial}_1 y| \geq |\tilde{\Phi}_1(y + h)|$, and set $l' := l + \tilde{\Phi}_1(h)$. Since $\tilde{\Phi}_1(h)$ is a measured logical operator and $l$ represents a logical operator surviving the measurement, $l'$ is again a nontrivial logical operator of $\mathcal{Q}$, whence
    \begin{align}
       \begin{aligned}
        |l + \tilde{\Phi}_1(y)| + |\tilde{\partial}_1 y| &\geq |l'| - |\tilde{\Phi}_1(y + h)| + |\tilde{\Phi}_1(y + h)| \\
        &= |l'| \geq d.
       \end{aligned}
    \end{align}
    For the other case, we have
    \begin{align}
         |l + \tilde{\Phi}_1(y)| + |\tilde{\partial}_1 y| \geq m \lambda(\partial_1) \geq d
    \end{align}
    as desired. If $\Phi_1(y)$ is a nontrivial logical operator only for $y \in \ker \partial_1$, then $\tilde{\Phi}_1(y) \neq l$, so $|l + \tilde{\Phi}_1(y)| \geq 1$ and the bound improves to $1 + m \lambda(\partial_1)$. We now, for completeness, give the first assertion, which holds even if there is a logical operator of type~\ref{surgery-logical-ancilla}. Since $\Phi_1$ is elementary and we can always decompose $y = y' + y''$ for $y' \in \mathbb{F}_2^{|\mathcal{V}|} \otimes (\mathbb{F}_2^m \setminus \ker D)$ and $y'' \in \mathbb{F}_2^{|\mathcal{V}|} \otimes \ker D$, a similar technique applies to conclude that for any logical operator of type~\ref{surgery-logical-data}, the distance is preserved under the action of the added (vertex) checks. It remains to prove that the distance is also preserved for any logical operator of type~\ref{surgery-logical-ancilla}. This holds because, from Lemma~\ref{lemma: surgery-cone-code-logical-operators} and Lemma~\ref{lemma: distance-preserving-soundness}, any $Z$-type logical operator of type~\ref{surgery-logical-ancilla} must necessarily have weight at least $m \geq d$. 
    \end{proof}
\end{theorem}

The results presented in this section hold generally for (hyper)graph surgery for low-rate and high-rate logical Pauli-product measurements, discussed in the literature~\cite{he2025extractorsqldpcarchitecturesefficient, Ide_2025, cowtan2025fastfaulttolerantlogicalmeasurements, CowtanHeWilliamsonYoder2025ParallelCodeSurgery, zheng2025highratesurgeryconstantoverheadlogical, ZhangLi2025TimeEfficientLogicalOperations}, while a key innovation is the use of a single ancilla to connect both $Z$-type and $X$-type measurements. Specifically, with a (slight) modification of the ancilla, we can similarly perform arbitrarily addressable, logically disjoint Pauli-product measurements, preserving the distance of the merged code, according to Theorem~\ref{thm: main-high-rate-surgery-distance-preserving} and Theorem~\ref{thm: formal-surgery-distance-preserving}. 

\section{Explicit, fault-tolerant surgery gadgets for canonical LP codes}\label{sec: canonical-LP-surgery-construction}

In this section, we present the explicit construction of the canonical surgery gadgets for the canonical lifted-product (LP) codes, which marries the algebraic basis characterization in Section~\ref{app-sec: lp-basis-characterization} and surgery techniques in Section~\ref{sec: code-surgery-general}.
\subsection{Lifted graph and spectral criterion}\label{sec: spectral-voltage-graph-theory}

We begin this section by giving a streamlined introduction on voltage graphs/graph liftings, which helps us to formally analyze the cyclic extractor. 

\begin{definition}[Orientation of a graph]\label{def: graph-orientation}
    Let $\mathcal G(\mathcal V,\mathcal E, \mathrm v)$ be an undirected graph (Definition~\ref{def: formal-graph-definition}), so that $\mathrm{v}(e) = \{\mathrm{v}_0(e), \mathrm{v}_1(e)\}$ is defined up to the exchange of $\mathrm{v}_0(e)$ and $\mathrm{v}_1(e)$. An \emph{orientation} of $\mathcal{G}$ is a choice of bits $\varepsilon \in \{0, 1\}^{\mathcal{E}}$, which assigns to each edge $e$ the ordered pair
    \begin{align}
        \mathrm{o}(e) := \big( \mathrm{v}_{\varepsilon(e)}(e),\ \mathrm{v}_{\varepsilon(e)+1}(e) \big) \in \mathcal{V} \times \mathcal{V},
    \end{align}
    the index being taken modulo $2$. Once an orientation is fixed, we relabel the endpoints of every edge so that $\mathrm{o}(e) = (\mathrm{v}_0(e), \mathrm{v}_1(e))$, i.e., $e$ is directed from $\mathrm{v}_0(e)$ to $\mathrm{v}_1(e)$, which turns $\mathcal{G}$ into the directed graph $(\mathcal{V}, \mathcal{E}, \mathrm{v}^{\mathrm{dir}})$ with the same vertices and edges; all formulas below are written in this convention. Reversing the orientation of $e$ means $\varepsilon(e) \mapsto \varepsilon(e) + 1$, which exchanges $\mathrm{v}_0(e)$ and $\mathrm{v}_1(e)$. For a self-loop at $u$, $\mathrm{v}_0(e) = \mathrm{v}_1(e) = u$ for either choice, and the orientation of $e$ is recorded by $\varepsilon(e)$ alone.
\end{definition}

We follow the notations from Ref.~\cite{golowich2024quantumldpccodestransversal}: 

\begin{definition}[Group-lift of graphs]\label{def: group-lift}
    Let $G$ be a finite group and let $\mathcal{G} = (\mathcal{V}, \mathcal{E}, \mathrm{v})$ be an undirected graph with an orientation $\mathrm{o}$, and let $L: \mathcal{E} \rightarrow G$ be a \emph{voltage assignment}. The $G$-lift of $\mathcal{G}$ is the directed graph $\hat{\mathcal{G}} = (\hat{\mathcal{V}}, \hat{\mathcal{E}}, \hat{\mathrm{v}}^{\mathrm{dir}})$ with vertex set $\hat{\mathcal{V}} = \mathcal{V} \times G$, edge set $\hat{\mathcal{E}} = \mathcal{E} \times G$, and
	\begin{align}
        \hat{\mathrm{v}}^{\mathrm{dir}}((e, g)) = \big( (\mathrm{v}_0(e), g),\ (\mathrm{v}_1(e), g \cdot L(e)) \big),
    \end{align}
    that is, the edge $(e, g)$ is directed from $(\mathrm{v}_0(e), g)$ to $(\mathrm{v}_1(e), g L(e))$. The voltage assignment is understood relative to the orientation: reversing the orientation of $e$ and replacing $L(e)$ by $L(e)^{-1}$ reverses the directed edge $(e, g)$ and relabels it as $(e, g L(e))$, so that the underlying undirected graph of $\hat{\mathcal{G}}$ does not depend on the orientation. For every $h \in G$, the maps $(v, g) \mapsto (v, hg)$ and $(e, g) \mapsto (e, hg)$ define an automorphism of $\hat{\mathcal{G}}$, so that $G$ acts freely on $\hat{\mathcal{G}}$ from the left with quotient $\mathcal{G}$; we refer to $\{v\} \times G$ and $\{e\} \times G$ as the vertex and edge \emph{fibres}.
\end{definition}

Operationally, the orientation assigns a direction to each edge. For example, consider an undirected edge $e$ with $\mathrm{v}(e) = \{\mathrm{v}_0(e), \mathrm{v}_1(e)\} =: \{ v_0, v_1\}$. An orientation $\mathrm{o}(e) = (v_0, v_1)$ implies that we assign $\mathrm{v}_0(e) = v_0$ and $\mathrm{v}_1(e) = v_1$. In this way, in what follows, we consider explicitly directed graph upon lifting, with the direction of each edge implicitly determined by the orientation $\mathrm{o}$. The voltage assignment $L$ is defined relative to the orientation: reversing the orientation of $e$ replaces $L(e)$ by $L(e)^{-1}$. This ensures that the Laplacian and adjacency matrices have real eigenvalues. Now we define the graph (co)chains: 

\begin{definition}[Graph (co)chain complex]\label{def: graph-cochain-complex}
    Let $\mathbb{K}$ be a field, let $\mathrm{o}$ be an orientation of $\mathcal{G}$ (Definition~\ref{def: graph-orientation}), and let $\mathcal{C}$ be a finite set of cycles of $\mathcal{G}$, each written as in Definition~\ref{def: path-cycle}, $c = (v_0, e_1, v_1, \ldots, e_r, v_r = v_0)$, together with the signs $\varepsilon_{c, e_i} \in \{\pm 1\}$ recording the direction in which $c$ traverses $e_i$ relative to the orientation,
    \begin{align}
        \varepsilon_{c, e_i} =
	\begin{cases}
	+1 & \text{if } (v_{i-1}, v_i) = (\mathrm{v}_0(e_i), \mathrm{v}_1(e_i)), \\
	-1 & \text{if } (v_{i-1}, v_i) = (\mathrm{v}_1(e_i), \mathrm{v}_0(e_i));
	\end{cases}
    \end{align}
    for a self-loop $e_i$ both cases apply. The \emph{graph cochain complex} $C(\mathcal{G}, \mathbb{K})$ is
	\begin{align}
        0 \rightarrow C^0(\mathcal{V}, \mathbb{K}) \xrightarrow{\partial_1} C^1(\mathcal{E}, \mathbb{K}) \xrightarrow{\partial_0} C^2(\mathcal{C}, \mathbb{K}) \rightarrow 0,
    \end{align}
	where the coboundary maps, indexed as in the graph chain of Definition~\ref{def: graph-chains}, are defined by
	\begin{align}
        (\partial_1 x)(e) &= x(\mathrm{v}_1(e)) - x(\mathrm{v}_0(e)), \qquad x \in C^0(\mathcal{V}, \mathbb{K}), \\
        (\partial_0 y)(c) &= \sum^{r}_{i=1} \varepsilon_{c, e_i}\, y(e_i), \qquad\quad\ y \in C^1(\mathcal{E}, \mathbb{K}).
    \end{align}
	The \emph{graph chain complex} consists of the adjoints $\partial^{*}_1$, $\partial^{*}_0$ with respect to the standard bilinear form, $\langle \partial x, y \rangle_{\mathbb{K}} = \langle x, \partial^{*} y \rangle_{\mathbb{K}}$, so that
	\begin{align}
        \partial^{*}_1(e) = \mathrm{v}_1(e) - \mathrm{v}_0(e), \qquad \partial^{*}_0(c) = \sum^{r}_{i=1} \varepsilon_{c, e_i}\, e_i .
    \end{align}
	We say that $\mathcal{C}$ is \emph{complete} if it spans the cycle space $\ker \partial^{*}_1$, and \emph{minimal} if it is moreover linearly independent.
\end{definition}

\begin{remark}
    Along a cycle $c$, the terms $x(v_i) - x(v_{i-1})$ telescope, so $\partial_0 \partial_1 = 0$ and $C(\mathcal{G}, \mathbb{K})$ is a cochain complex. The chain is exact at $C^1$, $\ker \partial_0 = \IM \partial_1$, if and only if $\mathcal{C}$ is complete, and $\partial_0$ is surjective if and only if $\mathcal{C}$ is minimal. We do not assume $\mathcal{C}$ minimal: an over-complete $\mathcal{C}$ is what preserves the symmetry of a lift (Lemma~\ref{lemma: graph-chain-via-group-lift}). Over $\mathbb{K} = \mathbb{F}_2$ the signs disappear and $\partial_1$, $\partial_0$ reduce to the graph chain $\mathbb{F}_2[\mathcal{V}] \xrightarrow{\partial_1} \mathbb{F}_2[\mathcal{E}] \xrightarrow{\partial_0} \mathbb{F}_2[\mathcal{C}]$ of Definition~\ref{def: graph-chains}, and completeness of $\mathcal{C}$ is the exactness condition $\IM \partial^{T}_0 = \ker \partial^{T}_1$ used throughout. 
\end{remark}

\begin{remark}
    Finally, reversing the orientation of $e$ multiplies row $e$ of $\partial_1$ and column $e$ of $\partial_0$ by $-1$, a diagonal change of basis of $C^1(\mathcal{E}, \mathbb{K})$, so kernels, images, and cohomology, and, as we see later, adjacency and graph Laplacian do not depend on the orientation.
\end{remark}

A convenient presentation is to parameterize the (co)boundary in terms of the group elements (group algebra): 

\begin{lemma}[Graph (co)chain via group lift]\label{lemma: graph-chain-via-group-lift}
    Let $\hat{\mathcal{G}}$ be the $G$-lift of $\mathcal{G}$ with voltage assignment $L$ (Definition~\ref{def: group-lift}), let $\mathbb{K}$ be a field, and let $\hat{\mathcal{C}}$ be the $G$-orbit of a finite set $\mathcal{C}$ of cycles of $\hat{\mathcal{G}}$. Identifying each fibre with the group algebra $\mathbb{K}[G]$, the cochain complex of $\hat{\mathcal{G}}$ over $\mathbb{K}$ is the base cochain complex evaluated over $\mathbb{K}[G]$,
    \begin{align}
        C^0(\mathcal{V}, \mathbb{K}[G]) \xrightarrow{\partial^G_1} C^1(\mathcal{E}, \mathbb{K}[G]) \xrightarrow{\partial^G_0} C^2(\mathcal{C}, \mathbb{K}[G]),
    \end{align}
    with the twisted coboundaries
    \begin{align}\label{eq: twisted-coboundary}
        \begin{aligned}
            \partial^G_1(x)(e) &= L(e)\, x(\mathrm{v}_1(e)) - x(\mathrm{v}_0(e)), \\
            \partial^G_0(y)(c) &= \sum^{r}_{i=1} \varepsilon_{c, e_i}\, g_i\, y(e_i),
        \end{aligned}
    \end{align}
    where $(e_1, g_1), \ldots, (e_r, g_r)$ are the edges of the cycle $c \in \mathcal{C}$ and $\varepsilon_{c, e_i}$ are its signs (Definition~\ref{def: graph-cochain-complex}). 
    \begin{proof}
        Identify $\hat{x}(v, \cdot) \in \mathbb{K}^{G}$ with $x(v) := \sum_{g \in G} \hat{x}(v, g)\, g^{-1} \in \mathbb{K}[G]$, and likewise for edges and cycles. By Definition~\ref{def: group-lift}, $(\partial_1 \hat{x})(e, g) = \hat{x}(\mathrm{v}_1(e), g L(e)) - \hat{x}(\mathrm{v}_0(e), g)$; multiplying by $g^{-1} = L(e)\,(g L(e))^{-1}$ and summing over $g$ gives $\partial^G_1$. The cycle $g \cdot c$ has edges $(e_i, g g_i)$, so $(\partial_0 \hat{y})(g \cdot c) = \sum_i \varepsilon_{c, e_i}\, \hat{y}(e_i, g g_i)$; multiplying by $g^{-1} = g_i\, (g g_i)^{-1}$ and summing over $g$ gives $\partial^G_0$. The standard form on $\mathbb{K}^{\mathcal{V} \times G}$ is the coefficient of the identity in $\sum_v x(v)^{*} y(v)$, so adjoints are the transposes with entrywise involution.
    \end{proof}
\end{lemma}

The chain complex is given by the adjoints $(\partial^G_1)^{*}$ and $(\partial^G_0)^{*}$, the transposes with the entrywise involution $g \mapsto g^{-1}$.

\begin{remark}
   Lifting a cycle basis of $\mathcal{G}$ never yields a complete $\hat{\mathcal{C}}$: a cycle of $\mathcal{G}$ lifts to cycles of $\hat{\mathcal{G}}$ only if the product of its voltages $L(e_i)^{\varepsilon_{c, e_i}}$ along the cycle equals $1$, while $\dim \ker \partial^{*}_1 = |G|\,(|\mathcal{E}| - |\mathcal{V}|) + 1$ for connected $\hat{\mathcal{G}}$.
\end{remark}

For $G = \mathbb{Z}_l$ written multiplicatively as $\{x^{g}\}$ and $\mathbb{K} = \mathbb{F}_2$, the group algebra is $R = \mathbb{F}_2[x]/(x^l + 1)$, the signs vanish, and row $e$ of $\partial^G_1$ is $e_{\mathrm{v}_0(e)} + L(e)\, e_{\mathrm{v}_1(e)}$ --- up to a unit, the row shape of the cycle-graphified $\partial_1$ used in Definition~\ref{def: lp-column-canonical-extractor}, with voltage $L(e)$ the ratio of the two cyclically adjacent monomial terms; $\partial^G_0$ is the $R$-linear cycle-edge incidence matrix $\partial_0$, with $\IM \partial^{T}_0 = \ker \partial^{T}_1$ exactly when $\hat{\mathcal{C}}$ is complete.

With this knowledge, we can define the graph Laplacian and the Cheeger constant.

\begin{definition}[Graph Laplacian and Cheeger constant]\label{def: graph-Laplacian-Cheeger}
    Let $\mathcal{G}(\mathcal{V}, \mathcal{E}, \mathrm{v})$ be an undirected (multi)graph. Its \emph{graph Laplacian} is
    \begin{align}
        L_{\mathcal{G}} := \partial^{*}_1 \partial_1 ,
    \end{align}
    where $\partial_1$ is the coboundary of Definition~\ref{def: graph-cochain-complex} over $\mathbb{C}$. For $S \subseteq \mathcal{V}$, let $\partial S$ be the set of edges with exactly one endpoint in $S$. The \emph{Cheeger constant} of $\mathcal{G}$ is
    \begin{align}
        h(\mathcal{G}) := \min_{\emptyset \neq S \subsetneq \mathcal{V}} \frac{|\partial S|}{\min(|S|, |\mathcal{V} \setminus S|)} .
    \end{align}
\end{definition}

\begin{proposition}[Spectral bounds for the Cheeger constant]\label{prop: Cheeger-spectral-bounds}
    Let $\hat{\mathcal{G}}$ be the $G$-lift of $\mathcal{G}$ and let $L_{\hat{\mathcal{G}}} = (\partial^G_1)^{*} \partial^G_1 \in \mathbb{C}[G]^{\mathcal{V} \times \mathcal{V}}$ be its graph Laplacian in the presentation of Lemma~\ref{lemma: graph-chain-via-group-lift}. Then $L_{\hat{\mathcal{G}}}$ is Hermitian, $L^{*}_{\hat{\mathcal{G}}} = L_{\hat{\mathcal{G}}}$, positive semidefinite, and independent of the orientation; its eigenvalues are real,
    \begin{align}
        0 = \mu_1 \leq \mu_2 \leq \cdots \leq \mu_{|\hat{\mathcal{V}}|},
    \end{align}
    and $\mu_2 > 0$ if and only if $\hat{\mathcal{G}}$ is connected. In that case, with $D_{\max}$ the maximum degree of $\hat{\mathcal{G}}$,
    \begin{align}\label{eq: Cheeger-spectral-bounds}
        \frac{\mu_2}{2} \leq h(\hat{\mathcal{G}}) \leq \sqrt{2 D_{\max}\, \mu_2}.
    \end{align}
    \begin{proof}
        By Lemma~\ref{lemma: graph-chain-via-group-lift}, under the regular representation $\partial^G_1$ is the coboundary $\partial_1$ of $\hat{\mathcal{G}}$, a real matrix, and $(\partial^G_1)^{*}$ is its transpose; hence $L_{\hat{\mathcal{G}}} = \partial^{T}_1 \partial_1$ is real symmetric and positive semidefinite. Reversing the orientation of the edges in $F \subseteq \mathcal{E}$ exchanges $\mathrm{v}_0(e)$, $\mathrm{v}_1(e)$ and replaces $L(e)$ by $L(e)^{-1}$ for $e \in F$ (Definition~\ref{def: group-lift}), that is, $\partial^G_1 \mapsto D\, \partial^G_1$ with $D := \mathrm{diag}(\delta_e)_{e \in \mathcal{E}}$, $\delta_e = -L(e)^{-1}$ for $e \in F$ and $\delta_e = 1$ otherwise. Since $D^{*} D = I$, $(D \partial^G_1)^{*} (D \partial^G_1) = (\partial^G_1)^{*} \partial^G_1$. The lower bound in Eq.~\eqref{eq: Cheeger-spectral-bounds} is the case $t = \infty$ of Lemma~\ref{lemma: spectral-criterion-soundness} below, and the upper bound is the Cheeger inequality for the unnormalized Laplacian~\cite{Mohar1989}.
    \end{proof}
\end{proposition}
\begin{remark}
    Proposition~\ref{prop: Cheeger-spectral-bounds} and its proof apply verbatim to any undirected (multi)graph, since only $L = \partial^{*}_1 \partial_1$ is used; in particular to the base graph $\mathcal{G}$. The lift enters only through the presentation of $L_{\hat{\mathcal{G}}}$ over $\mathbb{C}[G]$: evaluating its entries at the irreducible representations of $G$ decomposes the spectrum into blocks of size $|\mathcal{V}| \cdot \dim \lambda$, and for $G = \mathbb{Z}_l$, at $x \mapsto e^{2\pi i k/l}$ for $k = 0, \ldots, l-1$, into $l$ Hermitian blocks of size $|\mathcal{V}|$, which is how $\mu_2$ is computed in Section~\ref{sec: LP-canonical-extractor}.
\end{remark}

\begin{lemma}[Spectral criterion of small-set soundness]\label{lemma: spectral-criterion-soundness}
    Let $\hat{\mathcal{G}}$ be a connected $G$-lift with $\hat{n} := |\hat{\mathcal{V}}|$ vertices, spectral gap $\mu_2 > 0$ (Proposition~\ref{prop: Cheeger-spectral-bounds}), and edge-vertex incidence matrix $\partial_1$ over $\mathbb{F}_2$. For $t \geq 1$, let $m^{*}$ be the largest integer $1 \leq m \leq \hat{n}/2$ with $\mu_2\, m(\hat{n} - m)/\hat{n} < t$. Then $\partial_1$ is $(t, \rho_t)$-sound with
    \begin{align}\label{expr: main-spectral-bound-small-set-soudness}
        \rho_t \geq \mu_2\, \frac{\hat{n} - m^{*}}{\hat{n}} ;
    \end{align}
    if $m^{*} = 0$, the soundness condition is vacuous. Moreover, every $u \notin \ker \partial_1$ with $|\partial_1 u| < t$ satisfies $\mathrm{dist}(u, \ker \partial_1) \leq m^{*}$, and $t = \infty$ recovers $h(\hat{\mathcal{G}}) \geq \mu_2/2$.
    \begin{proof}
        By the proof of Proposition~\ref{prop: Cheeger-spectral-bounds}, $L_{\hat{\mathcal{G}}}$ acts on $C^0(\hat{\mathcal{V}}, \mathbb{C})$ as $\partial^{*}_1 \partial_1$ for the coboundary $\partial_1$ of $\hat{\mathcal{G}}$ over $\mathbb{C}$, whose reduction modulo $2$ is the incidence matrix $\partial_1$ of the statement. Let $S \subseteq \hat{\mathcal{V}}$ with indicator $\mathbf{1}_S \in \mathbb{C}^{\hat{\mathcal{V}}}$. Then $\mathbf{1}^{T}_S L_{\hat{\mathcal{G}}}\, \mathbf{1}_S = \|\partial_1 \mathbf{1}_S\|^2$; row $\hat{e}$ of $\partial_1$ is $x(\mathrm{v}_1(\hat{e})) - x(\mathrm{v}_0(\hat{e}))$, which at $x = \mathbf{1}_S$ equals $\pm 1$ on the edges of $\partial S$ and $0$ otherwise. Hence
        \begin{align}\label{eq: cut-quadratic-form}
            |\partial S| = \mathbf{1}^{T}_S L_{\hat{\mathcal{G}}}\, \mathbf{1}_S .
        \end{align}
        Decompose $\mathbf{1}_S = (|S|/\hat{n})\, \mathbf{1} + y$ with $y \perp \mathbf{1}$, so that $\|y\|^2 = |S| - 2|S|^2/\hat{n} + |S|^2/\hat{n} = |S|(\hat{n} - |S|)/\hat{n}$. Since $L_{\hat{\mathcal{G}}} \mathbf{1} = 0$ and $\mathbf{1}$ spans $\ker L_{\hat{\mathcal{G}}}$ by connectivity, the Courant--Fischer theorem gives
        \begin{align}\label{eq: cut-spectral-lower-bound}
            |\partial S| = y^{T} L_{\hat{\mathcal{G}}}\, y \geq \mu_2 \|y\|^2 = \mu_2\, \frac{|S|(\hat{n} - |S|)}{\hat{n}} .
        \end{align}
        Now let $u \notin \ker \partial_1$ with $|\partial_1 u| < t$, and let $S := \mathrm{supp}(u)$. Over $\mathbb{F}_2$, $(\partial_1 u)(\hat{e}) = u(\mathrm{v}_0(\hat{e})) + u(\mathrm{v}_1(\hat{e}))$ is $1$ exactly on $\partial S$, so $|\partial_1 u| = |\partial S|$. As $\hat{\mathcal{G}}$ is connected, $\ker \partial_1 = \{0, \mathbf{1}\}$, so $\mathrm{dist}(u, \ker \partial_1) = \min(|u|, |u + \mathbf{1}|) = \min(|S|, \hat{n} - |S|) =: m \in [1, \hat{n}/2]$. Since $|S|(\hat{n} - |S|) = m(\hat{n} - m)$, Eq.~\eqref{eq: cut-spectral-lower-bound} gives $\mu_2\, m(\hat{n} - m)/\hat{n} \leq |\partial S| < t$, and since $m(\hat{n} - m)$ is increasing in $m$ on $[0, \hat{n}/2]$, the definition of $m^{*}$ forces $m \leq m^{*}$; in particular no such $u$ exists if $m^{*} = 0$. Therefore
        \begin{align}
            |\partial_1 u| = |\partial S| \geq \mu_2\, \frac{\hat{n} - m}{\hat{n}}\; m \geq \mu_2\, \frac{\hat{n} - m^{*}}{\hat{n}}\; \mathrm{dist}(u, \ker \partial_1),
        \end{align}
        which is $(t, \rho_t)$-soundness with $\rho_t = \mu_2 (\hat{n} - m^{*})/\hat{n}$. Finally, for $t = \infty$ we have $m^{*} = \lfloor \hat{n}/2 \rfloor$, and $(\hat{n} - m)/\hat{n} \geq 1/2$ for $m \leq \hat{n}/2$, so Eq.~\eqref{eq: cut-spectral-lower-bound} gives $|\partial S| \geq \frac{\mu_2}{2} \min(|S|, \hat{n} - |S|)$ for every $S$, i.e., $h(\hat{\mathcal{G}}) \geq \mu_2/2$.
    \end{proof}
\end{lemma}

\begin{remark}
   Lemma~\ref{lemma: spectral-criterion-soundness} localizes every potential violation of $(t, \rho_t)$-soundness to supports of size at most $m^{*}$, so the bound can be sharpened by minimizing $|\partial S|$ exactly over $|S| = m$ for $m \leq m^{*}$, which is how the values in Table~\ref{tab: LP-canonical-extractor-resources} are certified.

\end{remark}

\subsection{Hardware implementability with LP code symmetry}\label{sec: LP-hardware}
We briefly sketch how the symmetry of LP codes and their extractors would improve hardware implementability. From a theoretical viewpoint, additional experimental constraints on the syndrome extraction cycle (SEC) time should remain independent of the lift size $l$.

 We can argue this from a Tanner graph viewpoint. To illustrate how code symmetry can aid experimental syndrome extraction, we construct a schedule by decomposing the Tanner graph edges into layers: 
\begin{align}\label{expr: tanner-partial-matching-layer}
    (\mathcal{T}_1, \cdots, \mathcal{T}_{\Delta T}),
\end{align}
where each $\mathcal{T}_t$ is a partial matching of the tanner graph, i.e., a set of edges such that no two edges share a common vertex.

A Tanner graph is a bipartite graph with vertices labelled with checks and qubits, respectively. For $\LP_l(A, B)$ over an Abelian ring $R$, with $A \in R^{m_A \times n_A}$ and $B \in R^{m_B \times n_B}$, the (co)boundary maps in Eq.~\eqref{eq: main-LP-boundary-maps-ring} are: 
\begin{align*}
    \partial^{M}_1 &= \left( A \otimes_R I_{m_B} \mid I_{m_A} \otimes_R B \right), \\
    (\partial^{M}_2)^* &= \left( I_{n_A} \otimes_R B^* \mid A^* \otimes_R I_{n_B} \right).
\end{align*} 
For convenience, we will only consider the Tanner graph corresponding to $\partial^{M}_1$ ($X$-type checks), as the analysis for $\partial^{M}_2$ is analogous. Let $\mathcal{G}(\mathcal{V}, \mathcal{E})$ (we do not need to consider cycles $\mathcal{C}$) be constructed as follows: 
\begin{align}
    \mathcal{V} &= \{ \text{rows of } \partial^{M}_1 \} \cup \{ \text{columns of } \partial^{M}_1 \}, \\
    \mathcal{E} &= \{ \text{monomials in each entry: } (\partial^M_1)_{ij}\neq 0 \}.
\end{align}
If $(\partial^M_1)_{ij}$ contains multiple monomials, each monomial defines a separate (multi)edge between the same check and qubit vertices. The number of edges equals the total number of nonzero monomials in $\partial^M_1$. The maximum degree of the Tanner graph equals the maximum number of nonzero monomials (the weight) in any row or column of $\partial^M_1$. By construction and Lemma~\ref{lemma: graph-chain-via-group-lift}, $\mathcal{G}(\mathcal{V}, \mathcal{E})$ over $R$ is equivalent to its lifted graph, which is the standard Tanner graph representation. Let $v^c_i$ for $i \in [m_A m_B]$ be the check vertex corresponding to the row index of $\partial^M_1$, and $v^q_j$ for $j \in [n_A m_B + m_An_B]$ be the qubit vertex corresponding to the column index of $\partial^M_1$. Then we can construct such a tuple of path-matching layers $(\mathcal{T}_1, \cdots, \mathcal{T}_{\Delta T})$ as in Eq.~\eqref{expr: tanner-partial-matching-layer} such that there are no two edges $e, e'$ such that $\mathrm{v}(e)$ and $\mathrm{v}(e)$ share any of the vertices $v^c_i$ and $v^q_j$. Hence, we set $\Delta T = \max(\omega_{\mathrm{col}}(\partial^M_1), \omega_{\mathrm{row}}(\partial^M_1))$, which is a constant and each $\mathcal{T}_t$ is a partial matching of the Tanner graph which can be still defined over $R$ and, hence, retains the (cyclic) symmetry. For many experimental platforms, including neutral atom platforms~\cite{Bhardwaj2026MittenCodes,Hong2026RateOneFifthLDPC}, the additional complexity of scheduling or reconfiguration within each $\mathcal{T}_t$ does not scale with the lift size $l$.

 For code families with fixed sizes of the prototype matrices $A$ and $B$ and varying lift size $l$, this symmetry can significantly reduce syndrome extraction complexity compared with random codes without these symmetries. 
\subsection{Construction of LP canonical extractors}\label{sec: LP-canonical-extractor}

We now present the explicit construction of the canonical extractor for the LP codes, where we start from the LP column-canonical extractor according to Definition~\ref{def: main-lp-column-canonical-extractor}. Intuitively, an extractor utilizes a fixed ancilla graph to measure an arbitrary logical operator, whereupon its fault tolerance depends on the (small-set) soundness $\rho_d$ and the added degree $\delta$. Specific to canonical LP codes, $Z$-type and $X$-type logical basis operators correspond to, respectively, columns and rows of physical fibres (See Fig.~\ref{fig: canonical-basis}). Hence, an extractor can be constructed borrowing this column/row structure, which we refer to as an \emph{LP column-canonical extractor}:

\begin{definition}[Explicit construction of LP column-canonical extractor]\label{def: lp-column-canonical-extractor}
    Let $\mathcal{Q} = \LP_l(A, B)$ be an LP canonical code with distance $d$ and $A \in R^{m_A \times n_A}$.
 Then we can construct an LP column-canonical extractor explicitly with the following protocol, for some target (small-set) soundness $\rho_d$:
    \begin{enumerate}
        \item \label{LP-col-extractor-graphify} Cycle-graphify $A \in R^{m_A \times n_A}$ according to Definition~\ref{def: graphify-parity-check} and expression~\eqref{expr: cycle-graphification-parity-check}. This always yields the \emph{$R$-linear} graph chain $\hat{\mathcal{G}}: \mathbb{F}_2[\hat{\mathcal{V}}] \xrightarrow{\hat{\partial}_1} \mathbb{F}_2[\hat{\mathcal{E}}] \xrightarrow{\hat{\partial}_0} \mathbb{F}_2[\hat{\mathcal{C}}]$.
 By Lemma~\ref{lemma: graph-chain-via-group-lift}, $\hat{\mathcal{G}}$ is the $\mathbb{Z}_l$-lift of a base graph $\mathcal{G}(\mathcal{V}, \mathcal{E}, \mathcal{C})$ with $|\mathcal{V}| = n_A$ vertices (the columns of $A$), and we write its (co)chain over the group algebra $R$, without hats:
        \begin{align}\label{eq: app-extractor-base-chain}
            \mathcal{G}: C^0(\mathcal{V}, R) \xrightarrow{\partial_1} C^1(\mathcal{E}, R) \xrightarrow{\partial_0} C^2(\mathcal{C}, R),
        \end{align}
        where $\partial_1$ and $\partial_0$ are the $R$-linear edge-vertex and cycle-edge incidence matrices given by the cycle-graphification, with $\IM \partial^T_0 = \ker \partial^T_1$ when $\hat{\mathcal{C}}$ is complete (over-complete rows allowed), per Lemma~\ref{lemma: graph-chain-via-group-lift} and the remark following it. In what follows, we work with this base presentation over $R$, identifying $\mathbb{F}_2[\hat{\mathcal{V}}] \cong C^0(\mathcal{V}, R)$, $\mathbb{F}_2[\hat{\mathcal{E}}] \cong C^1(\mathcal{E}, R)$, and $\mathbb{F}_2[\hat{\mathcal{C}}] \cong C^2(\mathcal{C}, R)$, dropping the hats. Note that we have that 
        \begin{align}
            \mathcal{E} := \bigsqcup_{a=1}^{m_A} \mathcal{E}_a
        \end{align}
        due to cycle-graphification of each row of $A$.

        \item \label{LP-col-extractor-boosting} Certify the (small-set) soundness with the spectral-criterion method of Lemma~\ref{lemma: spectral-criterion-soundness}. If not met, we add edges and new cycles until the desired (small-set) soundness $\rho_d$ is met and the updated graph chain is $R$-linear and exact. 
     
        \item \label{LP-col-extractor-static-connectivity-map} The $R$-linear ($Z$-type) connectivity maps are constructed to be $r_B^*$-elementary according to Definition~\ref{def: elementary-connectivity-surgery}, one copy per information column. For the symmetric code $\LP_l(A, A^*)$, $r_B^* = r_A = n_A - m_A$.  In particular, construct the static maps $(S^{(Z)}_1, S^{(Z)}_0)$ with
        \begin{align}\label{eq: static-Phi1-extractor}
            S^{(Z)}_1 =  \sum_{j\in[r_B^*]} I_{n_A} \otimes_R e^{(j)}_{\mathcal{B}_0} \cong \mathrm{vstack}\{ \underbrace{I_{n_A}, \dots, I_{n_A}}_{r_B^* \text{ times}} \},
        \end{align}

        where, with an abuse of notation, we drop the binarization $\mathbb{B}(\cdot)$. Note that $I_{\mathcal{V}} := I_{n_A} \in R^{n_A \times n_A}$ is the identity matrix transversely assigning vertices to qubits: 
        for the vertex $v_{m, i} \in \mathcal{V}$ of fibre position $m$ and column $i$ of $A$, we associate 
        \begin{align}\label{expr: vertex-connect-column-static}
            v_{m, i} \mapsto \begin{pmatrix}
                 q_L(i, 0, m) \\
                \vdots \\
               q_L(i, r_B^*-1, m)
            \end{pmatrix}.
        \end{align}
        Here, $S^{(X)}_1 = 0$ for the column-canonical extractor. Furthermore, 
        \begin{align}\label{eq: static-Phi0-extractor}
            S^{(Z)}_0 := \mathrm{vstack}\{\underbrace{\bigoplus^{m_A }_{a=1} \mathbf{1}_{\mathcal{E}_{a}}, \cdots, \bigoplus^{m_A }_{a=1} \mathbf{1}_{\mathcal{E}_{a}}}_{r_B^* \text{ times}} \}, 
        \end{align}
        where $\mathbf{1}_{\mathcal{E}_a}$ denotes the all-ones vector on the graphified edges of row $a$, assigned to the $X$-check of row $a$. Similarly $S^{(X)}_0 = 0$ for the column-canonical extractor.
    \end{enumerate}
\end{definition}

In addition to the monomial example in the main text (Section~\ref{subsec: main-LP-canonical-extractor} and Table~\ref{tab: LP-canonical-extractor-resources}), we give another explicit construction of an LP column-canonical extractor: 

\begin{example}[{Explicit LP column-canonical extractor for $\LP^{2 \times 3}_{36} = [[468, 36, 20]]$}]\label{example: LP-column-canonical-extractor-C36}
    The prototype matrix is given by 
    \begin{align*}
        A = \begin{pmatrix}
            1 & 1 + x^{11} + x^{29} & 1 \\
            1 & 0 & x^{8} + x^{13} + x^{16} + x^{17}
        \end{pmatrix} \in R^{2 \times 3},
    \end{align*}
    with $l = 36$, information column $j^* = 2$, and $r_A = r_B^* = 1$. Each row carries $r_a = 5$ monomial terms (the zero entry contributes none), so its graphification cycle traverses $5$ edge orbits; cyclically adjacent terms in the same column give intra-column edges (the single-column rows below). The base (co)chain over $R$, per Eq.~\eqref{eq: app-extractor-base-chain}, is
    \begin{align}
        \partial_1 &= \begin{pmatrix}
            1 & 1 & 0 \\
            0 & 1 + x^{11} & 0 \\
            0 & x^{11} + x^{29} & 0 \\
            0 & x^{29} & 1 \\
            1 & 0 & 1 \\
            1 & 0 & x^{8} \\
            0 & 0 & x^{8} + x^{13} \\
            0 & 0 & x^{13} + x^{16} \\
            0 & 0 & x^{16} + x^{17} \\
            1 & 0 & x^{17}
        \end{pmatrix}, \\
        \partial_0 &= \mathrm{vstack}\{ \mathbf{1}_{\mathcal{E}_1} \oplus \mathbf{1}_{\mathcal{E}_2},\; \partial'_0 \},
    \end{align}
    where the two all-ones covectors read the two graphification cycles ($\mathcal{E} = \mathcal{E}_1 \sqcup \mathcal{E}_2$) and $\partial'_0$ consists of further $R$-orbits of cycles completing exactness. The cycle set is over-complete, of binary rank $253 = |\hat{\mathcal{E}}| - |\hat{\mathcal{V}}| + 1$. The lifted graph chain has $(|\hat{\mathcal{V}}|, |\hat{\mathcal{E}}|, |\hat{\mathcal{C}}|) = (108, 360, {324})$, graph degree $\Delta = 10$, and $\mu_2 = 1.506$.
 The static maps of Definition~\ref{def: lp-column-canonical-extractor} are $1$-elementary and, since $r_A = 1$, the $Z$-type and symmetric-case measurements already cover all logical operators.
\end{example}

By Lemma~\ref{lemma: spectral-criterion-soundness} of the previous subsection, such an LP column-canonical extractor can always be constructed with any constant soundness $\rho_d \geq r_B^*$.
 In what follows, we consider the symmetric code $\mathcal{Q} = \LP_l(A, A^*)$, for which $r_B^* = r_A$, and assume 
\begin{align}
    \mathcal{X}[\mathcal{G}; (S^{(Z)}_1, S^{(Z)}_0)]
\end{align}
is an LP column-canonical extractor. It can measure any $Z$-type logical operator through the masking: 
\begin{align}
  \begin{aligned}
            \Phi^{(L_Z)}_1 &= \mathrm{diag}\{M^{(L_Z)}_{1, 1}, \cdots, M^{(L_Z)}_{1, r_A} \} S^{(Z)}_1 \\
            \Phi^{(L_Z)}_0 &= \mathrm{vstack}\{ \oplus_a\mathbf{1}_{\mathcal{E}_a} M^{(L_Z)}_{0, 1}, \cdots, \oplus_a\mathbf{1}_{\mathcal{E}_a} M^{(L_Z)}_{0, r_A} \},
        \end{aligned}
\end{align}
where $M^{(L_Z)}_{1, i}$ is the diagonal projection matrix to the support of the $Z$-type logical operator $\bar{L}^{(i)}_Z$ and 
\begin{align}
    M^{(L_Z)}_{0, i} := \bigoplus^{m_A }_{a=1} \Pi^{L^{(i)}_Z}_{a}, 
\end{align}
Here $\bar{L}_Z = \prod^{r_A}_{i=1} \bar{L}^{(i)}_Z$, and $\Pi^{L^{(i)}_Z}_{a}$ denotes, for each $X$-check index $a$, the projection to the set of path-matching edges in $\mathcal{E}_a$. This choice makes the commutation relations in Definition~\ref{def: main-graph-surgery-gadget} (resp.\ Definition~\ref{def: surgery-cone-code}) hold. Then this is a linear extractor (Definition~\ref{def: linear-extractor}) measuring arbitrary $Z$-type logical operators. It is distance-preserving and its merged code is sparse, by Lemmas~\ref{lemma: Z-type-linear-extractor},~\ref{lemma: graphify-parity-check-sum}, and~\ref{lemma: distance-preserving-soundness}. By the canonical extractor desiderata, Definition~\ref{def: canonical-extractor-desiderata}, we say that the linear extractor $\mathcal{X}[\mathcal{G}; (S^{(Z)}_1, S^{(Z)}_0)]$ measures a $Z$-type logical operator $\bar{L}_Z$ with soundness $\rho_d \geq r_A$ and added degree $\delta$. 

Now we detail the protocol for upgrading the LP column-canonical extractor $\mathcal{X}[\mathcal{G}; (S^{(Z)}_1, S^{(Z)}_0)]$ to measure an arbitrary logical operator. Let a general logical operator be given by 

\begin{align}
    \bar{L} = \underbrace{\bar{L}_X \bar{L}_Z \bar{L}^{\mathrm{sym}}_Y}_{(1)=: \bar{L}^{\mathrm{sym}}} \;\underbrace{\bar{L}^{\mathrm{asym}}_Y}_{(2)=: \bar{L}^{\mathrm{asym}}}, 
\end{align}
where we separate a generic logical operator into two parts: 
\begin{enumerate}
    \item \textbf{Symmetric case.} $\bar{L}_X \bar{L}_Z$ are $X$- and $Z$-type logical operators which are not conjugate to each  other. For the canonical LP basis, $\bar{L}_X$ and $\bar{L}_Z$ do not have any physical overlap. Furthermore, $\bar{L}_X$ does not have any physical overlap with the $Z$-component of $\bar{L}^{\mathrm{sym}}_Y$ or $\bar{L}^{\mathrm{asym}}_Y$, and vice versa. $\bar{L}^{\mathrm{sym}}_Y$ is given by the product: 
	\begin{align}
        \bar{L}^{\mathrm{sym}}_Y = \prod_{m \in [l]} \prod_{i, j:\, S^{(m)}_{ij} = 1}  \bar{Z}_{i, j, m} \bar{X}_{i, j, m},
    \end{align}
	where each $S^{(m)} \in \mathbb{F}^{r_A \times r_A}_2$ is a symmetric matrix, $S^{(m)}_{ij} = 1 \iff S^{(m)}_{ji} = 1$.
    \item \textbf{Asymmetric case.} $\bar{L}^{\mathrm{asym}}_Y$ is the $Y$-type logical operator given by the product: 
		\begin{align}\label{eq: asymmetric-Y-measurements}
           \begin{aligned}
             \bar{L}^{\mathrm{asym}}_Y &= \prod_{m \in [l]}\prod_{i, j: R^{(m)}_{ij}=1}  \bar{Z}_{i,j,m}\bar{X}_{i, j, m}, \\
             &= \prod_{m \in [l]} \prod_{i, j: R^{(m)}_{ij}=1}  i\bar{Y}_{i, j, m},
           \end{aligned}
        \end{align}

		where $R^{(m)} \in \mathbb{F}^{r_A \times r_A}_2$ has zero diagonal and $R^{(m)}_{ij} = 1 \implies R^{(m)}_{ji} = 0$.
\end{enumerate}

We first extend the LP column-canonical extractor to measure $\bar{L}_X \bar{L}_Z \bar{L}^{\mathrm{sym}}_Y$ and $\bar{L}^{\mathrm{asym}}_Y$ separately and fault-tolerantly by appropriately wiring the connectivity maps. Then we show that the two connectivity maps can be added. The sum measures $\bar{L}$, and the column weight of its $X$-component reaches $2r_A - 1$. The resulting surgery gadget is still distance-preserving whenever $\partial_1$ before the bridged edges are added is $(d, \rho_d)$-sound with $\rho_d \geq r_A$, because the bridged edges supply the boundary needed for the rewired qubits.
 We now study the symmetric case. The canonical LP code $\LP_l(A, A^*)$ is symmetric. Hence, if the LP column-canonical extractor $\mathcal{X}[\mathcal{G}; (S^{(Z)}_1, S^{(Z)}_0)]$ can measure $\bar{L}_Z$, then $\mathcal{X}[\mathcal{G}; (S^{(X)}_1, S^{(X)}_0)]$, with the $X$-type static maps $S^{(X)}_1 = \Psi_1 S^{(Z)}_1$ and $S^{(X)}_0 = \Psi_0 S^{(Z)}_0$, can measure $\bar{L}_X$ by the canonical $ZX$-duality. The statement analogous to the $Z$-type masking above holds. Now for $\bar{L}^{\mathrm{sym}}$, we compress the triplet indices by $\mu = (i, j, m)$. Let the connectivity maps be given by 
		\begin{align}
            	(\Phi^{(X_\mu)}_1, \Phi^{(X_\mu)}_0); \quad (\Phi^{(Z_\mu)}_1, \Phi^{(Z_\mu)}_0)
        \end{align}
		for each logical basis ($Y$-type) operator. Note that 
		\begin{align}
            \Phi^{(X_\mu)T}_1 \Phi^{(Z_\nu)}_1 = 0 \quad \text{for } \nu \neq \mu; \qquad \Phi^{(X_\mu)T}_1 \Phi^{(Z_\mu)}_1 \neq 0,
        \end{align}
		the latter being the conjugate pair, which meets at the single qubit $q_L(i, j, m)$. Note that 
		\begin{align}
            \Phi^{(X_\mu)T}_1 \Phi^{(Z_\mu)}_1 = \Phi^{(Z_{\tau(\mu)})T}_1 \Psi_1 \Phi^{(Z_\mu)}_1,
        \end{align}
which defines a mixed-type linear extractor according to Lemma~\ref{lemma: mixed-type-linear-extractor} measuring $\bar{L}^{\mathrm{sym}}_Y$. Furthermore, since $\bar{L}_Z$ and $\bar{L}_X$ do not have conjugate pairs, they do not have physical overlaps (Proposition~\ref{prop: main-property-canonical-basis}) with each other and with the $Z$-component or $X$-component of $\bar{L}^{\mathrm{sym}}_Y$. Hence, let 
\begin{align}\label{eq: phi1-symmetric-case}
    \begin{pmatrix}
        \Phi^{(L^{\mathrm{sym}}_X)}_1 \\
        \Phi^{(L^{\mathrm{sym}}_Z)}_1
    \end{pmatrix} = \begin{pmatrix}
        M^{(L^{\mathrm{sym}}_X)}_1  & 0 \\
        0 & M^{(L^{\mathrm{sym}}_Z)}_1
    \end{pmatrix} \begin{pmatrix}
        S^{(X)}_1 \\
        S^{(Z)}_1
    \end{pmatrix}
\end{align}
and 
\begin{align}\label{eq: phi0-symmetric-case}
    \begin{aligned}
        \Phi^{(L^{\mathrm{sym}}_X)}_0 &= \Psi_0\, \Phi^{(\tau(L^{\mathrm{sym}}_X))}_0, \\
        \Phi^{(L^{\mathrm{sym}}_Z)}_0 &= \mathrm{vstack}\{ \oplus_a \mathbf{1}_{\mathcal{E}_a} M^{(L^{\mathrm{sym}}_Z)}_{0, i} \}^{r_A}_{i=1},
    \end{aligned}
\end{align}
where $\tau(\bar{L}^{\mathrm{sym}}_X) := \prod \bar{Z}_{j, i, m}$ denotes the $Z$-type dual of $\bar{L}^{\mathrm{sym}}_X = \prod \bar{X}_{i, j, m}$, with $\tau$ acting on the indices as $(i, j, m) \mapsto (j, i, m)$. Here $\Phi^{(\tau(L^{\mathrm{sym}}_X))}_0$ is its $Z$-type masked map above, with the per-copy masks $M^{(L^{\mathrm{sym}}_Z)}_{0, i}$ chosen as in the $Z$-type masking. In this way the static maps $S^{(X)}_1$ and $S^{(X)}_0$ are initialized in the same manner (see Fig.~\ref{fig: illustration-LP-extractor}(a))~\footnote{The $X$-components of Eq.~\eqref{eq: phi1-symmetric-case} and Eq.~\eqref{eq: phi0-symmetric-case} admit two view points. First is that they are given and induced by the restriction to the support of the corresponding $X$-type logical operator and $Z$-checks for which they overlap. Second, we may treat these from the canonical $ZX$-duality with the physical involution operator $(\Psi_1, \Psi_0)$ such that $ M^{(L^{\mathrm{sym}}_X)}_1$ is folded by some $Z$-logical restriction map via $\Phi_1$ and $ M^{(L^{\mathrm{sym}}_X)}_0$ folded by some $X$-checks via $\Psi_0$}. The LP column-canonical extractor with the updated $X$-component of static connectivity maps can measure any logical operator of the symmetric case and ensures the LDPC property: 

\begin{lemma}\label{lemma: LP-extractor-symmetric-case}
    Assume the canonical $ZX$-duality of Definition~\ref{def: canonical-ZX-duality}. The LP column-canonical extractor $\mathcal{X}[\mathcal{G}; (S^{(Z)}_1, S^{(Z)}_0)]$ can be upgraded to measure any logical operator of the symmetric case $\bar{L}^{\mathrm{sym}}$ with connectivity maps given by Eq.~\eqref{eq: phi1-symmetric-case} and Eq.~\eqref{eq: phi0-symmetric-case}, with soundness $\rho_d$ and added degree $\delta + r_A$, which is distance-preserving as $\rho_d \geq r_A$. 
    \begin{proof}
        As $\rho_d \geq r_A$, the resulting surgery gadget $\mathcal{S}[\mathcal{G}; (\Phi^{(L^{\mathrm{sym}})}_1, \Phi^{(L^{\mathrm{sym}})}_0)]$ measuring any $\bar{L}^{\mathrm{sym}}$ has soundness $\rho_d$, and its degrees are accounted as follows. The $Z$-components satisfy the bounds of Lemma~\ref{lemma: graphify-parity-check-sum}: each stacked copy of the static map has column weight at most $1$, so $\omega_{\mathrm{col}}(\Phi^{(L^{\mathrm{sym}}_Z)}_0) \leq r_A$ over the $r_A$ copies. The $X$-components inherit the same bounds through Eq.~\eqref{eq: phi1-symmetric-case} and Eq.~\eqref{eq: phi0-symmetric-case}, since the involutions $\Psi_1$ and $\Psi_0$ are permutations and preserve row and column weights. The commutation relations hold by Lemma~\ref{lemma: mixed-type-linear-extractor} together with the chain identity $H_Z \Phi^{(L^{\mathrm{sym}}_X)}_1 = \Psi_0 H_X \Phi^{(\tau(L^{\mathrm{sym}}_X))}_1 = \Phi^{(L^{\mathrm{sym}}_X)}_0 \partial_1$. Hence, the added degree in addition to $\delta$ is given by $\omega_{\mathrm{col}}(\Phi^{(L^{\mathrm{sym}}_X)}_0) \leq r_A$.
 Furthermore, since $\omega_{\mathrm{col}}(\Phi^{(L^{\mathrm{sym}}_X)}_1), \omega_{\mathrm{col}}(\Phi^{(L^{\mathrm{sym}}_Z)}_1) \leq r_A$ and $\mathcal{G}$ is exact, by Lemma~\ref{lemma: distance-preserving-soundness}, any such surgery gadget measuring $\bar{L}^{\mathrm{sym}}$ is distance-preserving and, hence, $\mathcal{X}[\mathcal{G}; (S^{(Z)}_1, S^{(Z)}_0)]$ is distance-preserving for measuring any logical operator of the symmetric case.
    \end{proof}
\end{lemma}

\begin{figure*}[t]
    \centering
    \includegraphics[width=\textwidth]{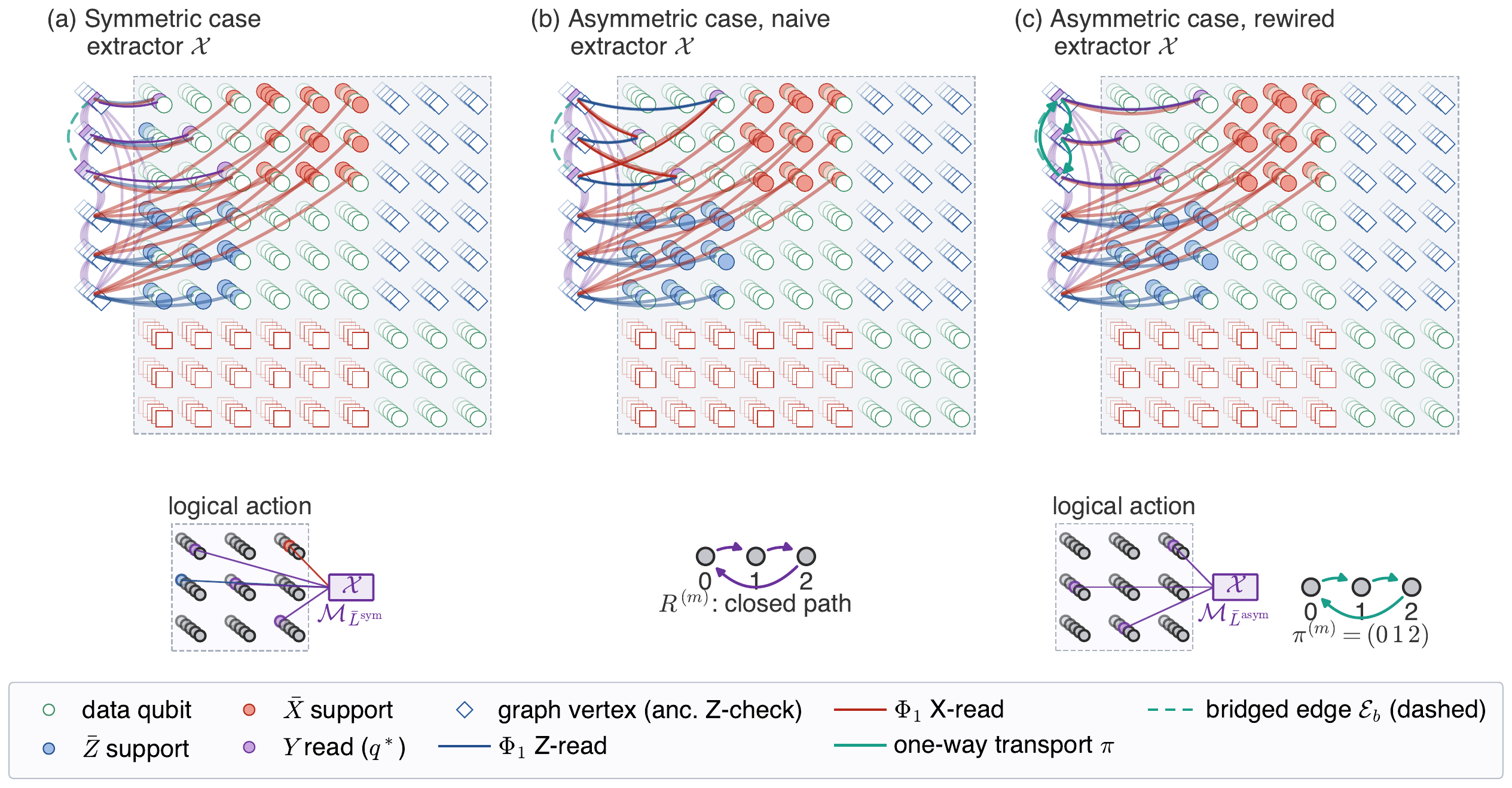}
    \caption{\textbf{Upgrading the LP column-canonical extractor to the LP canonical extractor.} The LP column-canonical extractor $\mathcal{X}[\mathcal{G}; (S^{(Z)}_1, S^{(Z)}_0)]$ of Definition~\ref{def: lp-column-canonical-extractor} for the canonical LP code $\LP^{3\times 6}_{5} = [[225, 65]]$ with $r_A = 3$. The $Z$-component (blue) and the $X$-component (red) of $\Phi_1$ are the static maps $S^{(Z)}_1$ and $S^{(X)}_1 = \Psi_1 S^{(Z)}_1$, where $\Psi_1$ is the qubit involution operator corresponding to $ZX$-duality. \textbf{(a)} Symmetric case: both components connect $q_L(i, i, m)$ to the same vertex $v_{m, i}$, so the connectivity map satisfies the commutation relations Eqs.~\eqref{eq: main-commutation_requirement} and measures $\bar{L}^{\mathrm{sym}}$ (Lemma~\ref{lemma: LP-extractor-symmetric-case}). \textbf{(b)} Asymmetric case with the static maps: $q_L(i, j, m)$ is connected to $v_{m, i}$ by the $Z$-component and to $v_{m, j}$ by the $X$-component, which violates the first condition of Eq.~\eqref{eq: main-commutation_requirement}. \textbf{(c)} The one-way permutation $\pi^{(m)}$ of the path matrix $R^{(m)}$ rewires the $X$-component along the bridged edges (teal), which restores the commutation relations and measures $\bar{L}^{\mathrm{asym}}$ (Lemma~\ref{lemma: LP-extractor-asymmetric-case}). The sum of the two maps measures $\bar{L} = \bar{L}^{\mathrm{sym}} \bar{L}^{\mathrm{asym}}$ and stays distance-preserving (Lemma~\ref{lemma: LP-extractor-bridged-soundness}, Theorem~\ref{thm: formal-LP-extractor}).}
    \label{fig: illustration-LP-extractor}
\end{figure*}

A final ingredient is to show that the LP column-canonical extractor can be upgraded to measure any logical operator of the asymmetric case $\bar{L}^{\mathrm{asym}}_Y$. The key bottleneck is that $ZX$-duality, and with it the mixed-type linear extractor, no longer applies in this case. There are genuine cases in which the $X$-component and the $Z$-component of $\Phi_1$ fail the first condition of the commutation relation in Definition~\ref{def: main-graph-surgery-gadget}. We solve this systematically using the developed generalized bridging techniques via graph isomorphism Lemma~\ref{lemma: graph-isomorphism} or one-way permutation Lemma~\ref{lemma: one-way-intertwining}. We present the LDPC and distance-preserving properties at the end for clarity. 

As in Definition~\ref{def: lp-column-canonical-extractor}, we label the graph vertices by $v_{m, i}$. Throughout, an \emph{information qubit} is the physical qubit $q_L(i, j, m)$ of an information fibre $(i, j)$, $i, j \in [r_A]$, and we set
\begin{align}\label{eq: LP-extractor-r-definition}
    r := \max(\omega_{\mathrm{col}}(\mathbb{B}(A)), \omega_{\mathrm{row}}(\mathbb{B}(A))). 
\end{align}
Then every parity check of $\mathcal{Q} = \LP_l(A, A^*)$, as represented as a row of $H_X$ or $H_Z$, acts on at most $r$ information qubits $q_L(i, j, m)$, $i, j \in [r_A]$, and every information qubit is acted on by at most $r$ parity checks of either type; Lemma~\ref{lemma: LP-extractor-asymmetric-case} is stated with this $r$. In Example~\ref{example: concrete-LP-extractor-degrees}, the two quantities that enter Lemma~\ref{lemma: LP-extractor-asymmetric-case} are computed directly. These are the number of rewired information qubits a parity check acts on, and the number of parity checks acting on an information qubit. Both are smaller than $r$, which gives the tighter degree bounds stated there.

Note that, in our previous setting, we allow directly that a given vertex can be connected to multiple columns in expression~\eqref{expr: vertex-connect-column-static}. In what follows, we temporarily restrict this setting by requiring $\omega_{\mathrm{row}}(R^{(m)}) \leq 1$ and $\omega_{\mathrm{col}}(R^{(m)}) \leq 1$ for every $m \in [l]$, so that, for $i \in [r_A], m \in [l]$, each vertex $v_{m, i}$ is connected to at most one information qubit among the factors of $\bar{L}^{\mathrm{asym}}_Y$. We now consider a natural matrix basis for $R^{(m)} \in \mathbb{F}^{r_A \times r_A}_2$. We say that $R^{(m)} \in \mathbb{F}^{r_A \times r_A}_2$ (with slight notational abuse) is a \emph{path matrix} if its nontrivial entries are traversed by consecutive row/column indices: 
\begin{align}
    R^{(m)}_{i_ki_{k-1}} = R^{(m)}_{i_{k-1}i_{k-2}} = \cdots = R^{(m)}_{i_2 i_1}=1
\end{align}
which induces an index path according to Definition~\ref{def: one-way-permutation-graph} (with orientation in mind):
\begin{align}
   ( v_{m, i_1}, v_{m, i_2}, \cdots, v_{m, i_k} ). 
\end{align}
Clearly, any generic matrix can be decomposed into a sum of path matrices with corresponding index paths. Evidently, each path matrix (whether closed or open) is one-to-one with vertices in $\mathcal{V}$ connecting to the information qubits; see a visual illustration in Fig.~\ref{fig: illustration-LP-extractor}(b)\&(c). This is a key observation which was anticipated from Lemma~\ref{lemma: bridge-measure-Y}. In this case, let the bridged edges be given by the set 
    \begin{align}\label{eq: bridged-edges-set-lp-extractor}
        \hat{\mathcal{E}}_b = \bigsqcup_{1 \leq i < j \leq r_A} \left\{ \mathrm{v}(e^b_{m, ij}) = \{v_{m, i}, v_{m, j}\}, m=1, \cdots l   \right\},
    \end{align}
    that is, the complete graph on the $r_A$ information vertices of each fibre position $m$, one orbit per pair. Every pair of rewired vertices is joined by a single bridged edge, and each information vertex is incident to $r_A - 1$ bridged edges of its fibre position.
\begin{remark}
    The bridged edges respect $R$-linearity. For $r_A = 3$ they can be represented compactly by: 
    \begin{align}
        \begin{pmatrix}
        I & I & 0 \\
        0 & I & I \\
        I & 0 & I
    \end{pmatrix},
    \end{align}
    where $\mathcal{E}_b$ denotes the edge set in this representation. 
    The bridged edges are added as new edges even when $v_{m,i}$ and $v_{m,j}$ are already adjacent in $\mathcal{G}$ (a multigraph, Definition~\ref{def: formal-graph-definition}); the bounds of Lemma~\ref{lemma: LP-extractor-asymmetric-case} count them as edges of $\mathcal{E}_b$. 
\end{remark}

By the canonical logical basis property, we only need to permute or rewire vertices within the same fibre position $m \in [l]$, since $X$-type and $Z$-type logical basis operators supported on different fibre positions never overlap (Proposition~\ref{prop: main-property-canonical-basis}). Let $\pi^{(m)}$ be the corresponding one-way permutations on these index paths (if the index paths are closed, the $\pi^{(m)}$ are simply permutations). Within one path matrix the index path visits every vertex at most once, by the restriction $\omega_{\mathrm{row}}(R^{(m)}), \omega_{\mathrm{col}}(R^{(m)}) \leq 1$; index paths of different path matrices of $R^{(m)}$ may share vertices, and the rewired $X$-components are summed over them. Hence, we can define the direct sum of the one-way permutations $\pi^{(m)}$ as
     \begin{align}
        J^{\mathcal{V}}_{\pi} = \bigoplus^{l}_{m=1} J^{\mathcal{V}}_{\pi^{(m)}}
    \end{align}
    up to a reordering of rows and columns, and $J^{\mathcal{V}}_{\pi^{(m)}} \in \mathbb{F}^{r_A \times r_A}_2$ (for $\pi^{(m)}$ a permutation, this is simply the natural representation of $\pi^{(m)}$). For simplicity, we write this action as an element of the direct sum of symmetric groups $S^{\oplus l}_{r_A}$, though it is a one-way permutation in general. By Lemma~\ref{lemma: one-way-intertwining}, there exists a corresponding $J^{\mathcal{E}}_{\pi}$ such that $J^{\mathcal{E}}_{\pi} \partial_1 = \partial_1 J^{\mathcal{V}}_{\pi}$ given by 
    \begin{align}
        (J^{\mathcal{E}}_{\pi})_{ee'} = \delta_{e' \in P_{\pi(e)}},
    \end{align}
where the path $P_{\pi(e)}$ can be represented by 
\[\begin{tikzcd}
	{\mathrm{v}_0(e)} & {\mathrm{v}_1(e)} \\
	{\pi(\mathrm{v}_0(e))} & {\pi(\mathrm{v}_1(e))}
	\arrow[no head, from=1-1, to=1-2]
	\arrow[two heads, from=1-1, to=2-1]
	\arrow[two heads, from=1-2, to=2-2]
	\arrow[no head, from=2-1, to=2-2].
\end{tikzcd}\]

Note that the above also works for the symmetric case for $Y$-type measurements, which correspond to elements in $S^{\oplus l}_{r_A}$ that are involutions. Hence, for path matrices, we can always find a corresponding $J^{\mathcal{V}}_{\pi}, J^{\mathcal{E}}_{\pi}$ which rewires the $X$-component of the connectivity map: 
\begin{align}
    \{\Phi^{(L^{\pi}_X)}_1, \Phi^{(L^{\pi}_X)}_0 \} \mapsto \{\Phi^{(L^{\pi}_X)}_1 J^{\mathcal{V}}_{\pi}, \Phi^{(L^{\pi}_X)}_0 J^{\mathcal{E}}_{\pi} \},
\end{align}
which satisfies the first condition of Eq.~\eqref{eq: main-commutation_requirement} by construction, since $H_Z \Phi^{(L^{\pi}_X)}_1 J^{\mathcal{V}}_{\pi} = \Phi^{(L^{\pi}_X)}_0 \partial_1 J^{\mathcal{V}}_{\pi} = \Phi^{(L^{\pi}_X)}_0 J^{\mathcal{E}}_{\pi} \partial_1$ by Lemma~\ref{lemma: one-way-intertwining}. Here $\bar{L}^{\pi}_X$ and $\bar{L}^{\pi}_Z$ denote the logical operators whose information qubits are supported along the path matrix $\oplus^l_{m} R^{(m)}$. Then $(\Phi^{(L^{\pi}_X)}_1 J^{\mathcal{V}}_{\pi})^T \Phi^{(L^{\pi}_Z)}_1$ is diagonal, and the commutation relations of Definition~\ref{def: main-graph-surgery-gadget} are satisfied. Returning to Eq.~\eqref{eq: asymmetric-Y-measurements}, let $R_{\pi} := \oplus^l_{m=1} R^{(m)}$, where each summand is a sum of disjoint path matrices. Applying the rewiring to every path matrix, we construct the surgery gadgets measuring an arbitrary $\bar{L}^{\mathrm{asym}}_Y$. We now bound its degrees using Lemma~\ref{lemma: graphify-parity-check-sum}: 

\begin{lemma}\label{lemma: LP-extractor-asymmetric-case}
    Let $R_\pi = \oplus^l_{m=1} R^{(m)}$ be a direct sum of path matrices, and let $\pi := \oplus^l_m \pi^{(m)}$ denote the direct sum of one-way permutations, with associated matrices $J^{\mathcal{V}}_{\pi}, J^{\mathcal{E}}_{\pi}$ defined above and in Lemma~\ref{lemma: one-way-intertwining}. Let the bridged edges be given by the set in Eq.~\eqref{eq: bridged-edges-set-lp-extractor}. Then the following bounds hold, where the sums run over the path matrices whose direct sum is $R_\pi$. 
    \begin{itemize}
        \item The column weight of $\sum_{\pi} \Phi^{(L^{\pi}_X)}_1 J^{\mathcal{V}}_\pi$: 
        \begin{align}
            \omega_{\mathrm{col}}\Big(\sum_{\pi} \Phi^{(L^{\pi}_X)}_1 J^{\mathcal{V}}_\pi\Big) \leq r_A.
        \end{align}
        \item The row weight of $\sum_{\pi} \Phi^{(L^{\pi}_X)}_0 J^{\mathcal{E}}_\pi$: 
        \begin{align}
            \omega_{\mathrm{row}}\Big(\sum_{\pi} \Phi^{(L^{\pi}_X)}_0 J^{\mathcal{E}}_\pi\Big) \leq r + \left\lfloor \omega_{\mathrm{row}}( \mathbb{B}(A))/2 \right\rfloor .
        \end{align}
        \item The column weight of $\sum_{\pi} \Phi^{(L^{\pi}_X)}_0 J^{\mathcal{E}}_\pi$:
        \begin{align}
            \omega_{\mathrm{col}}\Big(\sum_{\pi} \Phi^{(L^{\pi}_X)}_0 J^{\mathcal{E}}_\pi\Big) \leq \max(r_A,r).
        \end{align}

    \end{itemize}
    \begin{proof}
        We write $\mathcal{E}' = \bigsqcup_a \mathcal{E}_a \sqcup \mathcal{E}_b$ for the edges of the updated graph chain. Here $\mathcal{E}_a$ is the set of graphified edges of row $a$ of $A$ and $\mathcal{E}_b$ is the set of bridged edges. A deformed $Z$-check $h$ is supported on the graphification cycle of one row of $\mathbb{B}(A)$, since the $Z$-checks of $\LP_l(A, A^*)$ have the same vertex pattern as the $X$-checks by the canonical $ZX$-duality; we write $\mathcal{E}_{a(h)}$ for this cycle. Lemma~\ref{lemma: graphify-parity-check-sum} and the choice between its two solutions $u_h$ and $u_h + \mathbf{1}_{\mathcal{E}_{a(h)}}$ refer to this cycle only. All weights below are counted on the updated graph chain. The soundness $\rho_d$ of Theorem~\ref{thm: formal-LP-extractor} refers to $\partial_1$ before the bridged edges are added. Adding the bridged edges does not decrease $|\partial_1 u|$ for any $u$ and leaves $\ker \partial_1 = \{0, \mathbf{1}_{\mathcal{V}}\}$ unchanged; it enters the distance argument through Lemma~\ref{lemma: LP-extractor-bridged-soundness}. In Lemma~\ref{lemma: one-way-intertwining}, we choose the paths as follows. For each edge $e$ and each rewired endpoint $v_{m, j}$ with $\pi(v_{m, j}) = v_{m, i}$, the path $P_{\pi(e)}$ consists of $e$ itself and the single bridged edge of Eq.~\eqref{eq: bridged-edges-set-lp-extractor} joining $v_{m, j}$ and $v_{m, i}$. We also note the following fact. By the conjugate-pair property (Proposition~\ref{prop: main-property-canonical-basis}) and the restriction $\omega_{\mathrm{row}}(R^{(m)}), \omega_{\mathrm{col}}(R^{(m)}) \leq 1$, the column of $\Phi^{(L^{\pi}_X)}_1$ at a rewired vertex $v_{m, j}$ consists of the single qubit $q_L(i, j, m)$, which the rewiring moves to $v_{m, i}$.

        For the first bound. Right multiplication by $J^{\mathcal{V}}_{\pi}$ acts within each fibre position, according to Definition~\ref{def: one-way-permutation-graph}. It moves the column of $\Phi^{(L^{\pi}_X)}_1$ at a rewired vertex $v_{m, j}$, the single qubit $q_L(i, j, m)$ noted above, to $\pi(v_{m, j}) = v_{m, i}$, and leaves the columns at the other vertices unchanged. Summing over the path matrices of $R^{(m)}$, the column of $\sum_{\pi} \Phi^{(L^{\pi}_X)}_1 J^{\mathcal{V}}_{\pi}$ at $v_{m, i}$ is $\{q_L(i, j, m) : R^{(m)}_{ij} = 1\}$, of weight at most $r_A - 1$. Hence every column of the rewired sum has weight at most $r_A$.

        For the second bound. The row of a deformed $Z$-check $h$ consists of edges of $\mathcal{G}$ and bridged edges. On the edges of $\mathcal{G}$ it is the row of the unrewired sum $\sum_{\pi} \Phi^{(L^{\pi}_X)}_0$, which by Lemma~\ref{lemma: graphify-parity-check-sum}, applied to the sum, can be chosen with weight at most $\lfloor \omega_{\mathrm{row}}(\mathbb{B}(A))/2 \rfloor$. By the complete graph of Eq.~\eqref{eq: bridged-edges-set-lp-extractor}, each bridged edge is private to one rewired vertex $v_{m, j}$, and it enters the row of $h$ only if $h$ acts on the qubit $q_L(i, j, m)$ moved from $v_{m, j}$. Therefore the number of bridged edges in the row of $h$ is at most the row weight of $H_Z$ restricted to the information qubits $q_L(i, j, m)$, $j \in [r_A]$, which is at most $\omega_{\mathrm{row}}(\mathbb{B}(A)) \leq r$. This gives $\omega_{\mathrm{row}}\big(\sum_{\pi} \Phi^{(L^{\pi}_X)}_0 J^{\mathcal{E}}_\pi\big) \leq r + \lfloor \omega_{\mathrm{row}}(\mathbb{B}(A))/2 \rfloor$.

        For the third bound. The column of $\sum_{\pi} \Phi^{(L^{\pi}_X)}_0 J^{\mathcal{E}}_\pi$ at an edge counts the deformed $Z$-checks whose rows contain that edge. We compare the edges of $\mathcal{G}$ with the bridged edges. For an edge $e$ of $\mathcal{G}$, with the chosen paths $e \in P_{\pi(e')}$ only for $e' = e$, so the column at $e$ is that of the unrewired sum, which by the cycle graphification (Lemma~\ref{lemma: graphify-parity-check-sum}) and Eq.~\eqref{eq: phi0-symmetric-case} has weight at most $r_A$. For the bridged edges, by the complete graph of Eq.~\eqref{eq: bridged-edges-set-lp-extractor}, each rewired pair has its own bridged edge. This edge enters the rows of at most the $Z$-checks acting on the qubit $q_L(i, j, m)$ moved along it, since a row of $\Phi^{(L^{\pi}_X)}_0$ with support at $v_{m, j}$ belongs to a $Z$-check acting on $q_L(i, j, m)$. Its column weight is therefore at most the column weight of $H_Z$ at the information qubits, which is $\omega_{\mathrm{col}}(\mathbb{B}(A)) \leq r$ because the left block of $H_Z$ is $I_{n_A} \otimes_R A^*$. This gives $\omega_{\mathrm{col}}\big(\sum_{\pi} \Phi^{(L^{\pi}_X)}_0 J^{\mathcal{E}}_\pi\big) \leq \max(r_A, \omega_{\mathrm{col}}(\mathbb{B}(A))) \leq \max(r_A,r)$.

    \end{proof}
\end{lemma}

We now combine the symmetric case and the asymmetric case. Let $\mathcal{G}'$ be the graph chain obtained from $\mathcal{G}$ by adding the bridged edges of Eq.~\eqref{eq: bridged-edges-set-lp-extractor} as new edges, with edge-vertex incidence matrix $\partial'_1$. We add the maps of Lemma~\ref{lemma: LP-extractor-symmetric-case} for $\bar{L}^{\mathrm{sym}}$ and the rewired maps of Lemma~\ref{lemma: LP-extractor-asymmetric-case} for $\bar{L}^{\mathrm{asym}}$, and write
\begin{align}
    \Phi'^{(X)}_1 &:= \Phi^{(L^{\mathrm{sym}}_X)}_1 + \sum_{\pi} \Phi^{(L^{\pi}_X)}_1 J^{\mathcal{V}}_\pi, \nonumber \\
    \Phi'^{(X)}_0 &:= \Phi^{(L^{\mathrm{sym}}_X)}_0 + \sum_{\pi} \Phi^{(L^{\pi}_X)}_0 J^{\mathcal{E}}_\pi,
\end{align}
while the $Z$-components $\Phi^{(Z)}_1 := \Phi^{(L^{\mathrm{sym}}_Z)}_1 + \sum_{\pi} \Phi^{(L^{\pi}_Z)}_1$ and $\Phi^{(Z)}_0$ are not rewired. The sum satisfies the first two conditions of Eq.~\eqref{eq: main-commutation_requirement} by linearity. It satisfies the third condition because a qubit connected to one vertex by both the symmetric and the asymmetric maps is either a conjugate pair $\bar{X}_{i, j, m} \bar{Z}_{i, j, m}$ of the symmetric case or a diagonal qubit $q_L(i, i, m)$. Hence the surgery gadget $\mathcal{S}[\mathcal{G}'; (\Phi_1, \Phi_0)]$ with $\Phi_1 = (\Phi'^{(X)}_1, \Phi^{(Z)}_1)$ and $\Phi_0 = (\Phi'^{(X)}_0, \Phi^{(Z)}_0)$ is valid and measures $\bar{L}$. The key issue is that the column weight of $\Phi'^{(X)}_1$ can reach $2r_A - 1 = r_A + (r_A - 1)$, counting the symmetric case and the asymmetric case respectively. We now show that this does not deteriorate the soundness, and hence the distance-preserving property.
\begin{lemma}\label{lemma: LP-extractor-bridged-soundness}
    If $\partial_1$ is $(d, \rho_d)$-sound with $\rho_d \geq r_A$, then $\partial'_1$ is $(d, 1)$-sound relative to $\Phi'^{(X)}_1$ and relative to $\Phi^{(Z)}_1$. The surgery gadget $\mathcal{S}[\mathcal{G}'; (\Phi_1, \Phi_0)]$ is distance-preserving.
\end{lemma}
\begin{proof}
    We work with the binary representation of the lifted graph chains of $\mathcal{G}$ and $\mathcal{G}'$,
\begin{align}
    &\mathbb{F}_2[\hat{\mathcal{V}}] \xrightarrow{\hat{\partial}_1} \mathbb{F}_2[\hat{\mathcal{E}}] \xrightarrow{\hat{\partial}_0} \mathbb{F}_2[\hat{\mathcal{C}}] \quad \text{and} \nonumber \\
    &\mathbb{F}_2[\hat{\mathcal{V}}] \xrightarrow{\hat{\partial}'_1} \mathbb{F}_2[\hat{\mathcal{E}}'] \xrightarrow{\hat{\partial}'_0} \mathbb{F}_2[\hat{\mathcal{C}}'], \qquad \hat{\mathcal{E}}' = \hat{\mathcal{E}} \sqcup \hat{\mathcal{E}}_b,
\end{align}
where the second chain is the updated graph chain with the bridged edges of Eq.~\eqref{eq: bridged-edges-set-lp-extractor}. Let $y \in \mathbb{F}_2[\hat{\mathcal{V}}]$ with $|\hat{\partial}'_1 y| < d$. We bound the $X$-component first; the $Z$-component is not rewired and is treated at the end. We have
\begin{align}
    |\hat{\partial}'_1 y| = |\hat{\partial}_1 y| + |\hat{\partial}^b_1 y|,
\end{align}
where $\hat{\partial}^b_1$ is the restriction of $\hat{\partial}'_1$ to the rows of the bridged edges, so that $|\hat{\partial}^b_1 y|$ is the number of bridged edges with exactly one endpoint in $y$. Let $y$ be minimal, that is, $|y + \mathbf{1}_{\hat{\mathcal{V}}}| \geq |y|$. Since $|\hat{\partial}_1 y| \leq |\hat{\partial}'_1 y| < d$ and $\ker \hat{\partial}_1 = \ker \hat{\partial}'_1 = \{0, \mathbf{1}_{\hat{\mathcal{V}}}\}$, the soundness of $\hat{\partial}_1$ gives $|\hat{\partial}_1 y| \geq \rho_d |y| \geq r_A |y|$. Let $\hat{\mathcal{V}}_b \subset \hat{\mathcal{V}}$ be the set of the $r_A l$ information vertices $v_{m, i}$, $i \in [r_A]$, which are the endpoints of the bridged edges, and write $y = y_b + y_o$ with $y_b$ supported on $\hat{\mathcal{V}}_b$ and $y_o$ on $\hat{\mathcal{V}} \setminus \hat{\mathcal{V}}_b$. By construction, no two vertices connect to the same qubit, also after rewiring (Lemmas~\ref{lemma: LP-extractor-symmetric-case} and~\ref{lemma: LP-extractor-asymmetric-case}), so
\begin{align}
    |\Phi'^{(X)}_1(y)| &= |\Phi'^{(X)}_1(y_b)| + |\Phi^{(X)}_1(y_o)| \nonumber \\
    &\leq |\Phi'^{(X)}_1(y_b)| + r_A |y_o|,
\end{align}
and $|\hat{\partial}^b_1 y| = |\hat{\partial}^b_1 y_b|$. By the complete graph of Eq.~\eqref{eq: bridged-edges-set-lp-extractor}, $\hat{\partial}^b_1 = \bigoplus_{m=1}^{l} \hat{\partial}^{b, (m)}_1$, one block per fibre position $m$. Let $s_m$ be the number of information vertices of $y$ at fibre position $m$,
\begin{align}
    s_m := |\{ i \in [r_A] : y(v_{m, i}) = 1 \}|.
\end{align}
Each information vertex is incident to exactly $r_A - 1$ bridged edges, so
\begin{align}
    |\hat{\partial}^{b, (m)}_1 y| = s_m (r_A - 1) - 2\binom{s_m}{2} = s_m (r_A - s_m),
\end{align}
where the $2\binom{s_m}{2}$ overcounted edges join two vertices of $y$ at fibre position $m$ and cancel over $\mathbb{F}_2$. Similarly, a vertex $v_{m, i}$ with $i \in [r_A]$ connects to at most $r_A + (r_A - 1)$ information qubits: the qubits $q_L(j, i, m)$, $j \in [r_A]$, through the map of the symmetric case, and the qubits $q_L(i, j, m)$, $j \neq i$, moved to $v_{m, i}$ by the rewiring. Since no two vertices connect to the same qubit, the $s_m$ information vertices of $y$ at fibre position $m$ connect to at most
\begin{align}
    s_m r_A + s_m (r_A - 1) - 2\binom{s_m}{2} = s_m r_A + s_m (r_A - s_m)
\end{align}
information qubits. The $2\binom{s_m}{2}$ overcounted qubits are the $q_L(i, j, m)$ with both $v_{m, i}$ and $v_{m, j}$ in $y$, which are already counted in the column of $v_{m, j}$. Summing over $m$,
\begin{align}
    |\Phi'^{(X)}_1(y_b)| \leq r_A |y_b| + |\hat{\partial}^b_1 y|.
\end{align}
Altogether,
\begin{align}
    |\hat{\partial}'_1 y| &= |\hat{\partial}_1 y| + |\hat{\partial}^b_1 y| \nonumber \\
    &\geq r_A |y_b| + r_A |y_o| + |\hat{\partial}^b_1 y| \geq |\Phi'^{(X)}_1(y)|, \nonumber \\
    |\hat{\partial}'_1 y| &\geq r_A |y| \geq |\Phi^{(Z)}_1(y)|,
\end{align}
the second inequality because the $Z$-component is not rewired and has column weight at most $r_A$. Hence $\hat{\partial}'_1$ is $(d, 1)$-sound relative to $\Phi'^{(X)}_1$ and to $\Phi^{(Z)}_1$, and Lemma~\ref{lemma: distance-preserving-soundness} shows that $\mathcal{S}[\mathcal{G}'; (\Phi_1, \Phi_0)]$ is distance-preserving.
\end{proof}

We now restate the main theorem, Theorem~\ref{thm: main-LP-extractor}, formally:

\begin{table*}[!t]
\centering
\footnotesize
\setlength{\tabcolsep}{4pt}

\begin{tabular}{@{}l c ccc ccc@{}}
\toprule
 & Soundness & \multicolumn{3}{c}{Merged size} & \multicolumn{3}{c}{Merged degree} \\
\cmidrule(lr){2-2} \cmidrule(lr){3-5} \cmidrule(l){6-8}
 & $\rho_d$ & Qubits & $X$-checks & $Z$-checks & \shortstack{Stabilizer\\weight} & \shortstack{Qubit\\$X$-degree} & \shortstack{Qubit\\$Z$-degree} \\
\midrule
$\LP^{3\times5}_{33} = [[1122, 148, {\leq 20}]]$ & $9/2$ & $1617$\,{\scriptsize$(+495)$} & $858$\,{\scriptsize$(+363)$} & $660$\,{\scriptsize$(+165)$} & $11$\,{\scriptsize$(+3)$} & $9$\,{\scriptsize$(+4)$} & $5$\,{\scriptsize$(+0)$} \\
$\LP^{2\times3}_{36} = [[468, 36, 20]]$ & $18/7$ & $828$\,{\scriptsize$(+360)$} & {$540$\,{\scriptsize$(+324)$}} & $324$\,{\scriptsize$(+108)$} & $12$\,{\scriptsize$(+2)$} & {$8$\,{\scriptsize$(+3)$}} & $6$\,{\scriptsize$(+1)$} \\
$\LP^{3\times5}_{28} = [[952, 112, 17]]$ & $8/3$ & $1400$\,{\scriptsize$(+448)$} & {$840$\,{\scriptsize$(+420)$}} & $560$\,{\scriptsize$(+140)$} & $13$\,{\scriptsize$(+4)$} & $10$\,{\scriptsize$(+5)$} & $5$\,{\scriptsize$(+0)$} \\
$\LP^{3\times7}_{75} = [[4350, 1224, {\leq 20}]]$ & $9/2$ & $6375$\,{\scriptsize$(+2025)$} & $3300$\,{\scriptsize$(+1725)$} & $2100$\,{\scriptsize$(+525)$} & $16$\,{\scriptsize$(+6)$} & $13$\,{\scriptsize$(+6)$} & $7$\,{\scriptsize$(+0)$} \\
\bottomrule
\end{tabular}
\caption{\textbf{LP canonical extractors.} Full LP canonical extractors of Theorem~\ref{thm: formal-LP-extractor} for the four codes of Example~\ref{example: concrete-LP-extractor-degrees}, each satisfying the canonical extractor desiderata (Definition~\ref{def: canonical-extractor-desiderata}) with the soundness and added degree listed. $\rho_d$ is the certified (small-set) soundness of $\partial_1$ of $\mathcal{G}$. For $\LP^{3\times5}_{33}$ the row is the extractor of Table~\ref{tab: LP-canonical-extractor-resources}: the two information vertices of each fibre position are already adjacent through the all-ones row of $A$, and that edge serves as the bridge, so no bridged edge is added. For the other three codes the bridged edges of Eq.~\eqref{eq: bridged-edges-set-lp-extractor} are added as new edges ($|\mathcal{E}_b| = 0$, $28$, $450$, respectively) and the merged sizes count them. Parenthetical entries are increases over the data code. The merged degrees are the maxima over all logical operators $\bar{L} \in \bar{\mathcal{P}}_{k_c}$ of the stabilizer weight and of the qubit $X$- and $Z$-degrees; the added degree is $\delta = 3, 2, 4, 6$, respectively.}
\label{tab: LP-canonical-extractor-numerics}
\end{table*}

\begin{theorem}[Restatement of Theorem~\ref{thm: main-LP-extractor}]\label{thm: formal-LP-extractor}
    Let $\mathcal{X}[\mathcal{G}; (S^{(Z)}_1, S^{(Z)}_0)]$ satisfy the canonical extractor desiderata for $Z$-type logical operators (Definition~\ref{def: canonical-extractor-desiderata}) with soundness $\rho_d \geq r_A$ and added degree $\delta$. Let $\mathcal{G}'$ be the graph chain obtained by adding the bridged edges of Eq.~\eqref{eq: bridged-edges-set-lp-extractor} as new edges. Then $\mathcal{X}[\mathcal{G}'; (S^{(Z)}_1, S^{(Z)}_0)]$ satisfies the canonical extractor desiderata for every logical operator $\bar{L} \in \bar{\mathcal{P}}_{k_c}$, of any type, with soundness $\rho_d$ and added degree at most $\delta + r$, with $r$ as in Eq.~\eqref{eq: LP-extractor-r-definition}; hence, the LP canonical extractor is distance-preserving.

\end{theorem}

\begin{remark}
    Definition~\ref{def: formal-surgery-gadgets-desiderata}(ii) requires $\rho_d$ to dominate both component column weights, for every logical operator measured. This is not technically the case with rewiring, as $\Phi^{'(X)}_1$ would have column weights larger than $r_A$ owing to the increase of wiring qubits from information fibres. This does not, however, present a problem for distance preservation as carefully proven in Lemma~\ref{lemma: LP-extractor-bridged-soundness}.
\end{remark}

For specific codes, tighter bounds hold, which we state in the following example. 

\begin{example}[Degree bounds for several concrete canonical LP extractors]\label{example: concrete-LP-extractor-degrees}
    We list the concrete degree bounds for several canonical LP codes studied. The theory and the figures take the first $r_A$ columns of $A$ as the information columns; the matrices displayed below are not in this column order, and the information columns are stated in each case. 
    \begin{itemize}
        \item (LP canonical extractor for $\LP^{3 \times 5}_{33} = [[1122, 148, {\leq 20}]]$).
 Its LP column-canonical extractor satisfies the canonical extractor desiderata with certified soundness $\rho_d = 9/2 \geq 2$. In this case, we have the prototype matrix: 
        \begin{align*}
            A = \begin{pmatrix}
                1 & 1 & 1 & 1 & 1 \\
                1 & x^{14} & x^{19} & x^{11} & x^{26} \\
                1 & x^{13} & x^{2} & x^{15} & x^{21}
            \end{pmatrix},
        \end{align*}
        which is a quasi-cyclic monomial matrix. Since the first row of $A$ is all $1$, the vertices $v_{m, 3}$ and $v_{m, 4}$ are already adjacent in $\mathcal{G}$; {that edge serves as the bridged edge and no new edge is added, so $\mathcal{G}' = \mathcal{G}$ (Tables~\ref{tab: LP-canonical-extractor-resources} and \ref{tab: LP-canonical-extractor-numerics})}. Here $r_A = 2$, and each row of $\mathbb{B}(A)$ intersects precisely $2$ information qubits, so we can always use a one-way permutation instead of a transposition for the wiring. Since $A$ is monomial, the diagonal component of the $[r_A] \times [r_A] \times [l]$ grid belongs to the symmetric case. It turns out that, for this special case, it can be upgraded to a full canonical extractor with added degree $3$ rather than $4$. 

        \item (LP canonical extractor for $\LP^{3 \times 7}_{75} = [[4350, 1224, {\leq 20}]]$).  Its LP column-canonical extractor has certified soundness $\rho_d = 9/2 \geq r_A = 4$ before the bridged edges are added, and an added degree $3$, measuring arbitrary $Z$-type logical operators.  Its prototype matrix is 
        \begin{align*}
            A = \begin{pmatrix}
                1 & x^{71} & x^{73} & x^{68} & x^{33} & x^{50} & x^{47} \\
                x^{38} & x^{39} & x^{60} & x^{26} & x^{18} & x & x^{23} \\
                x^{73} & x^{6} & x^{5} & x^{42} & x^{20} & x^{22} & x^{73}
            \end{pmatrix},
        \end{align*}
        which is a quasi-cyclic monomial matrix, with information columns $0, 1, 2, 4$ of $A$ as displayed (Table~\ref{tab: LP-canonical-extractor-numerics}). Here $r_A = 4$, and each row of $\mathbb{B}(A)$ intersects precisely $4$ information qubits; likewise, the added degree can be reduced to $6$ rather than $7$. The bridged edges of Eq.~\eqref{eq: bridged-edges-set-lp-extractor} are needed, with $|\mathcal{E}_b| = \binom{r_A}{2} l = 450$ and vertex degree increase $r_A - 1 = 3$ on the information vertices. As a result, with the $450$ bridged edges added, it is a full LP canonical extractor for every $\bar{L} \in \bar{\mathcal{P}}_{k_c}$ with added degree at most $6$, distance-preserving by Theorem~\ref{thm: formal-LP-extractor}. The column weight of the rewired $X$-component reaches $2r_A - 1 = 7$, which Lemma~\ref{lemma: LP-extractor-bridged-soundness} absorbs into the boundary contributed by the bridged edges.

        \item (LP canonical extractor for $\LP^{2 \times 3}_{36} = [[468, 36, 20]]$). The prototype matrix is given in Eq.~\eqref{eq: canonical-matrix-C36}, which we restate: 
        \begin{align*}
            A = \begin{pmatrix}
                1 & 1 + x^{11} + x^{29} & 1 \\
                1 & 0 & x^{8} + x^{13} + x^{16} + x^{17}
            \end{pmatrix}.
        \end{align*}
        Here $r_A = 1$: the logical fibre grid is $1 \times 1$, every $Y$-type logical operator belongs to the symmetric case, and no wiring or bridged edges are needed. By Lemma~\ref{lemma: LP-extractor-symmetric-case}, the upgrade costs added degree $r_A = 1$, and soundness $\rho_d \geq r_A = 1$ suffices for distance preservation.

        \item (LP canonical extractor for $\LP^{3 \times 5}_{28} = [[952, 112, 17]]$). The prototype matrix is given in Eq.~\eqref{eq: saturating-C28-matrix}, which we restate: 
        \begin{align*}
            A = \begin{pmatrix}
                x^{8} + x^{11} + x^{22} & x^{18} & x & 0 & 0 \\
                0 & x^{4} & 0 & x^{14} + x^{19} & x^{18} + x^{25} \\
                0 & 0 & x^{12} & x + x^{10} & x^{3} + x^{27}
            \end{pmatrix}.
        \end{align*}
        Here $r_A = 2$, and the information columns (the last two columns of $A$ as displayed) carry entries of weight $2$. Hence each row of $\mathbb{B}(A)$ intersects at most $4$ information qubits, and each information qubit is checked by at most $4$ parity checks of either type. The diagonal component of the $[r_A] \times [r_A] \times [l]$ grid belongs to the symmetric case. Since $r_A = 2$, we then only need to rewire at most one information column, with added degree $4$
. Hence, we can build the full LP canonical extractor with added degree $\lfloor \omega_{\mathrm{row}}(\mathbb{B}(A))/2 \rfloor + 2 = 4$, using one-way permutations. Here $2$ is the number of rewired information qubits a parity check acts on (the two monomial terms of the other logical information column). Since its graph before bridging has certified soundness $\rho_d = 8/3 \geq r_A = 2$, it is distance-preserving for every $\bar{L} \in \bar{\mathcal{P}}_{k_c}$ by Theorem~\ref{thm: formal-LP-extractor}, although the rewired $X$-component has column weight up to $2r_A - 1 = 3 > 8/3$.
 
    \end{itemize}
    The detailed numerics are given in Table~\ref{tab: LP-canonical-extractor-numerics}.

\end{example}

\subsection{Parallel hypergraph surgery gadgets to canonical LP codes}\label{sec: LP-parallel-surgery-magic}

We now provide details to the high-rate, parallel surgery techniques and constructions presented in Section~\ref{sec: main-LP-parallel-logic}. First we show that, with a canonical LP code $\LP_l(A, A^*)$ such that $A$ consists of monomial entries, the LP intra-column/inter-column surgery gadgets present no additional logical operators of type~\ref{main-surgery-logical-ancilla}. 

\begin{lemma}\label{lemma: monomial-matrix-no-ancilla-logical}
Let $A \in R^{m_A \times n_A}$ be a monomial matrix such that it is full-row rank except at $b=1+x$ in the sense of Lemma~\ref{lemma: cokernel-A-odd}. Let the hypergraph be $\mathcal{G}: \mathcal{V} \xrightarrow{\mbb{B}(A)} \mathcal{E}$. Then the (inter-column) surgery gadgets constructed from $\mathcal{G}$ do not support any additional logical operators of type~\ref{main-surgery-logical-ancilla}.

   \begin{proof}
    We can consider the following commutative diagram where for the convenince, we work over the ring $A$: 
\[\begin{tikzcd}
	& { \mathcal{V}} & {\mathcal{E}} \\
	{\mathcal{Q}_2} & {\mathcal{Q}_1} & {\mathcal{Q}_0}
	\arrow["A", from=1-2, to=1-3]
	\arrow["{\Phi_1}", from=1-2, to=2-2]
	\arrow["{\Phi_0}", from=1-3, to=2-3]
	\arrow["{\partial^M_2}"', from=2-1, to=2-2]
	\arrow["{\partial^M_1}"', from=2-2, to=2-3]
\end{tikzcd}\]
where $\partial^M_2$ and $\partial^M_1$ are the (conjugate transpose) of parity-check matrices defined in Eq.~~\eqref{eq: main-LP-boundary-maps-ring}. Recall from Lemma~\ref{lemma: surgery-cone-code-logical-operators}, an $X$-type candidate logical operator of type~\ref{surgery-logical-ancilla} admits a representative that is supported exclusively on $\mathcal{E}$. In this case, let such a logical operator be $(0, L^{\mathcal{E}}_X)$, whose binary supports lie in the $\ker_R A^*$. By our assumption and the field decomposition of $A$ in Lemma~\ref{lemma: kernel-A-odd}, $\ker_R A^*$ is generated by 
\begin{align}
    \chi v^{(i)}_{\mathcal{A}_0}; \quad v^{(i)}_{\mathcal{A}_0} = e^{(i)}_{\mathcal{A}_0} + e^{(i+1)}_{\mathcal{A}_0}, \quad i \in [m_A-1],
\end{align}
where $\chi = 1 + x + \cdots + x^{l-1}$, and $e^{(i)}_{\mathcal{A}_0}$ is the unit vector. Note that in this case, we have that $\Phi_1$ and $\Phi_0$ are constructed as embedding map to $j$-th column of physical qubits and checks, respectively. In this case, we have (here we assume the code is symmetric, though the result holds generally)
\begin{align}
    H_X  = (A \otimes_R I_{n_A}, I_{m_A} \otimes_R A^*).
\end{align}
Since $A$ is a monomial matrix, $\chi A = \chi A(1)$, where $A(1)$ is the resulting matrix of $A$ evaluated at $x = 1$, which is an all-ones matrix.  Since $\Phi_0$ embeds $\mathcal{E} \cong \mathcal{A}_0$ into the $j$-th column of $X$-checks, $(0, \chi v^{(i)}_{\mathcal{A}_0})$ is a merged stabilizer if $\chi v^{(i)}_{\mathcal{A}_0}$ is the restriction to column $j$ of a redundancy $h$ of the $X$-checks, $h^{*} H_X = 0$. Take $h = \chi v^{(i)}_{\mathcal{A}_0} \otimes_R (e^{(j)}_{\mathcal{A}_1} + e^{(j')}_{\mathcal{A}_1})$ for any $j' \neq j$. In the first block of $H_X$, $\chi v^{(i)T} A = \chi\, v^{(i)}(1)^{T} A(1) = 0$ because $A(1)$ is the all-ones matrix and $v^{(i)}(1)$ has even weight; in the second block, $\chi (e^{(j)} + e^{(j')})^{T} A^* = \chi\,(A_{:,j} + A_{:,j'})(1)^{T} = 0$ because every column of $A$ evaluates to the all-ones vector at $x = 1$. Hence $h^{*} H_X = 0$, and the restriction of $h$ to column $j$ is $\chi v^{(i)}_{\mathcal{A}_0}$. This shows any such $X$-type candidate logical operator of type~\ref{surgery-logical-ancilla} must be a stabilizer of the merged code. 
   \end{proof}
\end{lemma}

We can similarly (not analogously) show that no $Z$-type logical operator of type~\ref{surgery-logical-ancilla} exists. We can state two simple corollaries to the above. 

\begin{remark}[No logical operators of type~\ref{main-surgery-logical-ancilla} under intra-column puncture and augmentation]
   For taking the intra-column puncture for $\mathbb{B}(A)$ discussed in Definition~\ref{def: main-LP-intracolumn-hyerpgraph-surgery}, the above Lemma~\ref{lemma: monomial-matrix-no-ancilla-logical} still holds provided that the puncturing does not lower the rank of $\partial_1$, so that $\ker \partial_1^{T} = \ker \mathbb{B}(A)^{T}$ and the $X$-type candidates are unchanged; this is the case for the intra-column punctures of information-fibre coordinates used in Table~{\ref{tab: high-rate-two-body-addressable}}. Similarly, for the augmentation by adding rows of weight at most $2$ supported on the information columns, the above Lemma~\ref{lemma: monomial-matrix-no-ancilla-logical} still holds whenever the added rows are mutually linearly independent over $\mathbb{F}_2$ and are linearly independent from the existing rows of $A$.
\end{remark}

However, augmentation with repeated checks (rows), which are linearly dependent, creates genuine logical operators of type~\ref{main-surgery-logical-ancilla} in the merged code. This occurs in the parallel measurements of joint logical operators reported in Table~{\ref{tab: high-rate-two-body-addressable}}. We need to add cycle checks to remove these logical operators. This way, there exists a natural choice of cycle basis given by the direct sum (assuming logically disjoint) of parity checks (of the form Eq.~\eqref{eq: parity-check-repetition}) which has column- and row-weights $2$, so that the process is well-structured and LDPC. The augmentation rows are added with multiplicity $t$ according to Eq.~\eqref{eq: bridged-edges-hypergraph}, together with the cycle checks of the length-$t$ repetition code recording their dependencies (of the form Eq.~\eqref{eq: parity-check-repetition}). Hence, the augmented $\partial_1$ then has $\lambda(\partial_1) = \min(2, t)$: a vector outside $\ker \mathbb{B}(A)$ has syndrome weight at least $2$ by Lemma~\ref{lemma: monomial-expansion-compute}, and a vector of $\ker \mathbb{B}(A)$ violating an augmented pair has syndrome weight $t$. By Theorem~\ref{thm: main-high-rate-surgery-distance-preserving}(iii), thickening length $m$ with $m \min(2, t) \geq d$ preserves the distance.

\begin{lemma}[Expansion of binarized monomial quasi-cyclic matrices]\label{lemma: monomial-expansion-compute}
    Let $l$ be odd and let $A \in R^{m_A \times n_A}$, $m_A \geq 2$, be a base matrix all of whose entries are nonzero monomials, which is full row rank except at $b = 1 + x$ in the sense of Lemma~\ref{lemma: cokernel-A-odd}. Then $\lambda(H) = 2$.

    \begin{proof}
    Write $A_{ij} = x^{a_{ij}}$ and let $u = (u_1, \ldots, u_{n_A})^T \in R^{n_A}$, so that the $i$-th block of the syndrome is $(Hu)_i = \sum_{j} x^{a_{ij}} u_j \in R$. Evaluating a polynomial $s \in R$ at $x=1$ returns its weight parity, $s(1) = |s| \bmod 2$, and every monomial evaluates to $1$; hence
    \begin{align}\label{eq: monomial-expansion-parity}
        (Hu)_i(1) = \sum_{j} u_j(1) \quad \text{for every } i \in [m_A],
    \end{align}
    i.e., the $m_A$ blocks of any syndrome carry one common weight parity. A syndrome of weight $1$ has odd weight on exactly one block and weight zero on the remaining $m_A - 1 \geq 1$ blocks, contradicting Eq.~\eqref{eq: monomial-expansion-parity}. Hence $|Hu| \geq 2$ for every $u \notin \ker H$. Conversely, by Lemma~\ref{lemma: cokernel-A-odd}(1), every $w \in R^{m_A}$ can be written as $w = \sum_{i,j} r_{ij} c_{b^{(i)}} e^{(i,j)}_{\mathcal{A}_0} + Ah$, and the first sum lies in $\IM_R A$ only if it vanishes. For $b^{(i)} \neq 1 + x$ the component $A_{(i)}$ has full row rank, so $\coker A_{(i)} = 0$ and no term occurs. For $b = 1 + x$ the idempotent is $\chi = 1 + x + \cdots + x^{l-1}$ and $c_{1+x} w = w(1)\chi$; for $w = (x^{a} + x^{b})\, e^{(i_0)}_{\mathcal{A}_0}$ we have $w(1) = 0$, so this component vanishes as well. Hence $w = Ah$, and $u := h$ satisfies $|Hu| = |w| = 2$, so that $\lambda(H) = 2$. For $m_A = 2$, $u = e^{(j)}_{\mathcal{A}_1}$ already works, since column $j$ of $A$ has weight $2$.
    \end{proof}
\end{lemma}
\subsection{Soundness criterion for hypergraph-surgery gadgets}\label{sec: soundness-criterion-for-canonical-product}

We now discuss the soundness criterion for hypergraph-surgery gadgets and its implications for distance preservation.

\begin{lemma}[Unconditional soundness for canonical matrices]\label{lemma: unconditional-soundness-canonical-matrix}
    Let $A \in R^{m_A \times n_A}$ be canonical and let $A' \in R^{m_A \times m_A}$ be its pivot block, so that $\mathbb{B}(A')$ is invertible. Then $\mathbb{B}(A')$ is $(t, \rho_t)$-sound (Definition~\ref{def: main-small-set-soundness}) with $\rho_t = 1/\gamma_A$ for every $t \geq 2$, where
	\begin{align}\label{eq: gamma-A-canonical}
        \gamma_A := \omega_{\mathrm{col}}(\mathbb{B}(A')^{-1}).
    \end{align}
\end{lemma}
\begin{proof}
    Since $\mathbb{B}(A')$ is invertible, $\ker \mathbb{B}(A') = 0$ and $\mathrm{dist}(u, \ker \mathbb{B}(A')) = |u|$ in Definition~\ref{def: main-small-set-soundness}, so $(t, \rho_t)$-soundness reads $|\mathbb{B}(A')u| \geq \rho_t\, |u|$ whenever $|\mathbb{B}(A')u| < t$. Writing $u = \mathbb{B}(A')^{-1} s$ for the syndrome $s = \mathbb{B}(A')u$, the vector $u$ is the sum of the $|s|$ columns of $\mathbb{B}(A')^{-1}$ indexed by the support of $s$, each of weight at most $\gamma_A$ of Eq.~\eqref{eq: gamma-A-canonical}; hence
    \begin{align}\label{eq: proof-soundness-chain}
        |u| = \big|\mathbb{B}(A')^{-1} s\big| \leq \gamma_A\, |s| = \gamma_A\, |\mathbb{B}(A')u|,
    \end{align}
    so $\rho_t = 1/\gamma_A$ holds for every $u$, with no restriction on the syndrome weight and hence independently of $t$. The bound is saturated already at $|s| = 1 < t$: the preimage $u = \mathbb{B}(A')^{-1} s$ of a weight-one syndrome $s$ supported on a heaviest column of $\mathbb{B}(A')^{-1}$ has $|u| = \gamma_A$, so no $\rho_t > 1/\gamma_A$ is possible for any $t \geq 2$.
\end{proof}

It turns out the classical distance is bounded by the soundness parameter $\rho_t$ of Lemma~\ref{lemma: unconditional-soundness-canonical-matrix} for canonical matrices:

\begin{lemma}[Codewords weight bound for canonical matrices]\label{lemma: codewords-weight-bound-canonical-matrix}
    Let $A$ be a canonical matrix, and write $(\gamma^{(1)}_A, \cdots, \gamma^{(m_A)}_A)$ for the column weights of $(A')^{-1}$, so that $\gamma_A = \max_k \gamma^{(k)}_A$. Then the minimal weight $d_{\mathrm{cl}}$ of codewords of $\ker \mathbb{B}(A)$ is bounded by
	\begin{align}\label{eq: codewords-column-bound}
        \begin{aligned}
	d_{\mathrm{cl}}  &\leq \min_{i \in 1, \cdots, r_A} \Big( 1 + \sum_{k=1}^{m_A} \gamma^{(k)}_A\, |A_{ki}| \Big) \\
    &\leq 1 + \gamma_A \min_{i \in 1, \cdots, r_A} |A_{:i}| \leq 1 + \gamma_A\, \omega_{\mathrm{col}}(A),
        \end{aligned}
	\end{align}
	for $A_{:i}$ the $i$-th information column of $A$ (the first $r_A$ indices) and $A_{ki}$ its $k$-th entry. Alternatively, we have that
	\begin{align}\label{eq: codewords-column-bound-inverse}
	\omega_{\mathrm{col}}(A) \geq \frac{1}{\gamma_A}(d_{\mathrm{cl}} - 1).
	\end{align}
	Furthermore, let $A'_{j:}$ be the $j$-th row vector of $A'$, with weight $|A'_{j:}| \geq 1$. We have that
	\begin{align}\label{eq: codewords-row-bound}
        \begin{aligned}
	d_{\mathrm{cl}} &\leq 1 + \frac{\gamma_A}{\underbrace{n_A -m_A}_{:=r_A}} \sum^{m_A}_{j=1}(\omega_{\mathrm{row}}(A) - |A'_{j:}|) \\
    &\leq 1 +  \frac{\gamma_A m_A}{r_A}(\omega_{\mathrm{row}}(A)-1);
        \end{aligned}
	\end{align}
    equivalently,
    \begin{align}\label{eq: codewords-row-bound-inverse}
    \omega_{\mathrm{row}}(A) \geq (d_{\mathrm{cl}}-1) \frac{r_A}{\gamma_A m_A} +1.
	\end{align}
\end{lemma}
\begin{proof}
    Up to the column permutation of Definition~\ref{def: canonical-matrix}, write $A = \big(\, A_{:1}, \cdots, A_{:r_A} \;\big|\; A' \,\big)$. For each information index $i \in 1, \cdots, r_A$, define the vector $g^{(i)}_A \in R^{n_A}$ whose information block is the $i$-th standard basis vector $e^{(i)}$ (the entry $1 \in R$ at position $i$) and whose pivot block is $(A')^{-1} A_{:i}$. Then
    \begin{align}\label{eq: proof-canonical-codeword}
        A\, g^{(i)}_A = A_{:i} + A' (A')^{-1} A_{:i} = 0
    \end{align}
    in characteristic $2$, so $\mathbb{B}(g^{(i)}_A) \in \ker \mathbb{B}(A)$ by $\mathbb{B}(A u) = \mathbb{B}(A)\mathbb{B}(u)$. Moreover, $g^{(i)}_A \neq 0$ since its information coordinate is $1$. We refer to $g^{(i)}_A$ as the \emph{canonical codeword} on the $i$-th information column, as constructed explicitly in Example~\ref{example: saturating-canonical-product-codes}. Its weight is $|g^{(i)}_A| = 1 + \sum^{m_A}_{k=1} |((A')^{-1} A_{:i})_k|$. Since $|fh| \leq |f|\,|h|$ for $f, h \in R$ (each pair of group elements contributes to a single coefficient of the product), exchanging the two sums gives
    \begin{align}\label{eq: proof-codeword-weight-chain}
        \sum^{m_A}_{k=1} |((A')^{-1} A_{:i})_k| \leq \sum^{m_A}_{k=1} \gamma^{(k)}_A\, |A_{ki}|,
    \end{align}
    where the column weight $\gamma^{(k)}_A = \sum_j |((A')^{-1})_{jk}|$ collects the $k$-th column of $(A')^{-1}$. Minimizing $1 + $ Eq.~\eqref{eq: proof-codeword-weight-chain} over $i$ gives the first inequality of Eq.~\eqref{eq: codewords-column-bound}. The second follows from $\gamma^{(k)}_A \leq \gamma_A$, and the third from $\min_i |A_{:i}| \leq \omega_{\mathrm{col}}(A)$. Eq.~\eqref{eq: codewords-column-bound-inverse} is a rearrangement of the third.

    For the row-based bound, we average instead of minimizing. The $j$-th row of $A$ splits into its information part and its pivot row $A'_{j:}$, so
    \begin{align}\label{eq: proof-row-averaging}
        \min_{i} |A_{:i}| \leq \frac{1}{r_A} \sum^{r_A}_{i=1} |A_{:i}| \leq \frac{1}{r_A} \sum^{m_A}_{j=1} \big( \omega_{\mathrm{row}}(A) - |A'_{j:}| \big),
    \end{align}
    which, inserted after the second inequality of Eq.~\eqref{eq: codewords-column-bound}, gives the first inequality of Eq.~\eqref{eq: codewords-row-bound}. The weight $|A'_{j:}| \geq 1$ holds automatically---a zero row of $A'$ would make $\mathbb{B}(A')$ singular---and yields the second inequality. Eq.~\eqref{eq: codewords-row-bound-inverse} is again a rearrangement.
\end{proof}

Hence, combined with Theorem~\ref{thm: basis-canonical-product-code}, Lemma~\ref{lemma: codewords-weight-bound-canonical-matrix} translates to bounds on the distance of the canonical-product codes.

\begin{theorem}\label{thm: no-go-soundness-distance-canonical-product-code}
    Let $\mathrm{LP}(A, B)$ be a canonical-product code, and let $d_{\mathrm{sys}}$ ($Z$-distance) and $d_{\mathrm{cosys}}$ ($X$-distance) be, respectively, its systolic and cosystolic distances. Let $\gamma_A$ be defined for $A$ as in Lemma~\ref{lemma: unconditional-soundness-canonical-matrix} and, similarly, $\gamma_{B^*}$ for $B^*$. Then we have that:
	\begin{align}\label{eq: no-go-systolic}
        \begin{aligned}
	d_{\mathrm{sys}} \leq&\ 1 + \min\Bigg( \frac{\gamma_A}{r_A} \sum^{m_A}_{j=1}\big(\omega_{\mathrm{row}}(A) - |A'_{j:}|\big), \\
    &\qquad\qquad \min_{i \in 1, \cdots, r_A} \sum_{k=1}^{m_A} \gamma^{(k)}_A\, |A_{ki}| \Bigg) \\
    \leq&\ 1+\min\left(\gamma_A\frac{m_A}{r_A}\,\omega_{\mathrm{row}}(A), \: \gamma_A\, \omega_{\mathrm{col}}(A) \right).
        \end{aligned}
	\end{align}
	Similarly, we have that
	\begin{align}\label{eq: no-go-cosystolic}
        \begin{aligned}
	d_{\mathrm{cosys}} \leq&\ 1 + \min\Bigg( \frac{\gamma_{B^*}}{r^*_B} \sum^{m^*_B}_{j=1}\big(\omega_{\mathrm{row}}(B^*) - |(B^*)'_{j:}|\big), \\
    &\qquad\qquad \min_{i \in 1, \cdots, r^*_B} \sum_{k=1}^{m^*_B} \gamma^{(k)}_{B^*}\, |B^*_{ki}| \Bigg) \\
    \leq&\ 1+\min\left(\gamma_{B^*}\frac{m^*_B}{r^*_B}\,\omega_{\mathrm{row}}(B^*), \: \gamma_{B^*}\, \omega_{\mathrm{col}}(B^*) \right),
        \end{aligned}
	\end{align}
    where $(\gamma^{(1)}_{B^*}, \cdots, \gamma^{(m^*_B)}_{B^*})$ are the column weights of $(\mathbb{B}((B^*)'))^{-1}$ with $\gamma_{B^*} = \max_k \gamma^{(k)}_{B^*}$, and $(B^*)'_{j:}$ is the $j$-th row vector of the pivot $(B^*)'$. Hence, for constant soundness and LDPC of canonical matrices, the canonical-product code $\mathrm{LP}(A, B)$ has distance $d = O(\gamma_A) =O(1)$.
\end{theorem}
\begin{proof}
    We prove the systolic bound. For the cosystolic bound Eq.~\eqref{eq: no-go-cosystolic}, simply apply the above lemmas to $B^*$, which is also canonical (Definition~\ref{def: canonical-product-code}).

    Since $\mathrm{LP}(A, B)$ is a canonical-product code, Theorem~\ref{thm: basis-canonical-product-code} provides the homology basis $\mathrm{span}\{G_A \otimes_R E_B\}$ of Eq.~\eqref{eq: basis-canonical-homology}, and the rows of $G_A$ are exactly the canonical codewords $g^{(i)}_A$ of Lemma~\ref{lemma: codewords-weight-bound-canonical-matrix}, since the completion of the information block $e^{(i)}$ to a codeword is unique by the invertibility of $\mathbb{B}(A')$. In particular, for any information indices $i \in 1, \cdots, r_A$ and $j \in 1, \cdots, r^*_B$, the $1$-chain $g^{(i)}_A \otimes_R e^{(j)}_{\mathcal{B}_0}$ represents a nonzero homology class, and its binarization has weight $|g^{(i)}_A|$ since $e^{(j)}_{\mathcal{B}_0}$ has a single entry $1 \in R$. The quantum distance is therefore upper-bounded by the classical codeword weight:
    \begin{align}\label{eq: proof-sys-embedding}
        d_{\mathrm{sys}} \leq \min_{i \in 1, \cdots, r_A} \big|\mathbb{B}\big(g^{(i)}_A \otimes_R e^{(j)}_{\mathcal{B}_0}\big)\big| = \min_{i \in 1, \cdots, r_A} |g^{(i)}_A|.
    \end{align}
    Applying the codeword weight bounds of Lemma~\ref{lemma: codewords-weight-bound-canonical-matrix}, Eqs.~\eqref{eq: codewords-column-bound} and~\eqref{eq: codewords-row-bound}, then gives Eq.~\eqref{eq: no-go-systolic}. The final coarse form uses $\sum_j (\omega_{\mathrm{row}}(A) - |A'_{j:}|) \leq m_A\, \omega_{\mathrm{row}}(A)$.

    For the last claim of the theorem, constant soundness means $\gamma_A = 1/\rho_t = O(1)$ by Lemma~\ref{lemma: unconditional-soundness-canonical-matrix}, and the LDPC property gives $\omega_{\mathrm{col}}(A) = O(1)$. Hence $d \leq 1 + \gamma_A\, \omega_{\mathrm{col}}(A) = O(\gamma_A)$.
\end{proof}

\begin{example}[Canonical-product codes saturating the distance bounds with constant soundness]\label{example: saturating-canonical-product-codes}
    It turns out the canonical-product codes constructed from prior examples $\LP_l(A, A^*)$ saturate the distance bounds of Theorem~\ref{thm: no-go-soundness-distance-canonical-product-code} with constant soundness $\rho_t = 1/\gamma_A${ $\in \{1/3, 1/4\}$} for all $t \geq 2$ (Lemma~\ref{lemma: unconditional-soundness-canonical-matrix}).
    \begin{itemize}
        \item ($\LP^{2 \times 3}_{36} = [[468, 36, 20]]$).
        Let $R = \mathbb{F}_2[x]/(x^{36}+1)$ and recall from Example~\ref{example: canonical-matrices} the canonical matrix $A = (\, A' \mid q_1 \,)$ of Eq.~\eqref{eq: canonical-matrix-C36}, with $B = A^*$ and with pivot inverse $A'^{-1}$ given in Eq.~\eqref{eq: canonical-pivot-C36}: its column weights are $(\gamma^{(1)}_A, \gamma^{(2)}_A) = (3, 4)$, so $\gamma_A = 4$, read from the second column. Here $\det(A') = 1 + x^{11} + x^{29}$, a self-inverse unit: $\det(A')^{-1} = 1 + x^{11} + x^{29}$. The bound of Eq.~\eqref{eq: codewords-column-bound} evaluates to
	\begin{align}\label{eq: saturating-C36-bound}
	1 + \gamma^{(1)}_A |q_{11}| + \gamma^{(2)}_A |q_{12}| = 1 + 3 \cdot 1 + 4 \cdot 4 = 20,
	\end{align}
    and it is saturated by the canonical codeword $\big( (A')^{-1} q_1 \mid 1 \big) =: (\, g_1 \mid g_2 \mid 1 \,)$, written out explicitly:
	\begin{align}\label{eq: saturating-C36-codeword}
        \begin{aligned}
	g_1 &= x^{8} + x^{13} + x^{16} + x^{17}, \\
    g_2 &= 1 + x + x^{6} + x^{8} + x^{9} + x^{10} + x^{11} + x^{13} \\
    &\quad + x^{16} + x^{17} + x^{19} + x^{24} + x^{27} + x^{28} + x^{29},
        \end{aligned}
	\end{align}
    of weight $|g_1| + |g_2| + 1 = 4 + 15 + 1 = 20$, with no cancellation in the products. This equals the exact code distance $d = 20$, verified via the complete low-weight-codeword search of Ref.~\cite{okada2026pairpartition}.

        \item ($\LP^{3 \times 5}_{28} = [[952, 112, 17]]$).
        Let $R = \mathbb{F}_2[x]/(x^{28}+1)$ and consider the canonical matrix, with $B = A^*$:
	\begin{align}\label{eq: saturating-C28-matrix}
        \begin{aligned}
	A &= \big(\, A' \;\big|\; q_1 \;\big|\; q_2 \,\big) \in R^{3 \times 5}, \\
    A' &= \begin{pmatrix} x^{8}{+}x^{11}{+}x^{22} & x^{18} & x^{1} \\ 0 & x^{4} & 0 \\ 0 & 0 & x^{12} \end{pmatrix}, \\
    \big(\, q_1 \;\big|\; q_2 \,\big) &= \begin{pmatrix} 0 & 0 \\ x^{14}{+}x^{19} & x^{18}{+}x^{25} \\ x^{1}{+}x^{10} & x^{3}{+}x^{27} \end{pmatrix}.
        \end{aligned}
	\end{align}
    One verifies the inverse of the (upper triangular) pivot block:
	\begin{align}\label{eq: saturating-C28-pivot-inverse}
	(A')^{-1} = \begin{pmatrix} 1{+}x^{14}{+}x^{17} & 1{+}x^{3}{+}x^{14} & x^{3}{+}x^{6}{+}x^{17} \\ 0 & x^{24} & 0 \\ 0 & 0 & x^{16} \end{pmatrix},
	\end{align}
    with column weights $(\gamma^{(1)}_A, \gamma^{(2)}_A, \gamma^{(3)}_A) = (3, 4, 4)$, so $\gamma_A = 4$, read from the second and third columns. Here $\det(A') = (x^{8}+x^{11}+x^{22})\, x^{4}\, x^{12} = x^{10} + x^{24} + x^{27}$, a unit with $\det(A')^{-1} = x + x^{12} + x^{26}$. The bound of Eq.~\eqref{eq: codewords-column-bound} evaluates, for either information column, to
	\begin{align}\label{eq: saturating-C28-bound}
        \begin{aligned}
	1 + \gamma^{(1)}_A \cdot 0 + \gamma^{(2)}_A \cdot 2 + \gamma^{(3)}_A \cdot 2 &= 1 + 4 \cdot 2 + 4 \cdot 2 \\
    &= 17 = 1 + \gamma_A \cdot 4,
        \end{aligned}
	\end{align}
    and it is saturated by the canonical codeword $\big( (A')^{-1} q_1 \mid 1 \mid 0 \big) =: (\, g_1 \mid g_2 \mid g_3 \mid 1 \mid 0 \,)$, written out explicitly:
	\begin{align}\label{eq: saturating-C28-codeword}
        \begin{aligned}
	g_1 &= 1 + x^{4} + x^{5} + x^{7} + x^{13} + x^{14} + x^{16} \\
    &\quad + x^{17} + x^{18} + x^{19} + x^{22} + x^{27}, \\
    g_2 &= x^{10} + x^{15}, \qquad g_3 = x^{17} + x^{26},
        \end{aligned}
	\end{align}
    of weight $|g_1| + |g_2| + |g_3| + 1 = 12 + 2 + 2 + 1 = 17$, with no cancellation in the products; this equals the exact code distance $d = 17$, verified via the complete low-weight-codeword search of Ref.~\cite{okada2026pairpartition}, and the codeword built on $q_2$ likewise has weight $17$. Here $n = (5 \cdot 5 + 3 \cdot 3) \cdot 28 = 952$ and $k = 2 \cdot 2 \cdot 28 = 112$.
{
        \item ($\LP^{2 \times 4}_{16} = [[320, 64, 10]]$).
        Let $R = \mathbb{F}_2[x]/(x^{16}+1)$ and recall from Example~\ref{example: canonical-matrices} the canonical matrix $A = (\, A' \mid q_1 \mid q_2 \,)$ of Eq.~\eqref{eq: canonical-matrix-C16}, with $B = A^*$ and with pivot inverse $A'^{-1}$ given in Eq.~\eqref{eq: canonical-pivot-C16}: its column weights are $(\gamma^{(1)}_A, \gamma^{(2)}_A) = (3, 3)$, so $\gamma_A = 3$, read from every column. Here $\det(A') = x^{2}$, a unit with $\det(A')^{-1} = x^{14}$. Every column of $A$ has weight $3$, and the bound of Eq.~\eqref{eq: codewords-column-bound} evaluates to
	\begin{align}\label{eq: saturating-C16-bound}
        \begin{aligned}
	1 + \gamma^{(1)}_A \cdot 2 + \gamma^{(2)}_A \cdot 1 &= 1 + 3 \cdot 2 + 3 \cdot 1 \\
    &= 10 = 1 + \gamma_A \cdot 3,
        \end{aligned}
	\end{align}
    for the information column $q_1$, and to $1 + 3 \cdot 1 + 3 \cdot 2 = 10$ for $q_2$. It is saturated by the canonical codeword $\big( (A')^{-1} q_1 \mid 1 \mid 0 \big) =: (\, g_1 \mid g_2 \mid 1 \mid 0 \,)$, written out explicitly:
	\begin{align}\label{eq: saturating-C16-codeword}
        \begin{aligned}
	g_1 &= x^{2} + x^{8} + x^{9} + x^{12}, \\
    g_2 &= 1 + x^{2} + x^{9} + x^{12} + x^{15},
        \end{aligned}
	\end{align}
    of weight $|g_1| + |g_2| + 1 = 4 + 5 + 1 = 10$, with no cancellation in the products; this equals the exact code distance $d = 10$, verified via the complete low-weight-codeword search of Ref.~\cite{okada2026pairpartition}, and the codeword built on $q_2$ likewise has weight $10$. Here $n = (4 \cdot 4 + 2 \cdot 2) \cdot 16 = 320$ and $k = 2 \cdot 2 \cdot 16 = 64$.

        \item ($\LP^{3 \times 5}_{20} = [[680, 80, 13]]$).
        Let $R = \mathbb{F}_2[x]/(x^{20}+1)$ and recall from Example~\ref{example: canonical-matrices} the canonical matrix $A = (\, A' \mid q_1 \mid q_2 \,)$ of Eq.~\eqref{eq: canonical-matrix-C20}, with $B = A^*$ and with pivot inverse $A'^{-1}$ given in Eq.~\eqref{eq: canonical-pivot-C20}: its column weights are $(\gamma^{(1)}_A, \gamma^{(2)}_A, \gamma^{(3)}_A) = (3, 3, 3)$, so $\gamma_A = 3$, read from every column. Here $\det(A') = x^{2} + x^{3} + x^{5} + x^{6} + x^{7} + x^{9} + x^{10} + x^{13} + x^{15} + x^{17} + x^{19}$, a unit with $\det(A')^{-1} = x + x^{2} + x^{6} + x^{7} + x^{11} + x^{14} + x^{17}$. The bound of Eq.~\eqref{eq: codewords-column-bound} evaluates to
	\begin{align}\label{eq: saturating-C20-bound}
        \begin{aligned}
	1 + \gamma^{(1)}_A \cdot 2 + \gamma^{(2)}_A \cdot 1 + \gamma^{(3)}_A \cdot 1 &= 1 + 3 \cdot 2 + 3 \cdot 1 + 3 \cdot 1 \\
    &= 13 = 1 + \gamma_A \cdot 4,
        \end{aligned}
	\end{align}
    for the information column $q_1$, and to $1 + 3 \cdot 1 + 3 \cdot 2 + 3 \cdot 1 = 13$ for $q_2$. It is saturated by the canonical codeword $\big( (A')^{-1} q_1 \mid 1 \mid 0 \big) =: (\, g_1 \mid g_2 \mid g_3 \mid 1 \mid 0 \,)$, written out explicitly:
	\begin{align}\label{eq: saturating-C20-codeword}
        \begin{aligned}
	g_1 &= x^{2} + x^{3} + x^{5} + x^{13} + x^{14} + x^{15}, \\
    g_2 &= x^{6} + x^{7} + x^{16}, \\
    g_3 &= x^{6} + x^{14} + x^{18},
        \end{aligned}
	\end{align}
    of weight $|g_1| + |g_2| + |g_3| + 1 = 6 + 3 + 3 + 1 = 13$, with no cancellation in the products. This equals the exact code distance $d = 13$, verified via the complete low-weight-codeword search of Ref.~\cite{okada2026pairpartition}, and the codeword built on $q_2$ likewise has weight $13$. Here $n = (5 \cdot 5 + 3 \cdot 3) \cdot 20 = 680$ and $k = 2 \cdot 2 \cdot 20 = 80$.
}
    \end{itemize}
\end{example}

\subsection{Parallel magic-state injection on canonical LP codes}\label{sec: LP-parallel-magic-injection}

We now give a detailed analysis on the parallel magic-state injection on canonical LP codes, presented in Section~\ref{sec: main-LP-parallel-logic}. As an application, we will discuss the partial transversal, high-rate magic state injection using disjoint unions of surface codes that support $|\bar{T}\rangle$. 
As shown in Figure~\ref{fig: magic-state-injection}, let $\mathcal{Q}'', \mathcal{Q}', \mathcal{Q}$ be given by $\LP_l(D, D^*)$, $\LP_l(A, D^*)$, and $\LP_l(A, B)$, where $D$ is the parity check for the length-$d_s$ repetition code given by Eq.~\eqref{eq: parity-check-repetition} and $D \in R^{d_s-1 \times d_s}$ ,$A \in R^{m_A \times n_A}$, $B \in R^{m_B \times n_B}$ are assumed throughout.

 For simplicity, we assume that the data LP canonical code $\LP_l(A, B)$ has purely canonical logical operators: $k_c = k$ so that $\mathbb{B}(A)$ has full row-rank. For our application, this is a harmless assumption since we can always treat the residual logical operators as gauge logical operators. We can have a systematic understanding of the code parameters of $\mathcal{Q}'$ and $\mathcal{Q}''$: 

\begin{lemma}[Code parameters of $\mathcal{Q}^{'}$ and $\mathcal{Q}^{''}$]\label{lemma: main-code-parameter-Q'-Q''}
    $\mathcal{Q}''= \LP_l(D, D^*)$ and  $\mathcal{Q}'= \LP_l(A, D^*)$ admit the code parameters, respectively, $[[(2d_s^2-2d_s+1)l, l, d_s]]$ and $[[ l(n_Ad_s + m_A(d_s-1)), lr_A, d_X=d_s, d_Z=d_A]]$, where
    \begin{align}
        d_A:=\min_{0\neq u\in\ker_{\mathbb{F}_2}\mathbb{B}(A)}|u|.
    \end{align}
    Furthermore, $\mathcal{Q}'$ and $\mathcal{Q}''$ are canonical LP codes admitting, respectively, minimal-weight canonical logical basis sets. 
    \begin{proof}
        The parameters of $\mathcal{Q}''$ follow by identifying it with a stack of $l$ surface codes of distance $d_s$. We now prove the parameters of $\mathcal{Q}'$. Since $\LP_l(A,B)$ is assumed to be canonical, $\LP_l(A,D^*)$ admits the following $X$-type and $Z$-type logical basis supports over $R$:
        \begin{align}\label{eq: transistor-code-basis}
            \supp \bar{X}^{\mathcal{Q}'}_{im}
                &=x^m e^{(i)}_{\mathcal{A}_1}\otimes_R\mathbf{1}_{\mathcal{D}_1}, \\
            \supp \bar{Z}^{\mathcal{Q}'}_{im}
                &=x^m u^{(i)}_{\mathcal{A}_1}\otimes_R e^{(0)}_{\mathcal{D}_1}.
        \end{align}
        Here $i\in[r_A]$, $m\in[l]$, $\mathbf{1}_{\mathcal{D}_1}$ is the all-one vector, and $u^{(i)}_{\mathcal{A}_1}\in\ker_R A$ has a unit entry on the $i$-th information index, as in Definition~\ref{def: main-canonical-basis-LP}. Here and below, the weight of an $R$-valued vector means the binary Hamming weight after applying $\mathbb{B}$.

        We first prove that $d_X=d_s$. Since the displayed $X$-type operators form a basis, every nontrivial $X$-type logical operator has a representative $a_{\mathcal{A}_1}\otimes_R\mathbf{1}_{\mathcal{D}_1}$ for some $a_{\mathcal{A}_1}\notin\IM_R A^*$. An arbitrary $X$-type stabilizer has support $(\partial^M_1)^* h$, where

        \begin{align}
            h=\sum_{j=0}^{d_s-1}h_j\otimes_R e^{(j)}_{\mathcal{D}_1},
            \qquad h_j\in\mathcal{A}_0.
        \end{align}
        After adding this stabilizer, the component on the $j$-th index of $\mathcal{D}_1$ is
        \begin{align}
            a_{\mathcal{A}_1}+A^*h_j.
        \end{align}
        This component cannot vanish: otherwise, $a_{\mathcal{A}_1}$ would lie in $\IM_R A^*$ and the logical operator would be trivial. Thus, each of the $d_s$ components has binary weight at least one. The other component of $(\partial^M_1)^*h$, supported on $\mathcal{A}_0\otimes_R\mathcal{D}_0$, can only add weight. Therefore, every nontrivial $X$-type logical operator has weight at least $d_s$. The displayed representative $\supp\bar{X}^{\mathcal{Q}'}_{im}$ has weight exactly $d_s$, since it contains one monomial on each of the $d_s$ indices of $\mathcal{D}_1$. Hence, $d_X=d_s$.

        We next prove that $d_Z=d_A$. Let $0\neq u\in\ker_R A$. A corresponding $Z$-type logical operator has representative $u\otimes_R e^{(0)}_{\mathcal{D}_1}$. An arbitrary $Z$-type stabilizer has support $\partial^M_2 h$, where
        \begin{align}
            h=\sum_{j=0}^{d_s-2}h_j\otimes_R e^{(j)}_{\mathcal{D}_0},
            \qquad h_j\in\mathcal{A}_1.
        \end{align}
        After adding this stabilizer, denote the components on the $d_s$ indices of $\mathcal{D}_1$ by $z_0,\ldots,z_{d_s-1}$. Since the first component of $\partial^M_2h$ is $(I_{\mathcal{A}_1}\otimes_R D^*)h$, these components satisfy
        \begin{align}
            z_0&=u+h_0, \\
            z_j&=h_{j-1}+h_j \quad (1\leq j\leq d_s-2), \\
            z_{d_s-1}&=h_{d_s-2}.
        \end{align}
        Every $h_j$ appears twice in their sum and therefore cancels over $\mathbb{F}_2$, giving $\sum_{j=0}^{d_s-1}z_j=u$. The triangle inequality for binary Hamming weight then gives
        \begin{align}
            \sum_{j=0}^{d_s-1}|z_j|\geq\left|\sum_{j=0}^{d_s-1}z_j\right|=|u|\geq d_A.
        \end{align}
        The other component of $\partial^M_2h$, supported on $\mathcal{A}_0\otimes_R\mathcal{D}_0$, can only add weight. Hence, every nontrivial $Z$-type logical operator has weight at least $d_A$. Taking $u$ to be a minimum-weight nonzero vector in $\ker_R A$, the representative $u\otimes_R e^{(0)}_{\mathcal{D}_1}$ has weight exactly $d_A$. Therefore, $d_Z=d_A$.
    \end{proof}
\end{lemma}

Hence, in what follows, we denote $\mathcal{Q}''= \LP_l(D, D^*) = [[(2d_s^2-2d_s+1)l, l, d_s]]$, $\mathcal{Q}'= \LP_l(A, D^*) = [[ l(n_Ad_s + m_A(d_s-1)), lr_A, d_X=d_s, d_Z=d_A]]$, and $\mathcal{Q} = \LP_l(A, B) =  [[ l(n_A m_B + m_A n_B), lr_Ar^*_B, d]]$, where $k:= lr_Ar^*_B$. For the magic state injection, we wish to teleport in a fully parallel fashion, which motivates the definition of partial transversal surgery, with which we build two separate gadgets.

\begin{enumerate}
    \item \label{magic-state-gadgets-XX} ($M_{\bar{X}\bar{X}}$ inter-column surgery gadget between $\mathcal{Q}''$ and $\mathcal{Q}'$).  Let $\mathcal{G}^{(i)}_X(\mathcal{V}^{(i)}_X, \mathcal{E}^{(i)}_X, \mathcal{C}^{(i)}_X)$ be constructed such that $\mathcal{V}^{(i)}_X \cong \mathcal{D}_1$, $\mathcal{E}_X \cong \mathcal{D}_0$, $\partial_1 = D$, and $\mathcal{C}_X=0$, $\partial_0=0$. For any row index $i \in [r_A]$, let 
    \begin{align}
        \Phi^{(i)}_1 = \mathbb{B}\left( \begin{array}{cc}
             e^{(0)}_{\mathcal{D}_1} \otimes_R I_{\mathcal{D}_1} \\
             e^{(i)}_{\mathcal{A}_1} \otimes_R I_{\mathcal{D}_1} 
        \end{array} \right); \quad \Phi^{(i)}_0 = \mathbb{B}\left( \begin{array}{cc}
             e^{(0)}_{\mathcal{D}_1} \otimes_R I_{\mathcal{D}_0} \\
             e^{(i)}_{\mathcal{A}_1} \otimes_R I_{\mathcal{D}_0} 
        \end{array} \right),
    \end{align}

    Let $\mathcal{G}_X := \oplus_{i\in [r_A]}\mathcal{G}^{(i)}_X$ and $\Phi_1 := \mathrm{hstack}(\Phi^{(0)}_1, \cdots, \Phi^{(r_A-1)}_1)$ and $\Phi_0 := \mathrm{hstack}(\Phi^{(0)}_0, \cdots, \Phi^{(r_A-1)}_0)$. Then we denote the surface-to-transistor \emph{row-transversal} surgery gadget by $\mathcal{S}[\mathcal{G}_X, (\Phi^{(X)}_1, \Phi^{(X)}_0)]_{\mathcal{Q}'' \rightarrow \mathcal{Q}'}$.

    \item \label{magic-state-gadgets-ZZ} ($M_{\bar{Z}\bar{Z}}$ inter-column  gadget between $\mathcal{Q}'$ and $\mathcal{Q}$).  Let $\mathcal{G}_Z(\mathcal{V}_Z, \mathcal{E}_Z, \mathcal{C}_Z)$ be constructed such that $\mathcal{V}_Z \cong \mathcal{A}_1$, $\mathcal{E}_Z \cong \mathcal{A}_0$, $\partial_1 = A$, and $\mathcal{C}_Z=0$, $\partial_0=0$. For any column index $j \in [r_B^*]$, let 
    \begin{align}
        \Phi^{(Z)}_1 = \mathbb{B}\left( \begin{array}{cc}
            I_{\mathcal{A}_1} \otimes_R e^{(0)}_{\mathcal{D}_1} \\
	 I_{\mathcal{A}_1} \otimes_R e^{(j)}_{\mathcal{B}_0}
        \end{array} \right); \quad \Phi^{(Z)}_0 = \mathbb{B}\left( \begin{array}{cc}
            I_{\mathcal{A}_0} \otimes_R e^{(0)}_{\mathcal{D}_1} \\
	I_{\mathcal{A}_0} \otimes_R e^{(j)}_{\mathcal{B}_0}
        \end{array} \right),
    \end{align}

    where $e^{(j)}_{\mathcal{B}_0}$ denotes the unit vector on $\mathcal{B}_0$ supported on the $j$th entry. We denote the transistor-to-data \emph{column-transversal} surgery gadget by $\mathcal{S}[\mathcal{G}_Z, (\Phi^{(Z)}_1, \Phi^{(Z)}_0)]_{\mathcal{Q}' \rightarrow \mathcal{Q}}$.
\end{enumerate}

Furthermore, by Lemma~\ref{lemma: monomial-matrix-no-ancilla-logical}, the row-transversal and column-transversal surgery gadgets do not encode any additional logical operators of type~\ref{surgery-logical-ancilla} for the respective merged codes. Notably, these two surgery gadgets are distance-preserving. 
  
\begin{lemma}
    The merged-code distance for the fibre-transversal surgery gadget and, consequently, the row-transversal surgery gadget $\mathcal{S}[\mathcal{G}_X, (\Phi^{(X)}_1, \Phi^{(X)}_0)]_{\mathcal{Q}'' \rightarrow \mathcal{Q}'}$ is at least $\min(d_s, d_A)$. The merged-code distance for the column-transversal surgery gadget $\mathcal{S}[\mathcal{G}_Z, (\Phi^{(Z)}_1, \Phi^{(Z)}_0)]_{\mathcal{Q}' \rightarrow \mathcal{Q}}$ is at least $\min(d_s, d_A, d) = \min(d_s, d)$, where $d$ is the distance of the data canonical LP code.
    \begin{proof}
    We now sketch the proof, which is a simple observation for transversal logical measurements along columns/rows. For the row-transversal surgery, it suffices to only prove for the $X$-distance. Denote the measured logical operators by the set: 
    \begin{align}\label{expr: row-parallel-measured-set}
     \left\{ \bar{X}^{\mathcal{Q}''}_{im} \bar{X}^{\mathcal{Q}'}_{im} \right\}_{i \in [r_A],m \in [l]},
    \end{align}

    where we denote the supports of $\bar{X}^{\mathcal{Q}''}_{im}$ by the $i$-th copy of $\mathcal{Q}''$:
    \begin{align}
        \supp \bar{X}^{\mathcal{Q}''}_{im} = x^m e^{(0)}_{\mathcal{D}_1} \otimes_R \mathbf{1}_{\mathcal{D}_1},
    \end{align}

    with the supports of $\bar{X}^{\mathcal{Q}'}_{im}$ given by Eq.~\eqref{eq: transistor-code-basis}. Since $\Phi^{(i)}_1$ is transversal, logical operators in the prescribed canonical basis that avoid the checked rows preserve the distance. Furthermore, any $X$-type logical operators supported entirely on either code patch preserve the distance. For an $X$-type logical operator supported on both code patches and intersecting the checked rows, multiply by stabilizers from the measured set~\eqref{expr: row-parallel-measured-set}. The resulting representative either avoids the checked rows or is supported on only one code patch. This concludes the proof, which extends analogously to preserve $Z$-distance for the row-transversal surgery gadget.
    \end{proof}

\end{lemma}

The magic-state injection consists only of transversal single-qubit measurements and the two joint measurement gadgets above. If we measure $\min(d, d_s)$ rounds of syndrome extractions per logical cycle, we arrive at Theorem~\ref{thm: main-magic-injection-distance-preserving} in Section~\ref{sec: main-LP-parallel-logic}.

\bibliography{references}

\end{document}